\documentclass[aps,prd,reprint,superscriptaddress,nofootinbib]{revtex4-2}
\usepackage[T1]{fontenc}
\usepackage{lmodern}
\usepackage{graphicx}
\usepackage[hidelinks]{hyperref}
\usepackage{amsmath,amssymb,amsthm,amsfonts,mathrsfs,bm}
\usepackage{booktabs,array}
\usepackage{tabularray}
\UseTblrLibrary{booktabs}
\SetTblrTemplate{contfoot}{empty}
\SetTblrTemplate{conthead}{empty}
\usepackage{capt-of}

\allowdisplaybreaks
\theoremstyle{definition}
\newtheorem{definition}{Definition}[section]

\theoremstyle{plain}

\newtheorem{proposition}[definition]{Proposition}

\makeatletter
\let\l@subsubsection\@gobbletwo
\def\appendix@toc{
  \let\l@subsection\@gobbletwo
  \let\l@subsubsection\@gobbletwo
}
\makeatother

\begin{document}

\preprint{CHIBA-EP-280v2}
\title{Understanding Color Confinement \\through Quantum Reference Frames and Relational Observables}
\author{Kei-Ichi Kondo}
\email{kondok@faculty.chiba-u.jp}
\affiliation{Department of Physics, Graduate School of Science, Chiba University, 1-33 Yayoi-cho, Chiba, Chiba 263-8522, Japan}
\affiliation{Research and Education Center for Natural Sciences, Keio University, 4-1-1 Hiyoshi, Yokohama, Kanagawa 223-8521, Japan}
\date{October 5, 2026}

\begin{abstract}
We present a formulation for color confinement on the basis of quantum reference frames (QRFs) and relational observables.
In the QRF approach to color confinement, colored quantities are not defined as isolated local fields, but rather as relational observables with respect to a color frame or a dressing field.
The Gauss law removes compactly supported local color charge from the physical bulk algebra,
while boundary fluxes, Wilson lines, and holonomies may remain as semi-local data.
This separates gauge invariance from frame independence.
Color confinement is characterized under the fixed protocol
by the absence of a finite-energy, physically accessible, asymptotically isolated nontrivial color-charged representation in the infinite volume limit.
A failure to globalize a long-distance color frame is one possible mechanism contributing to this absence, which clarifies the role played by topological defects in color confinement.

As concrete examples, we examine $(1+1)$-dimensional $SU(2)$ Yang--Mills theory and $(1+1)$-dimensional $U(1)$ gauge--Higgs model.
We also take up the two-dimensional U(1) gauge--Higgs model on $\mathbb{H}^2$ (Euclidean $AdS_2$) and three-dimensional $SU(2)$ gauge--Higgs model on $\mathbb{H}^3$ (Euclidean $AdS_3$) which are obtained based on conformal equivalence by dimensional reduction of four-dimensional SU(2) Yang--Mills theory restricted to symmetric-instanton sectors.
Through explicit calculations in these examples and in controlled sectors,
we provide nontrivial checks for internal consistency of the present formulation.
We also discuss prospects for four-dimensional Yang--Mills theory and gauge--Higgs theories.
The lower-dimensional models and symmetric-instanton sectors provide controlled consistency checks, whereas the proposed four-dimensional relation between relational-frame disorder and charged-sector exclusion is presented as a falsifiable dynamical hypothesis rather than a derivation of confinement.
Other confinement criteria--- including the Wilson-loop area law, generalized-symmetry such as center one-form symmetry, Kugo--Ojima, and residual-gauge-symmetry restoration ---have different assumptions and observables; rather than being identified as equivalent manifestations of one structure, they can be compared and partially organized within a common relational language.
\end{abstract}

\maketitle

\section{Introduction}

The confinement problem in Yang--Mills theory contains several logically distinct questions.  The most familiar one is \textbf{quark confinement}: in four-dimensional pure $SU(N)$ Yang--Mills theory, when color sources are introduced as \textbf{external probes}, the expectation value of a large \textbf{Wilson loop} in a representation $R$ with nonzero $N$-ality $q_R\neq0$ is expected to obey an \textbf{area law} \cite{Wilson:1974sk},

\begin{equation}
\begin{aligned}
  \langle W_R(C)\rangle
  &:=\left\langle \frac{1}{d_R}\mathrm{tr}_R\,\mathrm{P}
  \exp\left(ig_{\rm YM}\oint_C dx^\mu \mathscr{A}_\mu(x)\right)\right\rangle \\
  &\simeq \exp[-\sigma_R {\rm Area}(C)+\cdots] .
\end{aligned}
 \label{eq:WilsonAreaIntro}
\end{equation}
Here $\sigma_R$ denotes the string tension associated with the external representation $R$, and the ellipsis denotes subleading perimeter, curvature, and constant terms.

\begin{figure}[htbp]
  \centering
\includegraphics[width=.30\linewidth]{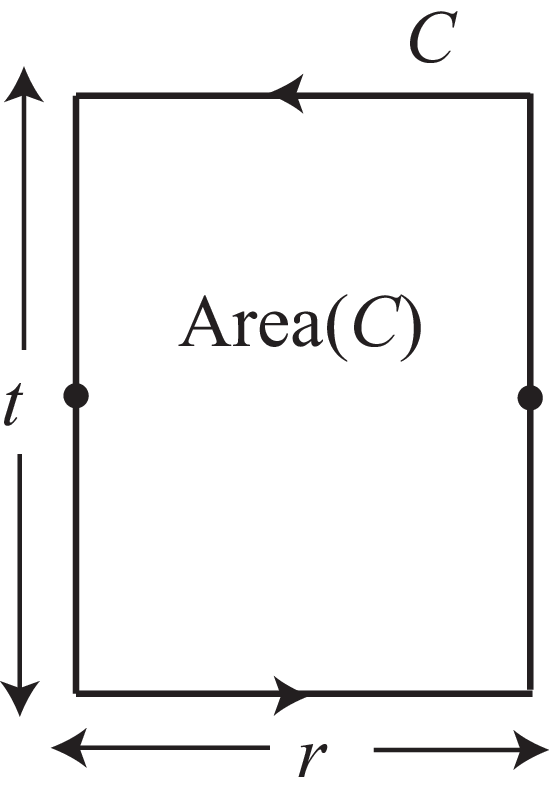}
  \caption{A rectangular Wilson loop $C$ used as a probe of a pair of external color sources.}
\label{fig:RectangularWilsonLoopIntro}
\end{figure}

For a rectangular loop $C=(r,t)$ of spatial size $r$ and Euclidean temporal size $t$, the area law \eqref{eq:WilsonAreaIntro} gives the static potential of the linear type as the leading term for the area ${\rm Area}(C)=rt$ surrounded by the loop $C$ (see Fig.~\ref{fig:RectangularWilsonLoopIntro}):
\begin{equation}
\begin{aligned}
  V_R(r)
  :=&-\lim_{t\to\infty}\frac{1}{t}\log \langle W_R(r,t)\rangle \\
  \simeq& \sigma_R r+\cdots \longrightarrow \infty
  \ (r\to\infty),
\end{aligned}
 \label{eq:LinearPotentialIntro}
\end{equation}
which expresses the energetic cost of separating a heavy source--antisource pair.
This is a precise and powerful diagnostic for the \textbf{confinement of external probes}.
\footnote{
In pure $SU(N)$ Yang--Mills theory
external sources with zero $N$-ality can be screened by gluons and hence
the asymptotic string tension is determined by the stable $N$-ality class $q_R\in\mathbb Z_N$ after asymptotic screening by adjoint gluons.
The representation-dependent Casimir scaling, when present, is an intermediate-distance behavior rather than the asymptotic classification \cite{Greensite:2003bk}.
}
It is, however, not identical to the full statement of color confinement.

By \textbf{color confinement} we mean the stronger statement that \textbf{a nontrivial unscreened color-charge class does not occur as a finite-energy, physically accessible, and asymptotically isolated sector of the theory.}
Equivalently, \textbf{isolated colored particles with a nonvanishing unscreened color charge, whether elementary or composite, should not appear as physical asymptotic states.}
This statement requires care.
The \textbf{physical Hilbert space} is obtained by imposing \textbf{the Gauss law} for \textbf{the proper gauge group}
\footnote{
The terms \textbf{proper gauge group} and \textbf{genuine boundary/asymptotic symmetry} are defined before their first technical use in Sec.~IIA, where the differentiable-generator qualification is also stated \cite{ReggeTeitelboim1974,HenneauxTeitelboim1992,Riello2021}.
}
 and the chosen boundary conditions; it need not be a singlet under genuine boundary or asymptotic symmetries \cite{Donnelly:2016auv,Riello2021}.
A local quark field $\psi^a(x)$ or a local gluon field $\mathscr{A}^{A}_{\mu}(x)$ is not by itself a gauge-invariant observable, therefore the absence of such bare colored fields from the gauge-invariant local algebra is almost tautological.
\textbf{The nontrivial question is whether a charge label can nevertheless survive after all admissible gauge-field dressings,
Wilson lines, matter reference systems, boundary frames, screening processes, defect-sector sums, and infrared limits have been taken into account.}
\footnote{The term \textbf{admissible}, as applied to frames, dressings, states, and relational realizations, is defined in Definition~\ref{def:admissible-qrf} in Sec.~\ref{subsec:gauss-boundary-relational}.
}
In other words, \textbf{color confinement is a question about whether a color label survives as a relational, gauge-invariant, and physically meaningful label at long distance.}
This distinction is the central motivation of the present paper.
\textbf{We formulate color confinement in terms of quantum reference frames (QRFs) and relational observables.}

\textbf{A quantum reference frame is a physical quantum system relative to which states, observables, and dynamical laws are specified.}
The modern QRF literature develops operational, perspective-neutral, and group-theoretic formulations of relational observables and frame changes \cite{AharonovKaufherr1984,BartlettRudolphSpekkens2007,GiacominiCastroRuizBrukner2019,VanrietveldeEtAl2020,DeLaHametteGalley2020}.  Complementary foundational, symmetry-based, and frame-transformation viewpoints may be found in Refs.~\cite{Rovelli1991,AngeloEtAl2011,LoveridgeMiyaderaBusch2018,
KrummHoehnMueller2021}.

These ideas are especially natural in gauge theories, where gauge-invariant observables can be expressed relative to dynamical frames and where edge modes, Wilson-line frames, or boundary dressings may provide reference data for subregions \cite{CarrozzaHoehn2022,AraujoRegadoHoehnSartini2025,FewsterJanssenRejzner2025}.  Recent developments toward relational quantum field theory and regional algebras provide complementary field-theoretic context \cite{Glowacki2024,FewsterEtAl2025}.
Related works on subsystem relativity, boundary QRFs, accessibility of global properties from internal perspectives, and lattice-gauge encodings provide further context for the operational questions considered here \cite{Ahmad2022,KabelBruknerWieland2023,HametteKabelBrukner2025,Lacambra2026}.

Global gauge-fixing sections or globally integrable dressings can be obstructed by the topology of gauge-orbit space and by Gribov-type ambiguities \cite{Gribov:1977wm,Singer1978,Zwanziger:1989mf,GomesRiello:2016}; analogous obstructions to global relational perspectives also occur in finite-dimensional constrained QRF systems \cite{VanrietveldeHoehnGiacomini2023}.
This does not imply that all global relational geometry fails to exist: field-space formulations with non-flat connection can remain meaningful even when a flat global gauge fixing is obstructed.

This semi-local viewpoint is natural for Yang--Mills theory because Wilson lines and dressings, boundary electric fluxes, and holonomies are neither pointlike local observables nor merely global numbers: they are associated with paths, surfaces, boundaries, or asymptotic regions \cite{Gomes:2018shn,CarrozzaHoehn2022,AraujoRegadoHoehnSartini2025,
FewsterJanssenRejzner2025}.

An important geometric and conceptual precursor to this gauge-theoretic reference-frame viewpoint is the field-space framework of Gomes and Riello \cite{GomesRiello:2016}.
There, the configuration space of fields and its gauge orbits are organized using infinite-dimensional principal-bundle geometry, and a connection one-form on field space defines horizontal and vertical directions. In this language the BRST structure acquires a geometric interpretation, while the Gribov problem motivates allowing a generally non-flat field-space connection in a global treatment.
The same framework was applied to bounded regions and given an observer or reference-frame interpretation.
The later work of Gomes, Hopfm\"uller and Riello \cite{Gomes:2018shn}  developed the boundary-charge and dressing aspects further.
In particular, it emphasized that the original gauge and matter fields can themselves supply relevant dressing or reference-frame data and that additional group-valued edge modes need not be postulated as new independent physical degrees of freedom.
These results provide an important conceptual bridge among BRST and Gribov issues, boundaries, dressings, and relational frames.
Following this field-space and dressing perspective, the existence of relational frame data does not by itself imply the introduction of a new independent edge degree of freedom.
This is one reason why \textbf{the present confinement criterion is formulated in terms of admissible and physically accessible frames
rather than the mere existence of a formal dressing.}

The elementary relational idea is illustrated by a translation-invariant two-particle system.
If particles $A$ and $B$ have positions $x_A$ and $x_B$ on a line, the simultaneous shift
\begin{equation}
 x_A\mapsto x_A+a,
 \qquad
 x_B\mapsto x_B+a
 \label{eq:TranslationRedundancyIntro}
\end{equation}
may be regarded as a redundancy of the absolute origin. Neither $x_A$ nor $x_B$ is invariant, whereas the difference
\begin{equation}
 X_{B|A}:=x_B-x_A
 \label{eq:RelationalPositionIntro}
\end{equation}
is invariant and represents the position of $B$ relative to the frame $A$. The gauge choice $x_A=0$ does not create a new observable; it merely chooses a representative of the relational observable.
Figure~\ref{fig:TwoParticleRelationalIntro} illustrates this point.

\begin{figure}[htbp]
\centering
  \centering
\includegraphics[width=.80\linewidth]{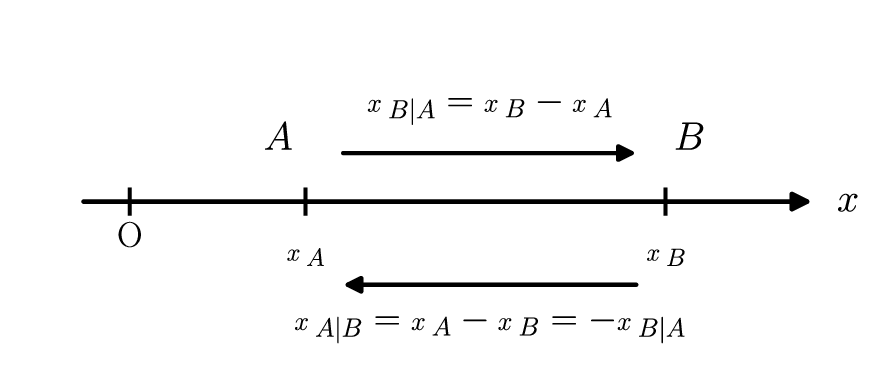}
 \caption{Two particles $A$ and $B$ on a line.  The quantities $x_{B|A}=x_B-x_A$ and $x_{A|B}=x_A-x_B=-x_{B|A}$ are relational observables with respect to the frames $A$ and $B$, respectively.  The choice $x_A=0$ corresponds to a gauge fixing.}
\label{fig:TwoParticleRelationalIntro}
\end{figure}

The same logic applies to color. A statement such as ``the quark is red'' presupposes a color basis.
A meaningful colored statement must specify how the system is oriented relative to a color frame.
A \textbf{frame field or dressing} $h$ may convert a gauge-covariant field into a relational expression, for example $h^{-1}\psi$ or $h^{-1}\mathscr{F}_{\mu\nu}h$.
Such an expression can be invariant under the original left action of the proper gauge group while remaining covariant under a right action that reorients the frame. Consequently,
\begin{equation}
 \text{gauge invariance}\ \not\Rightarrow\ \text{frame independence}.
 \nonumber
 \label{eq:GaugeVsFrameIntro}
\end{equation}
If a candidate color frame can be defined only locally, or if transport around a defect produces nontrivial monodromy or holonomy, the corresponding color label may become patch- or path-dependent. This is an \textbf{obstruction to a global color labeling}, but it is not by itself a proof of confinement.
It becomes confinement-relevant only when, after minimization over admissible relational realizations and the appropriate infrared and defect-sector prescriptions,
\textbf{the nontrivial charge class fails to survive as a finite-energy, physically accessible, asymptotically isolated sector.}

The dressing-field terminology used here is closely related to the Dressing Field Method (DFM), which provides a systematic geometric construction of gauge-invariant composite fields and reductions of gauge symmetry \cite{Fournel:2012yj,Attard:2017eji,Francois:2021yaz,Berghofer:2024fixing}.
Related developments for diffeomorphisms, general-relativistic gauge systems, and broader relational formulations may be found in Refs.~\cite{Francois:2023diffeo,Francois:2024grgft,Francois:2025nqm}.  Together with field-space geometry and horizontal dressings \cite{GomesRiello:2016,Gomes:2018shn}, boundary/edge and subregion charge structures \cite{Donnelly:2016auv,Riello2021,CarrozzaHoehn2022}, algebraic charged-sector analyses based on the Gauss law \cite{Buchholz}, and generalized-symmetry descriptions of center charge and genuine line operators \cite{Gaiotto:2014,BanksSeiberg:2010,KapustinSeiberg:2014}, these works constitute the main established antecedents of the present construction.

The present work does not claim QRFs, dressing fields, Wilson lines, edge data, or charged-sector language as new concepts.
\textbf{Its confinement-specific contribution is the operational layer placed before the infrared sector question is decided}:
\begin{enumerate}
\item[(i)] the frame, resource, and accessibility data required to make a non-Abelian charge component relationally readable are fixed as part of the protocol;
\item[(ii)] proper gauge redundancy, genuine boundary/asymptotic charge, frame change, dressing ambiguity, screening, and defect-sector averaging are kept logically distinct; and
\item[(iii)] with those data fixed, confinement or deconfinement is tested by the exclusion or survival of a finite-energy, physically accessible, asymptotically isolated nontrivial charged representation, allowing the exact and controlled models below and the conjectural four-dimensional targets to be compared within the same hierarchy.
\end{enumerate}
This is an operational organization of the charged-sector question, not a claim for a new universal local order parameter or for a QRF-only derivation of confinement dynamics.

The present QRF formulation keeps the central physical lesson that \textbf{nontrivial dynamical color charge should not survive as an asymptotically isolated bulk sector, but recasts the question in a gauge-invariant and boundary-sensitive language}.
\textbf{The relevant question is whether an admissible relational construction exists for which the nontrivial screening class survives in the infinite-volume limit as a finite-energy, physically accessible, asymptotically isolated charged representation.}
\footnote{
The \textbf{screening class} and charged representation are defined in III.B.
}

The viewpoint of this paper is also close in spirit to the Kugo--Ojima scenario \cite{Kugo:1979gm,Kugo:1995km,NakanishiOjima}. Kugo and Ojima seek conditions under which the global color charge is well defined and acts trivially on the physical BRST cohomology.
In Landau gauge this is commonly associated, under the assumptions entering the Kugo--Ojima construction, with the infrared condition $\tilde u(0)=-1$ for the Kugo--Ojima function in the Landau-gauge.

The construction directly addresses color charge rather than only the static potential, but its nonperturbative implementation assumes a suitable global BRST structure and is complicated by the Gribov problem \cite{Gribov:1977wm,Zwanziger:1989mf}.
It is also not naturally formulated in terms of boundaries, edge data, or the accessibility of relational frames, and the gauge-invariant character of the physical state space by itself does not sharply distinguish a confinement-like regime from a Higgs-like regime.
Moreover, the infrared behavior encountered in the scaling and decoupling families of Landau-gauge solutions \cite{FischerMaasPawlowski2009} means that the ghost-sector realization of the criterion should not be elevated to a gauge-independent definition of confinement.

The QRF criterion therefore does not replace the Kugo--Ojima quartet mechanism and does not derive it.
When the assumptions of the Kugo--Ojima BRST construction hold, it provides a gauge-fixed route toward the disappearance of colored asymptotic states.
No equivalence with the QRF criterion is asserted.
\textbf{The QRF formulation instead poses the charged-sector question after boundary, screening, and accessibility data have been fixed, without assuming a global nonperturbative BRST charge as part of the definition.}

Generalized global symmetry provides a complementary and particularly sharp language for center charge, genuine line operators, line endability, background higher-form fields, and anomaly constraints \cite{Gaiotto:2014,BanksSeiberg:2010,KapustinSeiberg:2014}.
In pure $SU(N)$ Yang--Mills theory, the electric center one-form symmetry diagnoses the $N$-ality of Wilson lines and the confinement properties of external probes.
However, it does not by itself assign an operational meaning to a Lie-algebra-valued color component or determine whether an adjoint relational excitation defines a finite-energy isolated sector.
In particular, an adjoint external source in pure Yang--Mills theory has zero $N$-ality and can be screened by gluons into a gluelump-type state. Thus \textbf{the center one-form symmetry does not define a nontrivial stable adjoint screening class.}
\textbf{The QRF question in the zero-$N$-ality sector is instead whether a full-color relational label can be defined and whether any such label survives screening as an independent finite-energy asymptotic charged sector.}
The two viewpoints therefore address overlapping but distinct questions:
\textbf{generalized symmetry classifies the infrared realization of extended operators}, whereas \textbf{QRF language asks what relational structure is required to read a color label and whether the corresponding charged sector survives.}

Topological defects enter the QRF description through the global properties of candidate frames.
A vortex, monopole, nontrivial bundle, or nontrivial holonomy can obstruct the globalization of a particular Higgs phase, color direction, or lift to a complete frame. This must be distinguished from ordinary path dependence caused by nonzero curvature and from a Gribov ambiguity in field or configuration space. More importantly, a defect or nontrivial holonomy alone is not sufficient for confinement.
\textbf{The obstruction must pair nontrivially with the charge class under consideration and must be accompanied by infrared disorder, defect proliferation, a divergent isolation energy, an area law, or another dynamical mechanism that eliminates the finite-energy isolated sector.}

The treatment of gauge--Higgs systems requires an additional qualification.
A nonvanishing scalar may provide accessible matter-frame data, but the quantum numbers of a gauge-invariant composite relative to that frame may be custodial or flavor quantum numbers rather than gauge color.
Moreover, with Higgs matter in the fundamental representation, confinement-like and Higgs-like regions need not be separated by a thermodynamic singularity in local gauge-invariant observables in the Fradkin--Shenker sense \cite{Fradkin:1978dv}.
The QRF distinction in such a theory is therefore an operational and semi-local classification relative to specified boundary conditions and preparation/accessibility protocols; it should not be presented as a universal local order parameter separating thermodynamic phases.

In a Higgs regime, for example, a nonvanishing accessible scalar configuration may supply matter-frame data, and a gauge-invariant composite such as $\Phi^{\dagger}\psi$ may describe an excitation relative to that frame. Its observable quantum numbers, however, may be custodial or flavor quantum numbers rather than gauge color, so the relevant charge algebra must be identified before the QRF criterion is applied.

Conversely, \textbf{in a confining regime the criterion is the absence of the admissible charged sector, not merely the absence of an isolated one-particle pole.}
Non-singlet relational expressions may still occur as finite-separation correlations, line-attached objects, boundary-anchored objects, or quantities conditioned on an external frame, without defining an asymptotically isolated colored sector.
The absence of an isolated pole is at most an auxiliary spectral diagnostic in a gapped situation with an ordinary particle interpretation.

{}

To avoid repeating the same logical caveats throughout the paper, the distinction among gauge kinematics, frame geometry, infrared dynamics, and charged-sector survival is consolidated here in a single status map Table~\ref{tab:status-map}; later sections refer back to this map.
Throughout the paper, exact or controlled lower-dimensional results, conditional ensemble results, and conjectural four-dimensional targets are kept explicitly distinct according to the logical status summarized in Table~\ref{tab:status-map}.

The paper is organized as follows.
Section~II develops the QRF and relational framework for gauge theory, including the proper gauge group, Gauss law, relational fields, bulk--boundary color charge, the distinction between gauge invariance and frame independence, and representative realizations of QRF-related structures.

Section~III formulates the operational QRF criterion for color confinement in terms of stable screening classes, physically accessible charged sectors, and sector energies, and compares this criterion with the Kugo--Ojima condition and other confinement diagnostics.

Section~IV tests the construction in exactly solvable $(1+1)$-dimensional Yang--Mills theory on the circle, the line, and an open interval, with particular emphasis on holonomy, Wilson-line dressings, boundary color charge, and the exact energy of nontrivial boundary-flux sectors.

Section~V analyzes path-dependent dressings, patchwise frames, topological defects, monodromy, and the global definability of QRF observables, while keeping geometrical obstructions distinct from the infrared dynamics required for confinement.

Section~VI studies the $(1+1)$-dimensional $U(1)$ gauge--Higgs model, including the boundary charge algebra, interval and circle observables, vortices, $N$-ality, and the relation between vortex disorder and charge-sector energies.

Section~VII introduces the symmetric-instanton reductions that connect four-dimensional Yang--Mills theory to lower-dimensional gauge--Higgs systems on hyperbolic spaces.

Sections~VIII and IX then analyze, respectively, the $SO(3)$-symmetric reduction to a two-dimensional $U(1)$ gauge--Higgs theory and the $SO(2)$-symmetric reduction to a three-dimensional $SU(2)$ gauge--Higgs theory, using radial or non-Abelian Wilson-line dressings, Higgs-direction frames, boundary data, vortices, and monopoles as controlled QRF laboratories.

Section~X returns to four-dimensional Yang--Mills theory and interprets color-direction fields and the Cho--Duan--Ge construction as partial QRF structures, while separating kinematical statements from conjectural infrared dynamics and discussing adjoint- and fundamental-scalar examples.

Section~XI gives a QRF interpretation of the Wilson-loop area law in terms of long-distance frame comparison and holonomy disorder.

Section~XII compares the QRF viewpoint with generalized-symmetry diagnostics, including center one-form symmetry, screening, line operators, and dynamical matter.

Finally, Section~XIII summarizes the resulting operational picture of color confinement, its controlled results, its limitations, and the main open problems.

The appendices collect technical material supporting the main text. Appendix A gives concrete realizations of right-frame transformations and clarifies their physical meaning in Wilson-line, boundary-edge-mode, dressing, and matter-frame constructions.
Appendix B summarizes the logical hierarchy of the QRF confinement statements and separates kinematical, geometrical, dynamical, and charged-sector statements.
Appendix C derives the differentiably improved Gauss generator in a bounded region and identifies the associated surface term with the boundary electric charge.
Appendix D derives the relational bulk--boundary charge identity and its specialization to the open interval.
Appendix E supplies algebraic details of the operational QRF criterion, including protocol-dependent twirling and the passage from finite-volume sectors to infinite-volume representations.
Appendix F gives the exact fixed-endpoint-frame correlator in $(1+1)$-dimensional Yang--Mills theory and shows explicitly how a nonzero finite-separation relational correlator coexists with a linearly divergent sector energy.

\begin{center}
\footnotesize
\begin{tblr}{colspec={Q[l,t,wd=0.24\columnwidth] Q[l,t,wd=0.48\columnwidth] Q[l,t,wd=0.20\columnwidth]},colsep=2pt,rowsep=1pt}
\toprule
Ingredient & Role in the present paper & Logical status \\
\midrule
Charged-sector question & Existence of a finite-energy, physically accessible, asymptotically isolated nontrivial charged representation & standard sector question; not new by itself \\
QRF operational input & Admissible relational frames/dressings, accessible boundary/asymptotic readout, right-frame orientations, and preparation resources & confinement-specific organizing input of this paper \\
Gauss-law/gauge invariance & Quotient by the proper gauge group and separation of genuine boundary/asymptotic charge & kinematical \\
Frame-lift/globalization obstruction & Failure of a chosen partial or complete frame to globalize & geometric; not a confinement criterion by itself \\
Infrared disorder/area-law dynamics & Defect ensembles, restricted-field dominance, or other dynamics that can raise the isolation energy & model dependent \\
$(1+1)$D Yang--Mills & $\Delta E_R(L)\propto L$ for every nontrivial boundary-flux sector & exact sufficient test \\
$(3+1)$D noncompact $U(1)$ & nonzero asymptotic flux with finite renormalized infrared Coulomb energy & deconfined positive control \\
Symmetric-instanton reductions & $\mathbb H^2$ vortex and $\mathbb H^3$ monopole frame-lift topology & exact within the symmetric ansatz \\
4D CDG mechanism & link between partial-frame disorder/defects and elimination of charged sectors & conjectural; falsifiable targets stated below \\
\bottomrule
\end{tblr}
\normalsize
\captionof{table}{A single status map of this paper.}
\label{tab:status-map}
\end{center}

\section{Quantum reference frames and relational observables in gauge theory}
\label{sec:qrf-framework}
\label{sec:ym-qrf}
This section contains only the kinematical and relational QRF framework needed independently of confinement: the proper gauge group and Gauss law, relational fields, boundary color, right-frame transformations, frame independence, and the classification of QRF-related structures.  The infrared charged-sector criterion for color confinement is discussed in Sec.~\ref{sec:qrf-confinement}.

\subsection{Proper gauge group, Gauss law, and color charge}
\label{subsec:gauss-boundary-relational}

Let $G$ be a compact non-Abelian gauge group.
The matter field $\psi(x)$ in the fundamental representation of the group $G$ transforms under a local gauge transformation $g(x)\in G$ as
\begin{equation}
  \psi(x)\mapsto
  g(x)\psi(x) .
\end{equation}
For a matter field $\psi(x)$ in a representation $R$, the local gauge transformation is given by $\psi(x)\mapsto R(g(x))\psi(x)$.

The non-Abelian Yang--Mills gauge field $\mathscr{A}_\mu(x)$ and the field strength $\mathscr{F}_{\mu\nu}(x)$ are denoted throughout by

\begin{align}
  \mathscr{A}_\mu(x)&=\mathscr{A}_\mu^A(x)T_A \in \mathfrak{g},
\nonumber\\
  \mathscr{F}_{\mu\nu}(x)
  &=\partial_\mu\mathscr{A}_\nu(x)-\partial_\nu\mathscr{A}_\mu(x)-ig_{\rm YM}[\mathscr{A}_\mu(x),\mathscr{A}_\nu(x)]
\nonumber\\
  &=\mathscr{F}_{\mu\nu}^A(x)T_A \in \mathfrak{g} ,
\end{align}
where Hermitian generators $T_A$ of the Lie algebra $\mathfrak{g}$ associated with the gauge group $G$ are
normalized by $2\mathrm{tr}(T_AT_B)=\delta^{AB}$ for $G=SU(N)$ in the fundamental normalization.

The Yang--Mills gauge field $\mathscr{A}_\mu(x)$ and the field strength $\mathscr{F}_{\mu\nu}(x)$ transform under a local gauge transformation $g(x)\in G$   as
\begin{align}
  &\mathscr{A}_\mu(x)\mapsto \mathscr{A}_\mu^g(x)
  g(x)\mathscr{A}_\mu(x)g(x)^{-1}
  +{i\over g_{\rm YM}}g(x)\partial_\mu g(x)^{-1},
  \nonumber\\
  &\mathscr{F}_{\mu\nu}(x)\mapsto
  g(x)\mathscr{F}_{\mu\nu}(x)g(x)^{-1} .
\end{align}
The canonical \textbf{color electric field} is defined by
\begin{equation}
\mathscr{E}^j(x)=\mathscr{F}^{0j}(x).
\end{equation}
Then the physical state is specified by imposing the \textbf{Gauss law}:
\begin{equation}
  \mathscr{G}^A(x)|\mathrm{phys}\rangle=0,
  \label{gauss-law}
\end{equation}
where $\mathscr{G}^A(x)$ is \textbf{the Gauss-law generator}
given by
\begin{equation}
  \mathscr{G}^A(x) :=\left(\mathscr{D}_j[\mathscr{A}]\mathscr{E}^j\right)^A(x)-\rho^A(x),
  \label{gauss-law-generator}
\end{equation}
where $\mathscr D_j[\mathscr{A}]$ is the covariant derivative in the adjoint representation $\mathscr D_j[\mathscr{A}]=\partial_j-ig_{\rm YM}[\mathscr A_j,\,\cdot\,]$ and
$\rho^A$ is the matter color-charge density.

Equation~\eqref{gauss-law} imposes the Gauss constraint for transformations belonging to the proper gauge group determined by the boundary prescription and differentiable generator (see II.C). Connectedness to the identity alone does not decide properness when genuine boundary or asymptotic symmetries are retained.

\begin{definition}[Proper gauge group and genuine boundary symmetry]\label{def:proper-genuine}
Fix a boundary condition $B$ on a spatial region $V$.  The \textbf{proper gauge group} $\mathcal{G}_0(B)$ is the subgroup of gauge transformations compatible with $B$ whose \emph{differentiable} canonical generators (see II.C) are constraints and are therefore quotiented out as redundancies.  Properness is determined by the allowed transformation and the boundary prescription, not by the accidental numerical vanishing of a charge on one state.  A boundary-nontrivial transformation whose differentiable generator retains a nonzero surface charge is not quotiented out.  When the boundary condition preserves such an action on the physical phase space or Hilbert space, it is a \textbf{genuine boundary symmetry}; analogously one speaks of a genuine asymptotic symmetry at infinity.  This distinction is independent of whether a transformation is ``large'' in the homotopical sense.

\end{definition}

Let $V$ be a $(D-1)$-dimensional spatial region with boundary $\partial V$.
The formal total \textbf{matter color charge} $Q^A(V)$ associated with the spatial region $V$ is obtained by integrating the Gauss law.  If one ignores boundary subtleties, one writes
\begin{equation}
  Q_{\rm matter}^A(V)=\int_V d^{D-1}x\,\rho^A(x)
  \approx \int_V d^{D-1}x\,\left(\mathscr{D}_j\mathscr{E}^j\right)^A .
\end{equation}
By the non-Abelian Gauss law this can be converted into a surface term only after specifying how the color index is compared at different space(time) points.
This is already the first appearance of the QRF issue.
\textbf{The Lie-algebra component label $A$ has no invariant meaning unless a color frame is specified throughout the region $V$ or at least on the boundary $\partial V$.}

\subsection{Quantum reference frame and relational fields}

\begin{definition}[Relational color frame/dressing transformation law]
Let $\Phi$ denote
the total set of dynamical variables: $\Phi=\{ \mathscr{A}, \psi, \phi,  \cdots \} $.
Let $G$ be a gauge group.
We write a gauge transformation as
\begin{equation}
  \Phi\longmapsto \Phi^g .
\end{equation}

Let a field $h(x)$ or dressing functional $h[\Phi](x)$ be a \textbf{relational color frame or dressing field} which transforms under a local gauge transformation $g(x)\in G$ as
\begin{equation}
  \boxed{
  h[\Phi](x) \mapsto h[\Phi^g](x)=g(x)h[\Phi](x)r_g(x)^{-1},
  }
  \label{h-transform}
\end{equation}
where $r_g$ is the residual transformation acting at the chosen root, base point, or boundary frame.
A dressing field or QRF is a field or functional $h[\Phi]$ whose transformation law compensates the transformation of the system as shown below.
If the root frame is kept fixed, we set $r_g=1$ for small gauge transformations.
If the boundary frame is dynamical, $r_g$ is the transformation of that boundary frame.
For transformations nontrivial at the root or boundary, the dressed variables transform under the residual boundary symmetry.  This is the precise sense in which \textbf{a QRF removes local gauge redundancy but does not necessarily remove all global or boundary color.}

This transformation law by itself does \emph{not} make $h$ a physical QRF.  We reserve the term \textbf{physical QRF} for a relational frame that in addition satisfies the operational accessibility requirements
summarized in Sec.~\ref{subsec:qrf-classification} and specified by the protocol in Sec.~\ref{subsec:operational-qrf}
; a Wilson line, Coulomb dressing, or gauge-fixing section without such an independently accessible readout remains a mathematical or QRF-type dressing/frame transporter.
The details of the right-frame transformations are given in Appendix A.
\end{definition}

Case (I)
First of all, we consider the case $r_g=1$ which includes the case of base-point-trivial small gauge transformations.
By introducing a \textbf{relational frame/dressing} $h(x)\in G$ that transforms as (\ref{h-transform}),
a gauge-invariant but frame-labeled expression is obtained.

First, we apply the idea of QRF to the matter field.
 \begin{definition}[Relational matter field]
A bare matter field $\psi(x)$ is gauge covariant, not gauge invariant,
and transforms as
\begin{equation}
\psi(x) \mapsto \psi^g(x)=g(x)\psi(x) ,
  \label{transformation-quark}
\end{equation}
Once a frame $h(x)$ satisfying \eqref{h-transform} is supplied, one can define the \textbf{relational matter field}, i.e.,
 the \textbf{relational representative} $\Psi_h(x)$ of a fundamental matter field $\psi(x)$ by
\begin{equation}
\boxed{ \Psi_h(x):=h(x)^{-1}\psi(x) }.
  \label{rel-quark}
\end{equation}
Indeed, \textbf{the relational field $\Psi_h$ is invariant under the simultaneous gauge transformation of the system $\psi$ and the frame $h$}:

\begin{align}
 \Psi_h \mapsto \Psi_h^g
 &=(h^g)^{-1} \psi^g=(gh)^{-1}g\psi
 \nonumber\\
 &=h^{-1}g^{-1}g\psi=h^{-1}\psi=\Psi_h.
\end{align}
Therefore $\Psi_h(x)$ is  a gauge-invariant \textbf{relational field: $\Psi_h(x)$ is the matter field as seen from the color frame $h$.}
\textbf{Its component label is not an absolute color.}  \textbf{It is a color component relative to the frame $h$.}
\end{definition}

\begin{definition}[Relational gauge field]
Next, one can define the relational Yang--Mills gauge field (connection 1-form) $\mathscr{A}^h(x)$ and the relational field strength (curvature 2-form) $\mathscr{F}^h(x)$ by \cite{Donnelly:2016auv,GomesRiello:2016,Gomes:2018shn, CarrozzaHoehn2022}

\begin{equation}
\begin{aligned}
&\boxed{\begin{aligned}
\mathscr{A}^h(x)
&:=h(x)^{-1}\mathscr{A}(x)h(x)
+{i\over g_{\rm YM}}h(x)^{-1}dh(x)
\end{aligned}}\\
&\boxed{\mathscr{F}^h(x):=h(x)^{-1}\mathscr{F}(x)h(x)} .
\end{aligned}
  \label{dressed-AF}
\end{equation}

In the notation of components, the relational Yang--Mills connection and the curvature are written as

\begin{equation}
\begin{aligned}
  \mathscr{A}^h_\mu(x)
  &:=h(x)^{-1}\mathscr{A}_\mu(x)h(x)
  +\frac{i}{g_{\rm YM}}h(x)^{-1}\partial_\mu h(x),\\
  \mathscr{F}^h_{\mu\nu}(x)
  &:=h(x)^{-1}\mathscr{F}_{\mu\nu}(x)h(x).
\end{aligned}
\end{equation}
These fields are invariant under the original local gauge redundancy, although they generally transform under the right-frame change $h\mapsto hk$ as shown later.
Indeed, a direct calculation gives
\begin{align}
(\mathscr{A}^h)^g
=&(h^g)^{-1}\mathscr{A}^g h^g
+{i\over g_{\rm YM}}(h^g)^{-1}dh^g
\nonumber\\
=&(gh)^{-1}\left(g\mathscr{A}g^{-1}
+{i\over g_{\rm YM}}gdg^{-1}\right)(gh)
\nonumber\\&\quad
+{i\over g_{\rm YM}}(gh)^{-1}d(gh)
\nonumber\\
=&h^{-1}\mathscr{A}h
+{i\over g_{\rm YM}}h^{-1}dg^{-1}gh
\nonumber\\&\quad
+{i\over g_{\rm YM}}h^{-1}g^{-1}dgh
+{i\over g_{\rm YM}}h^{-1}dh
\nonumber\\
=&h^{-1}\mathscr{A}h
-{i\over g_{\rm YM}}h^{-1}g^{-1}dgh
\nonumber\\&\quad
+{i\over g_{\rm YM}}h^{-1}g^{-1}dgh
+{i\over g_{\rm YM}}h^{-1}dh
\nonumber\\
=&h^{-1}\mathscr{A}h+{i\over g_{\rm YM}}h^{-1}dh
=\mathscr{A}^h,
\\
(\mathscr{F}^h)^g
=&(h^g)^{-1}\mathscr{F}^g h^g
=(gh)^{-1}g\mathscr{F}g^{-1}gh
\nonumber\\
=&h^{-1}g^{-1}g\mathscr{F}g^{-1}gh
=h^{-1}\mathscr{F}h=\mathscr{F}^h .
\end{align}
Thus \textbf{$\mathscr{A}^h(x)$ and $\mathscr{F}^h(x)$ are gauge-invariant relational fields
as seen from the color frame $h$.}

\end{definition}

Similarly, \textbf{the relational color electric field is defined by
\begin{equation}
  \boxed{ \mathscr{E}^j_h(x):=h(x)^{-1}\mathscr{E}^j(x)h(x) },
\end{equation}
which is invariant under the simultaneous original gauge transformation of the system and the frame.}
A \textbf{relational boundary color charge $Q_h^A(\partial V)$ measured with respect to the frame $h$} may therefore be written as
\begin{equation}
\begin{aligned}
  Q_h^A(\partial V)
  &=\int_{\partial V}dS_j\,2\mathrm{tr}
  \left(T^A \mathscr{E}^j_h\right) \\
  &=\int_{\partial V}dS_j\,2\mathrm{tr}
  \left(T^A h^{-1}\mathscr{E}^j h\right).
\end{aligned}
  \label{rel-boundary-charge}
\end{equation}
This formula is not the assertion that color has become absolute.  Rather, it says that \textbf{the color charge is meaningful only relative to a frame.}
\textbf{If the frame is not physically admissible, is not globally definable, or depends on a path in a topologically nontrivial way, the charge component $Q_h^A$ is not a globally defined observable}.

This is the key difference between a \textbf{formal color charge} and a \textbf{physical relational color charge}.
The former is obtained by writing a Lie-algebra component of a gauge-covariant expression.  The latter requires a prescription for comparing Lie-algebra directions.  In an Abelian theory this distinction is hidden because the adjoint action is trivial.  In a non-Abelian theory it is essential.

There are several ways to realize such a frame.
One can use

(i) a dressing functional $h(x;\mathscr{A})\equiv h[\mathscr{A}](x)$ of the gauge field $\mathscr{A}$;

(ii) a Wilson line from a reference point $x_0$ to $x$;

(iii) a Higgs field when a scalar condensate supplies a direction in color space;

(iv) boundary data at infinity.

For example, one can choose the Wilson-line dressing for a chosen path $\gamma:x_0\to x$ from the base point $x_0$ to the point $x$,
\begin{equation}
  W(x,x_0;\gamma) :=\mathcal{P}\exp\left(ig_{\rm YM}\int_{\gamma:x_0 \to x} \mathscr{A}\right),
\end{equation}
which transforms as (see Fig.~\ref{fig:Wilson_line_1})
\begin{equation}
  W(x,x_0;\gamma)\mapsto g(x)W(x,x_0;\gamma)g(x_0)^{-1} .
\end{equation}

\begin{figure}[htbp]
  \centering
\includegraphics[width=.85\linewidth]{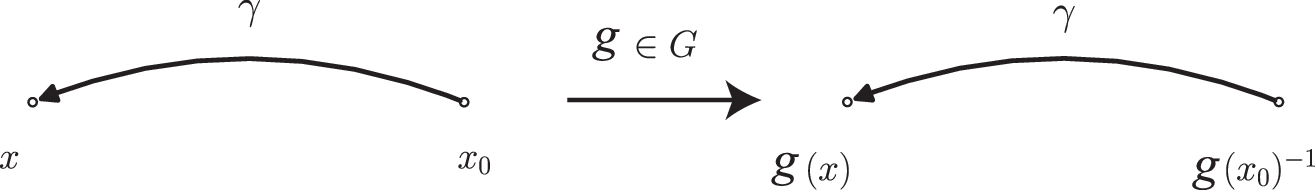}
  \caption{The Wilson line $W(x,x_0;\gamma)$ for a path $\gamma:x_0\to x$ from the base point $x_0$ to the point $x$, and its gauge transformation by $g\in G$.
}
  \label{fig:Wilson_line_1}
\end{figure}

Then we can use the Wilson-line QRF-type dressing/frame transporter:

\begin{align}
&h_{\gamma,h_0}(x) :=W(x,x_0 ;\gamma) h_0, \ h_0 :=h(x_0) .
\end{align}
Thus the relational matter field is obtained as
\begin{equation}
  \Psi_{\gamma,h_0}(x) =h_{\gamma,h_0}(x)^{-1}\psi(x)
=h_0^{-1}W(x_0,x;\gamma)\psi(x) .
  \label{wilson-dressed-quark}
\end{equation}
If a boundary or base-point frame $h_0$ transforms as
$h(x_0) \mapsto g(x_0)h(x_0) \Leftrightarrow h_0\mapsto g(x_0)h_0$,
then $\Psi_{\gamma,h_0}(x)$
is invariant under local gauge transformations.  It is also manifestly \textbf{semi-local}: \textbf{it depends on the path $\gamma$ and on the reference data at $x_0$. } Thus \textbf{a dressed colored object is not a strictly local observable.}

This semi-locality is not a drawback of the construction.  It is the correct physical content.
\textbf{In gauge theory, a charged excitation cannot be localized without accompanying its gauge field}; the classic Dirac construction provides an explicit gauge-invariant charged dressing in QED \cite{Dirac:1955uv}.  The dressing records how the charge is attached to the rest of the system, to a boundary, or to another charge.  The QRF language emphasizes that \textbf{the dressing is the frame-side data required to make the colored statement meaningful.}

Case (II)
Subsequently, we consider the case $r_g\not=1$
The relational fields defined in this case are invariant under local left gauge transformations, but covariant under the boundary or base-point group.
In fact, the relational matter field
$
 \Psi_h=h^{-1}\psi
$
transforms in general as
\begin{equation}
 \Psi_h\longmapsto R(r_g)\Psi_h,
\end{equation}
and is not a strict invariant. Similarly, the Yang--Mills field and the field strength transform

if $r_g$ is spacetime dependent, as
\begin{equation}
 \mathscr A^h\longmapsto r_g\mathscr A^h r_g^{-1}
 +\frac{i}{g_{\rm YM}}r_gdr_g^{-1},
 \quad
 \mathscr F^h\longmapsto r_g\mathscr F^h r_g^{-1}.
\end{equation}
It should be remarked that a right-frame transformation is not simply defined as a ``transformation other than a gauge transformation.''
It is a general term for an action that remains on the frame side after the original left gauge action has been handled by a dressing, or for an action that reorients the frame itself.

Depending on the model, a right action may be a residual gauge action, a physical boundary symmetry, a change of dressing choice, a flavor/custodial symmetry, or a second local gauge group.

See Appendix~\ref{Appendix:right-frame} for concrete examples and the physical meaning of right-frame transformations.

\subsection{Bulk Gauss law, relational boundary charge, and boundary color}
\label{subsec:relational-boundary-charge}

The Gauss law constrains compactly supported color charge in the bulk.  Let $V$ be a $(D-1)$-dimensional spatial region with boundary $\partial V$, and let $\omega(x)=\omega^A(x)T_A$ be a Lie-algebra-valued gauge parameter.
For a Yang--Mills field $\mathscr{A}_j$ with the canonical color electric field $\mathscr{E}^j$, the raw bulk Gauss functional is given by
\begin{equation}
  G_V[\omega]
  :=\int_V d^{D-1}x\,2\operatorname{tr}\left[\omega(x)\left(\mathscr{D}_j\mathscr{E}^j-\rho\right)(x)\right].
  \label{eq:raw-gauss-functional-V}
\end{equation}
Integration by parts gives
\begin{equation}
\begin{aligned}
 G_V[\omega]
 &=-\int_V d^{D-1}x\,2\operatorname{tr}
 \left\{(\mathscr{D}_j\omega)\mathscr{E}^j+\omega\rho\right\}\\
 &\quad+\int_{\partial V}dS_j\,2\operatorname{tr}
 \left\{\omega\mathscr{E}^j\right\}.
\end{aligned}
 \label{eq:gauss-integration-V}
\end{equation}
The first term shows the Gauss constraint, while the second term is \textbf{differentiable on phase space} only after the boundary term has been specified \cite{ReggeTeitelboim1974,HenneauxTeitelboim1992,Riello2021}.

If $\omega$ is zero at the boundary: $\omega|_{\partial V}=0$, the generator is pure gauge and one imposes
\begin{equation}
  (\mathscr{D}_j\mathscr{E}^j - \rho)|\mathrm{phys}\rangle=0 .
\end{equation}

If $\omega$ is nonzero at the boundary, $\omega|_{\partial V}\not=0$, we define the \emph{outward} boundary color electric flux by
\begin{equation}
  Q_{\partial V}^{\mathrm{out}}[\omega]
  :=\int_{\partial V}dS_j\,2\operatorname{tr}\{\omega\mathscr{E}^j\} .
  \label{app-Q-boundary}
\end{equation}
This is the standard separation between gauge redundancy and boundary or   asymptotic symmetry.  Throughout the paper the shorter relational notation $Q_{\partial V}[h,\xi]$ refers to the same outward orientation; the superscript ``out'' is displayed explicitly whenever a sign comparison with an improved generator is involved.

Strictly speaking, when $\omega|_{\partial V}\neq0$, the raw bulk functional $G_V[\omega]$ is not by itself a \textbf{differentiable Hamiltonian generator on the chosen phase space}: its variation contains an uncancelled surface term.
One must specify the boundary condition and add the corresponding \textbf{Regge--Teitelboim boundary improvement}.

For the standard electric-flux form, one may write
\begin{equation}
G_B[\omega]=G_V[\omega]-Q_{\partial V}^{\mathrm{out}}[\omega],
\end{equation}
so that $G_B$ is differentiable; with this convention
\begin{equation}
G_B[\omega]\approx -Q_{\partial V}^{\mathrm{out}}[\omega]
\end{equation}
on the Gauss-constraint surface $\mathscr{D}_j\mathscr{E}^j - \rho\approx 0$.

The overall sign assigned to the surface charge is conventional.  If $\omega|_{\partial V}=0$, the surface variation vanishes and the differentiable generator is a proper Gauss constraint.
In what follows, ``boundary charge'' always refers to the surface functional defined after this differentiability requirement has been imposed \cite{ReggeTeitelboim1974,HenneauxTeitelboim1992,Riello2021}.

To assign a non-Abelian component to the boundary flux one must also specify how color directions are compared along $\partial V$.  Let $h(x)$ be a boundary frame and define
\begin{equation}
\begin{aligned}
 \mathscr{E}_h^j(x)&:=h(x)^{-1}\mathscr{E}^j(x)h(x),\\
 \rho_h(x)&:=h(x)^{-1}\rho(x)h(x).
\end{aligned}
 \label{eq:relational-E-rho-V}
\end{equation}
Let $\xi(x)\in\mathfrak g$ denote the components of an adjoint-valued \textbf{physical test direction} in the frame $h(x)$; the same physical section in the original gauge description is given by
\begin{equation}
 \boldsymbol{\xi}_h(x):=h(x)\xi(x)h(x)^{-1}.
\end{equation}
The gauge-invariant \textbf{relational boundary color charge} ${Q}_{\partial V}[h,\xi]$ is the scalar pairing
\begin{align}
&\boxed{ {Q}_{\partial V}[h,\xi]
 := \int_{\partial V}dS_j\,
 2\operatorname{tr}\left\{\xi(x)\mathscr{E}_h^j(x)\right\} }
 \nonumber\\
\Leftrightarrow &\boxed{ {Q}_{\partial V}[h,\xi]= \int_{\partial V}dS_j\,
 2\operatorname{tr}\left\{\boldsymbol{\xi}_h(x)\mathscr{E}^j(x)\right\} }
 .
 \label{eq:relational-boundary-charge}
\end{align}

Both expressions are invariant under the original proper gauge transformation $\mathscr{E}^j\mapsto g\mathscr{E}^j g^{-1}$, $h\mapsto gh$, and $\boldsymbol{\xi}_h\mapsto g\boldsymbol{\xi}_h g^{-1}$.  Equation~\eqref{rel-boundary-charge} is the fixed-generator special case $\xi=T^A$:
\begin{align}
 Q_{h}^{A}(\partial V)=Q_{\partial V,h}^{A}:=&{Q}_{\partial V}[h,\xi=T^A].
 \label{eq:relational-boundary-charge2}
\end{align}
Thus $Q_{\partial V,h}^{A}$ is a gauge-invariant component of the boundary color electric flux \emph{relative to} the frame $h$; it is not an absolute non-Abelian color charge component.

The corresponding bulk/boundary relation can be stated without introducing a gauge-dependent integral of color components.  Define the relational color charges as
\begin{equation}
 \begin{aligned}
 {Q}_{V}^{\mathrm{matter}}[h,\xi]
 &:=\int_V d^{D-1}x\,
 2\operatorname{tr}\left\{\xi(x)\rho_h(x)\right\},\\
 {Q}_{V}^{\mathrm{gauge}}[h,\xi]
 &:=\int_V d^{D-1}x\,
 2\operatorname{tr}\left\{\bigl(\mathscr{D}_j[\mathscr{A}^h]\xi\bigr)(x)\mathscr{E}_h^j(x)\right\},\\
 {Q}_{V}^{\mathrm{total}}[h,\xi]
 &:={Q}_{V}^{\mathrm{matter}}[h,\xi]
 +{Q}_{V}^{\mathrm{gauge}}[h,\xi]
 \ \approx\ {Q}_{\partial V}[h,\xi],
 \end{aligned}
 \label{eq:bulk-boundary-relational-gauss-law}
\end{equation}
where
$
 \mathscr{D}_j[\mathscr{A}^h]\xi
 :=\partial_j\xi-i g_{\rm YM}[\mathscr{A}_j^h,\xi].
$
and the last relation holds on the Gauss-constraint surface and follows by multiplying $\mathscr{D}_j[\mathscr{A}^h]\mathscr{E}_h^j-\rho_h\approx 0$ by $\xi$, taking the trace, and integrating by parts.

The matter term alone ${Q}_{V}^{\mathrm{matter}}[h,\xi]$ is therefore not in general equal to the boundary flux ${Q}_{\partial V}[h,\xi]$: when $\mathscr{D}_j[\mathscr{A}^h]\xi\neq0$, the non-Abelian gauge-field contribution ${Q}_{V}^{\mathrm{gauge}}[h,\xi]$ is required.  If $\mathscr{D}_j[\mathscr{A}^h]\xi=0$, the gauge-field term vanishes ${Q}_{V}^{\mathrm{gauge}}[h,\xi]=0$ and the matter contribution equals the boundary flux ${Q}_{V}^{\mathrm{matter}}[h,\xi]\approx\ {Q}_{\partial V}[h,\xi]$.  In pure Yang--Mills theory $\rho=0$, but a nonzero boundary-flux sector may still exist; it is not an independent compactly supported bulk charge.

If the test transformation vanishes at $\partial V$, its differentiable generator belongs to the proper gauge group and annihilates physical states.  If it is boundary-nontrivial and the boundary condition retains that action as a genuine symmetry, ${Q}_{\partial V}$ is a physical boundary charge.  This is why a parameter nonzero at the boundary should not automatically be called a ``large'' gauge transformation: boundary nontriviality and homotopical largeness are distinct notions.  The transformation under changes of the color frame is analyzed separately in Sec.~\ref{subsec:color-charge}.

\subsection{Gauge invariance, frame changes, and frame independence}
\label{subsec:gauge-frame-independence}
\par

A first lesson of the QRF formulation is that \textbf{gauge invariance and frame independence are distinct notions}.
This distinction is familiar in QRF formulations \cite{VanrietveldeEtAl2020}.
\begin{definition}[Frame independence]
A QRF observable is constructed so as to be invariant under the proper gauge group.
In a canonical formulation in which the relevant domains and Poisson brackets are well defined, this means that the QRF observable has weakly vanishing Poisson brackets with the proper first-class constraints, and in that restricted sense it may be called a Dirac observable \cite{HenneauxTeitelboim1992}.
For the field-dependent, nonlocal, and boundary-sensitive dressings considered here, however, we will use the more precise statement ``invariant under the proper gauge group'' rather than assuming additional global domain properties.  However, it is typically an observable relative to a chosen reference frame.
\textbf{A relational observable is constructed to be invariant under the proper gauge group, but it may still transform when the reference frame itself is changed.}
The basic logical relation is
\begin{equation}
  \boxed{\hbox{gauge invariance}\not\Rightarrow \hbox{frame independence}.}
  \nonumber
\label{gauge-not-frame}
\end{equation}

\end{definition}

Let a second admissible frame $h'(x)$ be related to the first $h(x)$ by a right reorientation $k(x)$:
\begin{equation}
  h'(x)=h(x)k(x),
  \label{frame-change}
\end{equation}
where $k$ is invariant under the original local left gauge redundancy and hence represents a right reorientation of the relational frame rather than an ordinary gauge transformation.
The corresponding dressed variables satisfy
\begin{equation}
  \Phi^{h'}=(\Phi^h)^k .
\end{equation}
It follows directly from the definition above that
$O_{h'}[\Phi]=\mathcal O[(\Phi^h)^k]$.

The basic relational variables then transform as
\begin{align}
  \Psi_{hk}&=R(k^{-1})\Psi_h,\\
  \mathscr{A}^{hk}&=k^{-1}\mathscr{A}^{h}k
  +{i\over g_{\rm YM}}k^{-1}dk,\nonumber\\
  \mathscr{F}^{hk}&=k^{-1}\mathscr{F}^{h}k.
  \label{frame-change-AF}
\end{align}
Consequently, \textbf{a color component of a relational field is generally frame relative even though it is invariant under the original proper gauge transformations}.
For example, the components
\begin{equation}
  \mathscr{F}^{h,A}_{\mu\nu}=2\,\mathrm{tr}\bigl(T^A\mathscr{F}^h_{\mu\nu}\bigr)
\end{equation}
are in general frame dependent.
By contrast, invariant contractions such as the composite operator
\begin{equation}
 \operatorname{tr}(\mathscr{F}_{\mu\nu}^{h}\mathscr{F}^{h,\mu\nu})
 =\operatorname{tr}(\mathscr{F}_{\mu\nu}\mathscr{F}^{\mu\nu})
\end{equation}
are both gauge invariant and frame independent.
In this precise sense, QRF does not turn color components into absolute observables.  It turns them into components relative to a specified frame.

\subsection{Change of frame for relational color charge}
\label{subsec:color-charge}

The relational boundary-charge functional ${Q}_{\partial V}[h,\xi]$, its component $Q_{\partial V,h}^{A}$, and the bulk/boundary Gauss relation were defined in Sec.~\ref{subsec:relational-boundary-charge}.  We now analyze only how these quantities behave when the relational frame is changed.

We distinguish a \textbf{passive change of frame} from an \textbf{active reorientation of a physical reference frame}.
In order to represent the same physical adjoint test section (direction) $\xi_{h}(x)$ in the new frame, its components must transform as
\begin{equation}
 \xi(x)\longmapsto \xi'(x)=\xi^k(x):=k(x)^{-1}\xi(x)k(x).
\end{equation}
provided that $h$ transform as
\begin{equation}
 h(x)\longmapsto h'(x)=h(x)k(x).
\end{equation}
The frame components of the color electric field transform in the same way:
\begin{equation}
 \mathscr{E}_{h}^j(x)\longmapsto
 \mathscr{E}_{h'}^j(x)=\mathscr{E}_{hk}^j(x)
 =k(x)^{-1}\mathscr{E}_{h}^j(x)k(x),
\end{equation}
which leads to the \textbf{invariance of relational charges under the change of frame}:
\begin{align}
 {Q}_{\partial V}[hk,\xi^k]
  = {Q}_{\partial V}[h,\xi], \quad
 {Q}_{V}^{\mathrm{total}}[hk,\xi^k]
  =
 {Q}_{V}^{\mathrm{total}}[h,\xi].
 \label{eq:frame-independent-relational-charge}
\end{align}
In fact, at every boundary point, one finds
\begin{align}
& 2\operatorname{tr}\left\{\xi'(x)\mathscr{E}_{h'}^{j}(x)\right\} \nonumber\\
 &=2\operatorname{tr}\left\{
 k(x)^{-1}\xi(x)k(x)
 k(x)^{-1}\mathscr{E}_{h}^{j}(x)k(x)
 \right\}
 \nonumber\\
 &=2\operatorname{tr}\left\{
 k(x)^{-1}\xi(x)\mathscr{E}_{h}^{j}(x)k(x)
 \right\}
 \nonumber\\
 &=2\operatorname{tr}\left\{
 \xi(x)\mathscr{E}_{h}^{j}(x)
 \right\}.
\end{align}
The last equality follows from cyclicity of the trace. Hence,
the \textbf{frame-independence of the relational boundary charge} is obtained:
\begin{equation}
 \boxed{
 {Q}_{\partial V}
 \left[hk,k^{-1}\xi k\right]
 ={Q}_{\partial V}[h,\xi].
 }
\end{equation}
This remains true when $k(x)$ is position dependent and requires no assumption about the explicit forms of $\mathscr{A}_{\mu}$, $\mathscr{E}_{h}$, or $h$.
Thus \textbf{the relational charge functional is frame independent when the same physical test direction $\xi$ is held fixed and its components are transformed together with the frame $h$.}  This is the appropriate \textbf{passive notion of frame independence}.
\textbf{The fully paired scalar functional ${Q}_{\partial V}[h,\xi]$ is invariant under a passive change of frame when the components of the same physical test direction transform together with $h$}.

On the other hand, if the numerical components of $\xi$ are kept fixed:
\begin{equation}
 h\longmapsto hk,
 \qquad
 \xi\longmapsto\xi,
\end{equation}
then this is not a passive coordinate change but an \textbf{active reorientation of the physical color direction} used in the measurement. In this case, $Q_{\partial V}[hk,\xi]$ is not equal to $Q_{\partial V}[h,\xi]$ in general.
\begin{align}
 Q_{\partial V}[hk,\xi]
 &=\int_{\partial V}dS_{j}\,
 2\operatorname{tr}\left\{
 \xi k^{-1}\mathscr{E}_{h}^{j}k
 \right\}
 \nonumber\\
 &=\int_{\partial V}dS_{j}\,
 2\operatorname{tr}\left\{
 k\xi k^{-1}\mathscr{E}_{h}^{j}
 \right\}
 \not= Q_{\partial V}[h,\xi],
\end{align}
Equality holds only in special cases, for example when $[k(x),\xi(x)]=0$. Therefore, \textbf{a color component with respect to a fixed generator is frame relative and is not frame independent} in the strongest sense.

Quantities that are absolutely frame independent include, for example,
\begin{equation}
 2\operatorname{tr}\left\{\mathscr{E}_{h}^{2}\right\}
 =2\operatorname{tr}\left\{\mathscr{E}^{2}\right\},
\end{equation}
higher Casimir invariants, the label of a charge representation, and functionals in which the physical test section is included and all indices are fully contracted. In the quantum theory, the $SU(2)$ boundary-sector label $j$ and the Casimir $j(j+1)$ play this role, as shown in the next section VI.

By contrast, if one keeps the numerical generator \(T^A\) fixed while
actively replacing \(h\) by \(hk\), one changes the physical color
direction with respect to which the charge is measured.  The component
\(Q_{\partial V,h}^{A}\) then changes, as it should.  If the right-frame reorientation
\(k|_{\partial V}=k_0\) is constant on the boundary, the
Lie-algebra-valued quantity
\begin{equation}
 Q_{\partial V,h}
 :=
 \int_{\partial V}\mathrm{d}S_j\,
 \mathscr{E}_h^j
\end{equation}
transforms by a single adjoint conjugation,
\begin{equation}
 \boxed{
 Q_{\partial V,hk_0}
 =
 k_0^{-1}
 Q_{\partial V,h}
 k_0
 }.
 \label{eq:constant-right-frame-charge-conjugation}
\end{equation}
Its components are frame covariant, whereas \textbf{its conjugacy class,
Casimir invariants such as
\begin{equation}
 \operatorname{tr}
 \left(
 (Q_{\partial V,h})^{\,2}
 \right),
 \qquad
 \operatorname{tr}
 \left(
 (Q_{\partial V,h})^{\,3}
 \right),
\end{equation}
and the statement
\(Q_{\partial V,h}^{A}=0\) are frame independent.}
A statement such as $Q_{\partial V,h}^{A}\neq0$ is not an absolute statement unless the boundary frame has itself been physically specified.

If \(k(x)\) varies over the boundary, the integral of
\(k(x)^{-1}\mathscr{E}_h^j(x)k(x)\) is not related to the original
integral by one global conjugation.  In that case a single
Lie-algebra-valued total charge is defined only after a coherent
boundary frame, a common root, or an additional parallel-transport
prescription has been specified.  Otherwise the fundamental object is
the local boundary charge density or the smeared functional
\({Q}_{\partial V}[h,\xi]\), not an unqualified global vector
\(Q^A\).

For position-dependent boundary reorientations no such single global conjugation exists without additional parallel-transport/root data.
There are consequently four distinct situations.

\begin{enumerate}
 \item Under a passive change of frame coordinates, with
 \(\xi\longmapsto k^{-1}\xi k\), \textbf{the relational charge functional ${Q}_{\partial V}[h,\xi]$ is frame independent} as in
 \eqref{eq:frame-independent-relational-charge}.

 \item If a physical boundary frame is fixed by the preparation or by
 an external apparatus, \(Q_{\partial V,h}^{A}\) is a
 gauge-invariant and operationally meaningful relational component.
 An active rotation of that physical frame describes a different
 measurement and need not leave the component unchanged.

 \item If constant boundary reorientations are regarded as
 unobservable, only the conjugacy class, Casimir invariants, the
 representation or screening-sector label, and the vanishing of the
 full charge are frame-independent statements.

 \item If local boundary reorientations are retained as genuine
 surface symmetries, their charges are local boundary-flux densities.
 They must not be removed by an unconditional right-frame twirling.
\end{enumerate}

The correct conclusion is therefore not that a gauge-invariant relational color charge must always be frame dependent.
Rather, its component with respect to a fixed numerical generator is relative to a chosen physical color direction, while the fully paired charge functional is invariant under a passive change of frame.
\textbf{The Gauss law removes compactly supported bulk gauge charge as a redundancy, relates the relational total bulk charge to boundary electric flux, and still permits physical boundary-charge sectors whenever the boundary conditions retain the corresponding symmetry.}

\subsection{Classification and representative realizations of QRF-related structures}
\label{subsec:qrf-classification}

The word ``QRF'' will be used in a deliberately broad but classified sense, with the following terminology maintained throughout the paper.  A \textbf{physical QRF} is reserved for a reference system with independent quantum degrees of freedom or an operationally accessible readout that can be prepared, controlled, or conditioned upon.  A Wilson line, Coulomb dressing, or gauge-fixing section will instead be called a \textbf{mathematical dressing}, or more specifically a \textbf{QRF-type dressing/frame transporter}, unless an independently accessible reference system is also supplied.  A boundary frame or edge mode is called a \textbf{boundary QRF} only when its orientation is retained as physically accessible boundary data by the chosen boundary condition.  A Higgs field or matter system can provide a \textbf{physical matter QRF} where the relevant configuration is nonvanishing and its relational data are operationally accessible.  Finally, a color-direction field such as $\bm n[\mathscr A]$ or a normalized adjoint Higgs direction $\widehat\Phi$ is generally an \textbf{intrinsic partial frame}: it specifies only part of a color orientation and may inherit residual stabilizer freedom or Gribov ambiguities.  Similar transformation laws do not imply identical operational status, and the terms ``physical QRF'', ``QRF-type dressing/frame transporter'', and ``intrinsic partial frame'' will not be used interchangeably.

In this paper, the following objects are treated as QRF-related structures; they are not all identified as physical QRFs:

\begin{enumerate}
\item a field-dependent dressing constructed from the gauge field,
\item a local section of gauge fixing,
\item a Wilson line together with base-point data,
\item a boundary edge mode,
\item a dynamical Higgs field or matter reference system,
\item a CDG color-direction field $\bm n[\mathscr A]$.
\end{enumerate}
They may have the same transformation law, but their physical meanings are different. To call an object a physical QRF, one should at least state whether
\begin{itemize}
\item it has independent quantum degrees of freedom or an operationally accessible readout,
\item it is compatible with the boundary and regularity conditions and defines a well-posed relational construction in the regime under consideration,

\item the reference direction can be prepared, controlled, or conditioned upon,
\item it is more than a mere gauge-fixing section.
\end{itemize}

We use the following classification of QRF-related structures in Table~\ref{tab:gauge-frame1}.

\begin{widetext}
\begin{center}
\footnotesize
\begin{tblr}{colspec={Q[l,t,wd=0.17\textwidth] Q[l,t,wd=0.25\textwidth] Q[l,t,wd=0.53\textwidth]},colsep=2pt,rowsep=1pt}
\toprule
Type & Example & Role and caveat \\
\midrule
Mathematical dressing & Coulomb dressing, Wilson line & It makes a gauge-covariant quantity invariant, but it is not necessarily an independent physical system.\\
Boundary QRF & Edge mode, boundary frame & It represents boundary symmetries and charges. It depends on the boundary conditions.\\
Matter QRF & Higgs field, reference particle & It enables conditional observations, but zeros, defects, and quantum fluctuations must be taken into account.\\
Intrinsic partial-frame & color direction field $\bm n[\mathscr A]$ & It is an $SU(2)/U(1)$ axis, not a complete $SU(2)$ frame. It can have a Gribov problem.\\
\bottomrule
\end{tblr}
\normalsize
\captionof{table}{Classification of QRF-related structures}
\label{tab:gauge-frame1}
\end{center}
\end{widetext}

A useful geometric antecedent is the field-space formulation of Gomes and Riello \cite{GomesRiello:2016}.  A connection one-form on field space separates horizontal and vertical directions and is naturally related to geometric BRST structures.  In bounded regions this framework clarifies the relation among gauge reduction, boundary data, and dressings; the Gribov problem obstructs a global gauge-fixing section but does not by itself invalidate a field-space connection, which may have non-vanishing field-space curvature.  Gomes, Hopfm\"uller, and Riello \cite{Gomes:2018shn} further emphasize that gauge and matter fields themselves may supply relational data, without requiring every group-valued edge variable to be interpreted as an independent physical system.

\subsection{Abelian phase frames and non-Abelian Wilson lines}
\label{subsec:examples-frame}

The Abelian example is useful because it removes the complication of noncommuting frames.  Let a charged scalar field $\psi$ transform as
\begin{equation}
  \psi\mapsto e^{\mathrm{i}q\lambda}\psi,
\end{equation}
while let $\theta$ be a \textbf{phase reference frame} which transforms as $\theta\mapsto\theta+\lambda$.
Introduce the frame field
\begin{equation}
  h_\theta= e^{\mathrm{i}q\theta} .
\end{equation}
Then the relational variable
\begin{equation}
  \Psi_\theta=h_\theta^{-1}\psi=e^{-\mathrm{i}q\theta}\psi
\end{equation}
 is gauge invariant.  If the reference phase is changed to $\theta'=\theta+\beta$, however, it is frame dependent:
\begin{equation}
  \Psi_{\theta'}=e^{-\mathrm{i}q\beta}\Psi_\theta .
\end{equation}
Thus $\Psi_\theta$ is a charged amplitude relative to a reference phase.  The density
\begin{equation}
  \Psi_\theta^\dagger\Psi_\theta=\psi^\dagger\psi
\end{equation}
 is frame independent.

In non-Abelian Yang--Mills theory, the corresponding object is an open Wilson line.  For a path $\gamma$ from $y$ to $x$, let
\begin{equation}
  U_\gamma[\mathscr{A}]
  =\mathrm{P}\exp\left(\mathrm{i}g_{\rm YM}\int_{\gamma:y \to x}\mathscr{A}\right) .
\end{equation}
The \textbf{frame-dressed open Wilson line} is given by
\begin{equation}
  W_h(x,y;\gamma)=h^{-1}(x)U_\gamma[\mathscr{A}]h(y) .
  \label{frame-dressed-open-line}
\end{equation}
It is invariant under the original local gauge transformations, but under $h\mapsto hk$ it transforms as
\begin{equation}
  W_{hk}(x,y;\gamma)=k^{-1}(x)W_h(x,y;\gamma)k(y) .
\end{equation}
Hence \textbf{an open Wilson line is a gauge-invariant relational observable, but not a frame-independent scalar.} It measures how two endpoint color frames are related.

By contrast, \textbf{a traced closed Wilson loop defined by
\begin{equation}
  W_R(C) :={1\over d_R}\mathrm{tr}_R\, \left\{ \mathrm{P}\exp\left(\mathrm{i}g_{\rm YM}\oint_C\mathscr{A}\right) \right\} ,
\end{equation}
 is both frame independent and gauge invariant} because the conjugation at the base point cancels inside the trace.

The physical interpretation is summarized in Table~\ref{tab:gauge-frame}.  The important point is not that frame-dependent relational observables are unphysical.  They are physical relative observables once the frame system is included.
The point is rather that \textbf{an absolute claim about color confinement must be invariant under admissible changes of frame.}
The notation in the last rows is intentionally more explicit than an unqualified $Q_h$.  The scalar pairing ${Q}_{\partial V}[h,\xi]$ is the fundamental frame-covariant object for a general boundary-dependent frame.  A Lie-algebra-valued integrated charge $Q_{\partial V,h}$ admits frame-independent Casimirs by a single conjugation only when the relevant boundary reorientation is constant (or when additional data first transport all boundary fibers to a common root).

For confinement this gives a useful hierarchy.  One may first construct colored relational quantities in a gauge-invariant way.  They usually remain frame dependent.  The quantities that survive as unambiguously frame-independent observables are color singlet contractions, invariant traces, Casimirs, closed line operators, or structural statements such as the vanishing of the bulk color charge by the Gauss law.  Thus QRF separates two claims that are often confused: color can be written relationally in a gauge-invariant form; color is not thereby an absolute observable.

\begin{widetext}
\begin{center}
\footnotesize
\begin{tblr}{colspec={Q[l,t,wd=0.35\textwidth] Q[l,t,wd=0.18\textwidth] Q[l,t,wd=0.42\textwidth]},colsep=2pt,rowsep=1pt}
\toprule
Quantity & Gauge invariance & Frame (in)dependence \\
\midrule
$\Psi_\theta=e^{-iq\theta}\psi$ & invariant & generally dependent \\
$\Psi_\theta^\dagger\Psi_\theta$ & invariant & independent \\
$\mathscr{F}^{h,A}_{\mu\nu}$ & invariant & generally dependent \\
$\operatorname{tr}(\mathscr{F}^h_{\mu\nu}\mathscr{F}^{h,\mu\nu})$ & invariant & independent \\
open Wilson line $W_h(x,y;\gamma)$ & invariant & endpoint-frame covariant \\
traced closed Wilson loop $W_R(C)$ & invariant & independent \\
component $Q_{\partial V,h}^{A}$ & invariant & generally dependent for a fixed numerical generator \\
fully paired $Q_{\partial V}[h,\xi]$ with $\xi\mapsto k^{-1}\xi k$ & invariant & independent under passive frame-coordinate changes \\
Casimirs of $Q_{\partial V,h}$ & invariant & independent under constant boundary reorientations \\
\bottomrule
\end{tblr}
\normalsize
\captionof{table}{Gauge invariance and frame (in)dependence.  A relational dressing first removes the original proper gauge redundancy; frame independence is an additional statement about changes of the reference frame.}
\label{tab:gauge-frame}
\end{center}
\end{widetext}

\section{QRF color confinement}
\label{sec:qrf-confinement}
Having presented the general QRF kinematics in Sec.~\ref{sec:qrf-framework}, we now formulate the color confinement in the QRF.

\subsection{Semi-local observables, charged sectors, and boundary data}
\label{subsec:charged-sectors}
Section~\ref{sec:qrf-framework} established the general QRF construction: a charged relational observable may be gauge invariant while remaining semi-local and frame relative.
\textbf{For confinement, the relevant question is therefore not whether such an operator can be written at finite separation, but whether its nontrivial charge class survives the infrared limit as a finite-energy, physically accessible, asymptotically isolated charged sector, or is excluded there.}
This requires separating dynamical charged sectors, external probes, boundary sectors, and neutral line-attached excitations.

Four physically different constructions must be distinguished.

\begin{enumerate}
 \item A \textbf{dynamical charged sector} is a sector of the theory's
 physical state space that survives the infinite-volume limit and is
 generated by the dynamical field content.

 \item An \textbf{external-probe sector} is introduced by a heavy source,
 an open Wilson-line endpoint, a probe worldline, or a prescribed
 electric-flux boundary condition. The Wilson-loop area law directly
 measures the energy of such a line-attached probe pair.

 \item A \textbf{boundary or edge sector} is labeled by a representation
 of a genuine boundary symmetry or by boundary electric flux. It
 exists only for boundary conditions that retain the corresponding
 boundary degree of freedom as physical data.

 \item A \textbf{neutral pair or line-attached excitation} may have a
 nonzero gauge-invariant correlation function without defining an
 isolated charged sector. A nonzero bilocal correlator therefore does  not by itself establish deconfinement.
\end{enumerate}

This classification explains why the Wilson-loop area law and color confinement are closely related but not identical.
\textbf{An area law for \(W_R(C)\) for a rectangular Wilson loop implies that the energy of an external charge--anticharge pair in the corresponding unscreened charge class grows with their separation}.
In theories with dynamical screening, string breaking can modify the asymptotic Wilson-loop diagnostic, so the existence or absence of a dynamical charged sector must be analyzed separately.

The QRF criterion asks the broader \textbf{question of whether any admissible relational construction realizes the charge class as a finite-energy, physically accessible, asymptotically isolated sector}.
The two statements should be compatible in a complete confinement mechanism, but their equivalence is not assumed in a general gauge theory with dynamical matter or nontrivial boundaries.

\begin{definition}[operational protocol and operational data]\label{def:operational-data}
Before the confinement criterion is applied, the operational data must
be fixed independently of the desired conclusion. Schematically, one
specifies the \textbf{operational protocol} based on the operational data for \textbf{a fixed boundary/asymptotic sector} $s$:
\begin{align}
 \mathcal P=\Bigl(&B,\mathcal G_0(B),G_{\partial}(B),s,
 \mathfrak D^{\rm adm}_{B,s},K^{\rm tw}_{B,s},\nonumber\\
 &\mathfrak F_{B,s}(h),\mathfrak A_{B,s},
 \mathfrak A_{\rm acc},\mathcal S_{\rm prep},
 \text{order of limits}\Bigr),
\end{align}
where 
\begin{itemize}
\item 
\(B\) denotes the \textbf{boundary condition}.

\item
\(\mathcal{G}_0(B)\) is the \textbf{proper gauge group}whose differentiable generators are constraints for the chosen boundary prescription; 
in the open-interval examples in IV this reduces to transformations that are trivial at the endpoints.

\item
\(G_{\partial}(B)\) is the \textbf{genuine boundary symmetry} retained by the chosen boundary condition $B$.

\item
\(\mathfrak{D}^{\mathrm{adm}}_{B,s}\) is the set of \textbf{admissible dressings} in the \textbf{boundary or defect sector} \(s\).

\item
\(K^{\mathrm{tw}}_{B,s}\) is the \textbf{subgroup of right-frame reorientations} whose orientation is operationally inaccessible.

\item
\(\mathfrak{F}_{B,s}(h)\) denote the \textbf{dressed field algebra}, i.e., {algebra generated by dressed relational fields} for a fixed admissible frame \(h\). 
It is invariant under the proper gauge group, but it may transform under an accessible boundary symmetry or under a change of right frame.

\item
$\mathfrak A_{\rm acc}$ is the accessible observable algebra.

\item 
$\mathcal S_{\rm prep}$ is the allowed preparation/resource class.

\item
$\mathfrak{A}_{B,s}$ is the \textbf{frame-independent observable algebra} in the fixed sector $s$, schematically,
\begin{align}
 \mathfrak A_{B,s}:=
 \bigl(\mathfrak F_{B,s}(h)\bigr)^{K^{\rm tw}_{B,s}}.
\end{align}
If the allowed boundary/asymptotic labels define genuine superselection sectors, the corresponding family may, when appropriate, be assembled into a direct-sum or direct-integral decomposition.  Dynamical topological-defect sectors are conceptually different: unless they themselves label superselection sectors, they are not included as algebraic direct summands but are summed separately with the path-integral or ensemble weights determined by the dynamics and boundary conditions.  Below, after the relevant boundary/asymptotic sector has been fixed (or a legitimate superselection decomposition has been formed), $\mathfrak A_B$ denotes the corresponding frame-independent observable algebra.

\end{itemize}
\end{definition}

\begin{definition}[Twirling]\label{def:twirling}
\textbf{Twirling} means the group average over right-frame reorientations that are operationally inaccessible or redundant in the chosen protocol.  For a compact subgroup $K_{\rm tw}$ acting continuously on the dressed $C^*$-algebra by automorphisms $\alpha_k$, the normalized Haar average
\begin{align}
\mathcal T(X)=\int_{K_{\rm tw}}d\mu(k)\,\alpha_k(X)
\end{align}
defines a conditional expectation onto the corresponding fixed-point algebra.  This algebraic average is distinct from a dynamical path-integral or ensemble sum over topological-defect sectors.
\textbf{Twirling  is performed only over right-frame orientations that are physically inaccessible in the chosen protocol}. 
One must not twirl over an accessible boundary frame, a physical edge-mode symmetry, or a Higgs frame that is part of the experimental preparation. 
Moreover, twirling over a group orbit must be distinguished from a dynamical sum over defect sectors. The latter is performed separately with the path-integral or ensemble weights determined by the action, fugacity, and boundary conditions.
The boundary therefore changes the possible realization of charge, but does not alter the primary question. 
\end{definition}

The Gauss law enforces invariance under the proper gauge group after the corresponding gauge quotient; it does not, by itself, require singletness under genuine boundary or asymptotic symmetries.
On a manifold with boundary, boundary electric-flux data or, more invariantly, representation labels of the retained boundary symmetry can label superselection sectors or edge-mode states.  In a non-Abelian theory, individual color components require a boundary frame and need not furnish mutually commuting superselection labels.
In either case, \textbf{confinement concerns whether a nontrivial unscreened color charge can survive or not as a finite-energy isolated sector after all boundary and accessibility rules have been specified.}
Thus boundary color is not an exception to relationality; it is one of its cleanest manifestations.

\subsection{Operational definition of QRF confinement and a sector-energy sufficient test}
\label{subsec:operational-qrf}

\begin{definition}[screening class]\label{def:screening-class}
A \textbf{screening class} is an equivalence class of charge labels that can be transformed into one another by the dynamical screening excitations allowed in the theory.  Thus the infrared charge label is not, in general, the bare representation $R$ itself.

A \textbf{stable screening class} $[R]$ denotes the residual class after all screening channels specified by the operational protocol have been taken into account.  For example, in pure $SU(N)$ Yang--Mills theory adjoint-gluon screening preserves $N$-ality, so the asymptotically stable external-probe classes are labeled by $q_R\in\mathbb Z_N$, whereas the zero-$N$-ality class is asymptotically trivial.  In the charge-$N$ Abelian Higgs model, $[q]=q+N\mathbb Z\in\mathbb Z_N$.

The \textbf{preparation/resource class} $\mathcal S_{\rm prep}$ is a separate operational datum: 
it specifies which states, matter or boundary frames, and other reference resources may actually be prepared or conditioned upon when constructing and reading relational charge observables.  It therefore distinguishes a formally definable dressing from a physically available QRF.  The QRF confinement criterion asks whether a nontrivial stable screening class survives, with the allowed resources and readouts fixed in advance, as a finite-energy, physically accessible, asymptotically isolated charged representation in the infinite-volume limit.

\end{definition}

Two logically different statements are separated in this subsection.

First, the \textbf{operational definition} asks whether a nontrivial \textbf{stable screening class} survives as a finite-energy physically accessible asymptotically isolated charged representation after the protocol has been fixed.

Second, the \textbf{sector-energy test} is a calculable sufficient exclusion test: \textbf{divergence of the optimized isolation energy} rules such a representation out under the \textbf{thermodynamic approximation} assumptions stated below.  The definition is not itself claimed to be a universal order parameter, and the sufficient test is not taken as its converse.

The general relational and field-space constructions needed below were separated out in Section~\ref{sec:qrf-framework}.  The confinement criterion begins only after those kinematical structures have been fixed.  In particular, neither the failure of one global gauge-fixing section nor the obstruction of one selected dressing or frame is identified with confinement.
\textbf{The additional question is whether a nontrivial stable screening class survives all admissibility, accessibility, boundary, defect-sector, and infrared requirements as a finite-energy charged representation.}

\begin{definition}[Admissible relational realization for the confinement test]\label{def:admissible-qrf}

Within a fixed operational protocol, a frame, dressing, state, or relational realization is called \textbf{admissible at a fixed regulator} (e.g., a finite volume) if it obeys the prescribed Gauss constraint, boundary and regularity conditions, has finite norm and the required finite action or Hamiltonian expectation value at that fixed regulator or finite volume,
is compatible with the specified defect, holonomy, and boundary sector, and
is not counted twice as a physically identical realization.
Admissibility at fixed regulator does \emph{not} assume a uniform infrared energy bound when the regulator is removed; that is an output of the sector-energy test.  A \textbf{long-distance color QRF} is an admissible frame that can stably compare the relevant non-singlet color labels in the infrared regime.

\end{definition}

\begin{definition}[Quasilocal algebras]\label{def:quasilocal-algebra}

The \textbf{quasilocal observable algebra} is the $C^*$-norm closure of the union of the observable algebras associated with bounded regions.
It provides a common infinite-volume algebra on which finite-volume states can be compared by local weak-* convergence.
A state or representation is \textbf{locally normal} with respect to the vacuum representation if, for every bounded region $\mathcal O$, its restriction to $\mathcal A(\mathcal O)$ extends to a normal functional or normal representation of the vacuum local von Neumann algebra $\pi_0(\mathcal A(\mathcal O))''$.
Local normality expresses local regularity;
it does not mean that the state has the same local expectation values as the vacuum.
Thus a charged representation may be locally normal on every bounded region while remaining globally inequivalent to the vacuum because of nontrivial boundary or asymptotic charge data.

\end{definition}

\begin{definition}[local weak-* limit]\label{def: weak-*-limit}

Let $\mathfrak A_B$ now denote the common \textbf{quasilocal} frame-independent observable algebra for the fixed boundary condition and operational protocol.
Choose normalized states
$|\Psi_{[R],L}\rangle\in\mathcal H^{\mathrm{phys}}_{[R]}(L)$ whose energies approach $E_{[R]}(L)$, and denote their \textbf{expectation-value functionals} by $\omega_{[R],L}$.

A state $\omega_{[R]}$ is a \textbf{local weak-* limit} of this sequence if, for every $A\in\mathfrak A_B$ supported in a fixed bounded region,
\begin{equation}
\omega_{[R]}(A)
=
\lim_{n\to\infty}\omega_{[R],L_n}(A)
\end{equation}
for some sequence $L_n\to\infty$. When only a \textbf{subnet} is available, the corresponding directed weak-* limit is understood.
The \textbf{Gelfand--Naimark--Segal (GNS) construction} applied to $\omega_{[R]}$ gives
\footnote{
The \textbf{Gelfand--Naimark--Segal (GNS) construction} \cite{GelfandNaimark1943,Segal1947,HaagBook,BratteliRobinson}
 turns a state $\omega$ on an observable
$C^*$-algebra $\mathfrak A$ into a Hilbert-space representation
$(\mathcal H_\omega,\pi_\omega,\Omega_\omega)$ satisfying
$\omega(A)=\langle\Omega_\omega,\pi_\omega(A)\Omega_\omega\rangle$.
It is the standard bridge from the algebraic description of observables
to physical Hilbert spaces and superselection sectors.}
\begin{equation}
\bigl(\pi_{[R]},\mathcal H_{[R]},|\Omega_{[R]}\rangle\bigr).
\end{equation}
The representation $\pi_{[R]}$, rather than the finite-volume subspace $\mathcal H^{\mathrm{phys}}_{[R]}(L)$, is the object relevant to an \textbf{infinite-volume superselection sector}.

In this paper, a {local weak-* limit} means pointwise convergence of the state functionals on every observable supported in any fixed bounded region: for each such local observable $A$, one requires $\omega_{L_n}(A)\to\omega(A)$ along a sequence $L_n\to\infty$, or along a subnet when sequential convergence is insufficient.
We additionally assume that the resulting local functionals are compatible under inclusions of bounded regions, remain positive and normalized, and extend continuously to the $C^*$-norm closure; under these conditions they define a state on the quasilocal algebra $\mathfrak A_B$ to which the GNS construction applies.
\footnote{
A \textbf{subnet} is the analogue of a subsequence for a net indexed by a
directed set: $(x_{\phi(\beta)})_{\beta\in E}$ is a subnet of
$(x_\alpha)_{\alpha\in D}$ when $\phi:E\to D$ is eventually cofinal.
Subnets are needed because compactness guarantees convergent subnets,
but not necessarily convergent subsequences.
}

\end{definition}

Before stating the criterion, three related but distinct uses of the word ``sector'' must be separated.

\noindent
\textbf{
At finite volume one has charge-sector Hilbert spaces, or charge subspaces when the sectors are realized inside a common physical Hilbert space}.

\noindent
$\bullet$ In the \textbf{thermodynamic limit (infinite volume limit)} one obtains states on a common \textbf{quasilocal observable algebra} by taking 
\textbf{local weak-* limits} on local observables.

\noindent
$\bullet$ An infinite-volume superselection sector is then described by a unitary-equivalence class (or, more generally, the appropriate quasi-equivalence/disjointness class) of \textbf{physically admissible representations of that observable algebra}, possibly enlarged by the accessible \textbf{boundary or asymptotic charge algebra}; a particular limiting state supplies such a representation through its
\textbf{Gelfand--Naimark--Segal (GNS) construction}\cite{GelfandNaimark1943,Segal1947,HaagBook,BratteliRobinson}.

Thus a finite-volume charge subspace, an infinite-volume limiting state, and a superselection sector are related notions, but they are not identical.

\medskip
\subsubsection*{Operational data and admissibility assumptions}

Fix \textbf{a boundary condition} $B$, \textbf{a boundary/asymptotic sector} $s$, and an operational protocol $\mathcal P$ of the type described above. 
The following data and rules are to be fixed independently of the desired confinement conclusion.
\begin{enumerate}
\item
The Gauss law is imposed for the proper gauge group $\mathcal{G}_0(B)$.

\item
The dressed field algebra $\mathfrak F_{B,s}(h)$ and the frame-independent observable algebra $\mathfrak A_B$ are specified. The admissible class may include Wilson-line or Coulomb dressings, matter- or boundary-defined frames, and relational or horizontal representatives constructed with the aid of a field-space connection. Such a construction is not, by itself, assumed to define an independent physical QRF; its physical status is determined separately by the boundary conditions, accessibility, fixed-cutoff finiteness requirements, and the operational protocol.

\item Proper gauge transformations, genuine boundary or asymptotic symmetries, right-frame reorientations, dressing changes, and custodial or flavor symmetries are identified separately.

\item
The accessibility structure of the protocol is recorded explicitly by an \textbf{accessible observable algebra} $\mathfrak A_{\rm acc}\subseteq\mathfrak A_B$, an \textbf{allowed preparation/resource class} $\mathcal S_{\rm prep}$, and the \textbf{compact subgroup $K^{\rm tw}_{B,s}$ of right-frame reorientations} whose orientation is not operationally available.  ``Physically accessible'' below means accessible with respect to these protocol data, rather than merely formally definable.

\item
Twirling is performed only over $K^{\rm tw}_{B,s}$.  For the continuous action of this compact group on the dressed $C^*$-algebra, the normalized Haar average is the protocol-specified conditional expectation onto the corresponding fixed-point algebra.  Dynamical topological-defect sectors are not right-frame orientations and are summed separately with their physical path-integral or ensemble weights when required by the dynamics.

\item The finite-volume charged subspaces, the order of infrared limits, and the criterion for obtaining an infinite-volume finite-energy representation are specified.

\item When dynamical screening is present, the charge label is a \textbf{stable screening class} $[R]$, rather than a \textbf{bare representation} $R$.  The equivalence relation defining $[R]$ is fixed by the dynamical screening channels and by the relevant accessible charge algebra of the model.  For external probe charges in pure $SU(N)$ Yang--Mills theory, this quotient is organized asymptotically by the center charge, equivalently the $N$-ality; for boundary sectors it is a representation class of the retained boundary/asymptotic symmetry algebra; and with dynamical matter it must include matter-induced screening.
\end{enumerate}
Here ``admissible at finite volume'' means only that the state, dressing, or frame satisfies the stated boundary/regularity conditions and has finite norm, action, or Hamiltonian expectation value at each fixed regulator $L$.  It does \emph{not} include the infrared conclusion
\begin{equation}
\limsup_{L\to\infty}\Delta E_{[R]}(L)<\infty.
\end{equation}
The latter is an output of the sector-energy test below, not an input to admissibility.  This separation prevents the definition of the admissible class from assuming deconfinement through a uniform infrared energy bound.

\medskip
\begin{definition}[Finite-volume sector energy]\label{def:sector-energy}

Let $H_L$ be a \textbf{finite-volume Hamiltonian} after imposing the Gauss law for $\mathcal{G}_0(B)$.  Let $\mathcal H^{\mathrm{phys}}_{[R]}(L)$ denote the finite-volume physical Hilbert space for the sector $[R]$, or a physical subspace of a common Hilbert space when such a realization is available, with the boundary condition, source prescription, admissible dressings, topological-defect prescription, and accessibility rules held fixed.  
Define the \textbf{finite-volume sector energy} in the stable screening class $[R]$ by
\begin{equation}
E_{[R]}(L)
:=
\inf_{\substack{|\Psi\rangle\in\mathcal H^{\mathrm{phys}}_{[R]}(L)\\
\langle\Psi|\Psi\rangle=1}}
\langle\Psi|H_L|\Psi\rangle,
\end{equation}
and the excited energy or the excess energy as the difference
\begin{equation}
\Delta E_{[R]}(L):=E_{[R]}(L)-E_{[1]}(L).
\end{equation}
The infimum is taken over all admissible states and dressings realizing the specified charge class. This minimization is essential: the divergent energy of one artificially thin or otherwise nonminimal dressing does not by itself establish confinement.
\end{definition}

\begin{equation}
\Delta E_{[R]}(L)\longrightarrow +\infty \qquad (L\to\infty)
\end{equation}
excludes a finite-energy charged representation in the class $[R]$.  By contrast, boundedness,
\begin{equation}
\limsup_{L\to\infty}\Delta E_{[R]}(L)<\infty,
\label{finite-energy}
\end{equation}
is only a necessary energetic condition and does not by itself establish deconfinement; a local weak-* limit (or an appropriate subnet), a physically admissible GNS representation, and survival of the nontrivial charge information must also be established.

A bounded sequence of finite-volume excitation energies (\ref{finite-energy})
is an important energetic requirement for deconfinement, but it is not by itself sufficient to establish the existence of a nontrivial infinite-volume charged sector.
\textbf{The charge or its compensating flux can escape to infinity so that the local limit becomes the vacuum, or the putative charge label can fail to survive the thermodynamic limit.
This point is why the infinite-volume representation must be treated separately.}

When the trace over the corresponding finite-volume subspace exists, the same lowest-energy difference can be extracted from the $T\to\infty$ limit
\begin{equation}
\Delta E_{[R]}(L)
=
-\lim_{T\to\infty}\frac{1}{T}\log K_{[R]}(T,L),
\end{equation}
where $K_{[R]}(T,L)$ is defined by
\begin{equation}
K_{[R]}(T,L):=
\frac{Z_{[R]}(T,L)}{Z_{[1]}(T,L)},
\end{equation}
using the partition function $Z_{[R]}(T,L)$ in the class $[R]$ with the finite time $T$ and the spatial size $L$:
\begin{equation}
Z_{[R]}(T,L)
:=
\operatorname{Tr}_{\mathcal H^{\mathrm{phys}}_{[R]}(L)}
\exp(-T H_L),
\end{equation}
provided that the $T\to\infty$ limit exists.

\medskip
\begin{proposition}[protocol monotonicity at fixed charge problem]
Let $\mathcal P_1$ and $\mathcal P_2$ have the same boundary/asymptotic charge algebra, the same stable-screening equivalence relation, and the same infrared limiting prescription.  Write $\mathcal P_1\preceq\mathcal P_2$ when every relational realization, preparation resource, and readout allowed by $\mathcal P_1$ is also allowed by $\mathcal P_2$; in particular,
\begin{equation}
 \mathfrak A_{\rm acc}^{(1)}\subseteq \mathfrak A_{\rm acc}^{(2)},
 \qquad
 \mathcal S_{\rm prep}^{(1)}\subseteq \mathcal S_{\rm prep}^{(2)},
\end{equation}
and the admissible dressing/frame class of $\mathcal P_1$ is a subset of that of $\mathcal P_2$.  Then, for the same stable screening class $[R]$,
\begin{equation}
 [R]\ {\rm is}\ \mathcal P_2\text{-QRF-confined}
 \ \Longrightarrow\
 [R]\ {\rm is}\ \mathcal P_1\text{-QRF-confined},
\end{equation}
whereas
\begin{equation}
 [R]\ {\rm is}\ \mathcal P_1\text{-QRF-deconfined}
 \ \Longrightarrow\
 [R]\ {\rm is}\ \mathcal P_2\text{-QRF-deconfined}.
\end{equation}
The implication is set-theoretic: enlarging the admissible/accessibility resources can create additional candidate charged realizations but cannot remove one that was already available.
Conversely, confinement that survives the more permissive protocol is robust under restriction.
Changing the boundary charge algebra, the screening equivalence relation, or the infrared sector problem itself is not treated as a monotone enlargement of the same protocol; it defines a different physical charge problem.
\end{proposition}

\medskip
\subsubsection*{Infinite-volume limit and charged representation}

In the present operational sense, $\pi_{[R]}$ is a \textbf{charged representation} only if the nontrivial stable screening class $[R]\neq[1]$ survives in the action of the relevant accessible boundary or asymptotic charge algebra and the limiting representation is not merely the vacuum representation produced by moving the charge, its compensating endpoint, or its entire distinguishing flux distribution out of every bounded region. 
The charged representation may nevertheless be locally normal with respect to the vacuum representation.  Local normality is a statement of local regularity: on every bounded-region algebra the state belongs to the vacuum folium, or equivalently extends normally in the corresponding vacuum local von Neumann algebra.  It does not imply equality of local expectation values with the vacuum.  Thus local normality is compatible with global inequivalence and with nontrivial asymptotic charge data.
The finite-volume energy bound and the survival of a nontrivial charged representation must therefore be stated as logically separate requirements.

The following is not a theorem derived from bounded sector energy alone.  Rather, it records the additional thermodynamic-limit and charge-survival assumptions that must be verified before a bounded finite-volume sector energy can be interpreted as a finite-energy charged representation.
\begin{proposition}[Operational sufficient-condition checklist for a finite-energy charged representation]
Suppose that the finite-volume states can be compared as states on the common quasilocal algebra $\mathfrak A_B$, and suppose that there exists a sequence $L_n\to\infty$ such that:
\begin{enumerate}
\item the states $\omega_{[R],L_n}$ have a local weak-* limit $\omega_{[R]}$;
\item the excitation-energy functional is \textbf{lower semicontinuous}
\footnote{
A functional $E$ is \textbf{lower semicontinuous} with respect to the chosen convergence if
$\omega_\alpha\to\omega$ implies
$E(\omega)\leq\liminf_\alpha E(\omega_\alpha)$.
Hence a uniform upper bound on the approximating excitation energies passes to the limit.
Here this property is an assumption on the thermodynamic-limit energy functional, not a consequence of weak-* convergence alone.
}
 in the relevant thermodynamic limit, so that a uniform finite-volume excitation-energy bound yields finite excitation energy in the limiting representation;
\item survival requirement: a nontrivial accessible boundary or asymptotic charge criterion survives the limit and identifies the stable screening class $[R]\neq[1]$; and
\item non-escape requirement: the limiting GNS representation is not the vacuum representation obtained by moving the charged excitation, its compensating endpoint, or all of the flux carrying the distinguishing charge information out of every bounded region.
\end{enumerate}
If these independent assumptions hold, the uniform bound
\begin{equation}
\limsup_{n\to\infty}\Delta E_{[R]}(L_n)<\infty
\end{equation}
supplies the energetic part of a finite-energy charged GNS representation carrying $[R]$.
Items 3 and 4 are separate survival and non-escape requirements; consequently this checklist must not be read as a deconfinement theorem following from the energy bound by itself.
\end{proposition}

The qualifications are essential.
\textbf{Without the survival and non-escape assumptions, bounded finite-volume energy alone supplies only a necessary energetic condition.}
\textbf{It is not by itself a sufficient criterion for deconfinement.}

\begin{proposition}[{ Sector-energy sufficient exclusion test}]\label{prop:sector-energy-exclusion-test}
Fix an operational protocol and a nontrivial stable charge class $[R]\neq[1]$.  The exclusion statement uses the following assumptions:
\begin{enumerate}
\item[(A1)] \textbf{Thermodynamic approximability.}  If a physically admissible finite-energy charged representation carrying $[R]$ exists, it can be approximated on the common quasilocal algebra by finite-volume physical states in the same stable class whose excess energies are uniformly bounded.
\item[(A2)] \textbf{Lower semicontinuity.}  The thermodynamic-limit excitation-energy functional is lower semicontinuous with respect to the chosen local weak-* convergence, so that the finite-energy control used in the approximation is compatible with the limiting representation.
\item[(A3)] \textbf{Charge survival/non-escape.}  The distinguishing charge or flux information survives in the accessible boundary/asymptotic charge algebra and is not lost by moving the charge, its compensating endpoint, or all distinguishing flux outside every bounded region.
\end{enumerate}
If, subject to (A1)--(A3), every admissible sequence of finite-volume physical states that retains the nontrivial charge information satisfies
\begin{equation}
\Delta E_{[R]}(L)\longrightarrow+\infty
\qquad(L\to\infty),
\end{equation}
then no finite-energy asymptotically isolated charged GNS representation carrying $[R]$ can arise as a thermodynamic limit of those finite-volume sectors.
\end{proposition}
Indeed, if such a charged representation existed, (A1) would provide an admissible approximating sequence with uniformly bounded excess energy, while (A3) would ensure that this sequence retains the nontrivial charge information; this contradicts the stated divergence.  Assumption (A2) controls the passage of the energy functional to the limit.  The statement is therefore a sufficient exclusion test only; its converse is not asserted.

\begin{definition}[
Operational definition of $\mathcal P$-QRF confinement/deconfinement
]
Fix the operational protocol $\mathcal P$ specified above.  A nontrivial stable screening class $[R]\neq[1]$ is \textbf{$\mathcal P$-QRF-deconfined} if there exists an admissible sequence of finite-volume physical states, obeying the rules of $\mathcal P$, whose local weak-* limit has a GNS representation that
\begin{enumerate}
\item carries the nontrivial class $[R]$ under the relevant accessible boundary or asymptotic charge algebra;
\item has finite excitation energy above the vacuum sector; and
\item is asymptotically isolated, in the sense that it is not unavoidably tied to an opposite charge, a fixed boundary endpoint, or an externally supplied charged reference frame.
\end{enumerate}
The class $[R]$ is \textbf{$\mathcal P$-QRF-confined} if no such infinite-volume charged representation exists after all Gauss-constraint, boundary, dressing, screening, defect-sector, frame-accessibility, and infrared-limit rules contained in $\mathcal P$ have been imposed.
\end{definition}
Terminology will be kept distinct in what follows: the preceding statement is the \emph{operational definition} of $\mathcal P$-QRF confinement/deconfinement; Proposition~\ref{prop:sector-energy-exclusion-test} is a one-way \emph{sector-energy sufficient exclusion test}; Wilson-loop behavior, isolated-pole tests, relational frame correlators, and defect observables are \emph{diagnostics}.  These terms are not used as synonyms, and none of the diagnostics is asserted to be a universal necessary-and-sufficient order parameter.

\begin{definition}[intrinsic protocol class and intrinsic QRF confinement]
Let $\mathfrak P_{\rm int}$ be a specified class of admissible protocols whose reference resources are internal to the theory in the following operational sense.  The class may contain matter frames that can be prepared from finite-energy states of the theory, boundary frames retained as genuine physical data by the chosen boundary condition, gauge-field dressings or field-space relational representatives satisfying the fixed-regulator regularity and norm requirements, and defect-sector sums with weights determined by the action and boundary data.  It excludes an externally fixed colored reference frame at infinity when maintaining that frame requires an external infinite-energy resource or when that orientation is not part of the physical boundary system under study.  Resource admissibility is therefore not an energy condition alone: it also includes finite norm at the regulator, the prescribed regularity and boundary conditions, compatibility with the accessible algebra, and operational preparability or conditioning when a physical QRF is claimed.

Relative to this explicitly specified class, a nontrivial charge label is called \textbf{intrinsically QRF-confined} if it is $\mathcal P$-QRF-confined for every $\mathcal P\in\mathfrak P_{\rm int}$.  Any additional robustness requirement under deformations of the protocol should be imposed only after specifying what counts as an admissible deformation and which boundary/asymptotic charge problem is held fixed.  No absolute protocol-independent phase invariant is claimed without specifying this resource class.

\end{definition}

Here ``isolated'' is primarily a sector-theoretic and energetic notion. It means that the nontrivial charge class survives the infinite-volume limit without being unavoidably tied to an opposite charge, a fixed boundary endpoint, or an externally supplied charged reference frame.  It does \emph{not} exclude the asymptotic gauge-field flux or soft cloud required by Gauss' law for a genuine long-range charged state.  Accordingly, the definition is not restricted to strictly localizable Doplicher--Haag--Roberts (DHR)\cite{DHR}-type sectors
\footnote{
A DHR sector is a unitary-equivalence class of transportable representations of the quasilocal observable algebra satisfying the DHR localization criterion: for any bounded double cone $\mathcal O$, an equivalent representative can be chosen such that
$\pi|{\mathfrak A(\mathcal O')} \simeq \pi_0|{\mathfrak A(\mathcal O')}$.
Thus the charge is localizable inside $\mathcal O$ and is indistinguishable from the vacuum outside it. Long-range Coulomb charges generally do not satisfy this strict localization condition \cite{DHR}.
}
and allows, when appropriate, Coulombic or infraparticle-like charged representations.  Nor does ``isolated'' mean, in the general definition, that the spectrum contains an isolated
Lehmann--Symanzik--Zimmermann (LSZ) one-particle pole \cite{LehmannSymanzikZimmermann1955}.

Thus the \textbf{sector-energy sufficient exclusion test} should be interpreted as follows.
\textbf{Divergence of the optimized isolation energy for every admissible sequence that retains the nontrivial class is a sufficient energetic mechanism excluding the charged sector within the stated protocol.}
\textbf{By contrast, bounded $\Delta E_{[R]}$ is not by itself a complete deconfinement criterion; it must be accompanied by survival of the nontrivial charge representation in the infinite-volume limit.}

\medskip
\subsubsection*{Pole diagnostic when a stable particle interpretation is available}
Suppose, in addition, that the infrared theory is gapped, admits an ordinary particle spectral representation in the sector under consideration, and contains a stable finite-mass excitation carrying $[R]$. Let $O^j_{h,[R]}$ be admissible dressed relational fields that intertwine the vacuum sector and the charge sector $[R]$, with $j=1,\ldots,d_R$ an index in a representative $R$ of the screening class. Consider the invariantly contracted correlator
\begin{equation}
C_{[R]}(x-y)
=
\frac{1}{d_R}
\sum_{j=1}^{d_R}
\left\langle
O^j_{h,[R]}(x)
O^{j\dagger}_{h,[R]}(y)
\right\rangle_{\mathrm{phys}},
\end{equation}
where the expectation value includes only the twirling and defect-sector operations specified by the operational protocol.

If such an operator has nonzero overlap with a stable one-particle state of mass $M_{[R]}$, its momentum-space two-point function contains an isolated pole of the following form:
Using Euclidean momentum $p_E$, the corresponding statement is
\begin{equation}
\begin{aligned}
C_{[R]}(p_E)
&=\frac{Z_{[R]}}{p_E^2+M_{[R]}^2}
+\text{terms analytic near }p_E^2=-M_{[R]}^2,\\
&\qquad Z_{[R]}>0,\quad M_{[R]}<\infty,
\end{aligned}
\end{equation}
where the pole at $p_E^2=-M_{[R]}^2$ is understood after analytic continuation of the Euclidean correlator.  Equivalently, in Minkowski signature one would write $iZ_{[R]}/(p^2-M_{[R]}^2+i0)$.

Conversely, an isolated pole with nonzero residue in an admissible physically accessible operator that carries the nontrivial charge information is evidence for a finite-energy charged excitation in the corresponding sector.

The pole condition is only a diagnostic and is not the universal definition. In particular, the absence of a pole in one chosen correlator does not establish confinement: the operator may simply have zero overlap with the relevant state, or the lowest charged excitation may not be a stable one-particle state.
Only under an additional completeness assumption on the class of admissible charged interpolating operators could a systematic absence of such poles be promoted to a spectral test.

Moreover, in a Coulomb phase a charged excitation may be an infraparticle accompanied by a long-range soft gauge-field cloud \cite{KulishFaddeev,FMS,Buchholz,YFS,ChenFroehlichPizzo}. It can define a finite-energy, physically accessible charged sector even though it has no isolated sharp-mass LSZ pole. The primary question is therefore the existence of the charged representation, not the existence of a one-particle pole.

\medskip
\subsubsection*{Twirling, Higgs regimes, and the interpretation of correlators}
Twirling may remove a single operator carrying an open right-frame index from $\mathfrak A_B$, while the invariant contraction of a two-point function such as the one above remains nonzero. This algebraic projection is not itself confinement. The surviving correlator may describe a finite-energy charged sector, a neutral charge--anticharge pair, a line endpoint attached to a boundary, or a correlation conditioned on an external frame. Its physical interpretation is decided by the charge algebra, sector-energy, accessibility, and asymptotic-isolation conditions.

In a Higgs regime, a scalar configuration, a boundary condition, or an external reference system may make a frame physically accessible. The corresponding orientation must then not be twirled merely because it can be described as a right-frame transformation. Gauge-invariant composites may have ordinary particle poles, but the observable quantum numbers carried by such poles may be custodial or flavor quantum numbers rather than gauge color. The relevant observable, boundary, or asymptotic charge algebra must therefore be identified before the QRF confinement criterion is applied.

\medskip
\subsubsection*{Global, Gribov, and topological obstructions}
A Gribov obstruction can prevent the existence of a global gauge-fixing section or a globally integrable dressing, while a field-space connection with non-vanishing curvature may nevertheless provide a globally meaningful geometric description \cite{GomesRiello:2016}. Likewise, a vortex, monopole, nontrivial bundle, or nontrivial holonomy may obstruct a particular global frame or trivialization. None of these facts alone is sufficient for color confinement.

\textbf{Such obstructions become relevant to the present criterion only when they couple nontrivially to the charge class under consideration and, after minimization over all admissible relational realizations and the appropriate dynamical treatment of defect sectors, are accompanied by an infrared mechanism---for example infrared disorder, defect condensation, an area law, a divergent isolation energy, or another dynamical effect---that excludes every finite-energy, physically accessible, asymptotically isolated realization of the nontrivial charge class.}

The \textbf{novelty of the operational criterion} is therefore not the geometric use of field-space connections, dressings, boundaries, or reference-frame data themselves.
It is \textbf{the use of the infrared survival or absence of an admissible charged relational representation as the confinement question}.
Schematically,
\begin{align}
\text{field-space or dressing obstruction}
\not\Rightarrow
 \text{color confinement},
\nonumber
\end{align}
whereas, within the fixed operational protocol and under the assumptions entering the definition,
\begin{align}
&\text{QRF confinement}\nonumber\\
&\Longleftrightarrow
\text{absence of an admissible finite-energy,}\nonumber\\
&\text{asymptotically isolated, nontrivial charged representation}
\nonumber
\end{align}

\medskip
Consider noncompact $U(1)$ gauge theory coupled to a massive charged field, so that a genuine charged representation is present.  A Dirac/Coulomb dressed charged field carries nonzero asymptotic electric flux \cite{Dirac:1955uv}.  The same long-distance field-energy estimate may also be used for a heavy external charged probe, but in that case it is an external-probe benchmark rather than evidence for a dynamical charged sector:
\begin{equation}
 \lim_{R\to\infty}\int_{S_R^2}\mathbf E\cdot d\mathbf S=q,
\end{equation}
while its long-distance Coulomb field behaves as
\begin{equation}
 \mathbf E(\mathbf r)=\frac{q}{4\pi r^2}\widehat{\mathbf r}.
\end{equation}
The contribution of the field outside a large sphere of radius $R$ is therefore
\begin{equation}
 E_{\rm IR}(R)
 =\frac12\int_{|\mathbf r|>R}d^3r\,\mathbf E^2
 =\frac{q^2}{8\pi R},
\end{equation}
which is finite and decreases as the infrared separation scale $R$ is increased.  This infrared statement is distinct from the ultraviolet self-energy of an ideal point charge.  The latter must be regulated and absorbed into the usual renormalized charged-state energy (or avoided by using a finite-size heavy source); it is not an infrared confinement energy.

Soft photons imply that an exact charged QED sector is infraparticle-like and need not contain an isolated LSZ pole \cite{KulishFaddeev,Buchholz,YFS,ChenFroehlichPizzo}.  Nevertheless, after the ultraviolet renormalization just stated, the sector can have finite renormalized excitation energy and a nontrivial asymptotic electric-flux label.  Thus the standard Coulomb protocol is $\mathcal P$-QRF-deconfined in the charged-representation sense used here.  This positive control shows explicitly that imposing Gauss' law does not by itself imply confinement and that the primary criterion cannot require a sharp one-particle pole.

Further algebraic details of the operational criterion are given in Appendix~\ref{app:operational-qrf-revised}, and the exactly calculable fixed-frame correlator is given in Appendix~\ref{app:example_QRF}.

\subsection{Relation with the Kugo--Ojima criterion}
\label{subsec:kugo-ojima-qrf}
In the quantum non-Abelian gauge theory
in the manifestly Lorentz covariant gauge, Kugo and Ojima \cite{Kugo:1979gm,Kugo:1995km,NakanishiOjima} decompose the \textbf{global color Noether current} $J_{\mu}^{A}$ by using the Yang--Mills equation of motion as
\begin{align}
 J_{\mu}^{A}
 =
 \partial^{\nu}\mathscr{F}_{\mu\nu}^{A}
 +
 \left\{Q_{\mathrm{BRST}},
 \left(\mathscr{D}_{\mu}[\mathscr{A}]\overline{\mathscr{C}}\right)^{A}
 \right\}.
\end{align}
Here $Q_{\mathrm{BRST}}$ is the \textbf{BRST charge}, $\mathscr{D}_{\mu}[\mathscr{A}]:=\partial_\mu -ig_{\rm YM}[\mathscr{A}_\mu, ]$ is the covariant derivative in the adjoint representation and $\overline{\mathscr{C}}$ is the Faddeev--Popov anti-ghost field.
Then the \textbf{global color charge} $Q^{A}$ is given by
\begin{align}
 Q^{A}:=&\int_{\mathbb{R}^{3}}d^{3}x\, J_{0}^{A}(x)
 = G^{A}+N^{A}, \quad
 \nonumber\\
 G^{A}
 :=& {}\int_{\mathbb{R}^{3}}d^{3}x\,
 \partial^{j}\mathscr{F}_{0j}^{A},
 \nonumber\\
 N^{A}
 :=& {}\left\{Q_{\mathrm{BRST}}, \int_{\mathbb{R}^{3}}d^{3}x\,
 \left(\mathscr{D}_{0}[\mathscr{A}]\overline{\mathscr{C}}\right)^{A}
 \right\} .
\end{align}

The \textbf{Kugo--Ojima color confinement scenario} consists of the two steps:
\begin{enumerate}
\item[(i)]
$G^{A}=0$:
$G^{A}$ at spatial infinity  $R\to\infty$ vanishes,
if there is no massless vector pole in $\mathscr{F}_{\mu\nu}^{A}$.

\item[(ii)]
The condition
for the remaining $N^{A}$ to be a well-defined BRST-exact charge in Landau gauge is the well-known \textbf{Kugo--Ojima criterion}:
In the Landau gauge Kugo-Ojima function $u^{AB}(p^2)$ defined by a particular ghost--gluon correlation function has the following infrared limit $p^2\to 0$:
\begin{align}
 \tilde u^{AB}(p^2=0)=-\delta^{AB}.
\label{ko-condition}
\end{align}
\end{enumerate}
As a result, the global color charge \(Q^A\) acts trivially on the physical BRST cohomology
$
\mathcal{H}_{\mathrm{phys}}
=
\ker Q_{\mathrm{BRST}}/\operatorname{im}Q_{\mathrm{BRST}}
$
($Q_{\mathrm{BRST}} |\mathrm{phys}\rangle=0$):
\begin{align}
 \langle \mathrm{phys}^\prime | Q^{A}|\mathrm{phys}\rangle=0.
\end{align}
However, it should be remarked that the quantity $G^{A}$ is precisely the \textbf{chromoelectric flux} $\mathscr{E}_{j}^{A}$ at spatial infinity $R\to\infty$:
\begin{align}
 G^{A}=
 \lim_{R\to\infty}
 \int_{S_{R}^{2}}d^2S_{j}\,\mathscr{F}_{0j}^{A}
 = \lim_{R\to\infty}
 \int_{S_{R}^{2}}d^2S_{j}\,\mathscr{E}_{j}^{A} .
\end{align}
This construction directly addresses the fate of color charge in the physical state space, but its standard formulation presupposes, among other things, a globally defined nilpotent BRST charge, a well-defined global color charge, and the spectral assumptions entering the Kugo--Ojima construction.

\textbf{Finite-volume systems and systems with physical boundaries are not directly covered by the standard infinite-volume Kugo--Ojima derivation without specifying additional boundary and BRST structures.}

Therefore, it is not so immediate to compare the Kugo-Ojima criterion with the results of numerical simulations performed on a finite lattice (with periodic boundary conditions), apart from the gauge-fixing ambiguity due to Gribov copies and the absence of the standard BRST symmetry on the lattice.

An important geometric and conceptual precursor to the boundary-sensitive and reference-frame viewpoint adopted in this paper is the field-space framework developed by Gomes and Riello in \textit{The Observer's Ghost} \cite{GomesRiello:2016}.  In that framework, the space of fields is treated geometrically as an infinite-dimensional principal-bundle-type structure over gauge-orbit space, and a connection one-form on field space is introduced in order to define a functional covariant, or horizontal, decomposition.  This construction relates the geometry of field space to BRST symmetry, while the Gribov problem motivates allowing a generally non-flat field-space connection when a global treatment is required.  The same framework was also applied to bounded regions and was interpreted as encoding the structure associated with a non-asymptotic observer or reference system.  The later work of Gomes, Hopfm\"uller and Riello \cite{Gomes:2018shn} developed the boundary-charge and dressing aspects further and emphasized that the original gauge and matter fields themselves can provide the relevant reference-frame or dressing data, so that additional group-valued edge modes need not be introduced as independent physical degrees of freedom.

Thus, the conceptual passage from BRST and Gribov issues to boundaries, dressings, and relational or reference-frame structures is not claimed to originate in the present work.  We regard the field-space construction of Refs.~\cite{GomesRiello:2016,Gomes:2018shn} as an important geometric precursor to this aspect of the QRF viewpoint.  The new step pursued here is different in scope: we use the existence or obstruction of a \emph{globally defined, physically admissible, long-distance color frame} as an operational diagnostic of color confinement, and relate it to the \emph{existence of finite-energy isolated non-singlet relational sectors}.

The QRF formulation therefore shares the physical objective of the Kugo--Ojima construction: nontrivial color should not occur as a physical asymptotic charge in a confining phase.
Its primary criterion, however, is not a condition on a particular gauge-fixed ghost--gluon two-point Kugo-Ojima function.
Rather, it asks whether, after imposing the proper Gauss constraint and the boundary conditions, and after specifying admissible dressings, defect sectors, dynamical screening, and reference-frame accessibility, \textbf{a nontrivial color label can survive as a finite-energy, physically accessible, and asymptotically isolated relational sector.}

Accordingly, if a globally defined nilpotent BRST charge, a well-defined global color charge, the absence of the relevant massless contribution, and the other assumptions of the Kugo--Ojima construction all hold, the Kugo--Ojima analysis provides a gauge-fixed route toward the absence of colored asymptotic states in the BRST cohomology.  The Landau-gauge condition on $\widetilde u(0)$ is discussed only within that framework; it is not used here as a gauge-independent universal sufficient criterion.  No equivalence or converse theorem with the QRF criterion is assumed.  The latter is intended also for Hamiltonian, algebraic, lattice, and boundary formulations in which a particular global nonperturbative BRST realization is not part of the definition of confinement.

\noindent

This formulation makes explicit both the continuity with and the distinction from the Kugo--Ojima picture.
It retains the idea that color must disappear from the physical asymptotic spectrum, but it \textbf{formulates the charged-sector question in gauge-invariant and boundary-sensitive terms} rather than defining confinement by a specific infrared ghost criterion.
It also \textbf{separates proper bulk gauge redundancy from genuine boundary or asymptotic symmetries}: the Gauss constraint removes the former, whereas an allowed boundary charge may act nontrivially on physical states, with its non-Abelian component defined only relative to boundary or asymptotic frame data.
The role of the QRF criterion is therefore not to replace the Kugo--Ojima quartet mechanism when the latter is well defined, but to state the confinement question in a form that remains meaningful when global BRST assumptions, boundary conditions, or nonperturbative global obstructions become essential.

\subsection{Relation to other confinement diagnostics}
\label{subsec:qrf-comparison}
It is useful to summarize the logical difference between the standard viewpoints and QRF.
See Table~\ref{tab:viewpoints} below.
This table should not be read as a replacement of the other criteria.  Rather, the QRF formulation organizes them.  The \textbf{Wilson loop} tests the energy of a line-attached pair.  \textbf{Generalized symmetries} classify line operators and their infrared realization.
The \textbf{Kugo--Ojima criterion} emphasizes the {disappearance of color charge from the physical Hilbert space. }
\textbf{QRF language} explains why these statements are connected by asking whether a nontrivial unscreened color class survives, after all admissible relational realizations and infrared limits have been taken into account, as a finite-energy physically accessible isolated charged sector.  Global frame obstruction can contribute to the failure of such a sector, but is not required in every confining realization.

\begin{widetext}

\begin{center}
\footnotesize\begin{tblr}{colspec={Q[l,t,wd=0.15\textwidth] Q[l,t,wd=0.38\textwidth] Q[l,t,wd=0.42\textwidth]},colsep=2pt,rowsep=1pt}
\toprule
Viewpoint & Main diagnostic & Limitation addressed by QRF \\
\midrule
Wilson-loop criterion & Area law for the Wilson loop average $\langle W_R(C)\rangle$ and the linear static potential & Directly probes external sources, but does not by itself define the fate of all colored relational excitations. \\
Kugo--Ojima criterion & Trivial action of global color charge in BRST cohomology, often expressed by $\tilde u(p^2=0)=-1$ & Depends on covariant gauge, global BRST assumptions, and infrared ghost behavior; boundary data are not central. \\
Generalized-symmetry criterion & Realization of center one-form symmetry and behavior of genuine line operators & Powerful for line operators, but not automatically a criterion for relational non-singlet particle excitations or color frames. \\
QRF criterion & Existence of a finite-energy, physically accessible, asymptotically isolated nontrivial charged representation after optimization over admissible relational realizations and after the prescribed infrared limits & Specifies how the charge label is relationally realized and which frame, boundary, or asymptotic data are physically accessible; obstruction of a selected frame is a possible mechanism or diagnostic, not the defining criterion. \\
\bottomrule
\end{tblr}
\normalsize
\captionof{table}{The logical difference between the standard viewpoints and QRF.}
\label{tab:viewpoints}
\end{center}
\end{widetext}

\section{(1+1)-dimensional Yang--Mills theory as a minimal QRF test bed}
\label{sec:ym2-qrf}
\noindent\textit{QRF type(s) in this section:} mathematical Wilson-line dressing together with boundary/base-point frame data; the Wilson line is not assumed to be an independent quantum reference system.\par

Pure Yang--Mills theory in $1+1$ dimensions is the cleanest continuum example of the QRF confinement criterion.
After imposing the Gauss law, the physical information is compressed into global holonomy, boundary color electric flux, and Wilson-line dressings.
This makes it possible to see explicitly which variables are redundant, which variables remain in the physical algebra, and why a bare colored local operator does not descend to the physical Hilbert space.

This model is not a miniature copy of four-dimensional confinement.
This theory has no locally propagating degrees of freedom,
no transverse gluons, no magnetic monopoles, no center-vortex world sheets, and no fluctuating flux-tube transverse modes.  Nevertheless, it is a sharp test of the algebraic part of the QRF picture:
\textbf{the local color index can be made relational by a Wilson-line frame}, but the resulting non-singlet object is not a complete physical singlet; physical states on closed space are class functions, and external charges on the line create a uniform electric flux with linearly growing energy.

In this section only, we switch from the physical connection of Sec.~II to the \textbf{geometric-connection convention}
\begin{align}
 &\mathscr A_{\rm geom}:=g_{\rm YM}\mathscr A_{\rm phys},
 \nonumber\\
 &\mathscr F_{\rm geom}:=d\mathscr A_{\rm geom}-i\mathscr A_{\rm geom}\wedge\mathscr A_{\rm geom},
 \nonumber\\
 &\mathscr E^j:=\Pi^j={1\over g_{\rm YM}^2}\mathscr F_{\rm geom}^{0j}.
 \label{eq:ym2-geometric-convention}
\end{align}
For notational economy, $\mathscr A_{\rm geom}$ and $\mathscr F_{\rm geom}$ are denoted simply by $\mathscr A$ and $\mathscr F$ throughout Sec.~IV.  Consequently
$\mathscr D=d-i[\mathscr A,\,\cdot\,]$,
$W_\gamma=\mathrm P\exp(i\int_\gamma\mathscr A)$, and
$\mathscr A^h=h^{-1}\mathscr A h+i h^{-1}dh$, while the canonical Hamiltonian is
$H=(g_{\rm YM}^2/2)\int dx\,2\operatorname{tr}\mathscr E^2$.

\subsection{Continuum theory and Gauss law}
\label{subsec:ym2-continuum}

Let spacetime be $M=\mathbb{R}_t\times\Sigma_1$, with compact semisimple gauge group $G$.  In the temporal gauge
\begin{equation}
  \mathscr{A}_0(x)=0,
\end{equation}
the canonical Hamiltonian $H$ of pure Yang--Mills theory is given by
\begin{align}
  H &={g_{\rm YM}^2\over2}\int_{\Sigma_1}dx\,
  \mathscr{E}^A(x)\mathscr{E}^A(x)\nonumber\\&
  = {g_{\rm YM}^2\over2}\int_{\Sigma_1}dx\,
  2\operatorname{tr}\left\{\mathscr{E}(x)^2\right\},
  \label{ym2-H}
\end{align}
where $\mathscr{E}^A=\mathscr{F}_{10}^A$ is the canonical color electric field.
The Gauss law is obtained as
\begin{equation}
  \mathscr{G}^A(x):=(\mathscr{D}_1\mathscr{E})^A(x)-\rho^A(x)
\approx0 .
  \label{ym2-gauss}
\end{equation}
For pure Yang--Mills theory, $\rho$ vanishes: $\rho^A=0$.  With Hermitian generators, the covariant derivative in the adjoint representation reads

\begin{equation}
\begin{aligned}
  \mathscr{D}_1\mathscr{E}
  &=\partial_1\mathscr{E}-\mathrm{i}[\mathscr{A}_1,\mathscr{E}],\\
  (\mathscr{D}_1\mathscr{E})^A
  &=\partial_1\mathscr{E}^A+f^{ABC}\mathscr{A}_1^B\mathscr{E}^C .
\end{aligned}
\end{equation}

The equation $\mathscr{D}_1\mathscr{E}=0$ (Gauss law for $\rho^A(x)=0$) says that the electric field $\mathscr{E}$ is covariantly constant throughout the interval.
Thus, once a frame is chosen, the remaining color information is constant along the one-dimensional space, unless external or dynamical charges are inserted.
If the boundary conditions retain the gauge transformations at the two endpoints as genuine boundary symmetries, however, the electric flux at each endpoint may carry a physical boundary charge.

\subsection
{The Wilson line as a QRF-type dressing and frame transporter}
\label{subsec:ym2-wilson-qrf}

Choose a base point $x_0$.
We introduce the Wilson-line frame from $x_0$ to $x$:
\begin{equation}
  \boxed{
  h(x;\mathscr{A})
  :=\mathrm{P}\exp\left(\mathrm{i}\int_{x_0}^{x}dy\, \mathscr{A}_1(y) \right)} , \ h(x_0, \mathscr{A}) = \bm{1} .
  \label{ym2-h-frame}
\end{equation}

Figure~\ref{fig:Wilson_line_0}(left) shows this Wilson-line QRF-type dressing.  Under a gauge transformation $g(x)$, the corresponding frame transporter $h$ transforms as
\begin{equation}
 h(x;\mathscr{A}) \longmapsto h(x;\mathscr{A}^g)=g(x)h(x;\mathscr{A})g(x_0)^{-1} .
\end{equation}
Therefore $h$ compares the color orientation at $x$ with the color orientation at the base point $x_0$.
Moreover, the transformation by $h$ moves the connection to contour gauge or axial gauge,

\begin{equation}
\begin{aligned}
  \mathscr{A}_1^h(x) =&h(x;\mathscr{A})^{-1}\mathscr{A}_1(x)h(x;\mathscr{A})\\
  &+i\,h(x;\mathscr{A})^{-1}dh(x;\mathscr{A})=0,
\end{aligned}
\label{A1=0}
\end{equation}

away from possible global obstructions.  Thus in $1+1$ dimensions the frame transporter is not merely symbolic: the Wilson-line dressing implements axial or contour gauge relationally.

\begin{figure}[htbp]
  \centering
\includegraphics[width=.8\linewidth]{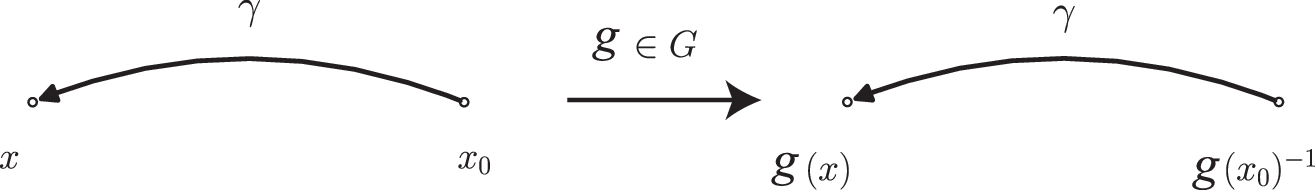}
\caption{
The Wilson line $W(x,x_0;\gamma)$ from the base point $x_0$ to the point $x$, and its gauge transformation by $g\in G$.  In one spatial dimension the path is fixed once the base point and endpoint are chosen, therefore the Wilson line gives a particularly concrete QRF-type dressing/frame transporter.
}
  \label{fig:Wilson_line_0}
\end{figure}

For a matter field $\psi(x)$ in a representation $R$, define the relational field
\begin{equation}
  \Psi_h(x):=R\bigl(h(x;\mathscr{A})\bigr)^{-1}\psi(x) .
  \label{ym2-rel-matter}
\end{equation}
For small gauge transformations with $g(x_0)=\bm{1}$, this field is invariant.  For transformations nontrivial at the base point $x_0$, the relational field transforms
\begin{equation}
  \Psi_h(x)\longmapsto R(g(x_0))\Psi_h(x) .
\end{equation}
Hence \textbf{the Wilson-line QRF-type dressing removes local gauge redundancy but leaves a residual color action at the base point or boundary.}
The bilocal singlet
\begin{align}
  M(x,y)
  &:=\Psi_h^\dagger(x)\Psi_h(y)
    =\psi^\dagger(x)R\bigl(W(x,y)\bigr)\psi(y),
  \label{ym2-bilocal}
\end{align}
with
\begin{equation}
  W(x,y):=\mathrm{P}\exp\left(\mathrm{i}\int_x^y\mathscr{A}_1(z)dz\right),
\end{equation}

 is a genuine gauge-invariant operator.
A bare colored local operator does not belong to the proper-gauge-invariant relational description.

More precisely, the Wilson-line-dressed but endpoint-frame-covariant field $\Psi_h(x)$ belongs to the \textbf{dressed/charged field algebra} $\mathfrak{F}$, whereas the endpoint-contracted operator $M(x,y)$ belongs to the frame-independent \textbf{observable algebra} $\mathfrak{A}$ when the corresponding endpoint-frame indices are not retained as accessible physical data.  Thus proper-gauge invariance of $\Psi_h$ should not be conflated with membership in the frame-independent observable algebra $\mathfrak{A}$.

\subsection{Yang--Mills theory on a circle: global holonomy and class functions}
\label{subsec:ym2-circle}

Let $\Sigma_1$ be a circle of length $L$: $\Sigma_1=S^1_L$, as in Figure~\ref{fig:circle_L}.

\begin{figure}[h]
  \centering
\includegraphics[width=.6\linewidth]{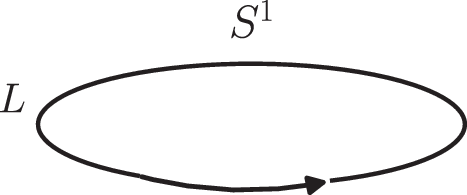}
  \caption{A circle of length $L$: $\Sigma_1=S^1$.  There is no boundary charge, but a global holonomy around the circle remains after local gauge redundancy has been removed.}
  \label{fig:circle_L}
\end{figure}

Locally, the Wilson-line frame can remove $\mathscr{A}_1$: $\mathscr{A}_1^h(x)=0$ as shown in the above (\ref{A1=0}).  Globally, however, the \textbf{holonomy} remains:
\begin{equation}
  \boxed{
  U:=\mathrm{P}\exp\left(\mathrm{i}\oint_{S^1}dx\, \mathscr{A}_1(x) \right).}
  \label{ym2-holonomy}
\end{equation}
Under a periodic gauge transformation, $U$ transforms as
\begin{equation}
  U\longmapsto\Omega(0)U\Omega(0)^{-1} .
\end{equation}
Therefore the physical information is not an element $U$ with a marked color basis, but its \textbf{conjugacy class}.

The \textbf{physical wave function} can be represented as a class function $\Phi(U)$ satisfying
\begin{equation}
  \Phi(U)=\Phi(\Omega U\Omega^{-1})
  \quad (\forall\Omega\in G).
  \label{ym2-class-function}
\end{equation}
Thus we obtain the \textbf{physical Hilbert space} of class functions:
\begin{equation}
  \mathcal{H}_{\rm phys}(S^1)\simeq L^2(G)^G ,
  \label{ym2-H-circle}
\end{equation}
under the assumption for the trivial principal bundle and the standard quantization.
By the Peter--Weyl theorem it decomposes into \textbf{characters} $\chi_R(U)$,
\begin{equation}
  \mathcal{H}_{\rm phys}(S^1)=\bigoplus_R\mathbb{C}\,\chi_R(U),
  \quad
  \chi_R(U)=\mathrm{tr}_R(U),
\end{equation}
where $R$ runs over irreducible representations of $G$.  The Hamiltonian $H$ reduces to the Laplace--Beltrami operator on the group manifold, and the characters satisfy the eigenvalue equation:
\begin{equation}
  H\chi_R(U)={g_{\rm YM}^2L\over2}C_2(R)\chi_R(U),
  \label{ym2-character-energy}
\end{equation}
up to the conventional normalization of the canonical color  electric field.

Equation~\eqref{ym2-H-circle} is the simplest exact illustration of the Gauss-law and relational kinematics underlying the QRF formulation;
\textbf{confinement requires the additional infrared sector-energy statement developed for the open interval below.}
\textbf{Physical states are not color vectors; they are class functions of the holonomy.}
\textbf{The physical algebra is generated by loop observables, singlet dressed bilocals, and other gauge-invariant composites.}
The subalgebra relevant to the present discussion is generated schematically by

\begin{equation}
\begin{aligned}
  \mathcal{A}_{\rm phys}(S^1)
  =\bigl\langle
  &\mathrm{tr}_R(U^n),\;
  \psi^\dagger(x)R(W(x,y))\psi(y),\\
  &\psi^\dagger(x)\psi(x),\ldots
  \bigr\rangle .
\end{aligned}
  \label{ym2-phys-alg-circle}
\end{equation}

For $G=SU(2)$, one may diagonalize $U\sim\exp(\mathrm{i}\theta\sigma_3)$, but the Weyl group identifies $\theta\sim-\theta$.  This residual identification is a simple finite-dimensional remnant of the general fact that \textbf{a global color frame cannot be selected uniquely.}

\subsection{Yang--Mills theory on the line: external charges and the linear potential}
\label{subsec:ym2-line}

Next we take $\Sigma_1=\mathbb{R}$ and insert a static external source in the representation $R$ at $x=0$ and an antisource in the conjugate representation $\overline R$  at $x=L$, as shown in Figure~\ref{fig:line_L}.

\begin{figure}[htbp]
  \centering
\includegraphics[width=.7\linewidth]{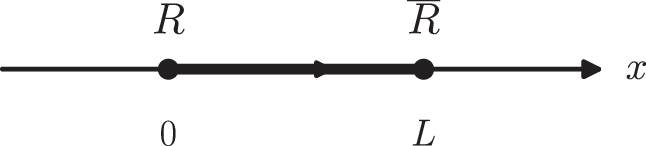}
  \caption{The line $\Sigma_1=\mathbb{R}$ with static external charges $R$ and $\overline R$, and the color electric flux generated between them.  In $1+1$ dimensions the Wilson line connecting the sources is a segment of uniform electric flux.}
  \label{fig:line_L}
\end{figure}

The gauge-field operator acts on $\mathcal H_{\rm gauge}$, while the two external sources act on $\mathcal{V}_R\otimes \mathcal{V}_{\overline R}$.  Thus, in the Wilson-line-dressed frame in which $\mathscr{A}_1^h=0$, the Gauss law is more precisely an operator identity on
$\mathcal H_{\rm gauge}\otimes \mathcal{V}_R\otimes \mathcal{V}_{\overline R}$:
\begin{equation}
 \partial_1\mathscr{E}_h^A(x)
 =
 (T_R^A\otimes\mathbf 1)\delta(x)
 +( \mathbf 1\otimes T_{\overline R}^A)\delta(x-L),
 \label{ym2-external-gauss}
\end{equation}
with $T_{\overline R}^A=-(T_R^A)^*$.  In the source singlet channel,
\begin{equation}
 \left(T_R^A\otimes\mathbf 1+\mathbf 1\otimes T_{\overline R}^A\right)
 |\mathrm{singlet}\rangle=0,
\end{equation}
therefore the second endpoint acts as the negative of the first.  With vanishing electric field outside $(0,L)$, the field between the endpoints is therefore represented on this channel by
\begin{equation}
\mathscr{E}_h^A(x)=
\begin{cases}
0, & x<0,\\
T_R^A\otimes\mathbf 1, & 0<x<L,\\
0, & x>L .
\end{cases}
\label{ym2-E-solution}
\end{equation}
Consequently,
\begin{equation}
 E_{\min}^R(L)
 =\frac{g_{\rm YM}^2}{2}\int_0^Ldx\,
 \sum_A(T_R^A\otimes\mathbf 1)^2
 =\frac{g_{\rm YM}^2}{2}C_2(R)L,
\end{equation}
which makes explicit why only the quadratic Casimir of the external representation remains.
Therefore, we have the linear potential:
\begin{equation}
  \boxed{V_R(L)={g_{\rm YM}^2\over2}C_2(R)L .}
  \label{ym2-linear-potential}
\end{equation}

From the QRF point of view, a single external color charge can be described relationally, but it leaves a residual color index at the reference point.  To make a singlet, it must be paired with an opposite charge by a Wilson line.  In one spatial dimension such a Wilson line is exactly the interval filled with electric flux, and the energy grows linearly with the separation.

This conclusion is independent of the particular admissible dressing.  Any state $|\Phi_h^R\rangle$ containing an external $R$--$\overline R$ pair belongs to the same external-charge sector.  By using the variational principle, we find
\begin{equation}
  {\langle\Phi_h^R|H|\Phi_h^R\rangle\over\langle\Phi_h^R|\Phi_h^R\rangle}-E_0
  \geq E_{\min}^R(L)
  ={g_{\rm YM}^2\over2}C_2(R)L .
  \label{ym2-variational}
\end{equation}
Thus one need not enumerate all possible QRF-type dressings $h$ one by one; the lowest energy of the external-charge sector gives a uniform linear lower bound.

The same result is obtained from the Wilson loop.  For a simple closed curve $C$ on the plane enclosing area $\mathcal{S}(C)$, the Wilson loop average exhibits the area law in the infinite volume limit
\begin{equation}
  \langle W_R(C)\rangle
  =\exp\left[-{g_{\rm YM}^2\over2}C_2(R)\mathcal{S}(C)\right]
  \label{ym2-area-law}
\end{equation}
 in pure two-dimensional Yang--Mills theory.
 A rectangular loop $L\times T$ gives $\langle W_R(L,T)\rangle\sim e^{-T V_R(L)}$ and hence again \eqref{ym2-linear-potential}.
 \textbf{The Wilson loop is a frame-independent observable.}  Therefore, \textbf{the linear color electric flux represented with a Wilson-line QRF-type dressing reduces to a standard frame-independent confinement diagnostic.}

\subsection{Yang--Mills theory on an open interval: boundary color charge}
\label{subsec:ym-open-interval}
\label{subsec:ym2-open-interval}

Let $\Sigma_1=I=(0,L)$ be an open interval of length $L$, as shown in Figure~\ref{fig:interval_L}.

\begin{figure}[htbp]
  \centering
\includegraphics[width=.8\linewidth]{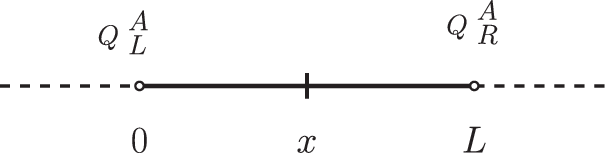}
  \caption{The open interval $\Sigma_1=I=(0,L)$, with the coordinates $x$: $0<x<L$.  Local color charge in the bulk $I$ is constrained by Gauss law, while boundary electric flux can remain at $x=0,L$ as a boundary observable.}
  \label{fig:interval_L}
\end{figure}

\subsubsection*{Gauss law, bulk color charge and boundary color charges}
 In temporal gauge $\mathscr{A}_0^A(x)=0$, the Hamiltonian and Gauss law are again \eqref{ym2-H} and \eqref{ym2-gauss}.
Smearing the Gauss law with a test function $\omega^A(x)$ gives
\begin{align}
  G[\omega]
  &:=\int_0^Ldx\,\omega^A\bigl((\mathscr{D}_1\mathscr{E})^A-\rho^A\bigr)  \nonumber\\
  &=-\int_0^Ldx\,(\mathscr{D}_1\omega)^A\mathscr{E}^A
    -\int_0^Ldx\,\omega^A\rho^A
    +\left[\omega^A\mathscr{E}^A\right]_{0}^{L} .
  \label{ym2-smeared-gauss}
\end{align}
Only transformations $\omega(x)$ which are zero at the boundary ($\omega(0)=0$ and $\omega(L)=0$) are pure gauge redundancies.  Transformations $\omega(x)$ which are nonzero at the boundary ($\omega(0)\not=0$ or $\omega(L)\not=0)$ generate \textbf{boundary symmetries}.

In the presence of the boundary, the \textbf{differentiable improved generator} may be written according to \cite{Donnelly:2016auv,Riello2021} as

\begin{equation}
\begin{aligned}
  \widetilde G[\omega]
  :=&\int_0^Ldx\,\omega^A\bigl((\mathscr{D}_1\mathscr{E})^A-\rho^A\bigr)\\
  &-\omega^A(L)\mathscr{E}^A(L)+\omega^A(0)\mathscr{E}^A(0).
\end{aligned}
  \label{ym2-smeared-gauss2}
\end{equation}
For the details on the differentiable improved generator, see Appendix \ref{app:differentiable_generator}.

On the constraint surface satisfying $(\mathscr{D}_1\mathscr{E})^A-\rho^A\approx0$, define the outward boundary flux by
\begin{equation}
 Q_{\partial I}^{\mathrm{out}}[\omega]
 :=\omega^A(L)\mathscr{E}^A(L)-\omega^A(0)\mathscr{E}^A(0).
\end{equation}
The differentiably improved generator then obeys the same convention as in Sec.~\ref{subsec:relational-boundary-charge},
\begin{equation}
 \widetilde G[\omega]\approx- Q_{\partial I}^{\mathrm{out}}[\omega].
\end{equation}
The endpoint generator charges are defined by
\begin{equation}
  Q_L^A:=\mathscr{E}^A(0),
  \quad
  Q_R^A:=-\mathscr{E}^A(L),
  \label{ym2-boundary-Q}
\end{equation}
so that $\widetilde G[\omega]=\omega^A(0)Q_L^A+\omega^A(L)Q_R^A=-Q_{\partial I}^{\mathrm{out}}[\omega]$.  Thus $Q_L^A,Q_R^A$ are coefficients of the boundary-symmetry generator and are not a second, oppositely signed definition of the outward flux.
The first conclusion is that \textbf{the Gauss law removes local bulk color charge, but boundary color charge $Q_L^A,Q_R^A$ can remain as electric flux $\mathscr{E}^A(0),\mathscr{E}^A(L)$.}

Take the left endpoint $x=0$ as the frame root (base point) and define the frame field $h(x)$ by

\begin{equation}
\begin{aligned}
  h(x)&:=\mathrm{P}\exp\left(i\int_0^x dy\,
  \mathscr{A}_1(y)\right)h(0) \in SU(2),\\
  &\qquad h(0)=\mathbf{1},\quad x\in(0,L).
\end{aligned}
  \label{ym2-left-frame}
\end{equation}

We introduce
\begin{equation}
  U_{0L}:=h(L)^{-1}
  =\mathrm{P}\exp\left(-i\int_0^L dx\, \mathscr{A}_1(x)\right),
\end{equation}

which transforms under the gauge transformation as
\begin{equation}
  U_{0L}\longmapsto g(0)U_{0L}g(L)^{-1} .
\end{equation}
Therefore the left and right boundary symmetries act by left and right multiplication $g(0),g(L)$.

Up to conventional signs, we find that the boundary charges are not merely constraints; they act as generators on the open Wilson line ending at the boundary:
\begin{equation}
  [Q_L^A,U_{0L}]=T^AU_{0L},
  \quad
  [Q_R^A,U_{0L}]=-U_{0L}T^A .
  \label{ym2-boundary-action}
\end{equation}

If the physical state belongs to an irreducible \textbf{boundary-flux sector} $R$, the Hamiltonian on the interval reduces to
\begin{equation}
  H={g_{\rm YM}^2L\over2}Q_L^AQ_L^A={g_{\rm YM}^2L\over2}C_2(R) .
  \label{ym2-interval-spectrum}
\end{equation}
This demonstrates a point that is sometimes hidden in slogans about confinement: \textbf{color charge is absent as a local bulk observable, but it can remain as boundary color electric flux.}
The frame-independent statement is not that every boundary color component is absolute, but that \textbf{the boundary sector and its Casimir are physical once the boundary conditions and frame data are specified.}

\subsubsection*{Relational color electric field}

Define \textbf{the relational color electric field} $\mathscr{E}_h(x)$ by
\begin{equation}
\mathscr{E}_h(x) := h(x)^{-1}\mathscr{E}(x)h(x)
\Leftrightarrow
\mathscr{E}(x) = h(x)\mathscr{E}_h(x)h(x)^{-1} .
\label{eq:relational-electric-field}
\end{equation}
\textbf{The relational color electric field $\mathscr{E}_h(x)$ is invariant under the original local gauge transformation.}
Equation \eqref{ym2-left-frame} implies
\begin{align}
\partial_xh(x)=i\mathscr{A}_1(x)h(x).
\end{align}

and hence the relation for the derivatives between $\mathscr{E}_h$ and $\mathscr{E}$ holds:

\begin{align}
\partial_x\mathscr{E}_h
&=\partial_x\left(h^{-1}\mathscr{E}h\right)
\nonumber\\
&=h^{-1}\left[
\partial_x\mathscr{E}-i[\mathscr{A}_1,\mathscr{E}]
\right]h
\nonumber\\
&=h^{-1}(\mathscr{D}_1[\mathscr{A}]\mathscr{E})h,
\quad x \in (0,L).
\label{eq:derivative-relational-electric-field}
\end{align}

Therefore, the Gauss law $\mathscr{D}_1[\mathscr{A}]\mathscr{E}(x)
=\rho(x)=0$
implies
\begin{equation}
\partial_x\mathscr{E}_h(x)=0, \quad x \in (0,L).
\label{eq:constant-relational-electric-field}
\end{equation}
Thus,
in the pure Yang-Mills theory the Gauss law: \textbf{the color electric field $\mathscr{E}(x)$ is covariantly constant $\mathscr{D}_1[\mathscr{A}]\mathscr{E}(x)=0$ throughout the interval}
implies that
 \textbf{the relational color electric field $\mathscr{E}_h(x)$ is constant ($x$-independent) throughout the interval}:
\begin{equation}
\mathscr{E}_h(x)=\boldsymbol{q} \in\mathfrak{su}(2), \quad x \in (0,L).
\label{eq:constant-q}
\end{equation}
Here
\(\boldsymbol{q}\in\mathfrak{su}(2)\) represents a \textbf{boundary color electric-flux sector}.
\footnote{
It is shown that $\mathscr{E}_h(x)$ is also independent of time $t$.
In temporal gauge $\mathscr{A}_0=0$, the Yang--Mills equation implies $\partial_t\mathscr{E}=0$.  Moreover, differentiating $\partial_x h=i\mathscr{A}_1h$ with respect to $t$ and using the Hamilton equation $\partial_t\mathscr{A}_1\propto\mathscr{E}$ gives $\partial_x(h^{-1}\partial_t h)\propto\mathscr{E}_h$.  Since the Gauss law gives $\partial_x\mathscr{E}_h=0$ and $h(0)=\mathbf{1}$ is fixed, $h^{-1}\partial_t h$ is proportional to $\mathscr{E}_h$.  Hence
$\partial_t\mathscr{E}_h=h^{-1}(\partial_t\mathscr{E})h+[\mathscr{E}_h,h^{-1}\partial_t h]=0$.
}

Integrating Equation (\ref{eq:derivative-relational-electric-field}) over the interval $I$ and inserting the Gauss law $\mathscr{D}_1[\mathscr{A}]\mathscr{E}(x)=\rho(x)$ gives
\begin{equation}
  h(L)^{-1}\mathscr{E}(L)h(L)-\mathscr{E}(0)
  =\int_0^Ldx\,h(x)^{-1}\rho(x)h(x) .
  \label{ym2-integrated-gauss}
\end{equation}
For pure Yang--Mills theory with $\rho(x)\equiv 0$ the right-hand side vanishes, so the color electric flux $\mathscr{E}(L)=-Q_R$ at the right endpoint, parallel transported to the left-end frame, equals the left boundary color charge $Q_L=\mathscr{E}(0)$.  To compare the two Lie-algebra-valued endpoint charges in one common relational frame, define

\begin{equation}
\begin{aligned}
  {Q}_{R\to L}
  &:=h(L)^{-1}{Q}_{R}h(L)
  =-h(L)^{-1}\mathscr{E}(L)h(L),\\
  {Q}_{L}&:=\mathscr{E}(0).
\end{aligned}
  \label{eq:right-charge-transported-left-frame}
\end{equation}
Then the pure Yang--Mills Gauss law gives
\begin{align}
  {Q}_{L}=-{Q}_{R\to L}
  \Leftrightarrow
  {Q}_{L}&=-h(L)^{-1}{Q}_{RL}h(L).
  \label{eq:nonabelian-endpoint-charge-relation}
\end{align}
Thus the untransported endpoint charges are related by parallel transport, not by the componentwise identity ${Q}_{L}=-{Q}_{R}$ in two independently chosen endpoint frames.  Therefore, \textbf{the bulk carries no independent color charge.}

\subsubsection*{Boundary-charge sectors in the quantum theory}

Let the proper gauge transformations $g(x)\in SU(2)$ satisfy
\begin{align}
g(0)=g(L)=\boldsymbol{1}.
\label{g0=gL=1}
\end{align}
On an open interval, when the Gauss law is smeared with a Lie-algebra-valued test function $\omega(x)$, endpoint terms appear upon integration by parts.
Therefore, transformations that vanish at the endpoints and transformations that do not vanish at the endpoints  are not physically equivalent.
Infinitesimal transformations satisfying
$
 \omega(0)=\omega(L)=0
$
at the endpoints, or finite transformations satisfying
$
 g(0)=g(L)=1,
$
are transformations called the \textbf{proper gauge group} $\mathcal G_0$, and are quotiented out as redundancies.
In other words, only
\begin{align}
 \mathcal G_0=\{g(x)\in SU(2) \mid g(0)=g(L)=1\}
\end{align}
is quotiented out.
Therefore, the remaining gauge transformation are $g(0)\not=1$ and $g(L)\not=1$.
This is the decisive difference from the circle case.

By contrast, transformations $\omega$ that are nontrivial at the endpoints $\omega|_{\partial I}\not=0$ can act physically as \textbf{boundary symmetries}, insofar as the chosen boundary conditions allow them.
Thus, Eq.~(\ref{g0=gL=1}) is not merely a gauge choice, but rather the definition that distinguishes
 \text{bulk gauge redundancy}
\text{from}
 \text{boundary symmetry}.
This distinction is the reason why $L^2(SU(2))$ does not later reduce to the space of class functions.

After quotienting by this proper gauge group $\mathcal G_0$, the ``only'' physical degree of freedom of the gauge field on the interval $I$ is the \textbf{total holonomy}:
\begin{equation}
U_{\mathscr{A}}
=
\mathcal{P}\exp\left[i\int_0^L\mathrm{d}x\,\mathscr{A}_1(x)\right]
\in SU(2).
\label{eq:total-holonomy}
\end{equation}
The operator $U_{\mathscr{A}}$ parallel-transports the color frame from $x=0$ to $x=L$, and is the group-valued part of the open Wilson line.
Under a general gauge transformation, $U_{\mathscr{A}}$ transforms as
$
 U_{\mathscr{A}} \longmapsto g(L)U_{\mathscr{A}} g(0)^{-1}
$.
Therefore, under a proper gauge transformation satisfying Eq.~(\ref{g0=gL=1}), $U_{\mathscr{A}}$ does not change:
$
 U_{\mathscr{A}}\longmapsto U_{\mathscr{A}},
$
and $U_{\mathscr{A}}$ remains after quotienting by the proper gauge group.
The fact that $U_{\mathscr{A}}$ is the ``only'' remaining degree of freedom in the bulk can be shown as follows. For a connection $\mathscr A_1$, define
\begin{equation}
 W_{\mathscr A}(x)
 :=\mathcal P\exp\left[
 i\int_0^x dy\,\mathscr A_1(y)
 \right] \in SU(2).
\end{equation}

Then $W_{\mathscr A}(0)=\bm{1}$ and $W_{\mathscr A}(L)=U_{\mathscr A}$.
Suppose that two connections $\mathscr A_1$ and $\mathscr A'_1$ have the same holonomy $U_{\mathscr{A}}$: $U_{\mathscr{A}}=U_{\mathscr{A}'}$. If one sets
\begin{equation}
 \tilde g(x):=W_{\mathscr A'}(x)W_{\mathscr A}(x)^{-1},
\end{equation}
then we have
$
 \tilde g(0)=1
$
and
$
 \tilde g(L)=U_{\mathscr{A}} U_{\mathscr{A}}^{-1}=1
$,
therefore $\tilde g(x)\in\mathcal G_0$ is a proper gauge transformation. Moreover, this $g(x)$ transforms $\mathscr A_1$ into $\mathscr A'_1$: $g(x)W_{\mathscr A}(x)=W_{\mathscr A'}(x)$.

Therefore, the configuration space obtained by quotienting by the proper gauge group $\mathcal G_0$
\begin{equation}
 \{\mathscr A_1(x)\}/\mathcal G_0
 \simeq SU(2),
\end{equation}
can be described solely by the holonomy $U_{\mathscr A}\in SU(2)$.
This is the reduction underlying Eq.~(\ref{eq:total-holonomy}).
Thus, in one spatial dimension, all local gauge-field profiles disappear through Gauss/gauge redundancy, and only the holonomy connecting the two endpoints remains.

If left and right boundary frames \(h_L\) and \(h_R\) are introduced, the \textbf{relational holonomy} defined by
\begin{equation}
\Omega
=
h_R^{-1}U_{\mathscr{A}}h_L \in SU(2), \quad  h_R, h_L \in SU(2)
\label{eq:relational-holonomy}
\end{equation}
is invariant under the original local gauge transformation.
This is shown as follows.
Under the original local gauge transformation, the boundary frames transform as
$
 h_L\longmapsto g(0)h_L
$
and
$
h_R\longmapsto g(L)h_R
$.
On the other hand, $U_{\mathscr{A}}$ transforms as
$
 U_{\mathscr{A}}\longmapsto g(L)U_{\mathscr{A}}g(0)^{-1}
$.
Therefore $\Omega$ transforms
$
 \Omega
 \longmapsto
 (g(L)h_R)^{-1}
 \bigl[g(L)U_{\mathscr{A}}g(0)^{-1}\bigr]
 (g(0)h_L)
=h_R^{-1}U_{\mathscr{A}}h_L
 =\Omega
$.
Thus, $\Omega$ is invariant under the original local gauge transformation.

The important point here is not to identify ``gauge invariant'' with ``frame independent.'' If the frames themselves are reoriented from the right according to
$
 h_L\longmapsto h_L k_L
$,
$
 h_R\longmapsto h_R k_R
$,
then
\begin{equation}
 \Omega\longmapsto k_R^{-1}\Omega k_L.
 \label{Omega}
\end{equation}
Thus, $\Omega$ is independent of gauge redundancy, while retaining physical information about the relative orientation of the left and right boundary frames.
In this sense, $\Omega$ is the ``color parallel transport between the left and right endpoints as seen from the frame at the right endpoint.''
These left and right actions later become the two boundary charges
$
 \widehat Q_L^A, \ \widehat Q_R^A.
$

The \textbf{physical Hilbert space} is therefore given by
\begin{equation}
\mathcal{H}_{\mathrm{phys}}
=
L^2
\left(
SU(2),\mathrm{d}\Omega
\right).
\label{eq:physical-hilbert-space}
\end{equation}
Since the classical reduced configuration space is $SU(2)$ itself, in configuration-space quantization the states become complex-valued wave functions on $SU(2)$ denoted by
$
 \Psi(\Omega)
$,
which belong to the Hilbert space $L^2(SU(2),d\Omega)$ of square-integrable functions on a compact group $G$ with the inner product defined by
\begin{equation}
 \langle\Phi|\Psi\rangle
 =\int_{SU(2)}d\Omega\,
 \Phi(\Omega)^*\Psi(\Omega),
\end{equation}
where $d\Omega$ is the Haar measure on $SU(2)$, usually normalized so that
$
 \int_{SU(2)}d\Omega=1.
$

The difference from the case of a circle is essential here. On a circle, the starting point and endpoint of the holonomy coincide, and the residual gauge transformation acts by conjugation,
$
 U\longmapsto kUk^{-1}
$,
therefore the physical wave function is restricted to a class function,
$
 \Psi(U)=\Psi(kUk^{-1})
$.
By contrast, on an open interval, the left and right boundary transformations have not been quotiented out as proper gauge redundancies (they are not included in the proper gauge group).
Therefore, a general $\Psi(\Omega)$ is allowed,
and the Hilbert space is the full space $L^2(SU(2))$, rather than the subspace of class functions.
Thus, Eq.~(\ref{eq:physical-hilbert-space}) expresses at the Hilbert-space level the fact that \textbf{boundary color charge degrees of freedom is retained as a physical sector}.

By the Peter--Weyl decomposition,
\begin{equation}
L^2
\left(
SU(2)
\right)
=
\bigoplus_{j=0,\frac{1}{2},1,\ldots}
V_j\otimes V_j^*.
\label{eq:peter-weyl}
\end{equation}
Here
$j=0,\frac12,1,\frac32,\ldots$ is the spin labeling an irreducible representation of $SU(2)$,
$V_j$ is the spin-$j$ representation space
with $\dim V_j=2j+1$,
and
$V_j^*$ is the dual representation space.
The Peter--Weyl theorem states that any square-integrable function on a compact group $G$ can be completely expanded in terms of the matrix elements of all irreducible representations of $G$. For $G=SU(2)$, the matrix element, i.e., Wigner $D$-function $D_{mn}^{(j)}(g)$ of the Wigner $D$-matrix $D^{(j)}(g)$ provide such a basis:
\begin{equation}
\begin{aligned}
 D_{mn}^{(j)}(\Omega)
 &=\langle j,m|D^{(j)}(\Omega)|j,n\rangle,\\
 &\qquad m,n=-j,-j+1,\ldots,j.
\end{aligned}
\end{equation}

To see why $V_j\otimes V_j^*$ appears, it is enough to note that under the frame reorientation (\ref{Omega})
$
 \Omega\longmapsto k_R^{-1}\Omega k_L,
$
one has
\begin{equation}
 D^{(j)}(k_R^{-1}\Omega k_L)
 =D^{(j)}(k_R^{-1})D^{(j)}(\Omega)D^{(j)}(k_L),
\end{equation}
which reads in components,
\begin{equation}
 D_{mn}^{(j)}(k_R^{-1}\Omega k_L)
 =\sum_{r,s}D_{mr}^{(j)}(k_R^{-1})D_{rs}^{(j)}(\Omega)D_{sn}^{(j)}(k_L).
\end{equation}
Thus the two matrix indices $m,n$ are representation indices in which
one matrix index $m$ carries the $SU(2)$ action at the right endpoint, while the other $n$ carries the $SU(2)$ action at the left endpoint. Therefore, the subspace with fixed $j$ forms a representation of the left and right boundary symmetries of the form
$
 V_j\otimes V_j^*
$.
For $SU(2)$, $V_j^*$ is equivalent to $V_j$ as a representation, but there is significance in writing $V_j^*$ explicitly. It makes clear that the two tensor factors are not ``two spins at the same place,'' but correspond instead to the two boundary actions operating from the left and right of the holonomy.
It is not the case that arbitrary $j_L,j_R$ independently appear on the left and right so that one has
$
 V_{j_R}\otimes V_{j_L}^*
$.
One necessarily has
$
j_L=j_R=j
$.
This is because the left and right boundary charges are not two independent color charges due to the Gauss law, but are the two ends of the same color electric flux.

Therefore, an orthonormal basis is given by
\footnote{
With the normalized Haar measure, the matrix elements satisfy the Schur orthogonality relation
\[
 \int_{SU(2)}d\Omega\,
 D_{mn}^{(j)}(\Omega)^*
 D_{m'n'}^{(j')}(\Omega)
 =\frac{1}{2j+1}
 \delta_{jj'}\delta_{mm'}\delta_{nn'}.
\]
Therefore, multiplying by $\sqrt{2j+1}$ gives
\[
 \int d\Omega\,
 \Psi_{jmn}(\Omega)^*
 \Psi_{j'm'n'}(\Omega)
 =\delta_{jj'}\delta_{mm'}\delta_{nn'}.
\]
This is the origin of the coefficient in Eq.~(\ref{eq:peter-weyl-basis}).
The dimension of the sector with fixed $j$ is
$
 (2j+1)^2
$.
The labels $m$ and $n$ correspond to magnetic quantum numbers relative to the left and right boundary frames. With the convention
$
 \Omega\to k_R^{-1}\Omega k_L,
$
used above, the $m$ index is acted on by the right boundary, whereas the $n$ index is acted on by the left boundary. However, which index is called left or right may be interchanged depending on the convention. What is physically important is that the two indices represent independent boundary orientations and that both belong to the same $j$ sector.
}
\begin{equation}
\Psi_{jmn}(\Omega)
=
\sqrt{2j+1}\,
D_{mn}^{(j)}(\Omega),
\label{eq:peter-weyl-basis}
\end{equation}
where it is useful to think of
\begin{itemize}
\item $j$: the representation label corresponding to the magnitude of the boundary color charge,
\item $m,n$: the components of the ``direction of the color charge'' relative to the left and right boundary frames.
\end{itemize}
In the convention above,
$
 \Omega\to k_R^{-1}\Omega k_L,
$
the $m$ index carries the right-boundary action, whereas the $n$ index carries the left-boundary action.
Therefore, when the boundary frames are rotated, $m$ and $n$ mix among themselves.
On the other hand, $j$ does not change.

The \textbf{left and right relational boundary color-charge operators} $\widehat{\boldsymbol{{Q}}}_{L},\widehat{\boldsymbol{{Q}}}_{R}$ act as the right-invariant and left-invariant vector fields on \(\Omega\), respectively.
Let infinitesimal reorientations of the left and right frames be
$
 k_L(t)=e^{itT^A}
$
and
$
 k_R(t)=e^{itT^A}
$
respectively.
Then
$
 \Omega\to \Omega e^{itT^A}
$
and
$
 \Omega\to e^{-itT^A}\Omega
$
are the left and right boundary actions, respectively. The corresponding Hermitian generators can, for example, be taken as the differential operators on the group manifold $SU(2)$:

\begin{align}
 (\widehat Q_L^A f)(\Omega)
 &:=-i\left.\frac{d}{dt}
 f(\Omega e^{itT^A})\right|_{t=0},\nonumber\\
 (\widehat Q_R^A f)(\Omega)
 &:=-i\left.\frac{d}{dt}
 f(e^{-itT^A}\Omega)\right|_{t=0}.
\end{align}
The overall signs depend on the conventions for the charges and outward normals, but this does not affect the Casimir.

For right multiplication,
$
 D^{(j)}(\Omega e^{itT^A})=D^{(j)}(\Omega)e^{itT_j^A}
$
yields
\begin{align}
 \widehat Q_L^A D_{mn}^{(j)}(\Omega)
 =\sum_{n'}D_{mn'}^{(j)}(\Omega)(T_j^A)_{n'n}.
\end{align}
Thus $\widehat Q_L^A$ acts on the $n$ index.
On the other hand,
$
 D^{(j)}(e^{-itT^A}\Omega)=e^{-itT_j^A}D^{(j)}(\Omega)
$
yields, if the sign convention is chosen as above,
\begin{align}
 \widehat Q_R^A D_{mn}^{(j)}(\Omega)
 =-\sum_{m'}(T_j^A)_{mm'}D_{m'n}^{(j)}(\Omega).
\end{align}
Thus $\widehat Q_R^A$ acts on the $m$ index.
This structure is exactly the quantum version of the relations for the open Wilson line in the above:
$
 [Q_L^A,U_{0L}]=T^A U_{0L}, \
 [Q_R^A,U_{0L}]=-U_{0L}T^A.
$

Defining the \textbf{quadratic Casimir operators} by
\begin{align}
 \widehat{\bm Q}_L^{\,2}
 :=\sum_{A=1}^{3}\widehat Q_L^A\widehat Q_L^A,
 \quad
 \widehat{\bm Q}_R^{\,2}
 :=\sum_{A=1}^{3}\widehat Q_R^A\widehat Q_R^A,
\end{align}
and using, in the spin-$j$ representation:
$
 \sum_{A=1}^{3}T_j^A T_j^A
 =j(j+1)\bm 1_{V_j},
$
one obtains
\begin{align}
 \widehat{\bm Q}_L^{\,2}\Psi_{jmn}
 =j(j+1)\Psi_{jmn}, \quad
 \widehat{\bm Q}_R^{\,2}\Psi_{jmn}
 =j(j+1)\Psi_{jmn}.
\end{align}
Therefore,
we obtain equal {quadratic Casimir operators}:
\begin{equation}
\widehat{\boldsymbol{{Q}}}_{L}^{\,2}
=
\widehat{\boldsymbol{{Q}}}_{R}^{\,2}
=
j(j+1) \bm{1}.
\label{eq:boundary-casimir}
\end{equation}
More precisely, the right-hand side denotes the eigenvalue on a fixed $j$ sector. As an operator identity, both operators are equal to the Laplace--Beltrami operator on $SU(2)$, whose eigenvalue in the $j$ sector is $j(j+1)$.
Here \(j\) is not an individual frame-dependent component.   Thus \textbf{\(j\) is a frame-independent label of the boundary color-charge sector.}

\subsubsection*{Exact calculation of the boundary color-charge sector energy}

By the Gauss law, the relational color electric field $\widehat{\mathscr{E}}_h$ is constant throughout the
interval:
\begin{align}
\widehat{\mathscr{E}}_h(x)
=\widehat{\mathscr{E}}_h(0)
=\widehat{\mathscr{E}}(0)
=\widehat{\boldsymbol{{Q}}}_{L} .
\end{align}
Therefore, in each \(j\) sector, the Hamiltonian (\ref{ym2-H})
becomes
\begin{align}
\hat H
=
\frac{g_{\mathrm{YM}}^2}{2}
\int_0^L\mathrm{d}x\,
2\operatorname{tr}
\left\{
\widehat{\mathscr{E}}_h(x)^2
\right\}
=
\frac{g_{\mathrm{YM}}^2L}{2}
\widehat{\boldsymbol{{Q}}}_{L}^{\,2}.
\label{eq:hamiltonian-casimir}
\end{align}
The boundary color-charge sector energy in the representation \(j\) is  exactly given by
\begin{equation}
E_j(L)
=
\frac{g_{\mathrm{YM}}^2L}{2}
j(j+1).
\label{eq:exact-sector-energy}
\end{equation}
By using the fact that the vacuum sector is \(j=0\), and
$
E_0(L)=0,
$
the \textbf{boundary color-charge sector energy} in the representation $j$ is finally obtained:
\begin{equation}
\Delta E_j(L)
=
E_j(L)-E_0(L)
=
\frac{g_{\mathrm{YM}}^2L}{2}
j(j+1).
\label{eq:sector-energy-difference}
\end{equation}
In particular, in the fundamental sector \(j=\frac{1}{2}\), one obtains
\begin{equation}
\Delta E_{\frac{1}{2}}(L)
=
\frac{3g_{\mathrm{YM}}^2}{8}L.
\label{eq:fundamental-sector-energy}
\end{equation}
This linear sector energy is consistent with the exact test for the $1+1$- dimensional Yang-Mills theory discussed in the paper.
The same linear dependence agrees with the exact static-probe energy extracted from a rectangular Wilson loop.

\subsubsection*{Verification using the color-charge sector partition function}

The same result can be obtained by using the \textbf{color-charge sector partition function} as follows.
The degeneracy of the \(j\) sector is
$
(2j+1)^2.
$
Therefore, the $j$-sector partition function is given by
\begin{equation}
Z_j(T,L)
=
(2j+1)^2
\exp
\left[
-T
\frac{g_{\mathrm{YM}}^2L}{2}
j(j+1)
\right].
\label{eq:sector-partition-function}
\end{equation}
The ratio to the vacuum-sector partition function is
\begin{equation}
\begin{aligned}
\mathcal{K}_j(T,L)
&=
\frac{Z_j(T,L)}{Z_0(T,L)}\\
&=
(2j+1)^2
\exp\left[
-T\frac{g_{\mathrm{YM}}^2L}{2}j(j+1)
\right].
\end{aligned}
\label{eq:sector-ratio}
\end{equation}
Consequently, we find
\begin{align}
-\lim_{T\to\infty}
\frac{1}{T}
\log\mathcal{K}_j(T,L)
=
\frac{g_{\mathrm{YM}}^2L}{2}
j(j+1)
=
\Delta E_j(L).
\label{eq:partition-sector-energy}
\end{align}
This directly realizes the sector-energy definition in the
superselection-sector criterion.

\subsubsection*{Bulk charge and boundary charge}

This model contains no matter field, therefore
$
\rho=0.
$
In the Wilson-line frame,
$\mathscr{A}_1^h=h^{-1}\mathscr{A}_1h+i h^{-1}\partial_xh=0$.
For a constant body-frame test direction \(\xi\),
$\mathscr{D}_1[\mathscr{A}^h]\xi=\partial_x\xi-i[\mathscr{A}_1^h,\xi]=0$.
The relational bulk contributions to color charges therefore vanish:
\begin{align}
{Q}_{I}^{\mathrm{matter}}[h,\xi]&=0,\qquad
{Q}_{I}^{\mathrm{gauge}}[h,\xi]=0,
\nonumber\\
{Q}_{I}^{\mathrm{total}}[h,\xi]&=0.
\label{charges3}
\end{align}
The vanishing of $Q_I^{\mathrm{gauge}}[h,\xi]$  is specific to a constant test direction $\xi$. For a general $x$-dependent test direction, the Wilson-line frame gives $\mathscr A_1^h=0$ and $\mathscr E_h=q$, but
\[
Q_I^{\mathrm{gauge}}[h,\xi]
=\int_0^L dx\,2\operatorname{tr}\{(\partial_x\xi)q\}
=2\operatorname{tr}\{[\xi(L)-\xi(0)]q\}.
\]
Thus an endpoint-selective test function $\xi(L)\not=\xi(0)$ has a nonzero gauge/transport contribution even in pure Yang--Mills theory. The constant choice probes the sum of the two endpoint color charges after they are expressed in a common relational frame, for which $Q_L+Q_{R\to L}=0$. Hence this result does not imply that the individual boundary color-charge generators vanish.
See Appendix~\ref{app:boundary-charge} for details on the consistency between the relational gauge charge in Sec.~\ref{subsec:relational-boundary-charge} and the boundary color-charge sector in Sec.~\ref{subsec:ym-open-interval}.

At finite \(L\), it is perfectly consistent to have
\begin{align}
{Q}_{L}\neq0,
\qquad
{Q}_{R}\neq0.
\end{align}

The two endpoint color charges must first be compared in a common frame.
In the left-end relational frame the individual boundary color charges with the transported right-end color charge ${Q}_{R\to L}$ defined in \eqref{eq:right-charge-transported-left-frame} satisfy
\begin{align}
{Q}_{L}=\boldsymbol{q},
\qquad
{Q}_{R\to L}=-\boldsymbol{q},
\end{align}
while in its own endpoint frame the untransported right-end color charge ${Q}_{R}$ obeys
\begin{equation}
{Q}_{R}=-h(L){Q}_{L}\,h(L)^{-1}.
\end{equation}
They are connected by the color electric flux that traverses the entire interval.
Thus the relational bulk color charge ${Q}_{I}^{\mathrm{total}}$ vanishes when the endpoint fluxes are first expressed in the same frame:
\begin{equation}
{Q}_{I}^{\mathrm{total}}
=
{Q}_{L}
+
{Q}_{R\to L}
=
0.
\label{eq:bulk-boundary-balance}
\end{equation}
This relation shows the following.
\begin{enumerate}
\item The generator of a compactly supported bulk gauge transformation is the Gauss constraint.
\item A boundary transformation acts as a genuine symmetry when the boundary conditions permit it.
\item Each individual boundary color charge may be nonzero.
\item The total boundary color charge is balanced by the total bulk color charge through the Gauss law.
\end{enumerate}
In pure Yang--Mills theory, the boundary flux is not a local object that exists independently of a compactly supported bulk color charge.  \textbf{The same Gauss-law data appear in two different forms: as a bulk constraint and as a boundary color charge.}

\textbf{The total bulk color charge of the interval vanishes.}
\textbf{Each boundary may nevertheless carry a nontrivial relational color charge, with the two color charges pointing in opposite directions.}
\textbf{The Gauss law does not eliminate the boundary color charges.}
\textbf{Instead, it connects the color charge at the left endpoint to the opposite color charge at the right endpoint through the color electric flux that traverses the entire interval.}

\subsubsection*{Criterion for color confinement}

For every nontrivial sector
$
j>0,
$
{we have} the infinite-interval limit:
\begin{equation}
\lim_{L\to\infty}
\Delta E_j(L)
=
+\infty \quad (j>0).
\label{eq:confinement-limit}
\end{equation}
Hence, \textbf{in the infinite-interval limit,
no finite-energy isolated sector carrying a nontrivial relational
boundary color charge exists.}
\textbf{This is color confinement according to the superselection-sector criterion.}

The essence of Eq.~(\ref{eq:boundary-casimir}) is not merely $SU(2)$ representation theory. It is also the quantum content of the classical Gauss law
$
 Q_L+Q_{R\to L}=0
$.
At the classical level, after the color charge at the right endpoint has been parallel-transported by the holonomy into the frame at the left endpoint, the left and right color charges are opposite to each other. Therefore, it is not generally correct to compare their components while leaving them in different frames and write
$
 Q_L^A=-Q_R^A
$.
However, parallel transport is an adjoint rotation and does not change the magnitude of the color charge, therefore
$
 Q_L^2=Q_R^2
$
holds in a frame-independent manner.

In the quantum theory, this ``magnitude of the color charge'' $Q_L^2=Q_R^2$ becomes the quadratic Casimir:
$
 C_2(j_L)=C_2(j_R).
$
For $SU(2)$,
$
 C_2(j)=j(j+1)
$
is one-to-one for $j\geq0$, and therefore
$
 j_L=j_R=j
$.
This is the physical reason why a single $j$ appears commonly in the two tensor factors in the Peter--Weyl decomposition.
Thus, the left and right boundary color charges are not two independent color sources. They are two boundary realizations appearing at the two ends of a single bulk color electric-flux sector.

It is particularly important to see that \textbf{$V_j\otimes V_j^*$ in (\ref{eq:peter-weyl}) is precisely the quantum version of the classical relation $Q_L+Q_{R\to L}=0$}.
There are not separate color charges on the left and right; rather, there are two left/right actions on a single color electric-flux degree of freedom.
Therefore $m,n$ can exist separately, while the Casimir $j$ is common.
The labels $m,n$ are labels for color charge components along the axes of the chosen boundary frames.
If the frames are rotated, $m,n$ mix within the same $j$. Therefore, $m,n$ are in general frame-dependent.
By contrast,
$
 j(j+1)
$
is a Casimir and is invariant under frame rotations.
Thus, $j$ is a frame-independent quantity characterizing the boundary color-charge sector.
This can be regarded as the quantum version of the classical relation
$
 \mathscr E_h(x)=q=q^A T^A,
$
in which the individual $q^A$ depend on the frame, whereas
$
 q^Aq^A
$
is invariant under frame rotations.
In the quantum theory, this invariant magnitude $q^2$ is discretized as
$
j(j+1).
$

For finite $L$, the $j>0$ boundary color-charge sectors exist.
\textbf{Confinement appears because
an infinite amount of energy is required to move one boundary color charge to infinity and isolate it}:
\begin{align}
E_{j>0}(L)\propto L \to \infty \quad (L \to \infty).
\end{align}
Color confinement therefore means neither
that a relational color charge cannot be defined
nor
that the boundary color-charge operator always vanishes.
The precise conclusion is:
\textbf{A gauge-invariant relational boundary color charge can be defined on a finite interval $L < \infty$, but the energy of every nontrivial color charge sector diverges in the infinite-volume limit $L \to \infty$.}

\subsection{QRF confinement in (1+1)-dimensional Yang--Mills theory}
\label{subsec:ym2-criterion}

As shown above, the $1+1$-dimensional model makes the QRF criterion concrete in three statements.

\noindent
(i) On a circle, we find
\begin{equation}
  \Phi(U)=\Phi(\Omega U\Omega^{-1}),
\end{equation}
therefore, no bare color index appears in the physical Hilbert space.

\noindent
(ii) In pure Yang--Mills theory the intrinsic physical algebra contains gauge-invariant loop/holonomy data and, when boundary symmetries are retained on an interval, boundary color-electric-flux invariants; it contains no dynamical matter field $\psi$:
\begin{equation}
  \mathcal{A}_{\rm phys}^{\rm YM}
  =\left\langle
  \mathrm{tr}_R(U^n),\;
  2\operatorname{tr}(Q_L^2),\;
  2\operatorname{tr}(Q_R^2),\ldots
  \right\rangle .
\end{equation}
Wilson-line bilinears containing external endpoint sources belong instead to a probe-extended algebra obtained after adjoining those nondynamical probe degrees of freedom; they are not elements of the intrinsic pure-Yang--Mills field algebra.

(iii) If an external colored sector is forced, the energy has the lower bound
\begin{equation}
  V_R(L)={g_{\rm YM}^2\over2}C_2(R)L .
\end{equation}
For the open-interval boundary-flux sector this is in fact the exact lowest excess energy,
\begin{equation}
 \Delta E_R(L)=\frac{g_{\rm YM}^2}{2}C_2(R)L ,
 \qquad R\neq\mathbf 1.
\end{equation}
Together with the thermodynamic-limit assumptions stated in Sec.~\ref{subsec:operational-qrf}---namely that any candidate charged representation in the same stable class is approximable by finite-volume states with uniformly bounded excess energy if it exists, that the relevant energy functional has the required lower-semicontinuity property, and that the distinguishing charge/flux does not escape from the accessible asymptotic algebra---the linear divergence excludes a finite-energy asymptotically isolated color charged representation in that class.  This is a \emph{sufficient} sector-energy exclusion test.  No converse theorem is claimed: bounded finite-volume energy still requires a nontrivial infinite-volume GNS limit and color charge survival.

Thus the QRF criterion in $1+1$ dimensions is not an infrared condition of Kugo-Ojima criterion.  It is a statement about the physical algebra after Gauss law and QRF-type dressing:
the only surviving data are the conjugacy class of global holonomy, boundary color-electric flux, and color-singlet Wilson-line-dressed operators.  Operators with a single uncontracted color index do not descend to the physical algebra.

This model also clarifies the limitation of the lower-dimensional test.
In pure \((1+1)\)-dimensional Yang--Mills theory, the absence of local propagating gluon degrees of freedom is automatic, and magnetic confinement mechanisms are absent.
Therefore, the model is not a substitute for four-dimensional Yang--Mills theory
including the dynamical gluon and glueball spectrum.

Moreover, the nontrivial $j$ sectors of $L^2(SU(2))$ on the interval are most precisely interpreted as boundary color-electric-flux sectors with fixed endpoint-frame data.  They do not represent dynamical colored bulk particles of pure Yang--Mills theory.
The exact statement tested here is that \textbf{such a proper-gauge-invariant boundary/flux sector has an energy proportional to $C_2(j)L$ and therefore cannot survive as a finite-energy isolated color charged sector when the endpoint separation is taken to infinity}.

Its value is that the QRF logic can be checked exactly:
the Wilson-line dressing/frame transporter removes local gauge redundancy, the Gauss law removes bulk color, and the remaining global or boundary data are precisely semi-local.
This model is a minimal exactly solvable test of the gauge-invariant relational boundary color charge and the superselection-sector criterion.

The most important conceptual conclusion is:
Color confinement is not the disappearance of the relational color charge.
\textbf{Color confinement is the absence of a finite-energy realization of a nontrivial relational color charge sector}.

The circle and interval geometries encode different surviving data and should not be conflated. On the circle, the gauge-invariant remnant is the conjugacy class of the global holonomy. On the open interval, endpoint transformations can remain as genuine boundary symmetries and the Peter--Weyl sectors are boundary-flux sectors. It is the latter geometry that supplies the direct sector-energy test of QRF confinement used below.

\section{Topological defects and global definability of QRF observables}
\label{sec:defects-qrf}
\noindent\textit{QRF type(s) in this section:} mathematical/pathwise dressings and patchwise frame data; the emphasis is on their globalization or lift, not on an independently dynamical reference system.\par

Topological defects affect QRF observables in two logically different ways.
They may disorder the vacuum dynamically and may screen or confine color charges;
they may also obstruct a particular global trivialization,
a lift of a partial frame to a complete frame, or a single-valued relational construction in a specified color charge representation.
These topological obstructions must be distinguished from ordinary path dependence generated by nonzero curvature of a smooth connection.
However, neither kind of obstruction, by itself, is identified with confinement.

\subsection{Path-dependent dressings}
\label{subsec:defect-path}

Independently of whether a topological defect is present, a charged bilocal can be made gauge invariant along a specified path by inserting a Wilson line between the two points $x,y$:
\begin{equation}
  \Phi^\dagger(x)\,\mathrm{P}\exp\left(\mathrm{i}\int_{\gamma:y\to x}\mathscr{A}\right)\Phi(y),
  \label{defect-bilocal}
\end{equation}
with representation matrices suppressed.  From the QRF viewpoint, \textbf{the path $\gamma$ is part of the reference-frame data}.  A \textbf{spanning tree or a contour gauge} is a systematic choice of such reference paths.  When all relevant paths are homotopic and no nontrivial holonomy is present, this construction gives a single-valued relational observable.

\begin{figure}[htbp]
  \centering
\includegraphics[width=.8\linewidth]{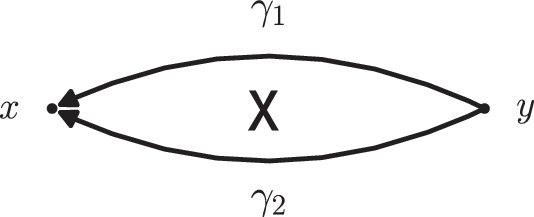}
  \caption{The Wilson lines $W(x;y;\gamma_1)$ and $W(x;y;\gamma_2)$ along two paths $\gamma_1,\gamma_2$ from $y$ to $x$, which cannot be moved into each other by a continuous deformation because of the presence of a defect denoted by the cross.  The two curves $\gamma_1,\gamma_2$ have the same endpoints $x$.}
  \label{fig:QRF_defect}
\end{figure}

If a  topological defect is present, however, two paths with the same endpoints need not be continuously deformable into each other.  Figure~\ref{fig:QRF_defect} shows the situation.  Let $\gamma_1$ and $\gamma_2$ be two paths from $y$ to $x$.
If the two paths pass to the left and right of a defect denoted by the cross, the Wilson lines have the different values:
\begin{equation}
  \mathrm{P}\exp\left(\mathrm{i}\int_{\gamma_1:y\to x}\mathscr{A}\right)
  \not=
  \mathrm{P}\exp\left(\mathrm{i}\int_{\gamma_2:y\to x}\mathscr{A}\right).
  \label{defect-holonomy-ratio-abelian-0}
\end{equation}
For  Abelian dressings, the ratio is given by

\begin{equation}
\begin{aligned}
  &\exp\left(\mathrm{i}\int_{\gamma_1:y\to x}\mathscr{A}\right)
  \left[\exp\left(\mathrm{i}\int_{\gamma_2:y\to x}\mathscr{A}\right)\right]^{-1}\\
  &\qquad=\exp\left(\mathrm{i}\oint_{\gamma_1\circ\gamma_2^{-1}:y\to y}
  \mathscr{A}\right).
\end{aligned}
  \label{defect-holonomy-ratio-abelian}
\end{equation}

For a non-Abelian gauge field the corresponding statement is path ordered:
\begin{equation}
  U_{\gamma_1}U_{\gamma_2}^{-1}
  =U_{\gamma_1\circ\gamma_2^{-1}},
  \label{defect-holonomy-ratio-nonabelian}
\end{equation}
up to the usual convention for orientation and base point.
If the closed curve winds around a defect, the holonomy can be nontrivial.  Then \textbf{the phrase ``the charged field dressed to the reference point'' is not unique unless one specifies on which side of the defect the dressing passes.}

\subsection{Patchwise frames, monodromy, and global definability}
\label{subsec:defect-definability}

The path dependence discussed above has to be separated from \textbf{the genuinely global question of whether a chosen QRF construction can be defined consistently over the whole space}.
Even on a topologically trivial bundle, a smooth connection with $\mathscr{F}\neq 0$ generally produces \textbf{path-dependent parallel transport}.
By contrast, a topological defect or nontrivial bundle structure may obstruct a chosen \textbf{global trivialization}, a \textbf{lift of a partial frame to a complete frame}, or \textbf{a single-valued realization of a relational field} in a specified representation.

A frame need not admit a single global representative.  A useful way to formulate this distinction is to work patchwise.  Let $\{U_\alpha\}$ be an open cover of space and let $h_\alpha$ be local frame representatives on the patches $U_\alpha$.  On an overlap $U_\alpha\cap U_\beta$, suppose that there exists a transition function $g_{\alpha\beta}(x)$ such that
\begin{equation}
  \psi^{(\alpha)}(x)=g_{\alpha\beta}(x)\psi^{(\beta)}(x),
  \qquad
  h_\alpha(x)=g_{\alpha\beta}(x)h_\beta(x).
\end{equation}
Then the corresponding relational fields $\psi_h(x)=h(x)^{-1}\psi(x)$ agree on the overlap $U_\alpha\cap U_\beta$  independently of the patches:
\begin{equation}
\begin{aligned}
  h_\alpha^{-1}(x)\psi^{(\alpha)}(x)
  &=h_\beta^{-1}(x)g_{\alpha\beta}^{-1}(x)
    g_{\alpha\beta}(x)\psi^{(\beta)}(x) \\
  &=h_\beta^{-1}(x)\psi^{(\beta)}(x).
\end{aligned}
\end{equation}
Thus \textbf{the relational observable itself may be globally defined even when the frame representatives are only local.}
  The mere absence of a single frame representative on all patches therefore does not imply that the relational observable is only locally defined.

The \textbf{obstruction} arises when this patching cannot be performed consistently for the specified QRF construction, or when continuation around a noncontractible loop or a topological defect returns a different frame representative.  In such a situation one may need branch cuts, several patches, nontrivial transition functions, or suitable boundary conditions.  For a frame constructed by parallel transport, continuation around a closed curve $C$ acts as
\begin{equation}
  h\longmapsto h\,\Omega_C,
  \qquad
  \Omega_C=
  \mathrm{P}\exp\left(
    \mathrm{i}g_{\rm YM}\oint_C \mathscr{A}
  \right),
  \label{qrf-frame-monodromy}
\end{equation}
where $\Omega_C$ is the \textbf{holonomy} of the connection along $C$.  When $C$ links a topological defect, $\Omega_C$ may be regarded as the corresponding topological-defect holonomy.  A relational field in a representation $R$ then transforms by the corresponding right action,
\begin{equation}
  \Psi_h\longmapsto R(\Omega_C)^{-1}\Psi_h.
\end{equation}
If this \textbf{monodromy} is nontrivial in the sector under consideration and cannot be absorbed into transition data compatible with the prescribed QRF protocol, the chosen relational construction is not globally single-valued.  In that case a statement such as ``this component is red'' becomes path dependent because the color label remains tied to the chosen frame.

This obstruction concerns the globalization of the \emph{specified} frame or transporter construction along the chosen class of paths.  It does not imply that every admissible relational construction fails, and it does not by itself establish confinement.  In non-Abelian examples one must furthermore specify precisely whether the obstruction concerns the principal bundle itself, a chosen local trivialization, or a lift of a partial color direction to a complete frame.  Likewise, the possible existence of a globally defined partial section must be distinguished from the existence of a complete global frame.

It is also essential not to identify nontrivial holonomy with the nonexistence of a global frame.  Even a smooth connection on a trivial principal bundle can have
$\Omega_C\neq\mathbf{1}$ when its curvature is nonzero.  Conversely, nontrivial bundle topology does not imply the absence of every physically meaningful section of an associated bundle.  The following logically distinct issues should therefore be kept separate:
\begin{enumerate}
  \item Ordinary path dependence caused by nonzero curvature.
  \item Nontriviality of a principal or associated bundle.
  \item Multivalued components in a particular local trivialization.
  \item Existence or nonexistence of a global section of a nontrivial associated bundle.
  \item Infrared disorder generated by an ensemble of topological defects.
  \item Absence of a finite-energy asymptotically isolated nontrivial charged sector.
\end{enumerate}
The first four are geometrical or kinematical statements, whereas the last two are dynamical and are the ones relevant to confinement in the operational sense adopted here.  A topologically nontrivial bundle does not by itself eliminate physically meaningful sections.  To infer confinement, an additional dynamical mechanism is required, such as sector averaging, proliferation of topological defects, divergent isolation energy, or loss of physically readable boundary data.

Accordingly, a Wilson line may remain a well-defined semi-local observable, a Wilson loop may remain a frame-independent observable, and a patchwise relational field may remain globally meaningful even when one particular global frame representative does not exist.
\textbf{Topological defects are relevant to the QRF description} not because topology alone proves confinement, but \textbf{because topological-defect-induced monodromy can frustrate a specified long-distance open QRF dressing}.
  When the physical ensemble of topological defects produces the required infrared disorder, \textbf{this obstruction can participate in eliminating a finite-energy asymptotically isolated nontrivial charged sector.}
The topological-defect picture and the QRF description are therefore complementary: the topological obstruction supplies a geometrical mechanism that can obstruct a specified open dressing, while the physical topological-defect ensemble and the infrared dynamics determine whether this obstruction becomes a confinement mechanism.

This observation provides \textbf{a QRF interpretation of, and a strong motivation for, the topological-defect-based extension of the Kugo--Ojima confinement picture} explored by Kondo and Fukushima \cite{KondoFukushima:2022,FukushimaKondo:2025rgs,FukushimaKondo:2023rgs}, where restoration of residual gauge symmetries is tied to the role of topological defects in color confinement.

\subsection{Closed versus open relational operators and the confinement implication}
\label{subsec:defect-open-closed}

The preceding distinction is particularly transparent when closed and open line operators are compared.
\textbf{A traced Wilson loop detects the defect holonomy in a frame-independent way:}
\begin{equation}
  W_R(C)={1\over d_R}\mathrm{tr}_R\,\Omega_C .
\end{equation}
A traced Wilson loop
is a closed, frame-independent observable.
Defects make particularly clear why closed and open line operators play different roles.
It can detect nontrivial holonomy and can be disordered by an ensemble of defects without requiring endpoint reference data.

An open Wilson line is different. It carries endpoint transformation laws.
It becomes a physical relational operator only after the endpoints are supplied with matter fields, boundary degrees of freedom, reference frames, or an appropriate singlet contraction.
If two paths $\gamma_1$ and $\gamma_2$ connect the same endpoints, their relative transporter is the holonomy around the closed contour $\gamma_1\circ\gamma_2^{-1}$.
Consequently, a defect can make the open operator depend on the homotopy class of the chosen dressing path even though the corresponding closed holonomy is itself gauge invariant.

This distinction is important for confinement.
\textbf{The Wilson-loop area law is a statement about the infrared behavior of closed, frame-independent line operators.}
\textbf{QRF color confinement is a statement about the impossibility of promoting open, non-singlet, frame-relative operators to isolated frame-independent asymptotic particle operators.}

\textbf{Topological defects can connect these statements because a holonomy that disorders closed loops can also obstruct a specified globalization or lift of an open
QRF-type dressing;}
\textbf{the resulting confinement statement still requires the corresponding infrared dynamics.
}
monodromy or failure of one chosen frame construction is not by itself sufficient to establish confinement.

\section{The (1+1)-dimensional $U(1)$ gauge--Higgs model}
\label{sec:u1-higgs-qrf}
\noindent\textit{QRF type(s) in this section:} matter QRF supplied locally by the Higgs phase, supplemented by mathematical Wilson-line dressings and boundary/holonomy data.\par

The $1+1$-dimensional $U(1)$ gauge--Higgs model is the Abelian counterpart of the preceding Yang--Mills example.
Because the gauge group is Abelian, the relation among QRF-type dressing, Gauss law, boundary charge, and global holonomy can be written explicitly.
At the same time, the presence of dynamical Higgs matter changes the physical meaning of confinement: color electric flux can be screened by charged matter, and for a charge-one Higgs field the Higgs and confinement regions may be analytically connected in the Fradkin--Shenker sense \cite{Fradkin:1978dv}.
\textbf{The QRF criterion for color confinement is therefore not simply the existence of a strictly linear potential.  It is the structure of the physical algebra: bare charged local operators are absent, while boundary fluxes, holonomies, and Wilson-line-dressed neutral composites remain.
}

\subsection{Continuum model and boundary charge algebra}
\label{subsec:u1-continuum}

We consider the (1+1)-dimensional $U(1)$ gauge--Higgs model where the complex Higgs field $\phi$ has charge $q\in\mathbb{Z}_{>0}$.
In the temporal gauge $A_0=0$, the continuum action is given by

\begin{equation}
\begin{aligned}
 S&=\int d^2x\,\mathscr{L},\\
 \mathscr{L}&={e^2\over2}E^2+|\partial_0\phi|^2-|D_1\phi|^2-V(|\phi|),\\
  E&=\frac{F_{01}}{e^2}, \qquad
  D_1\phi:=(\partial_1-\mathrm{i}qA_1)\phi,
 .
\end{aligned}
  \label{u1-lagrangian}
\end{equation}

The Gauss law is written as
\begin{equation}
  G(x):=\partial_1E(x)-\rho(x)\approx0,
  \quad
  \rho:=\mathrm{i}q\bigl(\pi_\phi\phi-\phi^\dagger\pi_\phi^\dagger\bigr),
  \label{u1-gauss}
\end{equation}
where $\pi_\phi=(D_0\phi)^\dagger$.

On an open interval $I=(0,L)$ of length $L$, we introduce the Gauss law smeared by $\omega(x)$:
\begin{equation}
  G[\omega]=\int_0^Ldx\,\omega(x)(\partial_1E-\rho).
\end{equation}
Integration by parts gives
\begin{equation}
  G[\omega]
  =-\int_0^Ldx\,(\partial_1\omega)E-\int_0^Ldx\,\omega\rho
  +[\omega E]_0^L .
\end{equation}
Thus gauge transformations with zero boundary values ($\omega(0)=0$ and $\omega(L)=0$) are proper gauge transformations generated by the constraint. Transformations with nonzero boundary values ($\omega(0)\not=0$ or $\omega(L)\not=0$) are boundary-nontrivial transformations; they generate boundary symmetries and boundary charges. They are not necessarily large in the homotopical sense.

By introducing a \textbf{differentiable improved generator} (see Appendix \ref{app:differentiable_generator} for the details on the differentiable improved generator),
\begin{equation}
  \widetilde G[\omega]
  =\int_0^Ldx\,\omega(\partial_1E-\rho)-\omega(L)E(L)+\omega(0)E(0),
  \label{improved generator}
\end{equation}
we define the outward boundary flux on the interval by
\begin{equation}
 Q_{\partial I}^{\mathrm{out}}[\omega]
 :=\omega(L)E(L)-\omega(0)E(0),
\end{equation}
and obtain on the constraint surface
\begin{equation}
 \widetilde G[\omega]\approx- Q_{\partial I}^{\mathrm{out}}[\omega].
\end{equation}
Equivalently, the endpoint generator charges are
\begin{equation}
  Q_L:=E(0),
  \quad
  Q_R:=-E(L),
  \label{u1-boundary-charges}
\end{equation}
so that $\omega(0)Q_L+\omega(L)Q_R=-Q_{\partial I}^{\mathrm{out}}[\omega]$.
integrating the Gauss law gives the relationship:
\begin{equation}
  E(L)-E(0)=\int_0^Ldx\,\rho(x)
  \Leftrightarrow
  Q_L+Q_R+\int_0^Ldx\,\rho(x)=0 .
  \label{u1-integrated-gauss}
\end{equation}
Thus \textbf{the total bulk charge is balanced by boundary electric flux.}
This is the Abelian version of the statement that \textbf{charge remains as {boundary data}, not as an independent local bulk observable.}

\subsection
{QRF-type dressing and physical algebra on the interval}
\label{subsec:u1-open-qrf}

On an open interval $I=(0,L)$ of length $L$, take the left endpoint $x=0$ as the base point and define the Wilson-line dressing/frame transporter
\begin{equation}
  \boxed{h_q(x):=\exp\left(\mathrm{i}q\int_0^xdy\,A_1(y)\right),} \quad h_q(0)=1 .
  \label{u1-R}
\end{equation}
Under the $U(1)$ gauge transformation,
\begin{equation}
  A_1\mapsto A_1+\partial_1\omega,
  \quad
  \phi\mapsto e^{\mathrm{i}q\omega(x)}\phi,
\end{equation}
the Wilson-line dressing/frame transporter transforms as
\begin{equation}
  h_q(x)\mapsto e^{\mathrm{i}q(\omega(x)-\omega(0))}h_q(x).
\end{equation}
Therefore the dressed Higgs field defined by
\begin{equation}
  \boxed{\Phi_{h_q}(x):=h_q(x)^{-1}\phi(x)}
  \label{u1-dressed-phi}
\end{equation}
transforms only under the \textbf{residual left-boundary symmetry},
\begin{equation}
  \Phi_{h_q}(x)\mapsto e^{\mathrm{i}q\omega(0)}\Phi_{h_q}(x) .
\end{equation}
Therefore, \textbf{the local gauge redundancy has disappeared, but the global boundary charge still remains.}

A gauge-invariant bilocal relational observable is given by
\begin{align}
  \phi^\dagger(x)\exp\left(\mathrm{i}q\int_y^x dz\,A_1(z)\right)\phi(y)
  &=\Phi_{h_q}^\dagger(x)\Phi_{h_q}(y).
  \label{u1-bilocal}
\end{align}
The boundary-symmetry-invariant observable algebra on the interval relevant to the present discussion is generated schematically by
\begin{equation}
  \mathcal{A}_{\rm obs}(I)
  =\left\langle
  E(0),  E(L),
  |\Phi_{h_q}(x)|^2,
  \Phi_{h_q}^\dagger(x)\Phi_{h_q}(y)
  \right\rangle .
  \label{u1-physical-algebra-interval}
\end{equation}
If the residual left-boundary symmetry is retained as a genuine physical symmetry, the proper-gauge-invariant dressed field $\Phi_{h_q}$ belongs naturally to the corresponding extended charged/field algebra rather than to the boundary-symmetry-invariant observable subalgebra; it transforms nontrivially under that retained boundary symmetry as shown below.
The boundary charge acts on the dressed field as
\begin{equation}
  [Q_L,\Phi_{h_q}(x)]=-q\Phi_{h_q}(x),
  \quad
  [Q_L,\Phi_{h_q}^\dagger(x)]=q\Phi_{h_q}^\dagger(x),
  \label{u1-boundary-action}
\end{equation}
up to conventional factors of $\mathrm{i}$.  This displays the same structure as the non-Abelian interval:
\textbf{QRF-type dressing removes local gauge redundancy, but the boundary symmetry is still physically represented.}

In the gauge--Higgs model the interpretation differs from pure Yang--Mills theory.
\textbf{Dynamical Higgs particles can screen electric flux, therefore, strict linear confinement of all charges is not the generic statement.}
\textbf{The robust QRF statement is that the bare charged field $\phi(x)$ is not a proper-gauge-invariant observable.  Neutral composites and boundary fluxes belong to the boundary-symmetry-invariant observable algebra, while a dressed boundary-charged field such as $\Phi_{h_q}$ may belong to the extended proper-gauge-invariant charged/field algebra when the residual boundary symmetry is retained physically.}

\subsection{Global holonomy on the circle}
\label{subsec:u1-circle}

On a circle $S^1$ of length $L$ there is no boundary charge.  Instead, there remains the global holonomy:
\begin{equation}
  U:=\exp\left(\mathrm{i}\oint_{S^1} dx\,A_1(x)\right)=e^{\mathrm{i}\Theta},
  \quad
  \Theta\in[0,2\pi).
  \label{u1-holonomy}
\end{equation}
In the pure gauge sector the local Gauss law implies that the electric field is constant, $E(x)=E_0$, and has the commutation relation:
\begin{equation}
  [E_0,U]=U,
  \quad
  E_0=-\mathrm{i}{d\over d\Theta} .
\end{equation}
The pure gauge Hilbert space is the rotor Hilbert space
\begin{equation}
  \mathcal{H}^{\rm gauge}_{\rm phys}\simeq L^2(U(1)),
  \quad
  \psi_n(\Theta)=e^{\mathrm{i}n\Theta},
  \quad n\in\mathbb{Z},
\end{equation}
with the Hamiltonian and the eigenvalues:
\begin{equation}
  H_{\rm gauge}={e^2L\over2}E_0^2,
  \quad
  H_{\rm gauge}\psi_n={e^2L\over2}n^2\psi_n .
\end{equation}

Even when the Higgs field is included, the local phase can be removed by the tree-like QRF.  But on the circle the dressing is not periodic unless the holonomy is trivial.  Defining
\begin{equation}
  h_q(x) :=\exp\left(\mathrm{i}q\int_0^xdy\,A_1(y)\right),
\end{equation}
and
\begin{equation}
  \Phi_{h_q}(x) :=h_q(x)^{-1}\phi(x),
\end{equation}
we find
\begin{equation}
  \Phi_{h_q}(L)=e^{-\mathrm{i}q\Theta}\Phi_{h_q}(0).
  \label{u1-twisted-boundary}
\end{equation}
Indeed, we find $\Phi_{h_q}(L) :=h_q(L)^{-1}\phi(L)=(e^{i\Theta})^{-q} \phi(0)=e^{-iq\Theta}\Phi_{h_q}(0)$.
Thus \textbf{the QRF removes local gauge redundancy but leaves the loop degree of freedom as physical data.}
  The physical algebra on the circle relevant to the present discussion is generated schematically by
\begin{equation}
  \mathcal{A}_{\rm phys}(S^1)
  =\left\langle
  U=e^{\mathrm{i}\Theta},\;|\Phi_{h_q}(x)|^2,\;\Phi_{h_q}^\dagger(x)\Phi_{h_q}(y)
  \right\rangle .
  \label{u1-phys-alg-circle}
\end{equation}

The same content is seen on a lattice.  On an open chain the
spanning-tree dressing/frame transporter
is the product of link variables from the left root to the site, and Gauss law expresses the total charge as the difference of endpoint fluxes.  On a periodic chain, removing one link leaves a single loop variable
\begin{equation}
  W=\prod_{n=1}^N U_n,
\end{equation}
while the local gauge phases are absorbed into tree dressings.
The lattice boundary-symmetry-invariant observable algebra is generated by the loop variable, local neutral densities, and dressed bilocals.  This is the discrete counterpart of \eqref{u1-physical-algebra-interval} and \eqref{u1-phys-alg-circle} \cite{Kogut:1974ag}.

\subsection{Compact Euclidean model, vortices, and $N$-ality}
\label{subsec:u1-vortices}

To see topological defect obstruction explicitly, consider the two-dimensional Euclidean compact $U(1)$ gauge--Higgs model with a charge-$N$ Higgs field,

\begin{equation}
\begin{aligned}
  S=\int d^2x\Bigl[&{1\over4e^2}F_{\mu\nu}F_{\mu\nu}
  +| (\partial_\mu-\mathrm{i}N A_\mu)\Phi |^2\\
  &+\lambda(|\Phi|^2-v^2)^2\Bigr].
\end{aligned}
  \label{u1-euclidean-action}
\end{equation}

In what follows, we consider the London limit $\lambda \to \infty$.  Then the radial mode $|\Phi(x)|$ of the Higgs field $\Phi(x)$ is frozen pointwise, $|\Phi(x)|=v$ and the action $S$ reduces to $S_L$:

\begin{equation}
\begin{aligned}
  \Phi(x)&=v e^{\mathrm{i}\chi(x)},\\
  S_L&=\int d^2x\left[
  {1\over4e^2}F_{\mu\nu}F_{\mu\nu}
  +
  v^2(\partial_\mu\chi-NA_\mu)^2
  \right].
\end{aligned}
  \label{u1-london}
\end{equation}

The gauge transformation is replaced by
\begin{equation}
  A\mapsto A+d\omega,
  \quad
  \chi\mapsto\chi+N\omega .
\end{equation}

\noindent
(i)
The phase of the Higgs field supplies local matter-frame data, represented more precisely by a local Higgs-frame lift (local section) of the $N$-fold covering map $U(1)\to U(1)$, $u\mapsto u^N$:
\begin{equation}
  \boxed{h_N(x):= \left( \frac{\Phi(x)}{|\Phi(x)|} \right)^{1/N}
  =\exp\left({\mathrm{i}\chi(x)\over N}\right).}
  \label{u1-h-frame-vortex}
\end{equation}
On overlaps, different local lifts can differ by a transition function in the transformation group $\mathbb{Z}_N$.  Thus $h_N$ should not be assumed to be a globally single-valued $U(1)$ function in a vortex sector; its $\mathbb{Z}_N$ transition data are precisely the global information relevant below.
For a charge-$q$ field $\psi_q$ with the local gauge transformation
\begin{equation}
  \psi_q\mapsto e^{\mathrm{i}q\omega}\psi_q,
\end{equation}
the relational operator defined by
\begin{equation}
  \Psi_{q,h}(x) :=h_N(x)^{-q}\psi_q(x)
  \label{u1-rel-charge-vortex}
\end{equation}
 is locally gauge invariant.

 \begin{figure}[htbp]
  \centering
\includegraphics[width=.3\linewidth]{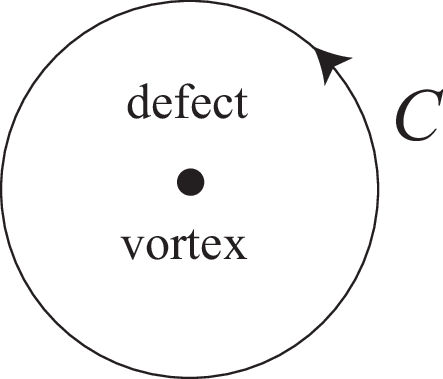}
  \caption{A vortex and the contour $C$ winding around the vortex.}
  \label{fig:vortex}
\end{figure}

\textbf{In a vortex background the frame is not globally single-valued.}
  Around a vortex of winding number $m$, $\chi$ obeys
\begin{equation}
  \oint_C d\chi=2\pi m, \quad m \in \mathbb{Z}.
\end{equation}
See Figure \ref{fig:vortex}.
Therefore we find around a vortex
\begin{equation}
 h_N\longmapsto e^{2\pi\mathrm{i}m/N}h_N ,
  \label{u1-psi-monodromy}
\end{equation}
and therefore a charge-$q$ relational field undergoes the monodromy:
\begin{equation}
 \Psi_{q,h}\longmapsto e^{-2\pi i mq/N}\Psi_{q,h}.
\end{equation}
If $q\not\equiv0\pmod N$, the relational representative constructed from this particular Higgs-phase lift acquires nontrivial $\mathbb{Z}_N$ monodromy.  What is established here is therefore the more specific statement that \textbf{the Higgs-phase realization of a complete charge-$q$ frame cannot be globalized as a single-valued function without patching or line attachment when the vortex monodromy acts nontrivially.}  This does not exclude Wilson-line dressings, patchwise sections, or associated-bundle descriptions; nor does the obstruction of this selected lift by itself prove the absence of every finite-energy charged sector.

\noindent
(ii)
The same obstruction is detected by Wilson loops.
Outside a finite-action vortex, it holds
\begin{equation}
  \partial_\mu\chi-NA_\mu\simeq0,
\end{equation}
which implies that \textbf{a vortex of winding number $m$ carries flux}
\begin{equation}
  \oint_C A={1\over N}\oint_C d\chi={2\pi m\over N} .
\end{equation}
See Figure \ref{fig:vortex}.
When a vortex of winding number $m$ is enclosed,
the Wilson loop of charge $q$ defined by
\begin{equation}
  W_q(C):=\exp\left(\mathrm{i}q\oint_C A\right),
\end{equation}
acquires the \textbf{Aharonov--Bohm phase}
\begin{equation}
  W_q(C)\longmapsto e^{2\pi\mathrm{i}mq/N}W_q(C).
\end{equation}

Assume a dilute gas in which vortices with $m=+1$ and antivortices with $m=-1$ are independently Poisson distributed, each with density $\zeta$ per unit area.  Thus the mean number of each species in the area $A(C)$ is $\zeta A(C)$ (and the mean total defect density is $2\zeta$).  Then we obtain
\begin{align}
 & \langle W_q(C)\rangle_{\rm vort} \nonumber\\
  &=\exp\left[
  \zeta A(C)(e^{2\pi\mathrm{i}q/N}-1)
  +\zeta A(C)(e^{-2\pi\mathrm{i}q/N}-1)
  \right]
  \nonumber\\
  &=\exp\left[-2\zeta A(C)
  \left(1-\cos{2\pi q\over N}\right)\right]
  \nonumber\\
  &=\exp\left[-\sigma_q A(C)\right].
  \label{u1-vortex-area-law}
\end{align}
The string tension is therefore given by
\begin{equation}
  \sigma_q=2\zeta\left(1-\cos{2\pi q\over N}\right),
  \label{u1-sigma-q}
\end{equation}
which means, $\sigma_q=0$ for $q\equiv0\pmod N$ and $\sigma_q>0$ for nontrivial $\mathbb{Z}_N$ charge $q\not\equiv 0\pmod N$.
Under the additional dilute Poisson-gas assumption used in this calculation, the same center monodromy carried by the specified Higgs-frame lift enters the Wilson-loop phase average and produces the area law for nontrivial $\mathbb Z_N$ charge.  The implication is therefore dynamical and ensemble dependent, rather than an equivalence between topology alone and confinement:

\begin{equation}
  \boxed{
  \begin{gathered}
  \text{specified Higgs-frame monodromy}
  \\
  +\ \text{dilute defect ensemble}\\
  \Longrightarrow
  \text{Wilson-loop area law for nontrivial }\mathbb{Z}_N\text{ charge}
  \end{gathered}
  }
\nonumber
  \label{u1-qrf-vortex-equivalence}
\end{equation}

A more precise dictionary for this controlled model is

\begin{equation}
\begin{aligned}
&\hbox{Higgs phase }\chi
\longleftrightarrow \hbox{specified local Higgs-frame lift},\\
&\hbox{vortex}
\longleftrightarrow \hbox{monodromy obstructing}\\
&\qquad\hbox{globalization of that lift},\\
&\hbox{nontrivial }\mathbb{Z}_N\hbox{ charge}
\longleftrightarrow \hbox{nontrivial monodromy}\\
&\qquad\hbox{of the specified charge-sensitive lift},\\
&\hbox{dilute vortex gas}
\longrightarrow \hbox{Wilson-loop area law}\\
&\qquad\hbox{under the Poisson assumption}.
\end{aligned}
\nonumber
\end{equation}

If the Higgs charge is one $(N=1)$, the Aharonov--Bohm phase for integer-charge $q$ Wilson loops is trivial and this $N$-ality selection disappears.
Therefore the charge-$N$ model is the minimal Abelian example in which Higgs-frame monodromy and confinement-type selection can both be seen explicitly.

\subsection{From the vortex area law to the charge-sector energy}

The vortex-gas calculation above can be connected directly to the charged-sector criterion, rather than being interpreted only as a Wilson-loop diagnostic.  This connection is important because the Wilson-loop area law and color confinement are closely related but are not identical statements in a gauge theory with dynamical matter.

The Wilson loop is a diagnostic for an external static probe.  By contrast, the color confinement criterion adopted in this paper is a sector statement: after all admissible dressings, boundary conditions, screening processes, defect-sector prescriptions, and infrared limits have been taken into account, \textbf{there should exist no finite-energy, physically accessible, asymptotically isolated realization of a nontrivial stable charge class}.  Thus an area law and the absence of all isolated non-singlet sectors are not equivalent in general.

The present Abelian model also does not diagnose the full Lie-algebra-valued color structure.
\textbf{It distinguishes the stable electric charge classes that remain after screening by dynamical charge-$N$ matter.}
\textbf{Therefore the precise statement established below is confinement of nontrivial $\mathbb Z_N$ screening classes, not general non-Abelian color confinement.}

\paragraph{Screening classes in the charge-$N$ Abelian Higgs model.}

For the present $U(1)$ model the \textbf{probe-charge lattice} is
\begin{equation}
 \Gamma_{\rm e}=\mathbb Z.
\end{equation}
The dynamical Higgs field carries charge $N$, therefore Higgs excitations can change the probe charge by an integer multiple of $N$.
Hence the \textbf{subgroup generated by the charges of dynamical screening excitations} is given by
\begin{equation}
 \Gamma_{\rm scr}=N\mathbb{Z},
 \qquad
 [q]=q+N\mathbb{Z}
 \in \frac{\mathbb{Z}}{N\mathbb{Z}}
 \simeq\mathbb{Z}_N.
 \label{eq:u1-screening-class}
\end{equation}
Consequently, we find the
identification for the screening class $[q]$ of the probe charge $q$:
\begin{equation}
 [q]=[0]
 \quad\Longleftrightarrow\quad
 q=0\pmod N.
 \label{eq:u1-trivial-screening-class}
\end{equation}
\textbf{A probe with $q=kN$ can be neutralized by dynamical charge-$N$ Higgs excitations.
By contrast, if $q\not\equiv0\pmod N$, no finite number of Higgs excitations can remove the residual class $[q]$.}

\textbf{For $N=1$, $\mathbb{Z}_1$ is trivial $\mathbb Z_1=\{[0]\}$ and every integer probe charge belongs to the trivial screening class.}
 The vortex Aharonov--Bohm factor for integer Wilson loops is then trivial.  The absence of an area law in this particular $\mathbb Z_N$ diagnostic should therefore not be interpreted as the appearance of a deconfined nontrivial color sector; rather,
\textbf{there is no nontrivial stable electric screening class left for this diagnostic to confine}.

\paragraph{Transfer-matrix interpretation and the static probe energy.}
Consider a rectangular Wilson loop of spatial extent $\ell$ and Euclidean-time extent $t$.  In a transfer-matrix interpretation one may write schematically
\begin{equation}
\begin{aligned}
 \langle W_q(\ell,t)\rangle
 =\sum_{n}|c_n(\ell)|^2
 \exp\left[-t\left(E_n^{(q,-q)}(\ell)-E_{\rm vac}\right)\right],
\end{aligned}
 \label{eq:u1-transfer-matrix-static-sector}
\end{equation}
where $E_n^{(q,-q)}(\ell)$ denotes an energy level in the sector containing a static probe of charge $q$ and an oppositely charged compensating probe at separation $\ell$.  Provided that the Wilson-loop operator has nonzero overlap with the lowest-energy state in this sector, the large-$t$ limit extracts the minimum static-pair energy,
\begin{equation}
\begin{aligned}
 V_{[q]}^{\rm probe}(\ell)
 &:=-\lim_{t\to\infty}\frac{1}{t}
 \log\langle W_q(\ell,t)\rangle\\
 &=E_0^{(q,-q)}(\ell)-E_{\rm vac},
\end{aligned}
 \label{eq:u1-static-pair-energy}
\end{equation}
up to the usual perimeter and probe self-energy renormalizations.

The vortex-gas area law therefore gives
\begin{equation}
\begin{aligned}
 V_{[q]}^{\rm probe}(\ell)&=\sigma_{[q]}\ell+O(1),\\
 \sigma_{[q]}&=2\zeta\left(1-\cos\!\left(\frac{2\pi q}{N}\right)\right).
\end{aligned}
 \label{eq:u1-static-pair-linear}
\end{equation}
The notation $\sigma_{[q]}$ emphasizes that the \textbf{asymptotic string tension} depends only on the \textbf{screening class $[q]$ of the probe charge $q$}.
Equivalently,
\begin{equation}
 \sigma_{[q]}
 \begin{cases}
 =0, & [q]=[0],\\
 >0, & [q]\neq[0].
 \end{cases}
 \label{eq:u1-string-tension-by-class}
\end{equation}
For a nontrivial class $[q]\neq[0]$,

\begin{equation}
\begin{aligned}
 V_{[q]}^{\rm probe}(\ell)
 =\sigma_{[q]}\ell+O(1)
 \longrightarrow+\infty \ (\ell\to\infty),
 \ [q]\neq[0].
\end{aligned}
 \label{eq:u1-probe-energy-divergence}
\end{equation}

Thus the compensating opposite charge cannot be moved to infinity at finite static-probe energy.

\paragraph{From the static probe energy to the screening-class energy.}
The static-probe energy and the stable screening-class energy should be distinguished.  Let
\begin{equation}
 \Delta E_{[q]}(L)
 :=E_{[q]}(L)-E_{[0]}(L)
 \label{eq:u1-screening-class-energy}
\end{equation}
be the sector energy minimized over all admissible relational realizations of the fixed stable screening class under the chosen protocol.  The Wilson-loop calculation directly establishes the divergence of the lowest static-pair energy in Eq.~\eqref{eq:u1-probe-energy-divergence}.  To infer the same divergence for the fully minimized screening-class energy, one must assume that the fixed protocol and infrared dynamics do not admit an alternative finite-energy realization of the same nontrivial class that escapes the static-probe sector.  Under this assumption, we have
\begin{equation}
\begin{aligned}
 \Delta E_{[q]}(L)
 =\sigma_{[q]}L+O(1)
 \longrightarrow \infty \ (L\to\infty), \ [q]\neq[0].
\end{aligned}
 \label{eq:u1-nontrivial-class-energy-divergence}
\end{equation}
Hence \textbf{there is no finite-energy asymptotically isolated realization of the nontrivial stable screening class}.  In this restricted and controlled sense,
\textbf{the vortex ensemble gives a dynamical realization of the operational QRF confinement criterion for $[q]\neq[0]$}.

\paragraph{Why the trivial class $[0]$ is not a deconfined charged sector.}
For $q=kN$, the area term vanishes because $[q]=[0]$.  A fixed external source of charge $kN$ can be surrounded by a dynamical screening cloud, and the renormalized energy of a widely separated screened source--antisource pair may approach a finite value,
\begin{equation}
 V_{kN}^{\rm probe}(\ell)
 \longrightarrow V_{kN}^{\rm probe}(\infty)<\infty
 \qquad (\ell \to\infty).
 \label{eq:u1-screened-probe-finite-energy}
\end{equation}
This does not mean that a nontrivial charged sector has become deconfined.
\textbf{The screening process has removed the charge label itself}:
\begin{equation}
 q=kN
 \quad\longrightarrow\quad
 [q]=[0].
 \label{eq:u1-screened-to-trivial}
\end{equation}
The infrared object is a neutral screened composite, not an isolated object carrying a nontrivial stable gauge charge.

It is therefore useful to distinguish two different energies.  If a fixed external source $q=kN$ is retained as part of the experimental protocol, its screening cloud may have a finite, nonzero preparation energy, as in Eq.~\eqref{eq:u1-screened-probe-finite-energy}.  By contrast, if one minimizes over all physical states in the stable screening class $[0]$, then the vacuum itself is included and
\begin{equation}
 \Delta E_{[0]}(L)
 =\inf_{|\Psi\rangle\in\mathcal{H}^{\rm phys}_{[0]}(L)}
 \langle\Psi|H_L|\Psi\rangle-E_{\rm vac}(L)
 =0,
 \label{eq:u1-trivial-class-energy}
\end{equation}
provided that no additional constraint forces an external source to remain in the state.  There is therefore no contradiction between a finite energy for a fixed screened probe and zero energy for the lowest state in the trivial screening class.

\paragraph{State-space interpretation.}
For boundary conditions under which the screening classes are meaningful, one may write schematically at finite volume
\begin{equation}
 \mathcal{H}_{\rm phys}(L)
 =\bigoplus_{[q]\in\mathbb{Z}_N}
 \mathcal{H}^{\rm phys}_{[q]}(L),
 \label{eq:u1-screening-sector-decomposition}
\end{equation}
with the understanding that the precise decomposition depends on the boundary conditions and on which boundary or asymptotic charge algebra is retained.

The sector $\mathcal{H}^{\rm phys}_{[0]}(L)$ is not merely a single vacuum vector.  It contains the vacuum, neutral Higgs excitations, neutral bound states, and completely screened configurations.  A completely screened $q=kN$ source belongs to this same infrared charge class; the finite energy of its local screening cloud does not by itself define a new nontrivial charged superselection sector.

In infinite-volume algebraic language, if there is no additional conserved quantum number, a screened configuration in the trivial class is naturally regarded as an ordinary finite-energy excitation in the vacuum GNS representation.  By contrast, for $[q]\neq[0]$, Eq.~\eqref{eq:u1-nontrivial-class-energy-divergence} excludes a finite-energy asymptotically isolated charged representation under the assumptions stated above.

The two outcomes are therefore logically distinct:
\begin{align}
 &[q]=[0]:\nonumber\\
 &\text{removal by dynamical screening, not deconfinement:}\nonumber\\
 &\text{no nontrivial stable charge label survives},
 \nonumber\\
 &[q]\neq[0]:\nonumber\\
 &\text{confinement by divergent isolation energy:}\nonumber\\
 &\text{the residual stable charge label survives screening,}\nonumber\\
 &\text{but no finite-energy isolated realization survives}.
\nonumber
\end{align}

\textbf{The absence of an asymptotically isolated unscreened charge is common to the two cases, but the mechanisms are different.}
\textbf{It is therefore preferable not to call the trivial class $[0]$ ``confined.''}
\textbf{It is screened to the vacuum charge class}.
\textbf{The term confinement is reserved here for a nontrivial stable screening class whose finite-energy isolated realization is dynamically eliminated.}

\paragraph{Relation to the general QRF confinement criterion.}
The logical chain established in this toy model may be summarized as
\begin{equation}
\left.
\begin{array}{c}
\text{charge-sensitive QRF monodromy}\\[1mm]
\text{assumed dilute (anti)vortex ensemble}
\end{array}
\right\}
\Rightarrow \sigma_{[q]}>0 .
\nonumber
\end{equation}
The subsequent inference to $\Delta E_{[q]}(L)\to\infty$ is a conditional corollary: it also assumes that minimization over the full stable screening class leaves no alternative finite-energy realization of the same nontrivial class. Thus single-defect monodromy is kinematical, whereas proliferation/ensemble weights and sector-energy divergence are dynamical statements.

The Wilson-loop area law here is a statement about the specified static-probe class. It implies absence of a finite-energy isolated dynamical sector only after minimization over the full stable screening class and after all matter-induced string-breaking/screening channels allowed by the protocol have been included.

This goes beyond stopping at the Wilson-loop area law because the final statement is formulated directly in terms of the survival or disappearance of a stable screening class in the isolated-charge limit.

The calculation establishes a confinement statement for the nontrivial $\mathbb Z_N$ probe screening classes of this compact Abelian gauge--Higgs model under the dilute-vortex-gas, transfer-matrix, and infrared assumptions stated above.
Its role is to provide a controlled example in which QRF monodromy, defect disorder, the Wilson-loop area law, and divergent isolation energy are connected in one and the same calculation.

\textbf{QRF is therefore not a mechanism that by itself already proves general color confinement dynamically.}
\textbf{It is a framework that formulates, in a gauge-invariant, boundary-sensitive, and sector-theoretic manner, what must be shown.}
\textbf{In this section, the dilute vortex gas provides a controlled dynamical realization of the operational QRF confinement criterion for nontrivial $\mathbb Z_N$ screening classes.}

\section{Symmetric instantons and dimensional reduction}
\label{sec:symmetric-instantons}
\noindent\textit{QRF type(s) in this section:} this section provides the geometric reduction framework rather than a new independent QRF; the reduced Higgs fields become matter or
intrinsic partial-frame data
in the subsequent $\mathbb{H}^2$ and $\mathbb{H}^3$ sections.\par

\subsection{Why symmetric instantons are relevant for confinement}

The dual-superconductor picture \cite{Nambu:1974,tHooft:1975,Mandelstam:1976} for quark confinement suggests that lower-dimensional magnetic defects, monopoles \cite{Dirac:1931} and vortices, should play a central role in the disordering of color at long distance.  In a pure four-dimensional Yang--Mills theory, however, there is no elementary scalar field and no ordinary smooth finite-action soliton of the Nielsen--Olesen \cite{NO:1973} type for vortices or 't Hooft--Polyakov \cite{tHP:1974} type for magnetic monopoles.
We restrict to smooth finite-action (anti-) self-dual Yang--Mills configurations and to their $SO(2)$- and $SO(3)$-symmetric sectors. No classification claim about all finite-action Euclidean Yang--Mills configurations is required for the argument. The normalization of the reduced fields, Hodge-star conventions, and the relation between reduced vortex/monopole numbers and the four-dimensional instanton number are understood with the conventions stated in the corresponding reduction formulas.
\textbf{A possible way to reconcile the defect picture with the present construction is to regard lower-dimensional defects as shadows of these symmetric instantons after imposing a spacetime symmetry and using the conformal invariance of the classical Yang--Mills equations}.

Let $\mathscr{A}=\mathscr{A}_\mu dx^\mu$ be a four-dimensional $SU(2)$ gauge field on Euclidean $\mathbb{R}^4$, with field strength
\begin{equation}
  \mathscr{F}=d\mathscr{A}-i\mathscr{A}\wedge\mathscr{A}.
\end{equation}
In what follows, the Yang--Mills coupling constant $g_{\rm YM}$ is absorbed into the gauge field $g_{\rm YM}\mathscr{A} \to \mathscr{A}$.
The self-duality equation
\begin{equation}
  \mathscr{F}=*_{\mathbb{R}^4}\mathscr{F}
  \label{eq:self-dual-four}
\end{equation}
is conformally invariant.  Therefore, after removing the fixed set of a rotation and rewriting the flat metric up to a conformal factor, one can interpret  \textbf{a symmetric instanton as a lower-dimensional defect on a hyperbolic space}.  This is the basic mechanism behind Atiyah's relation between circle-symmetric instantons and hyperbolic magnetic monopoles, and Witten's relation between spherically symmetric instantons and hyperbolic vortices \cite{Atiyah:1984,Witten:1977,ForgacsManton:1980,MantonSutcliffe:2004,Maldonado:2015}.
See e.g. \cite{Kondo:2025symmetric} for a review.

For the QRF interpretation, the point is not merely that the equations are reduced.
The point is that \textbf{the field component along the compact orbit of the symmetry becomes a Higgs-like reference direction}.
\textbf{It is a candidate color frame.}
\textbf{The defect number of the reduced solution measures the obstruction to extending this frame globally.}
Thus the symmetric-instanton sector supplies a controlled bridge:
\begin{align}
  &\text{four-dimensional instanton sector}
  \nonumber\\
  &\longrightarrow \text{hyperbolic defects}
  \nonumber\\
  &\longrightarrow \text{obstruction to globalizing the partial-frame lift}
  \nonumber\\
  &\longrightarrow \text{Wilson-loop disorder}.
 \nonumber
  \label{eq:symmetric-qrf-bridge}
\end{align}

This bridge is sector-limited, but it is much more rigid than an analogy: \textbf{the reduced actions, the boundary terms, and the defect equations are obtained by dimensional reduction of the Yang--Mills action and the self-duality equation.}

\subsection{The $SO(2)$ reduction: $\mathbb{R}^4\setminus\mathbb{R}^2\simeq \mathbb{H}^3\times S^1$}
\label{subsec:so2-reduction-general}

First consider the circle symmetric case, i.e., spatial rotation group $SO(2)\simeq S^1$.
Following the same coordinate convention as that used
in \cite{Kondo:2025symmetric},
we start from the four-dimensional Euclidean space $(x^1,x^2,x^3,x^4) \in \mathbb{R}^4$ and introduce polar coordinates in the $(x^1,x^2)$ plane:
\begin{equation}
  x^1=\rho\cos\varphi,
  \
  x^2=\rho\sin\varphi,
  \quad
  \rho>0,  \
  \varphi\in[0,2\pi),
  \label{eq:h3-polar-convention}
\end{equation}
and keep the remaining coordinates $x^3$ and $x^4$.
Then factoring out the $\rho^2$ as the \textbf{conformal factor} yields
\begin{align}
  ds^2(\mathbb{R}^4)
  &=(d\rho)^2+\rho^2(d\varphi)^2+(dx^3)^2+(dx^4)^2  \nonumber\\
  &=\rho^2\left[\frac{(d\rho)^2+(dx^3)^2+(dx^4)^2}{\rho^2}+d\varphi^2\right].
  \label{eq:r4-h3-s1-metric}
\end{align}
Thus, after removing the $\rho=0$ region, i.e., $x^3-x^4$ plane $\mathbb{R}^2$, the \textbf{conformal equivalence} is given by

\begin{equation}
\begin{aligned}
  &\mathbb{R}^4\setminus\mathbb{R}^2
  \simeq \mathbb{H}^3\times S^1,\\
  &\mathbb{H}^3=\{X=(\rho,x^3,x^4)\,|\,\rho>0\}, \\
  &S^1 =\{\varphi\mid\varphi\in[0,2\pi)\} .
\end{aligned}
  \label{eq:r4-h3-s1}
\end{equation}

To restore a length scale without introducing two independent radii, we choose the conformally equivalent product metric
\begin{equation}
  d\widetilde s^{\,2}
  :={\ell^2\over\rho^2}ds^2(\mathbb{R}^4)
  =ds^2_{\mathbb{H}^3,\ell}+\ell^2d\varphi^2,
  \label{eq:h3-s1-common-radius}
\end{equation}
where
\begin{equation}
  ds^2_{\mathbb{H}^3,\ell}={\ell^2\over \rho^2}\left((d\rho)^2+(dx^3)^2+(dx^4)^2\right).
  \label{eq:h3-metric}
\end{equation}
Thus the exact conformal product has the same radius $\ell$ for $\mathbb{H}^3$ and the orbit circle $S^1$; an independent radius $\ell_{S^1}$ is not introduced.

If it is continued to the Lorentzian region, we shall often write $(x^3,x^4)=(z,t)$ and use the metric of the three-dimensional \textbf{anti de Sitter space} $AdS_3$:
\begin{equation}
  ds^2_{AdS_3}={\ell^2\over \rho^2}\left((d\rho)^2+(dz)^2-(dt)^2\right).
\end{equation}

An $SO(2)$-symmetric configuration is independent of $\varphi$ up to a gauge transformation.  Since $\varphi$ is dimensionless and $g_{\varphi\varphi}=\ell^2$ in \eqref{eq:h3-s1-common-radius}, the canonically normalized three-dimensional adjoint scalar is defined by
\begin{equation}
  \Phi(X):={1\over\ell}\mathscr{A}_\varphi(X) \in\mathfrak{su}(2), \qquad X=(\rho,x^3,x^4) \in \mathbb{H}^3.
  \label{eq:phi-from-a-varphi}
\end{equation}
Equivalently, the coordinate component is $\mathscr{A}_\varphi=\ell\Phi$
,
and the remaining components define a three-dimensional gauge field $\mathscr{A}_\omega$ on $\mathbb{H}^3$:
\begin{equation}
  \mathscr{A}_\alpha(X) \in\mathfrak{su}(2) , \ X=(\rho,x^3,x^4) \in \mathbb{H}^3 .
  \label{eq:phi-from-a-varphi2}
\end{equation}

With the canonical scalar normalization in \eqref{eq:phi-from-a-varphi}, the four-dimensional Euclidean Yang--Mills action reduces to the $SU(2)$ gauge--Higgs model on the three-dimensional hyperbolic space $\mathbb{H}^3$ as
\begin{equation}
\begin{aligned}
  S_{4d}\big|_{SO(2)}
  &={2\pi\ell\over g_4^2}
  \int_{\mathbb{H}^3}d^3X\sqrt{g}\,
  \mathrm{tr}\Bigl(
  {1\over4}\mathscr{F}_{\alpha\beta}\mathscr{F}^{\alpha\beta}\\
  &\qquad\qquad
  +{1\over2}\mathscr{D}_\alpha\Phi\mathscr{D}^\alpha\Phi
  \Bigr)+S_\vartheta,
\end{aligned}
  \label{eq:so2-reduced-action-general}
\end{equation}
where $\alpha,\beta$ are $\mathbb{H}^3$ indices.  The factor $2\pi\ell$ is the circumference factor from the same orbit-circle radius $\ell$ appearing in \eqref{eq:h3-s1-common-radius}.  Had we instead kept the coordinate scalar $\Phi_{\rm coord}:=\mathscr{A}_\varphi$, its kinetic term would carry the factor $1/\ell^2$.
Here we have used
\begin{equation}
\begin{aligned}
  \mathscr{F}_{\alpha\beta}
  &:=\partial_\alpha\mathscr{A}_\beta-\partial_\beta\mathscr{A}_\alpha
  -i[\mathscr{A}_\alpha,\mathscr{A}_\beta],\\
  \mathscr{D}_\alpha\Phi
  &:=\partial_\alpha\Phi-i[\mathscr{A}_\alpha,\Phi].
\end{aligned}
\end{equation}

By completing the square in the reduced action, one obtains the Bogomolny bound.  With the canonical scalar $\Phi=\mathscr{A}_\varphi/\ell$, the four-dimensional self-duality (anti-self-duality) equation \eqref{eq:self-dual-four} reduces in the same way to the three-dimensional hyperbolic Bogomolny equation
\begin{equation}
  \mathscr{F}=*_{\mathbb{H}^3}\mathscr{D}\Phi.
  \label{eq:h3-bogomolny-general}
\end{equation}
With the metric of $\mathbb{H}^3$ coordinates $X^\alpha=(\rho,x^3,x^4)$ and $ds_{\mathbb{H}^3}^2=(\ell^2/\rho^2)\delta_{\alpha\beta}dX^\alpha dX^\beta$, the Bogomolny equation reads
\begin{equation}
  \mathscr{B}_\alpha(X)=\pm {\ell\over\rho}\mathscr{D}_\alpha\Phi(X),
  \quad
  \mathscr{B}_\alpha(X):={1\over2}\epsilon_{\alpha\beta\gamma}\mathscr{F}_{\beta\gamma}(X),
  \label{eq:ads3-bogomolny}
\end{equation}
where the factor $\ell/\rho$ follows from the Hodge star for the upper-half-space metric.  The formula reduces to $\mathscr{B}_\alpha=\pm\rho^{-1}\mathscr{D}_\alpha\Phi$ only in the convention $\ell=1$.
Therefore, \textbf{a finite-action $SO(2)$-symmetric instanton on $\mathbb{R}^4$ becomes a hyperbolic magnetic monopole on $\mathbb{H}^3$.}
The direction used here is the reduction from a regular four-dimensional symmetric configuration to the hyperbolic monopole.  The converse lift of an arbitrary hyperbolic monopole to a smooth four-dimensional circle-symmetric instanton requires the standard compatibility conditions, including the appropriate monopole mass relative to the circle action, periodicity/equivariance around the $S^1$ orbit, regularity at the removed fixed set, and the required asymptotic boundary behavior.  Thus ``correspondence'' below should not be read as an unconstrained bijection of all reduced configurations.
Since the $S^1$ orbit is compact, a finite three-dimensional action given by the solution of the Bogomolny equation in the reduced theory gives a finite four-dimensional action and may enter the symmetric-instanton sector of the four-dimensional path integral.

The remaining boundary term is proportional to $\int_{\partial \mathbb{H}^3}\mathrm{tr}(\Phi\mathscr{F})$, which measures the magnetic charge.
The Higgs direction $\widehat{\Phi}=\Phi/|\Phi|$ gives a local color direction and hence a natural \emph{partial} QRF valued in $SU(2)/U(1)$.  A nonzero monopole charge does not imply that $\widehat{\Phi}$ itself is absent as a global map; rather, it obstructs a global smooth \textbf{lift} of this partial frame to a complete $SU(2)$ frame.

\subsection{The $SO(3)$ reduction: $\mathbb{R}^4\setminus\mathbb{R}^1\simeq \mathbb{H}^2\times S^2$}
\label{subsec:so3-reduction-general}

The second reduction is the spherically symmetric case, i.e., the spatial rotation group $SO(3)$.
  Now we introduce polar coordinates in the three-dimensional subspace $(x^1,x^2,x^3)$:
\begin{equation}
\begin{aligned}
  x^1&=r\sin\theta\cos\varphi, \
  x^2 =r\sin\theta\sin\varphi, \
  x^3 =r\cos\theta,\\
  &r>0,\quad \theta\in[0,\pi),\quad \varphi\in[0,2\pi) .
\end{aligned}
  \label{eq:h2-polar-convention}
\end{equation}
Then by factoring out the conformal factor $r^2$ the flat metric on $\mathbb{R}^4$ becomes
\begin{align}
  ds^2(\mathbb{R}^4)
  &=(dx^4)^2+(dr)^2+r^2((d\theta)^2+\sin^2\theta\,(d\varphi)^2)\\
  &=r^2\left[{(dx^4)^2+(dr)^2\over r^2}
  +(d\theta)^2+\sin^2\theta\,(d\varphi)^2\right].
  \label{eq:r4-h2-s2-metric}
\end{align}
Thus we have the conformal equivalence after removing the $r=0$ region, i.e., $x^4$ axis $\mathbb{R}$:
\begin{equation}
\begin{aligned}
  &\mathbb{R}^4\setminus\mathbb{R}^1
   \simeq \mathbb{H}^2\times S^2,\\
  &\mathbb{H}^2 =\{(r,x^4)\,|\,r>0\}, \\
  &S^2 =\{(\theta,\varphi)\mid
  \theta\in[0,\pi),\ \varphi\in[0,2\pi)\} .
\end{aligned}
  \label{eq:r4-h2-s2}
\end{equation}

Restoring the same overall length scale in the exact conformal representative gives $\mathbb{H}^2_{\ell}\times S^2_{\ell}$, with a common radius $\ell$ for the hyperbolic and spherical factors.  In particular,
\begin{equation}
  ds^2_{\mathbb{H}^2}={\ell^2\over r^2}\left((dr)^2+(dx^4)^2\right),
  \label{eq:h2-metric}
\end{equation}
and the $S^2$ factor has radius $\ell$.  The reduced equations below are written in the standard dimensionless/unit-sphere normalization; factors of $\ell$ are restored through this common conformal scale.
After Lorentzian continuation $x^4\to t$, this is read as the two-dimensional anti de Sitter space $AdS_2$ in the coordinates $(r,t)$.

The most general $SO(3)$-invariant $SU(2)$ gauge field, up to gauge transformations, is given by the Witten ansatz.  With $j,A=1,2,3$ and $r^2=x^jx^j$, one writes

\begin{align}
  \mathscr{A}_4^A(x)
  &= {x^A\over r}A_0(r,x^4),
  \label{eq:witten-ansatz-a4-general}\\
  \mathscr{A}_j^A(x)
  &= {x^j x^A\over r^2}A_1(r,x^4)
  +{\delta^{jA}r^2-x^jx^A\over r^3}\phi_1(r,x^4)
  \nonumber\\
  &\quad+{\epsilon_{jAk}x^k\over r^2}
  \left(1+\phi_2(r,x^4)\right).
  \label{eq:witten-ansatz-aj-general}
\end{align}
The four-dimensional Euclidean Yang--Mills action reduces to a  $U(1)$ gauge--Higgs model with the complex scalar field on the two-dimensional hyperbolic space $\mathbb{H}^2$:
\begin{equation}
  \phi(r,x^4):=\phi_1(r,x^4)+i\phi_2(r,x^4),
\end{equation}
In the normalization used here, we keep the four-dimensional Yang--Mills coupling convention consistent with the preceding $SO(2)$ reduction and therefore retain the same overall factor $1/g_4^2$ explicitly in the reduced action:
\begin{equation}
\begin{aligned}
  S_{4d}\big|_{SO(3)}
  =
  {4\pi\over g_4^2}\int dx^4\int_0^\infty dr\Bigl[
  &{1\over4}r^2F_{\mu\nu}F_{\mu\nu}
  +(D_\mu\phi)^*D_\mu\phi\\
  &+{1\over2r^2}\left(|\phi|^2-1\right)^2
  \Bigr]+S_\vartheta .
\end{aligned}
  \label{eq:so3-reduced-action-general}
\end{equation}
where $\mu,\nu$ run over $(x^4,r)\in \mathbb{H}^2$.
The BPS equations are the hyperbolic vortex equations.  In a complex coordinate on $\mathbb{H}^2$, they may be written schematically as
\begin{equation}
  D_{\bar z}\phi=0,
  \quad
  B={\Omega\over2}\left(1-|\phi|^2\right),
  \label{eq:h2-vortex-general}
\end{equation}
where $\Omega$ is the conformal factor of $\mathbb{H}^2$.
Therefore, \textbf{a finite-action $SO(3)$-symmetric instanton on $\mathbb{R}^4$ becomes a hyperbolic vortex on $\mathbb{H}^2$.}
Again, the implication used in this paper is the reduction from a smooth four-dimensional symmetric instanton.
The converse \textbf{lift} of a hyperbolic vortex requires the standard regularity and boundary conditions, in particular the appropriate vacuum behavior at the conformal boundary of $\mathbb{H}^2$ and regularity at the fixed set of the $SO(3)$ action.
Only reduced configurations satisfying these extension conditions correspond to smooth finite-action four-dimensional instantons.

The two reductions may be summarized as (see e.g. \cite{Kondo:2025symmetric} for more details.)

\begin{table*}[t]
\centering
\small
\begin{tabular}{@{}llll@{}}
\toprule
symmetry & conformal form & reduced theory & defect \\
\midrule
$SO(2)$ & $\mathbb{R}^4\setminus\mathbb{R}^2\simeq \mathbb{H}^3\times S^1$ & $SU(2)$ gauge--Higgs on $\mathbb{H}^3$ & hyperbolic monopole \\
$SO(3)$ & $\mathbb{R}^4\setminus\mathbb{R}^1\simeq \mathbb{H}^2\times S^2$ & $U(1)$ gauge--Higgs on $\mathbb{H}^2$ & hyperbolic vortex \\
\bottomrule
\end{tabular}
\captionof{table}{Dimensional reduction of symmetric instanton based on conformal equivalence.}
\label{tab:dim-red}
\end{table*}

For QRF purposes, both reductions have the same conceptual structure: a component or phase of the reduced field supplies a \emph{specified local frame candidate}, while the corresponding topological defect may obstruct the globalization of that chosen frame or lift as a single-valued smooth representative.
This is a statement about the specified reduced construction; confinement requires the additional infrared/ensemble dynamics that eliminates the relevant finite-energy charged sector.

\section{$SO(3)$-symmetric reduction to a two-dimensional $U(1)$ gauge--Higgs model}
\label{sec:so3-u1-ads2}

\noindent\textit{QRF type(s) in this section:} matter QRF given locally by the reduced Higgs-phase lift, mathematical dressing supplied by the radial Wilson line, and boundary QRF data when the conformal-boundary orientation is physically accessible.\par

The reduction in this section is Euclidean. When $\mathbb H^2$ is referred to as Euclidean $AdS_2$, the self-duality/reduced first-order equations are not identified with Lorentzian canonical constraints or with a statement that the Lorentzian canonical momentum vanishes. Boundary-charge statements are made only after the Lorentzian phase-space interpretation and boundary conditions are specified.

\subsection{Reduced $AdS_2$ model and boundary data}

We now specialize the general $SO(3)$ reduction of Sec.~\ref{subsec:so3-reduction-general} to the QRF problem.
In the two-dimensional $U(1)$ gauge--Higgs model obtained from a spherically symmetric instanton, \textbf{the radial Wilson line explicitly supplies a QRF-type dressing/frame transporter, while the Higgs field supplies local matter-frame data where it is nonzero and accessible; vortices can obstruct the globalization of the specified Higgs-phase lift.}
\textbf{The physical information that remains after the local gauge orientation has been moved to the frame side is carried by the boundary holonomy, boundary electric flux, dressed Higgs operators, and vortex monodromies.}

We first fix the Abelian charge normalization used throughout this section.  The residual $U(1)$ connection will be denoted by $a_\mu$ and is normalized so that a four-dimensional fundamental probe has Abelian charge $q=1$.  In this normalization the reduced Higgs field coming from the off-diagonal adjoint components has charge $2$.  Equivalently, if $A^{\rm std}_\mu$ denotes the Abelian connection in the standard $T^3$-weight normalization of the reduced action, then
\begin{equation}
  A^{\rm std}_\mu=2a_\mu .
\end{equation}
Thus, a residual gauge transformation with parameter $\lambda$ is given by
\begin{equation}
  a\mapsto a-d\lambda,
  \quad
  \phi\mapsto e^{-2i\lambda}\phi,
  \quad
  \psi_+\mapsto e^{-i\lambda}\psi_+,
  \label{eq:ads2-fundamental-charge-normalization}
\end{equation}
where $\psi_+$ denotes a fundamental probe component of charge $+1$.  In this sign convention a charge-$q$ field transforms as $\Phi_q\mapsto e^{-iq\lambda}\Phi_q$ and has covariant derivative $D\Phi_q=(d-iq a)\Phi_q$.

Let $\tau$ denote Lorentzian time and regularize the radial interval by $r\in[\epsilon,R]$.  With the above normalization, the Lorentzian action corresponding to \eqref{eq:so3-reduced-action-general} is given by

\begin{equation}
\begin{aligned}
  S_M
  &=4\pi\int d\tau\int_\epsilon^R dr\,
  \left[2r^2f_{01}^2+(D_0\phi)^*D_0\phi\right.
  \\
  &\qquad\left.
  -(D_1\phi)^*D_1\phi
  -{1\over2r^2}\left(|\phi|^2-1\right)^2
  \right]+S_\vartheta,
\end{aligned}
  \label{eq:ads2-lorentzian-action}
\end{equation}
where
\begin{equation}
  D_\mu\phi:=(\partial_\mu-2ia_\mu)\phi,
  \quad
  f_{01}:=\partial_0a_1-\partial_1a_0.
\end{equation}
The factor $2r^2f_{01}^2$ is the same term as ${1\over2}r^2(F^{\rm std}_{01})^2$ written with $F^{\rm std}_{01}=2f_{01}$.

The topological term $S_{\vartheta}$ must be distinguished before and after Wick rotation.  With $\epsilon^{01}=+1$, we use the Euclidean convention $S_{\vartheta,E}$ inherited from the instanton action and a real Lorentzian convention $S_{\vartheta,M}$ as
\begin{equation}
\begin{aligned}
  S_{\vartheta,E}
  &=4\pi\int dx^4\int_\epsilon^Rdr\,
    \frac{i\vartheta}{8\pi^2}\epsilon^{\mu\nu}f_{\mu\nu},\\
  S_{\vartheta,M}
  &=4\pi\int d\tau\int_\epsilon^Rdr\,
    \frac{\vartheta}{8\pi^2}\epsilon^{\mu\nu}f_{\mu\nu},
\end{aligned}
\end{equation}
where the overall sign of the Lorentzian term is fixed by the chosen Wick-rotation convention.  The canonical analysis below is defined from the real Lorentzian action $S_M$, including $S_{\vartheta,M}$; in particular the regulator interval remains $[\epsilon,R]$ throughout.
The field $a_0$ is a Lagrange multiplier.  By using the electric displacement
\begin{equation}
  \Pi^1(r,\tau):={\partial\mathscr{L}_M\over\partial(\partial_0a_1)}
  =16\pi r^2f_{01}+\kappa_\vartheta,
  \label{eq:ads2-electric-displacement}
\end{equation}
where $\kappa_\vartheta\in\mathbb R$ is the real constant shift obtained from $S_{\vartheta,M}$ in the stated convention,
and the matter charge density
\begin{equation}
  \rho(r,\tau)=2i(\pi_\phi\phi-\phi^*\pi_{\phi^*}),
\end{equation}
the Gauss law is given by
\begin{equation}
  G(r,\tau):=\partial_r\Pi^1(r,\tau)-\rho(r,\tau)\approx0.
  \label{eq:ads2-gauss}
\end{equation}
For a smearing parameter $\lambda(r)$, the generator of the gauge transformation reads

\begin{align}
  G[\lambda]
  &=\int_\epsilon^Rdr\,\lambda(r)G(r)
  \nonumber\\
  &=-\int_\epsilon^Rdr\,(\partial_r\lambda)\Pi^1
    -\int_\epsilon^Rdr\,\lambda\rho
    +[\lambda\Pi^1]_\epsilon^R,
  \label{eq:ads2-gauss-parts}
\end{align}
where
$r=\epsilon$ is a conformal-boundary regulator, and $r=R$ is an infrared cutoff approaching the Poincar\'e horizon $r=\infty$ in the present Poincar\'e coordinates.

At finite $\epsilon$ and $R$, a gauge parameter that is allowed to be nonzero at an endpoint gives the corresponding surface term in the differentiable generator.  Whether that endpoint action survives as a genuine physical boundary or asymptotic symmetry after $\epsilon\to0$ or $R\to\infty$ is a separate question fixed by the boundary condition, the allowed falloff of $\lambda$ and $\Pi^1$, regularity at the Poincar\'e horizon, finiteness of the charge, and the symplectic-flux/differentiability conditions.  In particular, the $R$ endpoint is an infrared-regulator surface approaching the Poincar\'e horizon and is not automatically a second AdS conformal boundary.

The regulated surface generator is
\begin{equation}
  Q_\partial[\lambda]=[\lambda\Pi^1]_\epsilon^R
  =\lambda(R)\Pi^1(R)-\lambda(\epsilon)\Pi^1(\epsilon),
  \label{eq:ads2-boundary-charge}
\end{equation}
and the corresponding endpoint fluxes may be written as
\begin{equation}
  Q_L=\Pi^1(\epsilon),
  \quad
  Q_R^{\rm IR}=-\Pi^1(R).
  \label{eq:ads2-end-charges}
\end{equation}
Here $Q_L$ becomes a genuine conformal-boundary charge only for boundary conditions that retain the corresponding action, while $Q_R^{\rm IR}$ denotes an IR-regulator endpoint flux.  A genuine horizon charge is obtained from $Q_R^{\rm IR}$ only if the $R\to\infty$ limit satisfies the stated falloff, regularity, and differentiability requirements.

The \textbf{open Wilson line} along the radial interval, normalized for a fundamental charge, is given by
\begin{equation}
  U_{\epsilon\to r}(\tau)
  :=\exp\left(i\int_\epsilon^r ds\,a_1(s,\tau)\right),
  \label{eq:ads2-open-wilson}
\end{equation}
and the
\textbf{regulated endpoint holonomy} connecting the two regulated ends reads
\begin{equation}
  U_\partial(\tau)=U_{\epsilon\to R}(\tau)
  =\exp\left(i\int_\epsilon^R ds\,a_1(s,\tau)\right).
  \label{eq:ads2-boundary-holonomy}
\end{equation}
Figure~\ref{fig:QRF_AdS2_1_revised} represents this radial Wilson-line QRF-type dressing/frame transporter.

\begin{figure}[tbp]
  \centering
\includegraphics[width=.45\linewidth]{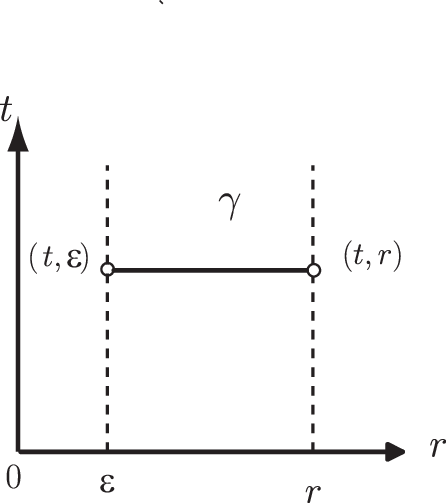}
  \caption{
The open Wilson line $U_\gamma(t,r)$ along the interval $[\epsilon,r]$ from $(t,\epsilon)$ to $(t,r)$.  At finite regulators it defines an endpoint-to-endpoint holonomy.  In the limits $\epsilon\to0$ and $r\to\infty$, the left endpoint approaches the AdS$_2$ conformal boundary whereas the right endpoint approaches the Poincar\'e horizon; whether the latter carries a genuine physical charge depends on the imposed horizon falloff and regularity conditions.
}
  \label{fig:QRF_AdS2_1_revised}
\end{figure}

Under the local gauge transformation in \eqref{eq:ads2-fundamental-charge-normalization}, this line transforms as
\begin{equation}
  U_{\epsilon\to r}\mapsto e^{-i\lambda(r)}U_{\epsilon\to r}e^{i\lambda(\epsilon)}.
\end{equation}
Therefore it is precisely the object needed to compare a fundamental phase at $r$ with the phase at the reference endpoint $\epsilon$.  For a charge-two field such as $\phi$, the corresponding dressing is $U_{\epsilon\to r}^{-2}\phi$.

Here we comment on charge normalization and the residual $\mathbb Z_2$ center.
We use the conventions:
\begin{equation}
 T^A={\sigma^A\over2},\quad g_\beta=e^{-i\beta T^3},
 \quad \lambda={\beta\over2},
 \quad a={1\over2}A_{\rm red}^3 .
\end{equation}
Then we have
\begin{equation}
 \psi_+\mapsto e^{-i\lambda}\psi_+,
 \quad
 \phi\sim A_{\rm off}^+\mapsto e^{-2i\lambda}\phi .
\end{equation}
It follows that, after the condensation of the charge-two Higgs field, the residual transformation with $\lambda=\pi$,
\begin{equation}
 g_{\lambda=\pi}=e^{-i\pi\sigma^3}=-\mathbf 1,
\end{equation}
is precisely the $\mathbb Z_2$ center of four-dimensional $SU(2)$.

\subsection
{Radial Wilson line as a relational frame transporter}

Using the left endpoint $r=\epsilon$ as the root, define the Wilson-line QRF-type dressing/frame transporter by
\begin{equation}
  \theta_h(r,\tau):=\int_\epsilon^r ds\,a_1(s,\tau).
  \label{eq:ads2-thetaA}
\end{equation}
Then the finite transformation by the radial Wilson-like QRF with parameter $\lambda=\theta_h$ gives
\begin{align}
  \widehat a_1(r,\tau)&:=a_1(r,\tau)-\partial_r\theta_h(r,\tau)=0,
  \label{eq:ads2-a1zero}\\
  \widehat a_0(r,\tau)&:=a_0(r,\tau)-\partial_0\theta_h(r,\tau),
  \label{eq:ads2-a0red}\\
  \Psi(r,\tau)&:=e^{-2i\theta_h(r,\tau)}\phi(r,\tau)
  =U_{\epsilon\to r}^{-2}\phi(r,\tau).
  \label{eq:ads2-psired}
\end{align}
This is the QRF-type dressing in the fundamental-charge normalization.
A fundamental probe of charge $+1$ would be dressed as $U_{\epsilon\to r}^{-1}\psi_+$, whereas the reduced Higgs field must be dressed with the square of the inverse fundamental Wilson line because it has charge $2$.
\textbf{The field $\Psi$ is the Higgs field viewed relative to the endpoint frame.  The bulk phase has been absorbed by the Wilson-line QRF-type dressing, while the residual endpoint phase and the total boundary holonomy remain.}

In the reduced gauge $\widehat a_1=0$,
we have
\begin{equation}
  f_{01}=-\partial_r \widehat a_0,
  \quad
  \Pi^1_{\rm red}=-16\pi r^2\partial_r \widehat a_0+\kappa_\vartheta,
\end{equation}
and the Gauss law becomes
\begin{equation}
\begin{aligned}
  &\partial_r\left(-16\pi r^2\partial_r \widehat a_0
  +\kappa_\vartheta\right)=\rho_{\rm red} ,\\
  &\rho_{\rm red} :=2i(\pi_\Psi\Psi-\Psi^*\pi_{\Psi^*}).
\end{aligned}
  \label{eq:ads2-gauss-red}
\end{equation}
Since $\kappa_\vartheta$ is constant, the local equation is essentially
\begin{equation}
  \partial_r(r^2\partial_r \widehat a_0)=-{1\over16\pi}\rho_{\rm red}.
\end{equation}
This makes the QRF interpretation transparent:
\begin{align}
  \text{bulk gauge redundancy} &:\quad a_1\sim a_1-\partial_r\lambda,\nonumber\\
  \text{QRF} &:\quad U_{\epsilon\to r}=\exp\left[i\int_\epsilon^r ds\,a_1(s,\tau)\right],\nonumber\\
  \text{reduction map} &:\quad (a_1,a_0,\phi)\mapsto(0,\widehat a_0,\Psi),\nonumber\\
  \text{boundary data} &:\quad Q_L,\ Q_R^{\rm IR},\ U_\partial.
\end{align}
Even after the QRF absorbs the radial gauge phase in the bulk, the boundary electric flux and boundary holonomy cannot be erased.  They are semi-local data, exactly of the type emphasized in the general QRF formulation.

The physical algebra contains the electric displacement $\Pi^1$, the boundary holonomy $U_\partial$, neutral local composites, and dressed bilocal fields.  For a charge-$q$ field $\Phi_q$, a typical dressed bilocal is
\begin{equation}
\begin{aligned}
  O_q(x,y;\gamma)&=\Phi_q^*(y)W_{q,\gamma}[x,y]\Phi_q(x),\\
  W_{q,\gamma}[x,y]&=\exp\left(iq\int_{\gamma:x\to y}a\right).
\end{aligned}
  \label{eq:ads2-dressed-bilocal}
\end{equation}
The local charged field $\Phi_q(x)$ itself is not an element of the physical algebra.  The reduced $AdS_2$ model therefore realizes in equations the principle stated earlier: the charged bare variable is replaced by a relational object tied to a frame or to a boundary.

\subsection{Higgs-phase QRF, vortices, and the $\mathbb{Z}_2$ center}
\label{subsec:higgs-qrf-vortices-z2}

The $SO(3)$-symmetric reduction contains a reduced Abelian gauge--Higgs system.
For the purpose of relating this reduced model to four-dimensional $SU(2)$ Yang--Mills theory, it is important to fix the normalization of the residual $U(1)$ charge once and for all.
This normalization is not a harmless convention, because the identification of the residual $\mathbb{Z}_2$ center and the monodromy of a fundamental relational operator depend on it.

We use Hermitian $SU(2)$ generators
\begin{equation}
  T^A=\frac{\sigma^A}{2},
  \quad
  [T^A,T^B]=i\epsilon^{ABC}T^C,
\end{equation}
and take the residual Abelian subgroup to be generated by $T^3$.
In order to match the sign convention used in \eqref{eq:ads2-fundamental-charge-normalization}, we write the residual transformation as
\begin{equation}
  g_\beta=\exp(-i\beta T^3).
\end{equation}
Then a fundamental doublet $\psi=(\psi_+,\psi_-)^t$ transforms as
\begin{equation}
  \psi_+ \mapsto e^{-i\beta/2}\psi_+,
  \quad
  \psi_- \mapsto e^{i\beta/2}\psi_- .
\end{equation}
Thus the standard $T^3$-weights of the fundamental representation are $\pm1/2$.
On the other hand, the off-diagonal adjoint generators
\begin{equation}
  T^\pm:=T^1\pm iT^2
\end{equation}
satisfy
\begin{equation}
  [T^3,T^\pm]=\pm T^\pm,
  \quad
  g_\beta T^\pm g_\beta^{-1}=e^{\mp i\beta}T^\pm .
\end{equation}
Therefore an off-diagonal adjoint component has twice the Abelian charge of a fundamental component.

In what follows we use the normalization in which a four-dimensional fundamental probe has Abelian charge $q=1$.
Equivalently, we introduce
\begin{equation}
  \lambda:=\frac{\beta}{2},
  \quad
  g_\lambda:=\exp(-2i\lambda T^3)=\exp(-i\lambda\sigma^3),
\end{equation}
and define the reduced Abelian connection $a$ by
\begin{equation}
  a:=\frac{1}{2}\mathscr{A}^3_{\rm red} .
\end{equation}
Then the reduced Abelian connection $a$ transforms as
\begin{equation}
  a\mapsto a-d\lambda,
\end{equation}
while a fundamental component of charge $+1$ transforms as
\begin{equation}
  \psi_+\mapsto e^{-i\lambda}\psi_+ .
\end{equation}
With this fundamental-charge normalization, the off-diagonal adjoint field produced by the $SO(3)$ reduction has charge $2$.
If
\begin{equation}
  \phi\sim \mathscr{A}^{+}_{\rm off}
\end{equation}
denotes the reduced complex Higgs field coming from the off-diagonal adjoint components, then, up to the sign convention used to choose $T^+$ rather than $T^-$,
\begin{equation}
  \phi\mapsto e^{-2i\lambda}\phi,
  \quad
  D\phi=(d-2ia)\phi .
  \label{eq:charge-two-Higgs}
\end{equation}
This is the convention used in the whole of Sec.~\ref{sec:so3-u1-ads2}.
If instead one keeps the standard $T^3$-weight normalization, the same statement would be that the fundamental has charge $1/2$ and the reduced Higgs has charge $1$.
The physics is unchanged, but the $\mathbb{Z}_2$ center is most transparent in the fundamental-charge normalization adopted here.

In the London regime the radial mode of the reduced Higgs field is frozen and one may write
\begin{equation}
  \phi(x)=v e^{i\chi(x)} .
\end{equation}
Equation \eqref{eq:charge-two-Higgs} gives
\begin{equation}
  D\phi
  =iv e^{i\chi}(d\chi-2a),
  \quad
  \chi\mapsto \chi-2\lambda .
\end{equation}
The Higgs phase therefore supplies a local color QRF for a fundamental probe:
\begin{equation}
  h_f(x):=e^{i\chi(x)/2},
  \quad
  h_f\mapsto e^{-i\lambda}h_f .
  \label{eq:higgs-frame-fundamental}
\end{equation}
More invariantly, for the fundamental doublet one may introduce the $SU(2)$-valued diagonal frame

\begin{equation}
\begin{aligned}
  H_f(x)&:=\exp(i\chi(x)T^3)
  =\begin{pmatrix}
    e^{i\chi(x)/2} & 0 \\
    0 & e^{-i\chi(x)/2}
  \end{pmatrix},\\
  H_f&\mapsto g_\lambda H_f .
\end{aligned}
\end{equation}
Then the relational field defined by
\begin{equation}
  \Psi_h(x):=H_f(x)^{-1}\psi(x)
\end{equation}
is invariant under the residual local $U(1)$ transformation.
Equivalently, for the charge-$+1$ component one has
\begin{equation}
  \Psi_{h,+}(x)=h_f(x)^{-1}\psi_+(x).
\end{equation}
This operator is a local relational expression relative to the Higgs QRF.
It is not an absolute color component; it is the fundamental probe as measured relative to the phase of the charge-two Higgs field.

(i)
Now consider a vortex of winding number $m$ in the reduced Abelian Higgs system.
On a large circle $C$ surrounding the vortex,
\begin{equation}
  \oint_C d\chi=2\pi m,
  \quad
  m\in\mathbb{Z} .
  \label{eq:vortex-winding}
\end{equation}
Therefore the fundamental Higgs frame has the monodromy
\begin{equation}
  h_f\longmapsto e^{i\pi m}h_f=(-1)^m h_f,
\end{equation}
or, equivalently,
\begin{equation}
  H_f\longmapsto H_f\exp(2\pi i m T^3)=(-1)^m H_f .
\end{equation}
For an odd vortex, $m\in2\mathbb{Z}+1$, the monodromy is precisely the nontrivial element of the $SU(2)$ center.
Consequently the fundamental relational operator changes sign:
\begin{equation}
  \Psi_h\longmapsto (-1)^m\Psi_h .
\end{equation}
Thus the relational representative built from the square-root Higgs frame $h_f=e^{i\chi/2}$ acquires a sign on going around an odd vortex.  The precise reduced-model statement is that \textbf{this Higgs-phase realization of a complete fundamental frame cannot be globalized as one single-valued function without patching or line attachment in the presence of nontrivial center monodromy.}  Other admissible Wilson-line or patchwise relational constructions are not ruled out by this topological statement alone.
By contrast, an adjoint or charge-two object is center neutral.
It uses the single-valued frame
\begin{equation}
  h_{\rm adj}(x):=e^{i\chi(x)},
\end{equation}
and is insensitive to the sign picked up by $h_f$.

(ii) The same obstruction is detected by Wilson loops.
Outside the core of a finite-action vortex the London equation implies
\begin{equation}
  d\chi-2a\simeq 0 .
\end{equation}
Combining this with \eqref{eq:vortex-winding} gives the Abelian flux carried by the vortex:
\begin{equation}
  \oint_C a
  =\frac{1}{2}\oint_C d\chi
  =\pi m .
  \label{eq:z2-vortex-flux}
\end{equation}
Thus the vortex number in the fundamental-charge normalization is
\begin{equation}
  Q_v={1\over\pi}\int_{\Sigma} f = {1\over\pi}\int_{\Sigma}  da
  = {1\over\pi}\oint_{\partial\Sigma=C}  a  = m,
  \quad
  f=da,
\end{equation}
so that a unit vortex has $Q_v=1$.
A Wilson loop of Abelian charge $q$ in the same normalization is
\begin{equation}
  W_q(C) :=\exp\left(iq\oint_C a\right) .
\end{equation}
If $C$ encloses a vortex of winding number $m$, then we have
\begin{equation}
  W_q(C)\longmapsto e^{i\pi q m}W_q(C)=(-1)^{qm}W_q(C) .
\end{equation}
Hence odd-charge probes, in particular the four-dimensional fundamental probe with $q=1$, detect the odd vortex by the nontrivial center phase, whereas even-charge probes do not.

Assume additionally that vortices and anti-vortices with $m=\pm1$ form independent homogeneous Poisson gases with fugacity (density) $\zeta$ per unit \emph{invariant hyperbolic area}.  For Poincar\'e coordinates $(r,x^4)$ on $\mathbb{H}^2$,
\begin{equation}
  dA_{\mathbb{H}^2}={\ell^2\over r^2}\,dr\,dx^4,
  \qquad
  \operatorname{Area}_{\mathbb{H}^2}(\Sigma_C)
  :=\int_{\Sigma_C}{\ell^2\over r^2}\,dr\,dx^4,
  \label{eq:h2-invariant-area}
\end{equation}
where $\partial\Sigma_C=C$.  Under this statistical assumption, the defect contribution to a large Wilson loop is

\begin{align}
  & \langle W_q(C)\rangle_{\rm vort}\nonumber\\
  &=\exp\left[
  \zeta \operatorname{Area}_{\mathbb{H}^2}(\Sigma_C)
  (e^{i\pi q}-1)
  +\zeta \operatorname{Area}_{\mathbb{H}^2}(\Sigma_C)
  (e^{-i\pi q}-1)
  \right]\nonumber\\
  &=\exp\left[
  -2\zeta \operatorname{Area}_{\mathbb{H}^2}(\Sigma_C)
  (1-\cos\pi q)
  \right] .
\end{align}
The Poisson distribution and its fugacity are additional semiclassical/statistical assumptions; they are not derived here from the four-dimensional path-integral measure of the symmetric-instanton sector.
\footnote{
The defect-induced exponential above should not, by itself, be identified with the flat-space area-versus-perimeter confinement criterion.
On $\mathbb H^2$, the distinction between asymptotic area and perimeter scaling is geometrically subtler than in flat space.  For a geodesic disk of radius $s$ and curvature radius $\ell$, the area $A(s)$ and the perimeter $P(s)$ are asymptotically proportional: indeed,
$A(s)=2\pi\ell^2[\cosh(s/\ell)-1]$ and $P(s)=2\pi\ell\sinh(s/\ell)$, so that $A(s)/P(s)\to\ell$ as $s\to\infty$.
The exponential above should therefore be interpreted here as a defect-induced exponential in the invariant hyperbolic area.
A static confinement criterion requires a specified family of loops or a corresponding Lorentzian observable, rather than importing the flat-space area-versus-perimeter distinction without qualification.
}
Thus the vortex-induced string tension $\sigma_q$ is given by
\begin{equation}
  \sigma_q=2\zeta(1-\cos\pi q)
  = \begin{cases} 0 &(\text{q:even}) \\ 4\zeta \not=0 &(\text{q:odd})
  \end{cases} .
\end{equation}
It vanishes for even $q$ and is nonzero for odd $q$.
This is exactly the $\mathbb{Z}_2$ selection expected from the center of $SU(2)$.

The charge normalization fixed above is therefore the bridge between the $SO(3)$-reduced Abelian Higgs model and four-dimensional $SU(2)$ color confinement:
\begin{equation}
\begin{aligned}
  \hbox{fundamental $SU(2)$ probe}
  &\longleftrightarrow q=1,\\
  \hbox{off-diagonal adjoint Higgs}
  &\longleftrightarrow q=2,
\end{aligned}
\nonumber
\end{equation}
so that
\begin{equation}
\begin{aligned}
  &\text{odd vortex monodromy of }h_f
  \longleftrightarrow
  {
  \text{nontrivial }\mathbb{Z}_2
  }  \\
  &\qquad
  \text{center phase of a fundamental Wilson loop}
   .
\end{aligned}
\nonumber
\end{equation}
In the QRF language, \textbf{the Higgs phase provides a local color frame, but odd vortices obstruct its \textbf{global lift} to a single-valued fundamental frame.}
This obstruction is not seen by center-neutral operators, but it is seen by fundamental relational operators and fundamental Wilson loops.

\subsection{Reading the $AdS_2$ reduction back into four dimensions}

In the $SO(3)$-symmetric sector, the two-dimensional vortex number $Q_v$ is the topological number $Q_{4d}$ inherited from the four-dimensional symmetric instanton.  Schematically,
\begin{equation}
  Q_{4d}=\pm Q_v, \quad Q_v:= {1\over\pi}\int_{\mathbb{H}^2} da.
\end{equation}

Thus a two-dimensional vortex is not an artificial defect added to the theory; it is the reduced image of a four-dimensional topological configuration.  The QRF dictionary in this sector is
\begin{align}
&\text{$SO(3)$-symmetric instanton}\nonumber\\
  &\longrightarrow \text{hyperbolic vortex phase of the reduced Higgs field}\nonumber\\
  &\longrightarrow \text{local Higgs-frame lift vortex number}\nonumber\\
  &\longrightarrow \text{center monodromy of the frame,}\nonumber\\
  &\qquad \text{fundamental Wilson-loop sign}\nonumber\\
  &\longrightarrow \text{detection of nontrivial center charge}.
\nonumber
\end{align}
This model gives a concrete calculation showing that the same center monodromy which obstructs the selected Higgs-frame lift also enters the Wilson-loop phase.  Under the additional dilute-Poisson-gas assumption stated above, the center monodromy yields the area-law factor in the reduced model.
The fugacity and vortex/antivortex weights are not derived here from the four-dimensional Yang--Mills path-integral measure.  Hence this is not a complete proof of confinement in the full four-dimensional theory: the argument is restricted to the symmetric sector, the reduced model is Abelian, and the vortex-gas treatment is semiclassical.  Nevertheless it is a useful reduced laboratory in which QRF data, defect monodromy, center charge, boundary data, and an area-law calculation can be compared within the same controlled setup.

\section{$SO(2)$-symmetric reduction to a three-dimensional $SU(2)$ gauge--Higgs model}
\label{sec:so2-su2-ads3}

\noindent\textit{QRF type(s) in this section:} {intrinsic partial frame} supplied by the adjoint-Higgs color direction $SU(2)/U(1)$, supplemented by non-Abelian Wilson-line dressings and boundary frame data.\par

Likewise, the hyperbolic monopole equations arise from a Euclidean symmetric-instanton reduction. Their use as geometric input is kept distinct from the Lorentzian Gauss constraint and canonical boundary-charge algebra on $AdS_3$.

We next turn to the $SO(2)$-symmetric reduction as described in Sec.~\ref{subsec:so2-reduction-general}.

\subsection{The Higgs direction as a QRF and the monopole obstruction}

In a region where $\Phi\neq0$, define the normalized Higgs direction as a QRF:
\begin{equation}
  \widehat{\Phi}(x)={\Phi(x)\over |\Phi(x)|}\in SU(2)/U(1)\simeq S^2.
  \label{eq:phihat-ads3}
\end{equation}
Figure~\ref{fig:QRF_AdS3_Higgs1_revised} represents this Higgs-direction QRF.

\begin{figure}[htbp]
  \centering
\includegraphics[width=.5\linewidth]{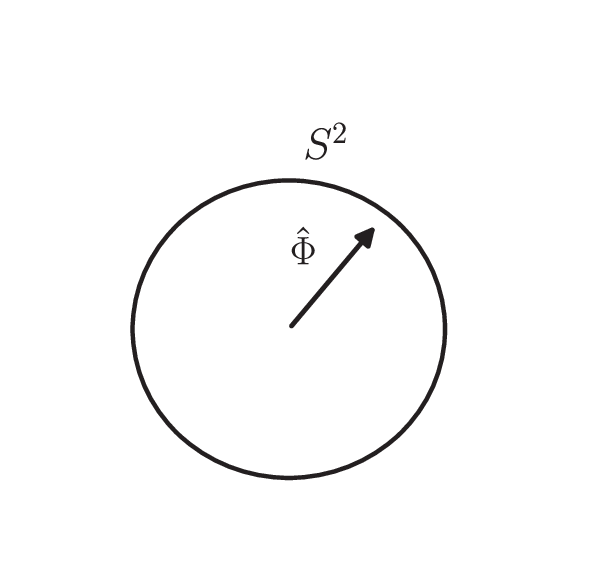}
  \caption{The normalized Higgs field $\widehat{\Phi}$ in the reduced $AdS_3$ or $\mathbb{H}^3$ theory.  It takes values in $SU(2)/U(1)\simeq S^2$ and supplies a local color direction, hence an {intrinsic partial frame}.}
  \label{fig:QRF_AdS3_Higgs1_revised}
\end{figure}

Locally one may choose $h(x)\in SU(2)$ such that
\begin{equation}
  \widehat{\Phi}(x)=h(x)T_3h(x)^{-1}, \quad h(x)\in SU(2).
  \label{eq:phihat-h-frame}
\end{equation}
This $h(x)$ is the \textbf{color frame} constructed from the Higgs field.  Under a gauge transformation $g(x)\in SU(2)$, the color frame $h(x)$ transforms as
\begin{equation}
  h(x)\mapsto g(x)h(x)e^{i\beta(x)T_3}, \ g(x) \in SU(2),
\end{equation}
where the right multiplication $e^{i\beta(x)T_3}$ by $U(1)$ is the residual Abelian ambiguity.
A fundamental field $\psi(x)$ with the transformation law $\psi(x)\mapsto g(x)\psi(x)$ may be written relationally as
\begin{equation}
  \psi_h(x)=h(x)^{-1}\psi(x).
\end{equation}
The local $SU(2)$ gauge transformation  acting from the left $g(x)$ has been removed, while the remaining components carry charges under the residual $U(1)$ acting from the right $e^{i\beta(x)T_3}$.  In this sense the Higgs direction Abelianizes the color frame.

Now surround a magnetic monopole by a sphere $S_\infty^2$.  The \textbf{magnetic charge} (monopole number) $k$ is the degree of the map $\widehat{\Phi}:S_\infty^2\to SU(2)/U(1)\simeq S^2$,
\begin{equation}
  k :={1\over8\pi}\int_{S_\infty^2}\widehat{\Phi}\cdot(d\widehat{\Phi}\wedge d\widehat{\Phi}).
  \label{eq:monopole-charge-phihat}
\end{equation}
For non-vanishing magnetic charge $k\neq0$, no globally smooth $h:S_\infty^2\to SU(2)$ satisfying \eqref{eq:phihat-h-frame} exists.  Indeed, such an $h$ would trivialize the \textbf{pullback} of the \textbf{Hopf fibration}:
\begin{equation}
  U(1)\longrightarrow SU(2)\longrightarrow SU(2)/U(1)\simeq S^2,
  \label{Hopf-fibration}
\end{equation}
but the \textbf{first Chern number} of the \textbf{pulled-back $U(1)$ bundle} is precisely $k$.  Hence
{}

\begin{equation}
  \boxed{\begin{gathered}
  \text{monopole charge } k\neq0
  \Longleftrightarrow
  \text{no global smooth lift}\\
  h:S_\infty^2\to SU(2)
  \text{ of the partial frame }\widehat{\Phi}.
  \end{gathered}}
  \nonumber
  \label{eq:monopole-qrf-obstruction}
\end{equation}

The global map $\widehat{\Phi}:S_\infty^2\to SU(2)/U(1)$ is the partial color-axis QRF.  The monopole number $k$ is the first Chern number of its pulled-back Hopf bundle and is therefore the obstruction to choosing a global smooth \emph{complete-frame lift} $h$, not to the existence of $\widehat{\Phi}$ itself.

It should be remarked that this fact alone cannot be the proof of confinement.
The area law requires an additional dynamical statement: a monopole--anti-monopole ensemble must disorder the dual photon or, more generally, the long-distance frame comparison, as supplemented in the final part of this section.

\begin{figure}[htbp]
  \centering
\includegraphics[width=.60\linewidth]{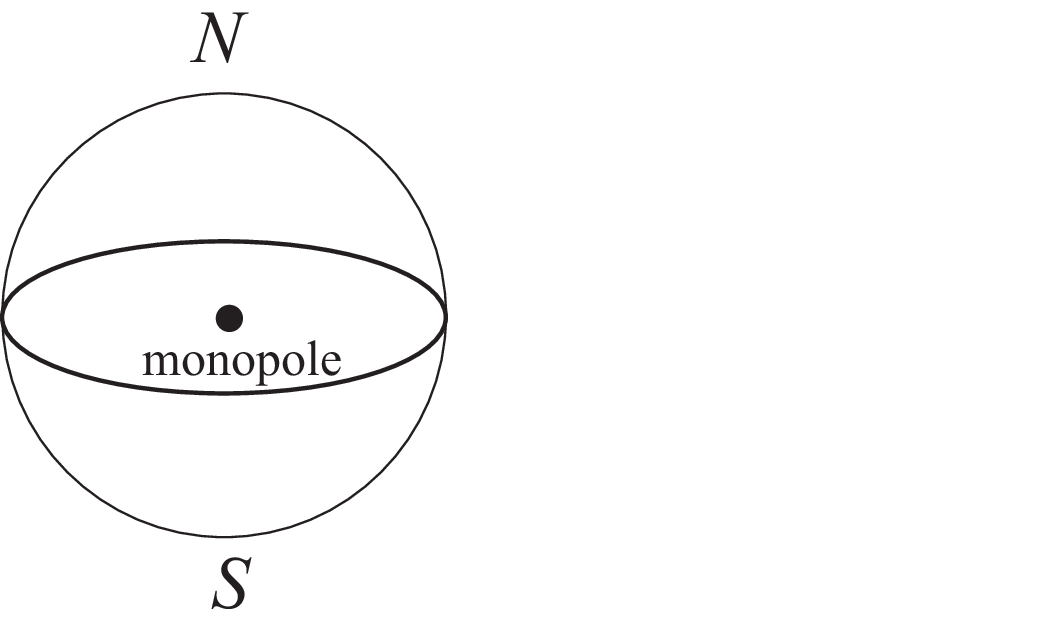}
  \caption{A magnetic monopole and the northern and the southern semi-sphere patches.}
  \label{fig:monopole}
\end{figure}

For a unit monopole, take
\begin{equation}
  \widehat{\Phi}(x)=\hat r^AT_A.
\end{equation}
On a northern patch choose
\begin{equation}
  h_N(\theta,\varphi)=e^{-i\varphi T_3}e^{-i\theta T_2}e^{i\varphi T_3},
\end{equation}
so that
\begin{equation}
h_N T_3 h_N^{-1}=\hat r^AT_A .
\end{equation}
A southern patch $h_S$ is needed
so that
\begin{equation}
h_S T_3 h_S^{-1}=\hat r^AT_A .
\end{equation}
Thus, on the equator $\theta=\pi/2$ the \textbf{transition function} is given by
\begin{equation}
  h_S^{-1}h_N=e^{2i\varphi T_3}.
\end{equation}
The winding of this transition function is the monopole charge.
See Figure~\ref{fig:monopole}.

The \textbf{Berry connection} of the upper component in the fundamental representation,
\begin{equation}
  a_N=-i\langle +|h_N^{-1}dh_N|+\rangle,
\end{equation}
may be chosen as
\begin{equation}
  a_N={1\over2}(1-\cos\theta)d\varphi,
  \quad
  f=da_N={1\over2}\sin\theta\,d\theta\wedge d\varphi,
\end{equation}
and therefore
\begin{equation}
  {1\over2\pi}\int_{S^2}f=1.
\end{equation}
This is the explicit calculation showing that \textbf{a single monopole prevents the partial direction $\widehat{\Phi}$ from being lifted over the whole sphere to one globally smooth complete $SU(2)$ frame.}

The Bogomolny equation used in this subsection is the Euclidean equation on $\mathbb H^3$ inherited from the four-dimensional Euclidean self-duality equation:
\begin{equation}
  \mathscr{F}_{\mathscr{A}}=*_{\mathbb H^3}D_{\mathscr{A}}\Phi
  \qquad\text{(Euclidean $\mathbb H^3$)}.
\end{equation}
It is used here only to characterize the monopole geometry and the complete-frame lift obstruction.  We do not infer from this Euclidean equation that the Lorentzian canonical electric momenta vanish after analytic continuation to $AdS_3$; the Lorentzian Gauss constraint and boundary charges are defined independently from the Lorentzian action in Sec.~IX.C.

\subsection
{Wilson lines as non-Abelian QRF-type dressings}

The reduced three-dimensional theory has a second, and in some respects more flexible, QRF-related construction.  One may use a Wilson-line dressing/frame transporter from a reference point $x_0$ to $x$,
\begin{equation}
  U(x;x_0,\gamma)=\mathcal{P}\exp\left(i\int_{\gamma:x_0\to x}\mathscr{A}\right),
  \label{eq:ads3-wilson-qrf}
\end{equation}
as a frame and transport local operators with it.
The \textbf{relational adjoint Higgs field} is introduced as
\begin{equation}
  \Phi_{\rm rel}(x)=U(x;x_0,\gamma)^{-1}\Phi(x)U(x;x_0,\gamma).
  \label{eq:ads3-rel-higgs}
\end{equation}
This construction is represented by  Fig.~\ref{fig:QRF_AdS3_Higgs2_revised}.

\begin{figure}[htbp]
  \centering
\includegraphics[width=.6\linewidth]{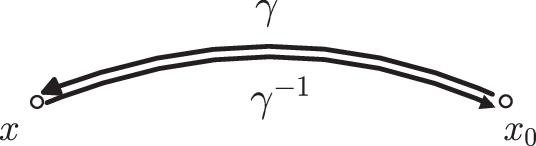}
\caption{
The relational Higgs field $\Phi_{\rm rel}(x)$ defined using the boundary-anchored Wilson-line QRF-type dressing/frame transporter $W(x;x_0;\gamma)$ along a path $\gamma:x_0\to x$ in $\mathbb{H}^3\simeq AdS_3$.  A physical QRF is present only if the endpoint/reference orientation is independently accessible in the chosen protocol.}
  \label{fig:QRF_AdS3_Higgs2_revised}
\end{figure}

For the $AdS_3$ interpretation it is natural to anchor the frame at the boundary.  Let $b\in\partial AdS_3$ be \textbf{a boundary reference point} and let $\gamma_b(X)$ be a chosen path to   \textbf{a bulk point} $X$.  Define the frame
\begin{equation}
\begin{aligned}
  U_{\gamma_b}(X;\mathscr{A})
  &=\mathcal{P}\exp\left(i\int_{\gamma_b(X)}\mathscr{A}\right),\\
  &\qquad b\in\partial AdS_3,\quad X\in AdS_3.
\end{aligned}
  \label{eq:ads3-boundary-frame}
\end{equation}
Fig.~\ref{fig:QRF_AdS3_Wilson_revised} depicts this boundary-anchored Wilson-line frame transporter.

\begin{figure}[htbp]
  \centering
\includegraphics[width=.6\linewidth]{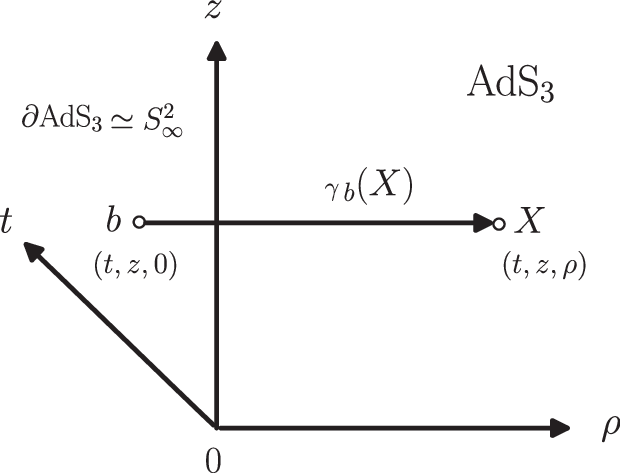}
  \caption{The Wilson line $W_{\gamma_b}(X)$ as a boundary-anchored frame along a path $\gamma_b(X)$ from a base point $b\in\partial AdS_3$ to a bulk point $X\in AdS_3$.  The base point $(t,z,0)$ is on the $(t,z)$ boundary plane at $\rho=0$.}
  \label{fig:QRF_AdS3_Wilson_revised}
\end{figure}

At the level of fields, the relational reduction map generated by the Wilson-line frame transporter is the contour-gauge-type transformation generated by
\begin{equation}
  g_{\gamma_b}(X)=U_{\gamma_b}(X;\mathscr{A})^{-1}.
\end{equation}
The relational Yang--Mills field $\mathscr{A}_\mu^{(\gamma)}(X)$, the relational adjoint scalar field $\Phi^{(\gamma)}(X)$ and the relational fundamental scalar field $\Psi^{(\gamma)}(X)$ are respectively defined by
\begin{align}
  \mathscr{A}_\mu^{(\gamma)}(X)
  &=g_{\gamma_b}(X)\mathscr{A}_\mu(X)g_{\gamma_b}(X)^{-1}
  +ig_{\gamma_b}(X)\partial_\mu g_{\gamma_b}(X)^{-1},\nonumber\\
  \Phi^{(\gamma)}(X)
  &=g_{\gamma_b}(X)\Phi(X)g_{\gamma_b}(X)^{-1},\nonumber\\
  \Psi^{(\gamma)}(X)&=g_{\gamma_b}(X)\Psi(X) .
\end{align}
On each chosen path, the tangential component of the Yang--Mills field to the curve $\gamma$ vanishes:
\begin{equation}
  \dot\gamma^\mu\mathscr{A}_\mu^{(\gamma)}(X)=0.
\end{equation}
For the radial path $\gamma_{b}:b=(t,z,0)\to X=(t,z,\rho)$, this gives the radial gauge:
\begin{equation}
  \mathscr{A}_\rho^{(\gamma)}(X)=0.
\end{equation}
The relational matter fields are then
\begin{equation}
\begin{aligned}
  \Phi_{\rm rel}(X)&=U_{\gamma_b}(X)^{-1}\Phi(X)U_{\gamma_b}(X),\\
  \Psi_{\rm rel}(X)&=U_{\gamma_b}(X)^{-1}\Psi(X).
\end{aligned}
\end{equation}

Thus \textbf{the non-Abelian local orientation in the bulk is absorbed into the frame side}; loop holonomies, boundary values, and dressed gauge-invariant combinations remain.

\subsection{Gauss law, boundary holonomy, and boundary charge on $AdS_3$}

The Bogomolny equation in Sec.~IX.A is a Euclidean equation on $\mathbb H^3$.  After analytic continuation, no vanishing of the Lorentzian electric momenta is assumed.  The following Gauss constraint and boundary charges are instead obtained directly from the Lorentzian canonical variables.
The canonical momenta are obtained from \eqref{eq:so2-reduced-action-general} with the same orbit radius $\ell$:
\begin{equation}
  \Pi^j={2\pi\ell\over g_4^2}\sqrt{-g}\,\mathscr{F}^{j0},
  \quad
  \Pi_\Phi={2\pi\ell\over g_4^2}\sqrt{-g}\,\mathscr{D}^0\Phi.
\end{equation}
The Gauss constraint is given by
\begin{equation}
  \mathcal{G}:=\mathscr{D}_j\Pi^j-i[\Phi,\Pi_\Phi]=0.
  \label{eq:ads3-gauss}
\end{equation}
If this generator is smeared with a Lie-algebra-valued test function $\omega$, then integration by parts gives
\begin{equation}
\begin{aligned}
  G[\omega]
  &=-\int_{\Sigma_t}d^2x\,\mathrm{tr}\{(\mathscr{D}_j\omega)\Pi^j\}
  -\int_{\Sigma_t}d^2x\,\mathrm{tr}\{\omega i[\Phi,\Pi_\Phi]\}\\
  &\quad+\oint_{\partial\Sigma_t}ds_j\,\mathrm{tr}(\omega\Pi^j).
\end{aligned}
  \label{eq:ads3-gauss-boundary-parts}
\end{equation}
If $\omega$ vanishes on $\partial\Sigma_t$, this is a constraint and annihilates physical states.  If $\omega$ does not vanish at the boundary, the last term remains as a boundary charge,
\begin{equation}
  Q_\partial[\omega]=\oint_{\partial\Sigma_t}ds_j\,\mathrm{tr}(\omega\Pi^j).
  \label{eq:ads3-boundary-charge}
\end{equation}
Therefore, as in the general Yang--Mills discussion, \textbf{the Gauss law removes local bulk color charge but does not erase boundary flux.}

For the class of reduced solutions, gauge choices, regularity requirements, and asymptotic boundary conditions considered in Ref.~\cite{Kondo:2025symmetric}, the $\rho\to0$ asymptotics Abelianize the reduced $SU(2)$ field to its maximal-torus component.  This is not asserted here as a consequence of the Gauss law or $SO(2)$ symmetry alone.  In a boundary coordinate system $(t,z)$ one has schematically
\begin{equation}
  \mathscr{A}_t(t,z,\rho)\to {\sigma_3\over2}a_t(t,z),
  \quad
  \mathscr{A}_z(t,z,\rho)\to {\sigma_3\over2}a_z(t,z),
\end{equation}
and hence
\begin{equation}
  \mathscr{F}_{tz}(t,z,\rho)\to {\sigma_3\over2}f_{tz}(t,z),
  \quad
  f_{tz} :=\partial_ta_z-\partial_za_t.
\end{equation}
After this Abelian asymptotic reduction has been imposed, path ordering is no longer needed for the maximal-torus boundary connection, and a fundamental Wilson loop reduces to

\begin{align}
  W_C[\mathscr{A}]
  &={1\over2}\mathrm{tr}_F\left\{\mathcal{P}
  \exp\left(i\oint_Cdx^\mu\mathscr{A}_\mu\right)\right\}
  \nonumber\\
  &\longrightarrow {1\over2}\mathrm{tr}_F\left\{
  \exp\left(i{\sigma_3\over2}\oint_Ca\right)\right\}
  \nonumber\\
  &=\cos\left({1\over2}\oint_Ca\right)
  =\cos\left({1\over2}\int_{\Sigma:\partial\Sigma=C}f\right).
  \label{eq:ads3-boundary-wilson}
\end{align}
The last equality is the Abelian Stokes relation in a patch containing $\Sigma$.  If the boundary $U(1)$ bundle requires nontrivial patching, the same formula must be understood with the corresponding transition data/quantized flux included rather than by naively applying a single-patch potential on all of $\Sigma$.  \textbf{Under these stated asymptotic assumptions, this is the
boundary holonomy that survives after the bulk orientation has been moved to the relational frame defined by the boundary-anchored transporter.}

\subsection{Flat-space template only: no $\mathbb H^3$ monopole-gas area-law derivation}

A single monopole exhibits the lift obstruction \eqref{eq:monopole-qrf-obstruction}, but a single monopole is not by itself an area-law derivation.  An ensemble of monopoles and antimonopoles is required in the flat-space Polyakov mechanism.  The following dual-photon formula is therefore used only as a flat-space-inspired template: it is not derived here from the monopole ensemble on $\mathbb{H}^3$.  If the hyperbolic ensemble admits an analogous Abelianized dual description after the correct $\mathbb{H}^3$ measure, Laplace--Beltrami operator, boundary conditions, curvature-dependent interaction and fugacity are included, one expects an effective \textbf{dual-photon} $\sigma$ description of the schematic form
\begin{equation}
  S_{\rm dual}=\int d^3x\left[
  {g_3^2\over32\pi^2}(\partial\sigma)^2-2\zeta\cos\sigma
  \right],
  \label{eq:dual-photon-action}
\end{equation}
where $\zeta$ denotes the monopole fugacity in the flat-space template.  If an appropriate $\mathbb{H}^3$ dualization produces the corresponding massive dual mode, a Wilson loop is expected to impose an appropriate discontinuity or saddle across a spanning surface $\Sigma$; an area law would then be controlled by the tension $T_{\rm DW}$ of the resulting domain-wall-like configuration,
\begin{equation}
  \langle W_{\rm fund}(C)\rangle
  \sim\exp[-T_{\rm DW}\operatorname{Area}(\Sigma)].
  \label{eq:ads3-monopole-area-law}
\end{equation}
Thus the statement relevant for $\mathbb{H}^3$ is conditional:
\begin{align}
  &\text{if the }\mathbb{H}^3\text{ monopole ensemble dualizes to an}\nonumber\\
  &\qquad\text{appropriate massive dual mode,}\nonumber\\
  &\text{then a Wilson-loop area law is expected from the}\nonumber\\
  &\qquad\text{corresponding saddle/domain wall}.
\nonumber
\end{align}
For the QRF interpretation, the important statement is the combined one:

\begin{align}
  &\text{$SO(2)$-symmetric instanton} \nonumber\\
  &\longrightarrow
  \text{hyperbolic monopole}\nonumber\\
  &\longrightarrow
  \text{nontrivial complete-frame lift data}\nonumber\\
  &\dashrightarrow
  \text{possible boundary-holonomy disorder}\nonumber\\
  &\qquad\text{and Wilson-loop area law}\nonumber\\
  &\text{if the }\mathbb{H}^3\text{ ensemble produces the required}\nonumber\\
  &\qquad\text{infrared dynamics}.
\nonumber
\label{eq:so2-final-correspondence}
\end{align}
The conclusion is sector-limited and, at the present stage, partly prospective.  In the $\mathbb{H}^3/AdS_3$ reduction of circle-symmetric instantons, monopole topology gives a controlled obstruction to a global complete-frame lift and supplies nontrivial boundary and holonomy data.  Whether an ensemble of these defects produces an area law on $\mathbb{H}^3$ requires a separate hyperbolic dualization and Wilson-loop saddle calculation.  Consequently, the present section establishes the geometric/QRF part of the mechanism but does not claim an $\mathbb{H}^3$ monopole-gas derivation of the area law.

This is also the endpoint of the controlled reduction argument: the hyperbolic-monopole equations and complete-frame lift obstruction are exact within the circle-symmetric ansatz, whereas defect proliferation, a monopole ensemble, and a four-dimensional Wilson-loop area law are separate dynamical hypotheses.  No implication from the reduced topology to dominance in the unrestricted four-dimensional path integral is used later as an established result.

\section{Four-dimensional Yang--Mills theory and color-direction fields as QRFs}
\label{sec:four-dimensional-color-direction}
\noindent\textit{QRF type(s) in this section:}
intrinsic partial frame $\bm n[\mathscr A]\in SU(2)/U(1)$ extracted from the Yang--Mills field, with complete dressings or boundary frames treated separately when required.\par

The reduced examples discussed above have a common lesson.  \textbf{A colored field is never observed as an absolute component.  It becomes meaningful only after one supplies a color frame, a dressing, or a boundary reference system.}
In the two-dimensional Abelian model the frame was an exponential of the gauge field.  In the $SO(3)$-symmetric reduction the Higgs phase of the reduced $U(1)$ model supplied a phase frame and vortices obstructed its globalization.
In the $SO(2)$-symmetric reduction, on the other hand, the adjoint scalar of the hyperbolic monopole supplied an $SU(2)/U(1)$ color axis, and magnetic monopoles obstructed its globalization.  The natural question is how the same structure can be formulated in the full four-dimensional Yang--Mills theory, where no elementary Higgs field is present.

\subsection{Pure Yang--Mills theory}
The discussion below is organized into three logically distinct layers:
(I) exact kinematics---the CDG decomposition, transformation laws, and restricted curvature;
(II) geometric statement---the interpretation of $\bm n[\mathscr A]$ as an intrinsic partial frame and possible obstruction to a lift to complete frame; and
(III) dynamical conjecture---monopole/vortex ensembles, restricted-field dominance, an area law, and charged-sector absence.
Only the first two layers are kinematical; the third is a conjectural confinement mechanism.

The controlled $\mathbb H^3$ monopole example of Sec.~IX.A fixes the terminology used here: a nonzero monopole charge need not destroy the global partial direction $\widehat\Phi$ itself; rather, it obstructs a globally smooth lift of that $SU(2)/U(1)$ direction to a complete $SU(2)$ frame.  The four-dimensional color-direction discussion below keeps the same distinction between a partial frame, its complete lift, and charge-sensitive center monodromy.

\subsubsection*{Layer I: kinematics and CDG construction}

This layer contains only the definition of the color-direction field, its transformation law, the CDG decomposition, the reduction condition, and the associated relational and restricted fields.

The proposal is to use \textbf{a color-direction field} $\bm{n}(x)$ extracted from the Yang--Mills field itself.
See e.g. \cite{Kondo:2014sta} for a review.
For $G=SU(2)$ one introduces a unit field $\bm{n}(x)$ as
\begin{equation}
  \bm{n}(x)=n^A(x)T_A \in su(2),
  \quad
  2\operatorname{tr}\bm{n}(x)^2=1,
  \label{eq:color-direction}
\end{equation}
which transforms according to the adjoint representation of the gauge group $G$ under the gauge transformation:
\begin{equation}
  \bm{n}(x)\mapsto \bm{n}^g(x)=g(x)\bm{n}(x)g(x)^{-1}, \ g \in G .
  \label{eq:color-direction-transform}
\end{equation}
Equivalently, at least locally, the \textbf{color direction field} $\bm{n}(x)$ is written as
\begin{equation}
  \bm{n}(x)=h_{\bm{n}}(x)T_3h_{\bm{n}}(x)^{-1},
  \quad
  h_{\bm{n}}(x)\in SU(2) .
  \label{eq:n-as-frame}
\end{equation}
The field $h_n(x)$ is \textbf{a local color frame}, whereas $\bm{n}(x)$ is only a partial frame because it fixes the Cartan axis but leaves the right action of the stabilizer $U(1)$ unfixed:
\begin{equation}
  h_{\bm{n}}(x)\sim h_{\bm{n}}(x)e^{i\alpha(x)T_3} .
  \label{eq:partial-frame-u1}
\end{equation}
Thus $\bm{n}$ is not itself a complete $SU(2)$ QRF.
It is \textbf{an intrinsic partial frame} valued in the coset $SU(2)/U(1)$.  This residual $U(1)$ is not a drawback of the construction; it is precisely the remnant Abelian frame in which monopoles and compact Abelian fluxes can appear.

A convenient dynamical way to extract $\bm{n}$ is the \textbf{Cho--Duan--Ge (CDG) decomposition} or \textbf{Cho--Duan--Ge--Faddeev--Niemi (CDGFN) decomposition} \cite{Cho:1980,DuanGe:1979,FaddeevNiemi:1999,Kondo:2014sta}.  With the convention that the generators are Hermitian, the Yang--Mills field $\mathscr{A}_\mu(x)$ is decomposed into the restricted field $\mathscr{V}_\mu(x)$ and the remaining field $\mathscr{X}_\mu(x)$:
\begin{equation}
  \mathscr{A}_\mu(x)=\mathscr{V}_\mu(x)+\mathscr{X}_\mu(x),
  \label{eq:cdg-decomposition}
\end{equation}
which is the gauge covariant decomposition in the sense that the decomposition retains the original form even after the gauge transformation
\begin{align}
\mathscr{A}_\mu^g(x)=&\mathscr{V}_\mu^g(x)+\mathscr{X}_\mu^g(x), \nonumber\\
  \mathscr{V}_\mu^g(x)
  =& g(x)\mathscr{V}_\mu(x)g(x)^{-1}
  +{i\over g_{\rm YM}}g(x)\partial_\mu g(x)^{-1},
  \nonumber\\
  \mathscr{X}_{\mu}^g(x)
  =& g(x)\mathscr{X}_{\mu}(x)g(x)^{-1},
\end{align}
so that
\begin{align}
  \mathscr{A}_\mu^g(x)
  =  g(x)\mathscr{A}_\mu(x)g(x)^{-1}
  +{i\over g_{\rm YM}}g(x)\partial_\mu g(x)^{-1}.
\end{align}
By solving the defining equations which are gauge covariant and gauge invariant:
\begin{equation}
  \mathscr{D}_\mu[\mathscr{V}]\bm{n}(x)=0,
  \quad
  2\operatorname{tr}(\mathscr{X}_\mu(x)\bm{n}(x))=0  .
  \label{eq:n-covariantly-constant-v}
\end{equation}
the restricted field $\mathscr{V}_\mu(x)$ and the remaining field $\mathscr{X}_\mu(x)$ are obtained as
\begin{align}
  \mathscr{V}_\mu(x)
  &=c_\mu(x)\bm{n}(x)
  +ig_{\rm YM}^{-1}[\bm{n}(x),\partial_\mu\bm{n}(x)],
  \nonumber\\
  c_\mu(x)&:=2\operatorname{tr}(\bm{n}(x)\mathscr{A}_\mu(x)),
  \label{eq:restricted-field}\\
  \mathscr{X}_\mu(x)
  &=-ig_{\rm YM}^{-1}[\bm{n}(x),
  \mathscr{D}_\mu[\mathscr{A}]\bm{n}(x)].
  \label{eq:remaining-field}
\end{align}
Thus \textbf{$\mathscr{V}$ is the connection that transports the partial color frame $\bm{n}$ without rotating it out of the Cartan axis, while $\mathscr{X}$ is the component that changes the frame.}
A standard \textbf{reduction condition} is obtained by minimizing the   functional of the remaining part with respect to $\bm{n}$:
\begin{equation}
\begin{aligned}
  \mathcal{R}[\mathscr{A},\bm{n}]
  &=\int d^4x\,2\operatorname{tr}\{(\mathscr{D}_\mu[\mathscr{A}]\bm{n})
  (\mathscr{D}^\mu[\mathscr{A}]\bm{n})\}\\
  &=g_{\rm YM}^2\int d^4x\,2\operatorname{tr}
  (\mathscr{X}_\mu\mathscr{X}^\mu).
\end{aligned}
  \label{eq:reduction-functional}
\end{equation}
The corresponding Euler equation as the stationary condition of the reduction functional can be written in the differential form:
\begin{equation}
  [\bm{n}(x),\mathscr{D}_\mu[\mathscr{A}]\mathscr{D}^\mu[\mathscr{A}]\bm{n}(x)]=0,
  \label{eq:reduction-condition}
\end{equation}
or, equivalently,
\begin{equation}
\mathscr{D}^\mu[\mathscr{V}]\mathscr{X}_\mu(x)=0.
\end{equation}
This condition removes the additional local rotations introduced with the color field without fixing the original Yang--Mills gauge degrees of freedom, but it need not select a unique global minimum.  It is therefore safer to introduce the set of admissible candidate partial frames
\begin{equation}
 \mathcal N[\mathscr A] := \{ \bm n \ | \
 [\bm n,\mathscr D_\mu[\mathscr A]\mathscr D^\mu[\mathscr A]\bm n]=0, \  2\operatorname{tr}\bm n^2=1 \},
\label{eq:candidate-frame-set}
\end{equation}
where $\bm n$ satisfies the prescribed boundary and regularity conditions.
A notation such as $\bm n[\mathscr A]$ is used only after a definite selection prescription has been specified, or when the observable under discussion is independent of the choice within $\mathcal N[\mathscr A]$.  If several local minima contribute physically, the measure or Jacobian that weights them must be part of the protocol.

In this language a colored quantity is not made gauge invariant by declaring a component to be absolute.  It is made relational.  For a matter field in the fundamental representation one would write, locally,
\begin{equation}
  \Psi_{\bm{n}}(x)=h_{\bm{n}}(x)^{-1}\psi(x) .
  \label{eq:fundamental-relational-n}
\end{equation}
For the gluonic field strength one writes
\begin{equation}
  \mathscr{F}_{\bm{n},\mu\nu}(x)=h_{\bm{n}}(x)^{-1}\mathscr{F}_{\mu\nu}(x)h_{\bm{n}}(x) .
  \label{eq:gluon-relational-n}
\end{equation}
These expressions are gauge invariant under the simultaneous transformation of the object and the frame, but they still carry the right-frame index associated with $h_{\bm{n}}$.
If the QRF is physically readable, the components of \eqref{eq:fundamental-relational-n} or \eqref{eq:gluon-relational-n} have operational meaning.
If the QRF is not globally definable or if its right-frame orientation is averaged over, only singlet combinations survive.

The restricted field $\mathscr{V}_\mu$ has an Abelianized curvature, since from \eqref{eq:restricted-field} one obtains
\begin{equation}
  \mathscr{F}_{\mu\nu}[\mathscr{V}](x)=\bm{n}(x)\,G_{\mu\nu}(x),
\label{eq:restricted-curvature}
\end{equation}
where $G_{\mu\nu}(x)$ is the gauge-invariant Abelian-like curvature given by

\begin{equation}
\begin{aligned}
  G_{\mu\nu}(x)
  &:=\partial_\mu c_\nu(x)-\partial_\nu c_\mu(x)\\
  &\quad+ig_{\rm YM}^{-1}2\operatorname{tr}
  \bigl(\bm{n}(x)[\partial_\mu\bm{n}(x),\partial_\nu\bm{n}(x)]\bigr) .
\end{aligned}
  \label{eq:gauge-invariant-abelian-field}
\end{equation}

\subsubsection*{Layer II: partial-frame geometry and lift obstruction}

At this stage the discussion is geometric: one asks what topological or patching data obstruct a smooth complete lift of the intrinsic partial frame.  No area law or confinement conclusion follows from this layer alone.

The second term in $G_{\mu\nu}(x)$ is the local two-form associated with the color-direction field.  It can be locally nonzero even for a smooth map $\bm{n}$ with vanishing global topological charge.  Topological information is obtained only after integrating the appropriate normalized two-form over a closed two-surface
Thus local nonvanishing of this term must not be identified with nontrivial winding by itself.  The magnetic current is obtained as the violation of the Bianchi identity up to the normalization:
\begin{equation}
  k^\mu(x)={1\over2}\epsilon^{\mu\nu\rho\sigma}\partial_\nu G_{\rho\sigma}(x) .
  \label{eq:magnetic-current-n}
\end{equation}
For a globally smooth single-valued $\bm{n}$ on a patch with no singular defect, the corresponding local Bianchi violation vanishes there.  A nonzero $k^\mu$ therefore detects defect, singularity, or nontrivial patching data associated with the chosen partial color frame and the stated CDG decomposition.  Which global obstruction is present---for example a nontrivial bundle class or a failure to lift/patch a frame smoothly---must be determined from the relevant surface integrals and transition functions; $k^\mu\neq0$ should not be described, without these qualifications, as ``precisely'' the obstruction to globalizing a complete QRF.

Only in this third layer are dynamical assumptions such as restricted-field dominance, monopole or center-vortex ensembles, infrared disorder, and an area law invoked.

The connection with the Wilson loop is supplied by the representation-dependent non-Abelian Stokes theorems of Refs.~\cite{Kondo:2008,Kondo:2014sta}.
Since the present paper does not need the details, we use only the structural content and label the following formula explicitly as schematic:
\begin{equation}
  W_R(C)\
  =\int \mathcal D\mu_{\Sigma,R}[\bm n]\,
  \exp\!\left[
  ig_{\rm YM}\int_{\Sigma:\partial\Sigma=C}G_R[\mathscr A,\bm n]
  \right].
  \label{eq:non-abelian-stokes-qrf}
\end{equation}
Here $\mathcal D\mu_{\Sigma,R}$ and $G_R$ are shorthand for the representation-dependent objects in the cited theorem.

In a restricted-field-dominance regime \cite{Kondo:2008,Kondo:2014sta}, $\mathscr{V}$ carries the long-distance confining field and $\mathscr{X}$ is massive or short-ranged.
The Wilson loop then probes the flux of the QRF curvature $G$ and the defects of $\bm{n}$.

The following chain is a \emph{conditional proposal} for four-dimensional Yang--Mills theory, not a theorem following from the CDG decomposition alone.  Monopole topology and center-vortex topology are kept distinct because they obstruct different parts of the relational-frame construction:
\begin{equation}
\begin{aligned}
\bullet &~\text{an admissible }\bm n\in\mathcal N[\mathscr A]\\
&\qquad\text{as an intrinsic partial frame}\\
&\Longrightarrow \text{restricted connection} \mathscr{V} \text{and curvature }G\\
\bullet &~ \text{monopole topology}\\
&\Longrightarrow \text{obstruction to a smooth complete lift of }\bm n,
\nonumber\\
\bullet &~\text{center-vortex topology}\\
&\Longrightarrow \text{center-valued monodromy/patching}\\
&\qquad\text{of a charge-sensitive complete frame}\\
&\Longrightarrow \text{possible long-distance disorder of relational color,}\\
&\qquad\text{subject to the infrared dynamics}.
\end{aligned}
\nonumber
\label{eq:four-dimensional-qrf-chain}
\end{equation}

The advantage of this formulation is that it does not require adding an external Higgs field.  The reference frame is extracted from the gauge field itself.  This is why the construction is potentially applicable to the pure Yang--Mills theory.

There are, however, important qualifications.

\noindent
(i) First, the candidate set $\mathcal N[\mathscr A]$ selected by the reduction condition can contain several local minima or Gribov-related representatives.
In QRF language, different local minima correspond to different candidate frames.  A confinement criterion cannot depend on an arbitrary choice among them.  One must therefore either formulate observables that are independent of the choice among admissible minima, or specify a dynamical/statistical prescription, including the measure or Jacobian, for summing over the copies.  An unweighted ``average over all admissible frames'' is not by itself a defined physical operation.

\noindent
(ii) Second, $\bm{n}$ is only an {intrinsic partial frame}; a complete $SU(2)$ frame requires
the lift \eqref{eq:n-as-frame} and the treatment of the residual $U(1)$ redundancy.

\noindent
(iii) Third, a defect of $\bm{n}$ is not by itself a proof of an area law.  One still needs a dynamical statement: condensation, percolation, a monopole gas, a center-vortex ensemble, a BF-type effective theory, or another mechanism that produces a nonzero string tension.

Nevertheless the QRF reading is useful for interpreting what the color-direction field measures.  It is not an absolute color direction, and its order or disorder is not a violation of Elitzur's theorem; the absence of spontaneous gauge-variant magnetization in a lattice Abelian Higgs--Kibble model provides a concrete related example \cite{Elitzur:1975}.  It is a dynamical partial frame used to compare color orientations.  When this frame is globally coherent, relational color can have long-distance meaning.  When it is disordered by defects, or when a \emph{specified} physical prescription with a defined measure sums over inaccessible/copy-related frame data, open non-singlet frame indices may disappear from the readout.  This kinematical or statistical projection must still be combined with the charged-sector criterion before it is interpreted as four-dimensional color confinement.

\subsubsection*{Falsifiable targets for the four-dimensional CDG/QRF proposal}

The conjectural QRF content can be separated from ordinary restricted-field dominance by measuring, on the same gauge ensemble and with a fixed reduction-condition/copy prescription, both frame coherence and charged-sector energetics.  A gauge-invariant partial-frame coherence correlator is for example
\begin{equation}
 C_{\bm n}(R=|x-y|;\gamma)
 :=
 \left\langle
 2\operatorname{tr}\!\left[
 \bm n(x)\,U_\gamma(x,y)\,
 \bm n(y)\,U_\gamma(y,x)
 \right]
 \right\rangle_{\mathcal P},
\label{eq:qrf-frame-coherence-test}
\end{equation}
where $U_\gamma$ is the full parallel transporter and the protocol $\mathcal P$ fixes how the candidate set $\mathcal N[\mathscr A]$ and reduction-condition copies are treated.  Its long-distance decay defines a directly testable \textbf{frame-coherence scale} when such a scale exists.
For a quantitative continuum coherence length, the protocol must also specify the path family $\gamma(x,y)$ and the renormalization of the open parallel transporter.  In particular, line self-energy divergences and, for piecewise-smooth paths, endpoint/corner/cusp renormalization must be removed by a fixed prescription before correlators obtained from different separations or path choices are compared.  The distance variable $R$ must likewise be tied to that path prescription (for example, straight Euclidean separation in a fixed geometry).  Without these choices, a decay length extracted from $C_{\bm n}(R;\gamma)$ is prescription dependent and should not be treated as a universal QRF coherence length.

Independently, the transfer matrix or projected partition functions can be used to determine the \textbf{optimized sector energy}
\begin{equation}
 \Delta E_{[R]}(L)
 =
 -\lim_{T\to\infty}\frac1T
 \log\frac{Z_{[R]}(T,L)}{Z_{[1]}(T,L)}
\label{eq:qrf-sector-energy-lattice-test}
\end{equation}
for a fixed stable screening class and boundary protocol.  These quantities can be correlated with the full and restricted Wilson-loop string tensions and with monopole/center-vortex observables.

The specifically QRF part of the four-dimensional proposal is therefore falsifiable.  For one fixed operational protocol and the same lattice ensemble, one should compare the optimized sector energy $\Delta E_{[R]}(L)$, the renormalized relational partial-frame correlator $C_{\bm n}(R;\gamma)$, the full and restricted Wilson-loop string tensions $\sigma_{\rm Wilson}$, and a specified defect observable or density $\rho_{\rm defect}$.  Two explicit failure conditions are particularly important.  First, if the selected frame loses long-distance coherence or fails to globalize while another admissible relational realization still supports a nontrivial stable screening class with finite isolation energy, then the obstruction of that one frame is insufficient as a confinement mechanism.  Second, if finite-size scaling shows no reproducible correlation between loss of partial-frame coherence/defect disorder and exclusion of the nontrivial charged sector, then the proposed QRF mechanism is not supported, even if conventional restricted-field or monopole-dominance observables remain nontrivial.  Conversely, correlated scaling among these observables would support the proposal but would not by itself establish causation or a universal confinement theorem.  No numerical result for \eqref{eq:qrf-frame-coherence-test} or \eqref{eq:qrf-sector-energy-lattice-test} is claimed here; these quantities are proposed as direct tests of the four-dimensional conjecture.

\subsection{Including adjoint-scalar and fundamental-scalar examples}
\label{subsec:color-direction-examples}

It is useful to compare the pure Yang--Mills construction with gauge--scalar models.

\begin{figure}[hbt]
\centering
\includegraphics[height=45mm] {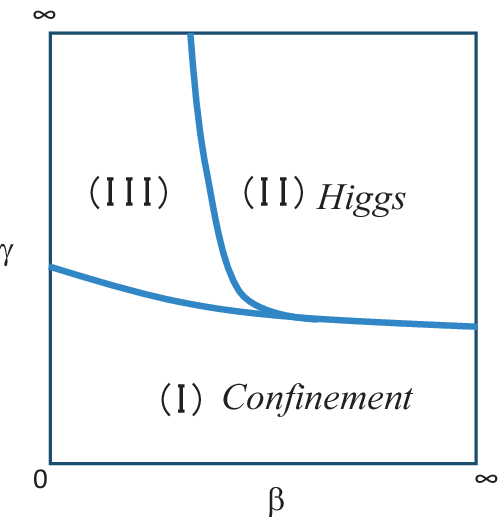}
\caption{
Regions of the lattice $SU(2)$ gauge--adjoint-scalar model distinguished by the gauge-invariant \textit{nonlocal}/relational diagnostic of Ref.~\cite{ShibataKondo:2024adjoint}, adapted/redrawn from that reference.  The labels (I) and (III) denote confinement-like regions with different relational-frame behavior, whereas (II) denotes a Higgs-like region.  The diagram is not intended to assert that every boundary shown is a thermodynamic phase transition in \textit{local} gauge-invariant observables.
}
\label{fig:phase_diagram1}
\end{figure}

First, we consider the case: an adjoint scalar field $\phi(x)$ is present \cite{ShibataKondo:2024adjoint}.  The normalized scalar direction
\begin{equation}
  \hat\phi(x)={\phi(x)\over \sqrt{2\operatorname{tr}(\phi(x)^2)}}
\end{equation}
is another \textbf{partial color frame}.  The gauge-field-derived color direction $\bm{n}(x)$ and the scalar direction $\hat\phi(x)$ can then be compared by the gauge-invariant relational quantity
\begin{equation}
  q(x)=2\operatorname{tr}\{\bm{n}(x)\hat\phi(x)\} .
  \label{eq:adjoint-relational-q}
\end{equation}
On a lattice one may use the order parameter convention of Ref.~\cite{ShibataKondo:2024adjoint}, schematically
\footnote{
The normalization in this lattice expression, including the use of the lattice adjoint scalar $\phi_x$ rather than the normalized direction $\widehat\phi_x$, follows the convention of Ref.~\cite{ShibataKondo:2024adjoint} and is not meant to be the literal discretization of Eq.~\eqref{eq:adjoint-relational-q}.  For a direct continuum--lattice comparison one should normalize the scalar direction consistently and adjust the overall trace normalization accordingly.
}
\begin{equation}
  Q={1\over N_{\rm site}}\sum_x {1\over2}\operatorname{tr}(n_x\phi_x) .
  \label{eq:lattice-q-qrf}
\end{equation}
In QRF terms, $Q$ is the average alignment of two frames: the frame $\bm{n}$ extracted from the gauge field and the frame $\phi$ supplied by the adjoint scalar.  A nonzero value of $\langle |Q|\rangle$ is not an expectation value of a gauge-variant scalar.  $\langle |Q|\rangle$ is a \textbf{relational order parameter} and is therefore compatible with Elitzur's theorem \cite{Elitzur:1975}.

For the adjoint scalar field $\phi$, it explains why \textbf{a confinement region can split into subregions} \cite{ShibataKondo:2024adjoint}.
One may have
\begin{equation}
\begin{aligned}
  &\text{(I)}~\text{frame-disordered confinement:}
  \quad \langle |Q|\rangle\simeq0,\\
  &\text{(III)}~\text{frame-locked confinement:}
  \quad \langle |Q|\rangle>0,\\
  &\quad\text{but residual compact }U(1)\text{ still confines},\\
  &\text{(II)}~\text{Higgs or Coulomb-like regime:}
  \quad \langle |Q|\rangle>0,\\
  &\quad\text{residual long-distance gauge force is not confining} .
\end{aligned}
\nonumber
\label{eq:three-region-qrf}
\end{equation}
The new information is not simply whether a Wilson loop confines fundamental probes.  Rather, the relational observables can distinguish which frame structures and defect data characterize different regimes; whether these data actually realize confinement is decided separately by the infrared charged-sector dynamics.

\begin{figure}[h]
\includegraphics[width=45mm]{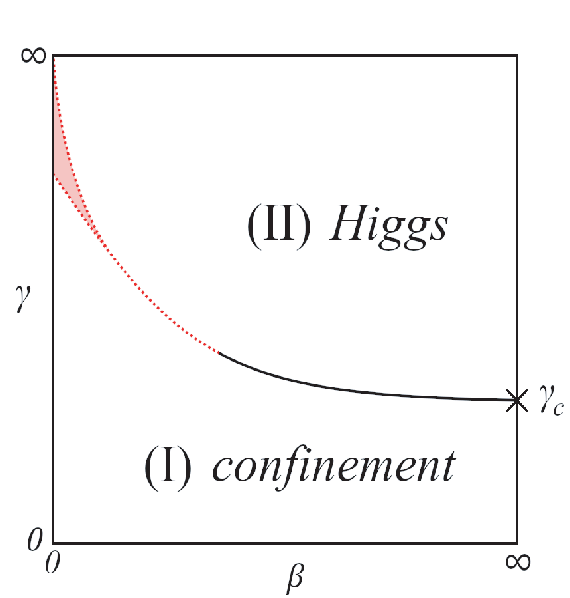}
\caption{
Regions of the lattice $SU(2)$ gauge--fundamental-scalar model distinguished by the gauge-invariant \textit{nonlocal} diagnostic of Ref.~\cite{IkedaKatoKondoShibata:2024fundamental}, adapted/redrawn from that reference.  The labels (I) and (II) are used here as confinement-like and Higgs-like relational regions; the shaded band indicates the finite-volume uncertainty discussed in the cited work.  In accordance with Fradkin--Shenker continuity, this caption does not interpret the diagram as a universal thermodynamic phase boundary defined by \textit{local} gauge-invariant observables.
}
\label{fig:phase_diagram2}
\end{figure}

Next, we consider the case in which a scalar field in the fundamental representation is present \cite{IkedaKatoKondoShibata:2024fundamental}.  This case differs in an important way from the adjoint-scalar example above.  A nonzero adjoint scalar specifies only a color axis and hence a partial frame in $SU(2)/U(1)$, whereas a nonzero fundamental $SU(2)$ doublet has no nontrivial stabilizer (stability subgroup) under the left $SU(2)$ action and can therefore supply a complete local $SU(2)$ frame.  Possible zeros of the scalar field must be excluded or treated patchwise.

For a fundamental scalar doublet $\Phi(x)$, define on a patch where $\Phi^\dagger(x)\Phi(x)\neq0$
\begin{equation}
 \Theta(x)
 :=\frac{1}{\sqrt{\Phi^\dagger(x)\Phi(x)}}
 \bigl(i\sigma_2\Phi^*(x),\,\Phi(x)\bigr)
 \in SU(2).
 \label{eq:fundamental-frame-theta}
\end{equation}
In the radially fixed lattice model of Ref.~\cite{IkedaKatoKondoShibata:2024fundamental}, the corresponding matrix-valued scalar $\widehat\Theta_x\in SU(2)$ is used directly.
Under the local gauge transformation $\Omega_x$ and the right global transformation $\Gamma$ of that model,
\begin{align}
 &\widehat\Theta_x\longmapsto
 \Omega_x\widehat\Theta_x\Gamma,
 \quad
 \bm n_x\longmapsto\Omega_x\bm n_x\Omega_x^\dagger,
 \nonumber\\
 &\Omega_x\in SU(2)_{\rm local},\quad
 \Gamma\in\widetilde{SU(2)}_{\rm global}.
\end{align}
Thus \textbf{the scalar frame removes the left local color orientation, but it leaves a right/global orientation that must be kept distinct from the original gauge redundancy}.

The fundamental-scalar analogue of the adjoint overlap $q(x)$ is the matrix-valued \textbf{relational color direction}
\begin{equation}
 \bm r_x^{\Theta}
 :=\widehat\Theta_x^\dagger\bm n_x\widehat\Theta_x,
 \quad
 \bm r_x^{\Theta}\longmapsto
 \Gamma^\dagger\bm r_x^{\Theta}\Gamma.
 \label{eq:fundamental-scalar-relational-direction}
\end{equation}
It is invariant under the local gauge group and covariant only under the right/global action.  In QRF language, $\bm r_x^{\Theta}$ is the gauge-field-derived partial color frame $\bm n_x$ as read in the complete matter frame $\widehat\Theta_x$.  Its components are therefore \textbf{relational components in the scalar-frame coordinates}, rather than absolute color components.

On the lattice one forms the spacetime average $\bm R$ and its rotationally insensitive norms $\|\bm R\|_1,\|\bm R\|_2$,
 \begin{align}
 &\bm R:=\frac{1}{N_{\rm site}}\sum_x\bm r_x^{\Theta}
       =R^A\sigma_A,
 \quad
 R^A=\frac12\operatorname{tr}(\sigma_A\bm R),\nonumber\\
 &\|\bm R\|_1:=\sum_{A=1}^{3}|R^A|,
 \quad
 \|\bm R\|_2:=
 \left[\frac12\operatorname{tr}(\bm R^\dagger\bm R)\right]^{1/2}.
 \label{eq:fundamental-scalar-r}
 \end{align}
The matrix $\bm R$ itself is gauge invariant but transforms covariantly under $\widetilde{SU(2)}_{\rm global}$.  Ref.~\cite{IkedaKatoKondoShibata:2024fundamental} therefore diagnoses the new separation using the magnitude of this scalar--color correlation, in particular the finite-volume-subtracted expectation values $\langle\|\bm R\|_n\rangle^{\rm sub}$ ($n=1,2$) and their susceptibilities, rather than by treating $\bm R$ itself as a gauge-symmetry-breaking expectation value.  This construction is compatible with Elitzur's theorem: no local gauge symmetry is spontaneously broken by a nonzero value of the gauge-invariant relational diagnostic.

For the fundamental-scalar model, the result corresponding to the three-region summary of the adjoint model is instead the following two-region pattern:
\begin{align}
 &\text{(I) \textbf{frame-disordered} confinement-like region:}\nonumber\\
 &\quad \langle\|\bm R\|_n\rangle^{\rm sub}\longrightarrow0,
\nonumber\\
 &\text{(II) \textbf{frame-locked} Higgs-like region:}\nonumber\\
 &\quad \langle\|\bm R\|_n\rangle^{\rm sub}>0\nonumber\\
 &\quad\text{in the infinite-volume extrapolation}.
 \nonumber
\end{align}
Thus the nonlocal scalar--color diagnostic measures whether \textbf{the intrinsic gauge-field frame and the complete fundamental-scalar frame possess persistent long-distance relative alignment}.  The crucial difference from the adjoint case is that there is no residual local compact $U(1)$ stabilizer of a nonzero fundamental doublet.  Consequently, the adjoint-model region (III), described above as a frame-locked confinement region associated with a residual compact $U(1)$ sector, has no direct counterpart in the phase diagram of Ref.~\cite{IkedaKatoKondoShibata:2024fundamental} and must not be imported into the fundamental-scalar discussion.

The terminology ``confinement-like'' and ``Higgs-like'' in this paragraph refers to the relational/nonlocal classification used in Ref.~\cite{IkedaKatoKondoShibata:2024fundamental}; it should not be confused with spontaneous breaking of the local gauge group.  In the nonlocally deformed description containing the color-direction field determined by the reduction condition, $\bm R$ also probes the realization of the right global symmetry $\widetilde{SU(2)}_{\rm global}$.  This global symmetry is conceptually distinct from the local gauge redundancy and should also be distinguished from a claim that the original local gauge--scalar theory develops a new local thermodynamic order parameter.

This qualification is essential for consistency with \textbf{Fradkin--Shenker continuity}.  With fundamental Higgs matter, confinement-like and Higgs-like regimes can be analytically connected without a thermodynamic singularity in local gauge-invariant observables in the Fradkin--Shenker analyticity region \cite{Fradkin:1978dv}.  The diagnostic of Ref.~\cite{IkedaKatoKondoShibata:2024fundamental} contains the nonlocally determined functional $\bm n[U]$ and therefore lies outside the class of strictly local observables to which that analyticity statement directly applies.  The new relational separation should consequently not be presented as a universal thermodynamic phase boundary of the original local theory.

Moreover, dynamical fundamental matter explicitly removes the electric $\mathbb Z_2^{(1)}$ center one-form symmetry as an exact symmetry, therefore the distinction detected by $\bm R$ is not a center-one-form-symmetry order parameter \cite{Gaiotto:2014}.

From the QRF viewpoint, the robust conclusion is correspondingly modest but useful.
Region (II) exhibits \textbf{persistent locking between a gauge-field-derived partial frame and a matter-derived complete frame}, whereas region (I) does not exhibit that particular long-distance locking.
This \textbf{frame-locking diagnostic} is not by itself the QRF confinement criterion.
In a theory with dynamical fundamental matter, screening channels must first be included when the stable charged classes are identified; one must then ask whether any nontrivial class survives as a finite-energy, physically accessible, asymptotically isolated charged representation under the protocol of Sec.~\ref{sec:qrf-confinement}.

The lesson for pure four-dimensional Yang--Mills theory is therefore not that confinement requires every intrinsic color frame to become incoherent.  Rather, in the absence of an elementary scalar one must construct an admissible intrinsic partial frame $\bm n\in\mathcal N[\mathscr A]$ or an admissible complete dressing $h[\mathscr A]$ from the gauge field itself.  The long-distance order or disorder of such relational frame data can provide a diagnostic or participate in a dynamical confinement mechanism, but the final statement of QRF confinement remains the independent charged-sector exclusion test.  Failure of one selected frame to remain coherent is neither necessary nor sufficient by itself.

\section
{A QRF interpretation of the Wilson-loop area law}
\label{sec:area-law-qrf}
\textit{QRF status:} mathematical dressing/frame transporter; the closed Wilson loop is a frame-independent observable, not an independent physical QRF.
\par

The Wilson loop is normally introduced as a gauge-invariant operator measuring the response of the vacuum to an external static color source.
For a large rectangular loop $C=r\times t$, the area law of the Wilson loop average
\begin{equation}
  \langle W(C)\rangle\sim e^{-\sigma rt}
  \label{eq:area-law-rectangular}
\end{equation}
implies a linear potential $V(r)\sim\sigma r$.
This is the standard quark-confinement interpretation due to Wilson \cite{Wilson:1974sk}.
The QRF viewpoint does not deny this interpretation.
It adds a second reading:
\textbf{a Wilson loop records gauge-invariant conjugacy-class information in the holonomy produced by transporting a color frame around a closed curve}; nontrivial holonomy is not itself a failure of transport consistency.

Indeed, the parallel transporter
\begin{equation}
  U_\Gamma(x,y)=\mathrm{P}\exp\left(ig_{\rm YM}\int_{\Gamma:y \to x} \mathscr{A}\right)
  \label{eq:open-transporter-area-section}
\end{equation}
compares the color frame at $x$ with the color frame at $y$ along the path $\Gamma$.  An open Wilson-line-dressed mesonic operator
\begin{equation}
  M_\Gamma(x,y)=\bar\psi(x)U_\Gamma(x,y)\psi(y)
  \label{eq:meson-relational-area}
\end{equation}
is therefore a relational comparison of two colored endpoints.  A closed Wilson loop is the special case in which the comparison is brought back to the initial frame and traced.  In this sense \textbf{the Wilson loop}
\begin{equation}
  W(C)= {1\over N}\operatorname{tr}\,\left\{ \mathrm{P}\exp\left(ig_{\rm YM}\oint_C\mathscr{A}\right) \right\}
  \label{eq:wilson-loop-area-section}
\end{equation}
\textbf{measures the gauge-invariant conjugacy-class information of the holonomy mismatch obtained after parallel transport around the closed curve $C$}.
Parallel transport itself is well defined for a smooth connection; a nontrivial holonomy is therefore not, by itself, a failure of consistency.  When an appropriate defect ensemble dynamically disorders this holonomy, the result may be interpreted in the QRF language as \textbf{degradation of long-distance frame comparison}.

As a \textbf{toy independent-cell model} for a possible QRF interpretation, suppose that a long-distance frame comparison is degraded by statistically independent effective cells on a spanning surface.
If a stable long-distance color frame is available, perimeter-dominated behavior may instead arise from fluctuations localized near the contour.
In the toy area-disordering picture, let each cell of area $a^2$ contribute an average model-dependent overlap factor $\kappa$ ($0<\kappa<1$).  Then one obtains

\begin{equation}
\begin{aligned}
  \langle W(C)\rangle
  &\simeq \prod_{\hbox{cells in }\Sigma}\kappa
  =\kappa^{\operatorname{Area}(\Sigma)/a^2}
  \\
  &=
  \exp\left[(\ln \kappa)
  {\frac{\operatorname{Area}(\Sigma)}{a^2}}\right]
  =
  \exp\left[-|\ln\kappa|
  {\operatorname{Area}(\Sigma)\over a^2}\right].
\end{aligned}
  \label{eq:cellular-area-law}
\end{equation}

\begin{figure}[htbp]
  \centering
\includegraphics[width=.50\linewidth]{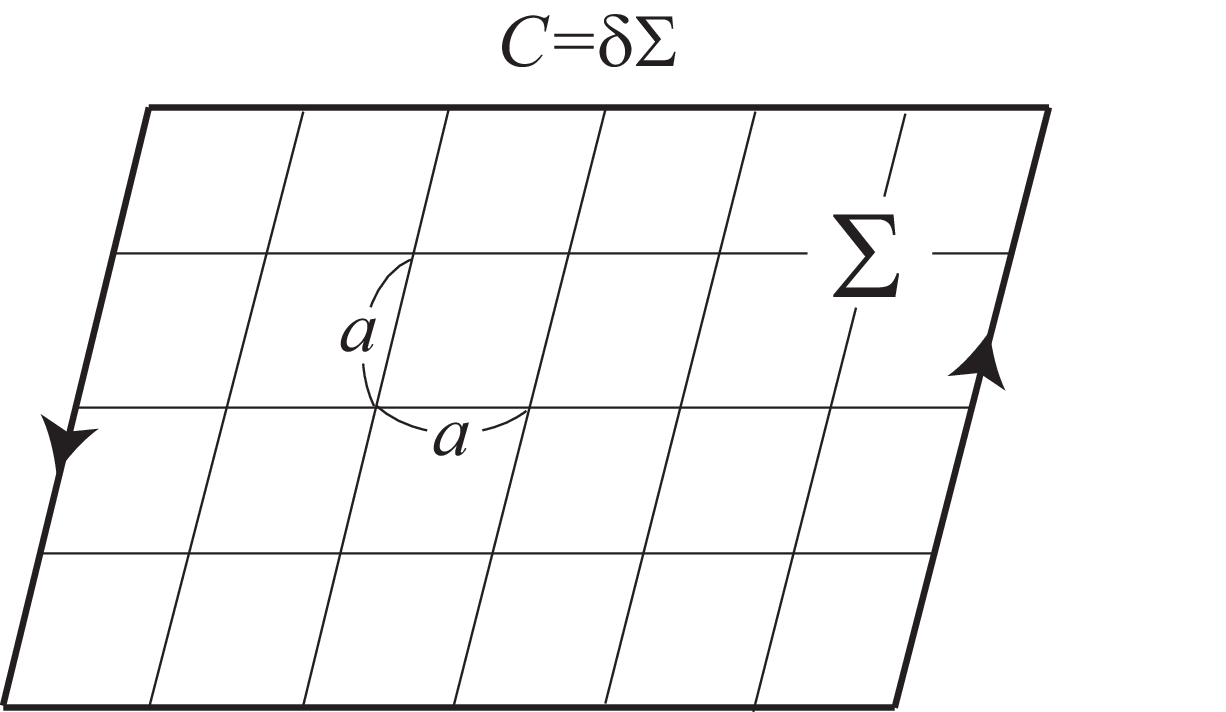}
  \caption{The surface $\Sigma$: $C=\partial \Sigma$ surrounded by the Wilson loop $C$ is divided into a collection of a small area $a^2$.}
  \label{fig:area-law}
\end{figure}
This toy independent-cell model is only a \textbf{heuristic QRF interpretation} of an already assumed \textbf{area-disordering mechanism}; it is not a derivation of the Wilson-loop area law from QRF principles:
\textbf{area law admits a heuristic interpretation as loss of relational frame coherence.}

Within this toy picture \textbf{the string tension can be read as an effective decay rate of the model-dependent overlap parameter $\kappa$ per unit area}.  No equality with an independently defined information-theoretic quantity is asserted.
At present the information-theoretic statement is only a prospective one: a quantitative notion of QRF coherence would require a specified parallel-transport quantum channel together with a definite observable such as a channel fidelity, mutual information, relative entropy, or depolarizing parameter, followed by a derivation relating that quantity to the Wilson-loop string tension.
Until such a construction is supplied, the only claim made here is the weaker one that an already established area law can be interpreted heuristically as long-distance disorder of relational frame comparison.
The same logic can be made more topological in a center-vortex picture \cite{tHooft:1978,Greensite:2003bk}.  Suppose a center vortex piercing $\Sigma$ multiplies a fundamental Wilson loop by a center element $z=e^{2\pi i/N}$.  If each plaquette of the surface is pierced with probability $p$, then the Wilson loop average  for $SU(2)$, i.e., $N=2$  behaves schematically as
In the simple real-valued $SU(2)$ toy formula below we restrict to the dilute regime $0\le p<1/2$ (in practice $p\ll1$), so that $1-2p>0$ and the logarithm defining the string tension is real without an additional phase convention.
\begin{equation}
\begin{aligned}
  \langle W_{\rm fund}(C)\rangle
  &\sim \bigl((1-p)\cdot 1+p\cdot z\bigr)^{
  \operatorname{Area}(\Sigma)/a^2}\\
  &=\exp[-\sigma_R\operatorname{Area}(\Sigma)].
\end{aligned}
  \label{eq:center-vortex-area-law}
\end{equation}
For $SU(2)$, $z=-1$, so in this regime $\sigma_R=-a^{-2}\log(1-2p)>0$; outside it one should use $-a^{-2}\log|1-2p|$ together with the corresponding sign/phase information.
\footnote{
For the general $SU(N)$, i.e., $N>2$
the Wilson-loop average is not guaranteed to be real valued.
If only a center phase $z \in \mathbb{C}$ with one orientation  is included, then
$
 [(1-p)+pz]^{A/a^2}
$
is generally complex.
In a CP-symmetric or orientation-symmetric ensemble, $z$ and $z^{-1}$ should be included simultaneously. For example, one may introduce a probability distribution $P(z)$ and write
\begin{equation}
\begin{aligned}
 \langle W_R(C)\rangle
 &\simeq
 \left[(1-p)+p\sum_{z\in Z_N}P(z)\omega_R(z)\right]^{A/a^2}\\
 &=\exp[-\sigma_R\operatorname{Area}(\Sigma)+i\theta_R \operatorname{Area}(\Sigma)] ,
\end{aligned}
\end{equation}
with $P(z)=P(z^{-1})$ imposed. Ultimately it is better to write something like
$
 \sigma_R=-\frac{1}{a^2}
 \log\left|(1-p)+p\langle\omega_R(z)\rangle\right|.
$
}

The usual phrase is that \textbf{the Wilson-loop phase is disordered by vortices}.  \textbf{In QRF language, the same defect data can obstruct the globalization of a specified color-frame construction through $\Sigma$ by producing nontrivial monodromy}.  These are \textbf{complementary descriptions of the same defect data}, not automatically identical dynamical statements: they become dynamically related only after the physical defect ensemble and the infrared limit are specified.  Thus phase disorder of the line operator should not be identified kinematically with frame obstruction alone.

A magnetic monopole or dual-superconductor mechanism \cite{Nambu:1974,tHooft:1975,Mandelstam:1976} admits an analogous reading.  In the Abelianized low-energy theory,
it has been shown by Polyakov \cite{Polyakov:1977} that monopole condensation gives a mass to the dual gauge field, or equivalently a dual photon.  The Wilson loop imposes a discontinuity or a domain wall on a spanning surface, and the wall tension gives the area law.
\textbf{In the QRF language the same domain wall is the cost of enforcing a color-frame comparison across a region in which the preferred frame is not globally coherent.
The flux tube is the physical remnant of the attempt to keep a non-singlet relational object isolated.}

This reading clarifies the relation between the area law and color confinement.  A fundamental external-source sector requires a gauge-field dressing to relationally connect the source to other reference data.  In a confining vacuum the lowest-energy state in that sector develops a flux tube and an energy cost growing with separation; a particular thin Wilson-line dressing is a trial relational realization and need not coincide with the energy-minimizing flux profile.  The fundamental Wilson-loop area law therefore strongly indicates that an external fundamental color source cannot be isolated at finite energy, without identifying an arbitrary dressing path with the physical flux tube.

As summarized in Table~\ref{tab:status-map}, the area law and the charged-sector criterion are related but logically distinct.  The Wilson loop diagnoses the energy of external sources, whereas QRF confinement concerns the survival of a finite-energy, physically accessible, asymptotically isolated nontrivial charged representation.  With dynamical fundamental matter the Wilson string can break, so the latter operational question remains meaningful even when an asymptotic fundamental Wilson-loop area law is absent.

The comparison with Kugo--Ojima is already stated in Sec.~\ref{subsec:kugo-ojima-qrf}: it concerns a different state-space construction and no equivalence theorem is asserted.  Likewise, QRF by itself does not determine the string tension.  The limited conclusion needed here is that an established Wilson-loop area law admits a heuristic relational-frame interpretation; making this quantitative requires an independently defined transport channel or information measure and a derivation of its relation to the string tension $\sigma$.

\section{Comparison with generalized-symmetry diagnostics}
\label{sec:generalized-symmetry-comparison}

Generalized global symmetry gives \cite{Gaiotto:2014,BanksSeiberg:2010,KapustinSeiberg:2014} one of the sharpest modern formulations of confinement.  In pure $SU(N)$ Yang--Mills theory the \textbf{electric center one-form symmetry} $\mathbb{Z}_N^{(1)}$ acts on Wilson lines.  A Wilson line in representation $R$ carries the $N$-ality $q_R\in\mathbb{Z}_N$.
In the standard gapped confining setting of pure $SU(N)$ Yang--Mills theory, \textbf{external-probe confinement is associated with an unbroken electric center one-form symmetry}, and \textbf{the area law for Wilson loops with $q_R\ne0$} is the corresponding standard diagnostic:
\begin{equation}
  \langle W_R(C)\rangle\sim\exp[-\sigma_{q_R}\operatorname{Area}(C)],
  \quad q_R\ne0 .
  \label{eq:one-form-area-law}
\end{equation}
Generalized symmetry also organizes genuine line operators, line endability, boundary conditions, background two-form fields, anomaly inflow, and mixed anomalies.  These are genuine strengths.  They give topological constraints that are independent of a particular gauge choice and often survive strong coupling.

The QRF viewpoint is not a replacement for generalized symmetry.  It addresses a different but overlapping question.
\textbf{Generalized symmetry asks how extended objects charged under higher-form symmetries transform.}
\textbf{QRF asks whether a physical reference frame exists that allows a non-singlet color label to be read out at all.}
 The difference is most visible in the following points.

The screening class must be fixed before the two descriptions are compared.

\noindent
$\bullet$ In pure $SU(N)$ Yang--Mills theory with nonzero $N$-ality, the asymptotically stable external-probe class is precisely the center class $q_R\in\mathbb Z_N$;
\textbf{QRF language adds an operational account of how the color label is read, not an additional conserved charge}.

\noindent
$\bullet$ For zero $N$-ality, gluon screening makes the asymptotic external screening class trivial, so no new stable adjoint superselection charge is introduced here.
\textbf{QRF variables may still carry finite-separation adjoint relational information} and can be used to describe how that information is screened, but an independent asymptotic adjoint charge would have to be established separately.

\noindent
$\bullet$
With dynamical fundamental matter, the electric center one-form symmetry is explicitly broken and the stable class must be defined after matter screening;
\textbf{The QRF question can still be posed relative to the chosen boundary and resource protocol, but it is no longer a center-symmetry order parameter.}

\noindent
(i)
First, \textbf{center one-form symmetry detects $N$-ality, not the full Lie-algebra-valued color representation}.
In fact, a fundamental Wilson line is charged under the center, whereas an adjoint Wilson line and a gluon have zero $N$-ality.
After gluon screening, the stable external adjoint screening class is therefore trivial in pure Yang--Mills theory; the present formulation does not postulate an additional adjoint superselection charge.

	\textbf{What QRF language can still track is finite-separation relational adjoint color information, or a genuine boundary representation} $\lambda_\partial$ when the boundary condition retains one.
\textbf{Whether any nontrivial full-color asymptotic charge survives is an independent charged-sector question, not a consequence of zero center charge.}

 In QRF language, the gluonic field strength becomes a relational object
\begin{equation}
  \mathscr{F}_{h,\mu\nu}(x)=h(x)^{-1}\mathscr{F}_{\mu\nu}(x)h(x) .
  \label{eq:gluon-qrf-generalized}
\end{equation}
If the right-frame orientation is physically inaccessible in the specified protocol, Haar twirling projects onto the $K_{\rm tw}$-fixed subalgebra.  For a full-group twirl, an operator carrying a nontrivial irreducible open adjoint frame index is projected out and full-group singlet contractions remain, for example
$\operatorname{tr}\mathscr{F}_{\mu\nu}\mathscr{F}^{\mu\nu}$ and
$\operatorname{tr}\mathscr{F}\wedge\mathscr{F}$.  For a proper subgroup $K_{\rm tw}$, however, additional $K_{\rm tw}$-fixed components may survive; twirling is therefore a fixed-point projection, not generically a projection onto full-group singlets.  This kinematical projection is not, by itself, a proof of dynamical gluon confinement: one must still establish the absence of a finite-energy physically accessible charged sector using the sector-energy, charge-algebra, accessibility, and infrared-representation conditions of Sec.~\ref{subsec:operational-qrf} and Appendix~\ref{app:operational-qrf-revised}.  The invariant contractions above are examples of glueball channels that remain in the fixed-point algebra.

\noindent
(ii)
Second, \textbf{generalized symmetry becomes less directly diagnostic when fundamental dynamical matter is present, because the electric center one-form symmetry is explicitly broken and Wilson lines can end on dynamical matter.}
 \textbf{QRF can still ask whether a long-distance relational color frame is available under the specified protocol.}

\noindent
$\bullet$ In a Higgs-like regime an accessible scalar field may provide matter-frame data, and a composite such as $\Phi^\dagger\psi$ can be a gauge-invariant relational excitation.

\noindent
$\bullet$ In a confinement-like regime the corresponding protocol may fail to support a finite-energy, physically accessible, asymptotically isolated nontrivial charged sector.

This is an operational and semi-local distinction; it is not asserted to define a thermodynamic phase boundary in local gauge-invariant observables.

\noindent
(iii)
Third, generalized symmetry precisely controls topological charges of line operators and their endpoints, but it does not by itself specify the Lie-algebra-valued boundary color frame.
The Gauss law on a spatial region $V$ gives

\begin{equation}
\begin{aligned}
  G[\epsilon]
  &=\int_V d^3x\,2\operatorname{tr}
  \{\epsilon\mathscr{D}_j\mathscr{E}^j\}\\
  &=-\int_V d^3x\,2\operatorname{tr}
  \{(\mathscr{D}_j\epsilon)\mathscr{E}^j\}+\int_{\partial V}dS_j\,2\operatorname{tr}
  \{\epsilon\mathscr{E}^j\} .
\end{aligned}
  \label{eq:gauss-boundary-generalized}
\end{equation}
The bulk term vanishes on physical states when $\epsilon$ is a redundancy.  If $\epsilon$ is nonzero at the boundary, the last term is a \textbf{boundary charge}.  To assign a component to this boundary charge one needs a \textbf{boundary color frame}.
QRF makes this explicit by writing
\begin{equation}
  Q_{\partial V,h}^A
  =\int_{\partial V}dS_j\,2\operatorname{tr}\{T^A h^{-1}\mathscr{E}^j h\} .
  \label{eq:qrf-boundary-generalized}
\end{equation}
\textbf{The disappearance of physically readable non-singlet boundary color is then not merely the absence of a one-form charged line; it is the disappearance of a physically readable boundary color frame supporting a non-singlet charge.}

\noindent
(iv)
Fourth, defects have two complementary descriptions.
\textbf{In generalized symmetry, a center vortex is naturally described by its center flux and by the phase it gives to a Wilson line.}
  Around a loop $C$ linking the defect one may have
$
  h(2\pi)=h(0)z,
  \quad z\in Z_N,
$
 or, more generally, a nontrivial $SU(N)/Z_N$ bundle.
Generalized symmetry captures the center element $z$ extremely well.

\textbf{In QRF language, the same center vortex obstructs the globalization of a charge-sensitive $SU(N)$ lift needed to read a probe with nonzero $N$-ality.  Center-neutral adjoint or $PSU(N)=SU(N)/Z_N$ frame data need not suffer from the same obstruction.}
This is why \textbf{the QRF formulation provides a common relational vocabulary for comparing center vortices, magnetic monopoles, Gribov-copy structure, residual-gauge-symmetry restoration, and the disappearance of colored relational observables.}
 It does not assert that these mathematically distinct structures are one and the same confinement mechanism.

The comparison table (Table~\ref{tab:gen-symm-qrf}) below keeps the kinematical effect of frame averaging distinct from the dynamical charged-sector criterion.
  In particular, projecting an inaccessible open frame index is not itself a proof that the corresponding adjoint-colored sector is dynamically confined.

Thus the two viewpoints are complementary: generalized symmetry gives the sharper topological classification of genuine line operators, $N$-ality, anomaly constraints, and line endability, while QRF language specifies the relational resources needed to read a non-singlet color label and asks whether the corresponding charged representation survives the infrared limit.

The Kugo--Ojima comparison is given in Sec.~\ref{subsec:kugo-ojima-qrf}; here it is enough to emphasize that generalized symmetry, QRF, and Kugo--Ojima address overlapping questions with different observables and assumptions rather than three equivalent formulations.  Table~\ref{tab:gen-symm-qrf} summarizes the comparison relevant to this section.

\begin{widetext}

\begin{center}
\footnotesize
\begin{tblr}
{colspec={Q[l,t,wd=0.20\textwidth] Q[l,t,wd=0.35\textwidth] Q[l,t,wd=0.41\textwidth]},colsep=2pt,rowsep=1pt}
\toprule
Question & Generalized symmetry & QRF/relational observables \\
\midrule
Fundamental Wilson loop in pure Yang--Mills & Diagnosed by electric center one-form symmetry $\mathbb{Z}_N^{(1)}$ and $N$-ality;
in the standard gapped confining setting, an area law is the standard diagnostic of unbroken $\mathbb{Z}_N^{(1)}$. & Read as loss of relational color-frame information across the Wilson surface spanning the Wilson loop. \\
\midrule
Adjoint gluon confinement & Not directly implied by center one-form symmetry because adjoint fields have zero $N$-ality.
& Dynamical adjoint confinement additionally requires the absence of a finite-energy, physically accessible, asymptotically isolated adjoint-colored representation;
singlet glueball channels remain.
Frame averaging over an inaccessible orientation removes an open adjoint frame index only kinematically.
\\
\midrule
Boundary charges & Organizes line endability, background fields, and anomaly inflow. & Requires a boundary frame to define Lie-algebra valued color components;
Gauss law leaves only boundary relational charge. \\
\midrule
Fundamental matter & Center one-form symmetry is explicitly broken and Wilson lines can end. & Still distinguishes whether relational color can be read with respect to a physical Higgs or matter frame. \\
\midrule
Topological defects & Captures center flux and higher-form charges. & Interprets the same defect data as possible obstructions to globalizing a specified frame or lift;
confinement additionally requires the relevant infrared dynamics. \\
\bottomrule
\end{tblr}
\normalsize
\captionof{table}{Comparison between generalized symmetry and QRF.}
\label{tab:gen-symm-qrf}
\end{center}

\end{widetext}

\section{Conclusion and discussion}
\label{sec:conclusion}

The purpose of this paper has been to formulate \textbf{color confinement as a question about the existence of physically admissible charged sectors defined by quantum reference frames (QRFs) and relational observables}.
	The starting point is that \textbf{a color component is not an absolute attribute of a local field.}
A quark field, a gluon field, an electric flux, or a boundary charge acquires a color component only after a color frame, a dressing prescription, or boundary reference data have been specified.  Thus, rather than assigning a direct physical meaning to a bare field such as the quark field $\psi(x)$ or the gluon field $\mathscr{A}_{\mu}(x)$, one considers relational quark field $\psi_h(x)$ and the relational gluon field $\mathscr{A}^{h}_{\mu}(x)$
\begin{equation}
\begin{aligned}
 \psi_h(x)&:=h(x)^{-1}\psi(x),\\
 \mathscr{A}^{h}_{\mu}(x)
 &:=h(x)^{-1}\mathscr{A}_{\mu}(x)h(x)+\frac{i}{g_{\mathrm{YM}}} h(x)^{-1}\partial_{\mu}h(x),
\end{aligned}
\label{dressed-field}
\end{equation}
and also the relational color electric field
\begin{equation}
\mathscr{E}^{j}_{h}(x)
 :=h(x)^{-1}\mathscr{E}^{j}(x)h(x)
\label{dressed-field2}
\end{equation}
where the frame field $h(x)$ transforms under the original left action of the proper gauge group as
\begin{equation}
 h(x)\longmapsto g(x)h(x), \qquad g(x)\in G.
\end{equation}
The dressed quantities in (\ref{dressed-field}) and (\ref{dressed-field2}) are invariant under that left action.

This gauge invariance, however, does not by itself make their non-singlet components independent of the chosen frame.
A central distinction established in this paper is therefore \textbf{the distinction between gauge invariance and frame independence}.  Under an arbitrary smooth right-frame reorientation
\begin{equation}
 h(x)\longmapsto h(x)k(x), \qquad k(x)\in G,
\end{equation}
the relational electric field $\mathscr{E}^{j}_{h}(x)$ transforms as
\begin{equation}
 \mathscr{E}^{j}_{h}(x) \longmapsto \mathscr{E}^{j}_{hk}(x)
 =k(x)^{-1}\mathscr{E}^{j}_{h}(x)k(x).
\end{equation}
If $\xi(x)$ denotes the components, in the frame $h$, of the same physical adjoint test section, then in the frame $hk$ its components are
\begin{equation}
 \xi(x)\longmapsto \xi'(x)=k(x)^{-1}\xi(x)k(x).
\end{equation}
Consequently the fully paired relational boundary charge
\begin{equation}
 Q_{\partial V}[h,\xi]
 :=\int_{\partial V}dS_j\,
 2\,\operatorname{tr}\!\left\{\xi(x)\mathscr{E}^{j}_{h}(x)\right\}
\end{equation}
satisfies
\begin{equation}
 Q_{\partial V}[hk,k^{-1}\xi k]
 =Q_{\partial V}[h,\xi].
\end{equation}
This identity is valid for arbitrary position-dependent $k(x)$ and expresses \textbf{passive frame-coordinate invariance}: \textbf{the physical test direction is held fixed while its components and the frame basis are transformed together.}
By contrast, if a numerical generator $T^A$ is held fixed while $h$ is actively replaced by $hk$, the measured physical color direction changes and the component
\begin{equation}
 Q^A_{\partial V,h}
 :=\int_{\partial V}dS_j\,
 2\,\operatorname{tr}\!\left\{T^A h(x)^{-1}\mathscr{E}^{j}(x)h(x)\right\}
\end{equation}
generally changes.
Thus \textbf{a fixed color component is relational}, whereas \textbf{the fully paired functional, conjugacy-class data, Casimir invariants, and the representation or stable screening-class label are frame-independent data}.
This distinction prevents frame averaging or twirling from being confused with a dynamical confinement mechanism.

The main result is \textbf{an operational organization of the confinement question}, not a new local order parameter.
The standard charged-sector problem asks whether a nontrivial stable screening class survives in the infinite-volume theory as a finite-energy, physically accessible, asymptotically isolated representation.

\textbf{The QRF-specific input is to specify, before answering that question, which relational frames or dressings, boundary/asymptotic readouts, right-frame orientations, and preparation resources make the non-Abelian charge label physically meaningful.}
Gauss-law kinematics, failure of a particular frame to globalize, infrared disorder, and absence of a charged sector are therefore distinct logical statements, summarized in the status map in the Introduction.

\noindent
$\bullet$ The exactly solvable $(1+1)$-dimensional $SU(2)$ Yang--Mills model gives the sharpest test of this structure.
On an open-interval of the spatial length $L$ the exact excitation energy is obtained  for every nontrivial boundary-flux in the representation $R$ as
\begin{equation}
 \Delta E_R(L)=\frac{g_{\rm YM}^2}{2}C_2(R)L .
\end{equation}
This means \textbf{color confinement in the sense that
a finite-energy isolated color charge with any non-trivial representation $R:C_2(R)\not=0$ is excluded in the infinite volume limit $L\to \infty$},

On a circle, after the global Gauss constraint is imposed, the gauge-invariant configuration variable reduces to the conjugacy class of the holonomy; this gives a complementary illustration of why {bare local color labels are not physical observables.}

\noindent
$\bullet$ The $(1+1)$-dimensional compact $U(1)$ gauge--Higgs models provide a second controlled laboratory.  In a charge-$N$ Higgs model, the local Higgs phase supplies a natural phase reference where it is nonzero and accessible, while vortices produce nontrivial monodromy for probes carrying nonzero residual $\mathbb{Z}_N$ electric charge.  Under the explicitly stated dilute Poisson-gas assumption, summing over vortices and antivortices disorders the phase comparison and yields an area law for the corresponding nontrivial residual charge classes.  This is a controlled result within the assumed statistical model; it is not a theorem that the multivaluedness of one Higgs-frame representative alone eliminates every charged sector.

\noindent
$\bullet$ The noncompact $(3+1)$-dimensional $U(1)$ Coulomb phase gives an example of deconfinement: nonzero asymptotic electric flux survives with finite renormalized infrared energy even though the charged sector is infraparticle-like rather than a sharp one-particle pole.

\noindent
$\bullet$ The symmetric-instanton reductions provide geometrically controlled bridges between four-dimensional Euclidean self-dual Yang--Mills sectors and lower-dimensional defect systems.
The $SO(3)$-symmetric reduction yields the two-dimensional $\mathbb H^2$ hyperbolic vortex system and its center-sensitive frame monodromy,
while the $SO(2)$-symmetric reduction yields the three-dimensional $\mathbb H^3$ hyperbolic magnetic monopole and the obstruction to lifting a partial $SU(2)/U(1)$ direction to a complete frame.
These statements are exact within the reduced sectors.  However, defect dominance, proliferation, or a four-dimensional area law require additional dynamics.

\noindent
$\bullet$ For four-dimensional pure Yang--Mills theory, the CDG color direction is treated only as an intrinsic partial frame selected from the candidate set $\mathcal N[\mathscr A]$ with a specified copy prescription.
Restricted-field dominance and defect disorder do not by themselves establish QRF confinement.
The proposed mechanism has instead been given quantitative targets: the relational partial-frame coherence correlator \eqref{eq:qrf-frame-coherence-test} and the charged-sector energy \eqref{eq:qrf-sector-energy-lattice-test}, to be compared with full/restricted Wilson loops and defect observables.
Their joint behavior can confirm or falsify the specifically QRF part of the conjecture.

\noindent
$\bullet$ The relation to Wilson loops, generalized one-form symmetry, and Kugo--Ojima is therefore one of comparison rather than identification.
Center symmetry classifies stable external $N$-ality classes; the QRF language does not create a new stable adjoint charge after gluon screening.
The Kugo--Ojima construction is treated as a gauge-fixed route under its BRST assumptions, not as a universal gauge-independent criterion.

The remaining open problem is dynamical: to determine whether the four-dimensional Yang--Mills measure eliminates all admissible finite-energy nontrivial charged sectors and, if so, how that elimination is correlated with the disorder of the relational frame data described here.

\section*{Acknowledgments}
This work was supported by Grant-in-Aid for Scientific Research, JSPS KAKENHI Grant
Number (C) No.23K03406.
During preparation of this manuscript, the author used ChatGPT (OpenAI; GPT-5.6 Sol and GPT-5.5) for language editing, restructuring selected passages, and preliminary consistency checks of derivations and bibliography.
The author independently rederived or verified all scientific suggestions against the relevant calculations and cited sources, reviewed all AI-assisted text, and takes full responsibility for the accuracy, citations, scientific claims, and conclusions of the manuscript.

\section*{Data Availability}
No data were created or analyzed in this study.

\appendix
\section{Right-Frame Transformations: concrete examples and the physical meaning}\label{Appendix:right-frame}

A right-frame transformation generally takes the form
\begin{equation}
 h(x)\longmapsto h(x)k(x).
 \label{eq:rightchange}
\end{equation}
However, the meaning of $k$ depends on the model.
A right-frame transformation $h\mapsto hk$ does not always represent the same type of physical symmetry.
See the Table~\ref{tab:Right-Frame}.
In this article, these possibilities are distinguished by their transformation laws and by concrete models.

In particular, a right action can exist even in pure Yang--Mills theory.  In an extended-phase-space realization of a bounded region in which boundary-frame data are retained, it is realized as the local surface reorientation group $\operatorname{Map}(\partial V,G)$, generated by the boundary electric flux.  Whether this right action is retained as a genuine physical boundary symmetry depends on the chosen boundary conditions and operational protocol; it is not an automatic consequence of the mere existence of a boundary.

A right index is essentially an ``index on the reference-frame side used to read color,'' and only when the frame is constructed from a matter field can that index coincide with a flavor or custodial index.

\begin{center}
\begin{widetext}
\footnotesize\begin{tblr}{colspec={Q[l,t,wd=0.19\textwidth] Q[l,t,wd=0.24\textwidth] Q[l,t,wd=0.52\textwidth]},colsep=2pt,rowsep=1pt}
\toprule
Realization & Meaning of $k$ (right factor) & Physical character or meaning\\
\midrule
i) Base point of a Wilson line & $g(x_0)$ & Residual action of the same gauge group at the base point or boundary\\
ii) Boundary edge mode & $k(s)$, $s\in\partial V$ & Local physical surface symmetry generated by electric flux\\
iii) Two dressings & $k=h^{-1}h'$ & Change of reference-frame choice, not necessarily a dynamical Hamiltonian symmetry\\
iv) Color--flavor frame & global flavor/custodial rotation & May be an independent global internal symmetry\\
v) Standard Model Higgs frame & local $U(1)_Y\subset SU(2)_R$ & Local right gauge action on the Higgs matrix after using the Higgs field as an $SU(2)_L$ dressing\\
vi) Left--right symmetric model & local $SU(2)_R$ & Independent second local gauge group acting on the right of a bifundamental frame\\
\bottomrule
\end{tblr}
\normalsize
\captionof{table}{Right-frame transformations: concrete examples and the physical meaning.}
\label{tab:Right-Frame}

\end{widetext}
\end{center}

\subsection*{Concrete examples and the physical meaning}
\subsection*{Example i) (Pure Yang--Mills): a Wilson line and a global right action at the base point}

Let the Wilson-line dressing/frame transporter from a base point $x_0$ to $x$ be
\begin{equation}
 h_\gamma(x)
 =W(x,x_0;\gamma)
 =\mathcal{P}\exp\left[
 ig_{\rm YM}\int_{x_0}^{x}\mathscr{A}
 \right].
\end{equation}
Under a gauge transformation, the Wilson line transforms as
\begin{equation}
 W(x,x_0;\gamma)
 \longmapsto
 g(x)W(x,x_0;\gamma)g(x_0)^{-1}.
 \label{eq:wilsontransform}
\end{equation}
Therefore, comparing the Wilson-line transformation with the general frame law~\eqref{h-transform}, the right factor is given by
\begin{equation}
 r_g=g(x_0).
\end{equation}

Although $x$ varies, there is only one base point $x_0$, so $g(x_0)$ is a global right action with respect to the bulk coordinate $x$.
If proper gauge transformations are restricted by $g(x_0)=\bm{1}$, this right action disappears. On the other hand, if $x_0$ is regarded as a physical boundary or external reference system and $g(x_0)$ is allowed, it remains as a color rotation at the base point or boundary.

This right group is not a new flavor group. It is simply the value, at the base point, of the same Yang--Mills gauge group $G$, appearing on the right.
Thus, a right action exists even in pure Yang--Mills theory, but its first example is not an ``independent internal symmetry''; it is the action of the same group remaining at the root.

\subsection*{Example ii) (Pure Yang--Mills): boundary edge modes and a local right surface symmetry}

When a spatial region $V$ has a boundary $S=\partial V$, one may introduce in the Donnelly--Freidel extended phase space a group-valued boundary variable $\varphi:S\to G$ \cite{Donnelly:2016auv}.  In the convention used in this example the gauge group acts on the right, $\varphi\mapsto\varphi g$, whereas the independent surface symmetry acts on the left, $\varphi\mapsto k^{-1}\varphi$.  Consequently, if $h_{\rm DF}:=\varphi^{-1}$, then
\begin{align}
 \text{gauge:}\quad &h_{\rm DF}(s)\longmapsto g(s)^{-1}h_{\rm DF}(s),\\
 \text{surface:}\quad &h_{\rm DF}(s)\longmapsto h_{\rm DF}(s)k(s).
 \label{eq:localrightsurface}
\end{align}
This is the inverse gauge-action convention relative to the main-text choice $h[\mathscr A^g]=g\,h[\mathscr A]$ and is kept explicitly distinct from it.  If one merely introduces $\widetilde g:=g^{-1}$, the frame law becomes $h_{\rm DF}\mapsto\widetilde g h_{\rm DF}$, but the same bulk transformation is then $\mathscr A\mapsto\mathscr A^{\widetilde g^{-1}}$, not $\mathscr A^{\widetilde g}$.  Thus the two notations are related by a consistent change of gauge-action convention, not by a new physical symmetry.  The independent right/surface reorientation $h_{\rm DF}\mapsto h_{\rm DF}k$ is the structure relevant for the comparison made here.
Therefore,
\begin{equation}
 \boxed{
 k:S\to G,
 \qquad
 h(s)\mapsto h(s)k(s)
 }
\end{equation}
is an explicit example of a \textbf{local right-frame symmetry} in pure Yang--Mills theory.
Its group is
\begin{equation}
 G_S=\operatorname{Map}(S,G),
\end{equation}
and it is local on the boundary rather than throughout the bulk. Its generator is the normal electric field expressed in the frame $h$, namely the non-Abelian surface charge.

This example is important.
\textbf{A right action can appear even in the absence of flavor and can be a physical boundary symmetry with a Hamiltonian generator, rather than a merely formal freedom.}

\subsection*{Example iii) (Pure Yang--Mills): a local transformation between two dressings}

Construct two dressings $h[\mathscr{A}]$ and $h'[\mathscr{A}]$ from the same gauge field $\mathscr{A}$, and suppose that under proper gauge transformations they transform identically as
\begin{equation}
 h[\mathscr{A}^g]=g h[\mathscr{A}],
 \qquad
 h'[\mathscr{A}^g]=g h'[\mathscr{A}].
\end{equation}
Then
\begin{equation}
 k[\mathscr{A}](x):=h[\mathscr{A}](x)^{-1}h'[\mathscr{A}](x)
\end{equation}
is gauge invariant and satisfies
$
 h'=hk
$.
Since $k(x)$ is generally position dependent, it defines a local right-frame transformation.

However, this $k(x)$ is not usually a physical Hamiltonian symmetry generated by a conserved charge. For example, it may connect
different Wilson-line paths,
Coulomb and axial dressings,
different boundary roots,
different field-space connections (horizontal representation),
and is closer to a redefinition of relational coordinates. If every such $k(x)$ were twirled as an independent local symmetry, physically distinct dressings or boundary electric fluxes could be removed excessively.
Therefore, among right-frame transformations one must distinguish at least
 \text{a frame-change groupoid}
 \text{a Hamiltonian symmetry group}.

\subsection*{Example iv)  (A color--flavor model): an example in which the right index is genuinely flavor}

Consider a matrix scalar $\Phi^a{}_i(x)$ with $N$ colors and $N$ flavors.  If the constraint $\Phi^\dagger\Phi=v^2\bm{1}$ is imposed, then $U:=\Phi/v$ is unitary, $U\in U(N)$.  To use it as an $SU(N)$ color frame one must in addition fix the determinant, for example by imposing $\det\Phi=v^N$ (or an equivalent fixed-phase condition).  Under this additional condition one may set
\[
 h(x):=\frac{1}{v}\Phi(x)\in SU(N).
\]
If the determinant phase is not fixed, the natural frame is instead $U(N)$-valued.

If a local color transformation $g_c(x)$ and a global flavor transformation $R_f$ act as
\begin{equation}
 \Phi(x)\longmapsto g_c(x)\Phi(x)R_f^{-1},
\end{equation}
then
\begin{equation}
 h(x)\longmapsto g_c(x)h(x)R_f^{-1}.
\end{equation}

In this model, the left index $a$ is color and the right index $i$ is flavor; the right-frame group is literally the flavor group.
If color--flavor locking occurs, the relative orientation between the left color frame and the right flavor frame becomes a physical order parameter.
This is, however, only one realization of a right action.
This realization is model dependent: the right index of a Wilson line in pure Yang--Mills theory does not form a flavor multiplet.
Therefore, one should understand that

\begin{center}
 a right-frame index may be flavor-like, but it is not always a flavor index.
\end{center}

\subsection*{Example v) (The Standard Model): local right $U(1)_Y$ and global $SU(2)_R$}

\subsubsection*{The Higgs matrix}

For the Standard Model Higgs doublet $H$, define $\widetilde H$ by
\begin{equation}
 H=\begin{pmatrix}H^+\\H^0\end{pmatrix},
 \qquad
 \widetilde H=i\sigma_2H^*,
\end{equation}
and form the $2\times2$ matrix
\begin{equation}
 \Sigma=(\widetilde H,\,H).
\end{equation}
Under a local $SU(2)_L\times U(1)_Y$ gauge transformation,
\begin{equation}
 \Sigma(x)\longmapsto
 g_L(x)\Sigma(x)
 \exp\left[-\frac{i}{2}\beta(x)\sigma_3\right].
 \label{eq:SMmatrix}
\end{equation}
Thus, one finds
 \text{a local }$SU(2)_L$\text{ acts on the left},
 \text{and a local }$U(1)_Y$\text{ acts on the right}.

In the region where $\rho=\sqrt{H^\dagger H}\neq0$, one may write
\begin{equation}
 U:=\frac{\Sigma}{\rho}\in SU(2),
\end{equation}
which transforms as
\begin{equation}
 U\longmapsto
 g_L U q_Y,
 \qquad
 q_Y(x)=\exp\left[-\frac{i}{2}\beta(x)\sigma_3\right].
 \label{eq:Utransform}
\end{equation}
Using $U$ as an $SU(2)_L$ dressing, the $SU(2)_L$-invariant composite connection
\begin{equation}
 \mathscr{W}_\mu^U
 =U^\dagger\mathscr{W}_\mu U
 +\frac{i}{g_{\rm YM}}U^\dagger\partial_\mu U
\end{equation}
transforms under the surviving right hypercharge action as

\begin{equation}
 \mathscr{W}_\mu^U
 \longmapsto
 q_Y^\dagger\mathscr{W}_\mu^Uq_Y
 +\frac{i}{g_{\rm YM}}q_Y^\dagger\partial_\mu q_Y.
 \label{eq:residualhyper}
\end{equation}

This is a concrete example of the residual $U(1)$ gauge symmetry that remains after $SU(2)$ dressing in the Dressing Field Method.
Thus, the Standard Model contains a local right transformation:
A local hypercharge $U(1)_Y$ acts on the right of the Higgs frame.

\subsubsection*{Difference from custodial symmetry}

When the Higgs potential is considered by itself, it has the familiar global
$SU(2)_L\times SU(2)_R$ form,
\begin{equation}
 \Sigma\longmapsto L\Sigma R^\dagger,
 \qquad
 L\in SU(2)_L,\quad R\in SU(2)_R,
 \label{eq:globalLR}
\end{equation}
and the vacuum
\begin{equation}
 \langle\Sigma\rangle=\frac{v}{\sqrt2}\bm{1}
\end{equation}
preserves the diagonal custodial subgroup
\begin{equation}
 SU(2)_V=\{(L,R)\mid L=R\}.
\end{equation}
In the full Standard Model, however, $SU(2)_L$ is gauged, $U(1)_Y$ gauges
only the $T_R^3$ subgroup of $SU(2)_R$, and hypercharge together with
nondegenerate Yukawa couplings explicitly breaks the global $SU(2)_R$.
Consequently custodial $SU(2)_V$ is only an approximate symmetry of the
full Standard Model.  The local right $U(1)_Y$ gauge action should therefore
be distinguished from the approximate global custodial symmetry.

\subsection*{Classification of right actions}

The preceding discussion shows that it is useful to classify right-frame transformations into the following four types.
\begin{enumerate}
\item[A.]
Root/boundary residual action of the same gauge group:
A typical example is $g(x_0)$ of a Wilson line. It is global in the bulk, but becomes local on the boundary if a root is assigned to each boundary point.

\item[B.]
Physical surface symmetry:
This is the transformation $h(s)\mapsto h(s)k(s)$ of an edge mode $h(s)$. It commutes with bulk gauge transformations and is generated by non-Abelian electric flux.

\item[C.]
Dressing ambiguity/frame change:
This connects two relational descriptions through $h'=hk$. Although $k(x)$ may be local, it need not be a symmetry of the action or a transformation carrying a conserved charge.

\item[D.]
Independent internal symmetry:
A flavor symmetry, custodial symmetry, or a second gauge group may act on the right. When a matter frame exists, the right-frame index can coincide with a flavor index.
\end{enumerate}

\subsection*{Relation to twirling}
Which right action should be twirled depends on this classification.

\begin{itemize}
\item
Proper gauge redundancy is removed by the Gauss law and gauge quotient.
\item
An accessible boundary surface symmetry carries physical charge and therefore must not be twirled unconditionally.
\item
Only when an external frame is unavailable should its right-frame orientation be twirled.
\item
A dynamical sum over vortex, monopole, or other defect sectors is not Haar twirling over a right-frame group, but a path-integral or ensemble sum weighted by the action.

\item
An independent custodial or flavor charge is not discarded merely
      because it acts on the right of a matter frame.

\item
Twirling over the full dressing ambiguity requires showing that the dressings are genuinely operationally indistinguishable.
\end{itemize}
In particular, the \textbf{boundary surface group} $\operatorname{Map}(S,G)$ of pure Yang--Mills theory generates physical boundary charges. In a protocol in which the boundary frame is measurable, twirling over this group would artificially erase charged boundary sectors.

\subsection*{Relation to the main-text classification}
The distinctions used in the main text are developed across
Secs.~II.D--II.F: Secs.~II.D and II.E distinguish gauge invariance,
passive frame changes, and active frame reorientations, while Sec.~II.F
classifies representative QRF-related structures.  For the concrete
examples collected here, it is useful to keep distinct notation for the
corresponding right actions rather than denoting all of them by one universal
group symbol. We therefore distinguish
\begin{equation}
 G_{\rm root}, \
 G_{\partial} =\operatorname{Map}(\partial V,G), \
 K_{\rm dress}, \
 G_{\rm cust/flavor}.
\end{equation}

A right action on a dressing field need not represent a single universal symmetry. It may be the residual action of the original gauge group at a root or boundary, a physical surface symmetry acting on edge-frame data, a change between admissible dressings, or an independent custodial/flavor symmetry. Only the subgroup whose orientation is operationally inaccessible in the chosen protocol is to be twirled. In pure Yang--Mills theory with a boundary, the edge-frame field carries a local right action of $\operatorname{Map}(\partial V,G)$ generated by the non-Abelian electric flux; this physical surface symmetry must be distinguished from proper gauge redundancy and from a mere change of dressing.

In the Standard Model, a local $U(1)_Y$ acts on the right of the Higgs matrix. The full $SU(2)_R$ is an approximate global symmetry and is the parent symmetry of the custodial $SU(2)_V$. A fully local $SU(2)_R$ is not present in the Standard Model itself; it is gauged in extensions such as left--right symmetric models.
A right action also exists in pure Yang--Mills theory. 
At the base point of a Wilson line, it appears as a global root action, while the introduction of a boundary edge mode gives the local and physical right surface symmetry
\begin{equation}
 h(s)\mapsto h(s)k(s),
 \qquad
 k\in\operatorname{Map}(\partial V,G).
\end{equation}
Thus, a right action does not require the presence of matter flavor.

For a bounded region $V$, an extended phase-space description may contain a
boundary frame.  In the convention in which its inverse is denoted by $h(s)$,
$s\in\partial V$, the physical surface action takes the form
\begin{equation}
 h(s)\longmapsto h(s)k(s),
 \qquad
 k\in\operatorname{Map}(\partial V,G).
\end{equation}
This action is distinct from the proper bulk gauge redundancy.  It is a
boundary symmetry generated by the normal electric flux and can therefore
carry a physical boundary charge \cite{Donnelly:2016auv,Riello2021}.  In the
Donnelly--Freidel convention one may instead introduce a boundary variable
$\varphi$ with $h=\varphi^{-1}$; the gauge action and the surface action then
act from opposite sides and commute.  Consequently, when the boundary frame
is part of the accessible physical data, averaging over this surface group
would erase genuine charged boundary sectors and is not appropriate.

\section{Logical map of the QRF confinement statements}

The logical hierarchy used in the paper is
\begin{align}
&\text{QRF/gauge kinematics}\nonumber\\
&\quad\longrightarrow\quad
\text{admissible relational charge data},\nonumber\\
&\qquad\quad\quad\text{frame globalization or defect obstruction}\nonumber\\
&\qquad\quad\quad\text{(possible mechanism)},\nonumber\\
&\text{infrared dynamics}\nonumber\\
&\quad\longrightarrow\quad
\text{sector-energy and line-operator diagnostics},\nonumber\\
&\text{charged-sector survival test}\nonumber\\
&\quad\longrightarrow\quad
\text{QRF confinement or deconfinement}.
\nonumber
\end{align}
Here the arrows
 are not equivalences.

\begin{center}
\begin{widetext}
\footnotesize
\begin{tblr}{colspec={Q[l,t,wd=0.22\textwidth] Q[l,t,wd=0.43\textwidth] Q[l,t,wd=0.30\textwidth]},colsep=2pt,rowsep=1pt}
\toprule
Layer & Statement & Main-text location \\
\midrule
General QRF kinematics & Proper gauge redundancy, relational fields, boundary charge, frame changes, and QRF classification & Sec.~II \\
Geometrical/global issue & Globalization of a chosen frame, lift, or transporter may fail because of topology or monodromy & Sec.~V and later reduced models \\
Infrared dynamics & Defect disorder, Wilson-loop behavior, screening, and the growth or boundedness of sector energies are model dependent & Secs.~IV, VI, VIII--XI \\
QRF confinement criterion & A nontrivial stable screening class fails to survive the prescribed infrared limit as a finite-energy, physically accessible, asymptotically isolated charged representation & Sec.~III.B \\
\bottomrule
\end{tblr}
\normalsize
\captionof{table}{Logical map of the QRF confinement statements.}
\label{tab:map-QRF-conf}
\end{widetext}
\end{center}

In particular,
\begin{equation}
\boxed{
\begin{aligned}
&\text{Gauss law alone} \not\Rightarrow \text{confinement},\\
&\text{failure of one frame globalization} \not\Rightarrow \text{confinement},\\
&\text{nontrivial holonomy} \not\Rightarrow \text{absence of a global frame},\\
&\text{absence of one particle pole} \not\Rightarrow \text{absence of every}\\
&\qquad\qquad\qquad\qquad\qquad\qquad\quad~~\text{charged sector}.
\end{aligned}}
\nonumber
\end{equation}

The Kugo--Ojima criterion, Wilson-loop area law, center one-form symmetry, and other diagnostics are compared in Secs.~III.C and III.D. They are not re-tabulated here because the purpose of the appendix is only to prevent the logical levels above from being conflated.

\section{Differentiable Gauss Generators, Boundary Charges, and Surface Terms}
\label{app:differentiable_generator}
Except where the comparison with Sec.~IV is stated explicitly, this appendix
uses the physical-connection convention of Sec.~II,
$\mathscr D_j=\partial_j-i g_{\rm YM}[\mathscr A_j,\cdot]$.
Section~IV instead uses
$\mathscr A_{\rm geom}=g_{\rm YM}\mathscr A_{\rm phys}$, so that the same
formulas are obtained there with $g_{\rm YM}$ absorbed into the geometric
connection.

This section collects, reorganizes, and removes overlaps among the explanations of the differentiable improved Gauss generator, the origin of the boundary charge, the Abelian analogue, and the relevant references.
The main point is that \textbf{the Gauss-law functional on a space with boundary is not automatically a well-defined Hamiltonian generator when gauge parameters are allowed to be nonzero at the boundary.}
\textbf{It must be supplemented by a boundary term, and this boundary term is precisely what becomes the boundary electric charge on the constraint surface.}

Sections IV and VI use the improved Gauss generator on an interval.  The
interval dynamics and the resulting sector energies are already derived in
the main text.  Here we retain only the phase-space argument that explains
why the surface term is required and how the boundary charge arises.

\subsection*{Functional differentiability in a bounded region}

Let $V$ be a spatial region with boundary $\partial V$.
The canonical Yang--Mills variables are $\mathscr{A}_j^A(x)$ and $\mathscr{E}^{jA}(x)$.
In the presence of a matter color charge density $\rho^A$, the Gauss constraint is given by
\begin{equation}
 G^A(x):=(\mathscr{D}_j\mathscr{E}^j(x))^A-\rho^A(x)\simeq0.
\end{equation}
For a fixed smearing function $\omega^A(x)$, consider the naive Gauss functional given by

\begin{equation}
\begin{aligned}
 G_0[\omega]
 &:=\int_V d^{D-1}x\,\omega^A(x)G^A(x)\\
 &=\int_V d^{D-1}x\,\omega^A(x)
 \left((\mathscr{D}_j\mathscr{E}^j(x))^A-\rho^A(x)\right).
\end{aligned}
\end{equation}
The naive functional $G_0[\omega]$ is not differentiable if \(\omega\) is allowed to be nonzero at the boundary \(\partial V\).
The word ``differentiable'' in the phrase ``differentiable improved generator'' does not mean differentiability with respect to the coordinate \(x\), nor differentiability with respect to the smearing function \(\omega^A(x)\).  It means functional differentiability with respect to the canonical phase-space variables
$
\mathscr{A}_j^A(x),\ \mathscr{E}^{jA}(x),
$
and the matter variables.

A functional \(F[\mathscr{A},\mathscr{E}]\) is called  \textbf{differentiable} on the chosen phase space if, for every allowed variation \((\delta\mathscr{A},\delta \mathscr{E})\), its first variation can be written as

\begin{align}
  \delta F
  &=\int_V d^{D-1}x\,\left(
  \frac{\delta F}{\delta \mathscr{A}_j^A(x)}
  \delta\mathscr{A}_j^A(x)
  +\frac{\delta F}{\delta \mathscr{E}^{jA}(x)}
  \delta \mathscr{E}^{jA}(x)\right)
  \nonumber\\
  &+\text{matter part} .
  \label{eq:functional-differential}
\end{align}

Unprocessed boundary variations, such as
\begin{equation}
  \left.C_A\delta \mathscr{E}^{jA}\right|_{\partial V},
  \quad
  \left.C_A\delta\mathscr{A}_j^A\right|_{\partial V},
\end{equation}
must not remain on the right-hand side of $\delta F$ because such terms cannot be represented as bulk functional derivatives.  Thus differentiability means that the functional has a well-defined gradient on the phase space with the allowed boundary conditions.

Integrating by parts gives

\begin{align}
G_0[\omega]
=&-\int_V d^{D-1}x\,
 (\mathscr{D}_j\omega)^A\mathscr{E}^{jA}
-\int_V d^{D-1}x\,\omega^A(x)\rho^A(x)
 \nonumber\\
&+\int_V d^{D-1}x\,\partial_j(\omega^A\mathscr{E}^{jA})
\nonumber\\
=&-\int_V d^{D-1}x\,
 (\mathscr{D}_j\omega)^A\mathscr{E}^{jA}
-\int_V d^{D-1}x\,\omega^A(x)\rho^A(x)
 \nonumber\\
&+\int_{\partial V}d^{D-2}S_j\,\omega^A\mathscr{E}^{jA}.
\end{align}
The last term is the boundary term:
\begin{equation}
 \int_{\partial V}d^{D-2}S_j\,
 \omega^A \mathscr{E}^{jA}.
  \label{eq:boundary-term-identity1}
\end{equation}
Equation \eqref{eq:boundary-term-identity1} is only an identity obtained by partial integration.  It is not yet the construction of the correct canonical generator.  The latter requires functional differentiability on the chosen phase space.

Let us see where the naive Gauss generator fails.  Since the smearing function is fixed, \(\delta\omega^A=0\).  The part of the variation of \(G_0[\omega]\) with respect to \(\mathscr{E}\) that contains a derivative is
\begin{equation}
  \delta_\mathscr{E} G_0[\omega]
  =\int_V d^{D-1}x\, \omega^A (\mathscr{D}_j\delta \mathscr{E}^{j})^A .
\end{equation}
For every allowed variation the first variation
must be expressible in terms of bulk functional derivatives, with no unprocessed boundary variation left over.
The part of $\delta G_0$ that contains a derivative of $\delta \mathscr{E}^{jA}$ is
\begin{equation}
 \delta_\mathscr{E} G_0[\omega]
 =
 -\int_V d^{D-1}x\,
 (\mathscr{D}_j\omega)^A\delta \mathscr{E}^{jA}
 +
 \int_{\partial V}dS_j\,
 \omega^A\delta \mathscr{E}^{jA}.
\end{equation}
If $\omega$ is allowed to be nonzero on the boundary and the normal electric flux is allowed to vary, the last term is an obstruction to functional differentiability.
The first term gives a legitimate bulk functional derivative.
Whereas the second term is the obstruction.
If $\omega$ is zero at the boundary $\omega|_{\partial V}=0 $, this term vanishes.
In that case \(G_0[\omega]\) is differentiable and generates a pure gauge redundancy.
If, however, \(\omega\) is allowed to be nonzero at the boundary $\omega|_{\partial V}\not=0$
and the boundary electric fluxes \(\mathscr{E}|_{\partial V}\) are kept as variable boundary data,
then the second term is not zero in general.  Hence \(G_0[\omega]\) is not differentiable on this phase space.
The required surface term is therefore
\begin{equation}
 B[\omega]
 :=
 -\int_{\partial V}dS_j\,\omega^A\mathscr{E}^{jA},
\end{equation}
so that

\begin{equation}
\boxed{\begin{aligned}
 \widetilde G[\omega]
 &:=G_0[\omega]+B[\omega]\\
 &=\int_V d^{D-1}x\,
 \omega^A\left((\mathscr{D}_j\mathscr{E}^j)^A-\rho^A\right)
 -\int_{\partial V}dS_j\,\omega^A\mathscr{E}^{jA},
\end{aligned}}
\end{equation}
has no uncanceled boundary variation.  This is the standard Regge--Teitelboim \textbf{surface-term improvement of the Hamiltonian gauge generator for the chosen boundary phase space} \cite{ReggeTeitelboim1974,GervaisSakitaWadia1976,Wadia1977,LeeWald1990,Troessaert2013,
HenneauxTroessaert2018,Riello2021}.

\subsection*{Why differentiability is the condition for a canonical generator}

In the canonical formalism, the \textbf{symplectic 2-form} is given by
\begin{equation}
 \Omega
 =
 \int_V d^{D-1}x\,
 \delta \mathscr{E}^{jA}\wedge\delta\mathscr{A}_j^A
 +\text{matter part}.
\end{equation}
A functional $F$ generates a canonical transformation only if its \textbf{Hamiltonian vector field} $X_F$ exists and satisfies
\begin{equation}
 \iota_{X_F}\Omega=\delta F.
 \label{eq:hamiltonian-vector-field}
\end{equation}

If unprocessed boundary variations remain in \(\delta F\), the right-hand side of \eqref{eq:hamiltonian-vector-field} cannot be represented as a \textbf{contraction} of the bulk symplectic form.  In that case \(X_F\) does not exist, or it is not an allowed tangent vector field on phase space.  Therefore \(F\) is not a well-defined \textbf{generator of canonical transformations}.

For \(\widetilde G[\omega]\), the problematic boundary variation has been canceled.  With the Poisson-bracket convention
\begin{equation}
 \{\mathscr{A}_j^A(x),\mathscr{E}^{kB}(y)\}
 =\delta_j{}^k\delta^{AB}\delta(x-y),
\end{equation}
one obtains, up to sign conventions,
\begin{align}
 \{ \widetilde G[\omega], \mathscr{A}_j^A \}
 &= \left(\mathscr{D}_j\omega\right)^A,
 \\
 \{ \widetilde G[\omega], \mathscr{E}^{jA}\}
 &=f^{ABC}\omega^B\mathscr{E}^{jC}.
\label{eq:gauge-on-E}
\end{align}
Thus \textbf{the improved functional generates the gauge transformation canonically.}
This is the sense in which functional differentiability is directly connected with the ability to generate transformations by Poisson brackets.
Thus the surface improvement is not an arbitrary addition to the Gauss constraint;
it is what makes the \textbf{gauge transformation Hamiltonian} on the specified phase space.

An uncanceled surface term in $\delta F$ cannot be represented by contraction with the bulk symplectic form on the chosen phase space.  Hence the naive
$G_0[\omega]$ is not a well-defined generator when its boundary variation is
allowed to fluctuate.
After the improvement, the Hamiltonian vector field exists.
Then the transformation of an observable \(O\) is given by
\begin{equation}
  \delta_FO=\{O,F\} .
\end{equation}

\subsection*{The non-Abelian Gauss constraint on an open interval}

Consider (1+1)-dimensional non-Abelian Yang--Mills theory on the open interval
$
  I=(0,L)
$
in the temporal gauge $\mathscr{A}_0^A(x)=0$.  Let the canonical variables be
$
  \mathscr{A}_1^A(x),\ \mathscr{E}^A(x)
$
where \(A\) is an adjoint color index.  The Gauss constraint is given by
\begin{equation}
  G^A(x):=(\mathscr{D}_1\mathscr{E})^A(x)-\rho^A(x)\simeq 0 .
  \label{eq:nonabelian-gauss-constraint}
\end{equation}
Here \(\mathscr{D}_1\) is the covariant derivative in the adjoint representation, and \(\rho^A\) is the matter color-charge density.  The naive smeared Gauss functional is  given by
\begin{equation}
  G_0[\omega]
  :=\int_0^L dx\,\omega^A(x)\left((\mathscr{D}_1\mathscr{E})^A(x)-\rho^A(x)\right),
  \label{eq:naive-gauss}
\end{equation}
where \(\omega^A(x)\) is a fixed smearing function.  Integrating by parts gives
\begin{equation}
  G_0[\omega]
  =-\int_0^L dx\,(\mathscr{D}_1\omega)^A \mathscr{E}^A
   -\int_0^L dx\,\omega^A\rho^A
   +\left[\omega^A\mathscr{E}^A\right]_0^L .
  \label{eq:partial-integration}
\end{equation}
The boundary term is
\begin{equation}
  \left[\omega^A\mathscr{E}^A\right]_0^L
  =\omega^A(L)\mathscr{E}^A(L)-\omega^A(0)\mathscr{E}^A(0).
  \label{eq:boundary-term-identity}
\end{equation}
Equation \eqref{eq:partial-integration} is only an identity obtained by partial integration.  It is not yet the construction of the correct canonical generator.  The latter requires functional differentiability on the chosen phase space.

Let us see where the naive Gauss generator fails.  Since the smearing function is fixed, \(\delta\omega^A=0\).  The part of the variation of \(G_0[\omega]\) with respect to \(\mathscr{E}\) that contains a derivative is
\begin{equation}
  \delta_\mathscr{E} G_0[\omega]
  =\int_0^L dx\,\omega^A(\mathscr{D}_1\delta \mathscr{E})^A .
\end{equation}
Integrating by parts gives
\begin{equation}
  \delta_\mathscr{E} G_0[\omega]
  =\left[\omega^A\delta \mathscr{E}^A\right]_0^L
   -\int_0^L dx\,(\mathscr{D}_1\omega)^A\delta \mathscr{E}^A .
  \label{eq:variation-naive}
\end{equation}
The second term gives a legitimate bulk functional derivative.  The first term,
\begin{equation}
  \left[\omega^A\delta \mathscr{E}^A\right]_0^L
  =\omega^A(L)\delta \mathscr{E}^A(L)-\omega^A(0)\delta \mathscr{E}^A(0),
  \label{eq:bad-boundary-variation}
\end{equation}
is the obstruction.  If \(\omega(0)=0=\omega(L)\), this term vanishes.  In that case \(G_0[\omega]\) is differentiable and generates a pure gauge redundancy.
If, however, \(\omega\) is allowed to be nonzero at the boundary and the boundary electric fluxes \(\mathscr{E}(0),\mathscr{E}(L)\) are kept as variable boundary data, then \(\delta \mathscr{E}(0)\) and \(\delta \mathscr{E}(L)\) are not zero in general.  Hence \(G_0[\omega]\) is not differentiable on this phase space.

The boundary variation in \eqref{eq:bad-boundary-variation} is canceled by adding
\begin{equation}
  B[\omega]
  :=-\left[\omega^A\mathscr{E}^A\right]_0^L
  =-\omega^A(L)\mathscr{E}^A(L)+\omega^A(0)\mathscr{E}^A(0).
  \label{eq:improvement-term}
\end{equation}
Since \(\omega\) is fixed, the variation of $B[\omega]$ with respect to $\mathscr{E}$ reads
\begin{equation}
  \delta_{\mathscr{E}} B[\omega]
  =-\left[\omega^A\delta \mathscr{E}^A\right]_0^L .
\end{equation}
Therefore the improved functional
$
  \widetilde G[\omega]
  :=G_0[\omega]+B[\omega]
$
is differentiable:

\begin{equation}
\boxed{\begin{aligned}
  \widetilde G[\omega]
  &:=G_0[\omega]+B[\omega]\\
  &=\int_0^L dx\,\omega^A
  \left((\mathscr{D}_1\mathscr{E})^A-\rho^A\right)\\
  &\quad+\omega^A(0)\mathscr{E}^A(0)
  -\omega^A(L)\mathscr{E}^A(L)
\end{aligned}}
  \label{eq:improved-gauss-interval}
\end{equation}

This is the differentiably improved Gauss generator on the open interval.  It is the expression denoted by Eq.~\eqref{ym2-smeared-gauss2}
 in the main text.
More explicitly, after the boundary term is included, the variation contains no leftover boundary piece:
\begin{align}
  \delta\widetilde G[\omega]
  =&-\int_0^L dx\,(\mathscr{D}_1\omega)^A\delta\mathscr{E}^A
  \nonumber\\
  &+g_{\rm YM}\int_0^L dx\,\omega^Af^{ABC}
  \delta\mathscr{A}_1^B\mathscr{E}^C
  -\int_0^L dx\,\omega^A\delta\rho^A .
  \label{eq:variation-improved}
\end{align}
The sign of the structure-constant term depends on the convention for the covariant derivative.  The important point is convention-independent:
\textbf{the improved generator has ordinary bulk functional derivatives and hence defines a Hamiltonian vector field on the chosen phase space}.

The addition of \eqref{eq:improvement-term} is not an arbitrary modification of the constraint.
It is required once the phase space is chosen so that boundary electric fluxes are allowed to vary and gauge parameters nonzero at the boundary are admitted.
Up to constants that are not varied, or up to changes corresponding to different boundary conditions, the term is fixed by differentiability.

On the constraint surface define the outward electric flux by
\begin{equation}
 Q_{\partial I}^{\mathrm{out}}[\omega]
 :=\omega^A(L)\mathscr E^A(L)-\omega^A(0)\mathscr E^A(0).
 \label{eq:boundary-charge}
\end{equation}
Then the differentiably improved generator obeys
\begin{equation}
 \widetilde G[\omega]\simeq- Q_{\partial I}^{\mathrm{out}}[\omega].
\end{equation}
The coefficients $Q_L^A:=\mathscr E^A(0)$ and $Q_R^A:=-\mathscr E^A(L)$ used below are endpoint \emph{generator} charges; together they generate $-Q_{\partial I}^{\mathrm{out}}$. This keeps the outward-flux convention of Sec.~\ref{subsec:relational-boundary-charge} explicit and prevents the same symbol $Q_\partial$ from being used for opposite signs.

Transformations satisfying $\omega(0)=\omega(L)=0$ are proper gauge
redundancies for the boundary prescription used here.  Boundary-nontrivial
transformations are not automatically gauge redundancies or physical
symmetries; when the chosen boundary conditions retain them as genuine
boundary symmetries, their differentiable generator obeys, on the constraint
surface,
\begin{equation*}
 \widetilde G[\omega]\simeq-Q_{\partial I}^{\rm out}[\omega]
 =\omega^A(0)Q_L^A+\omega^A(L)Q_R^A .
\end{equation*}
Thus the physical surface generator is the endpoint-generator combination
above, while $Q_{\partial I}^{\rm out}$ denotes the outward electric flux
with the opposite sign.  The surface term is fixed by functional
differentiability for the chosen boundary phase space.

In the non-Abelian theory, a further QRF point must be kept in mind.  The color component of a boundary electric charge is not an invariantly meaningful number unless a boundary color frame is specified.  Moreover, the left and right boundary charges act on an open Wilson line by noncommuting Lie-algebra actions from the two endpoints.  This is \textbf{the boundary version of the statement that color is meaningful only relationally.}

\subsection*{The Abelian gauge--Higgs analogue}

The same construction applies to the Abelian \(U(1)\) gauge--Higgs model on the interval.  Let the Abelian Gauss constraint be
\begin{equation}
  \partial_xE(x)-\rho(x)\simeq 0 .
\end{equation}
The naive smeared functional is
\begin{equation}
  G_0[\omega]
  =\int_0^L dx\,\omega(x)\left(\partial_xE(x)-\rho(x)\right).
\end{equation}
Integration by parts gives
\begin{equation}
  G_0[\omega]
  =-\int_0^L dx\, (\partial_x\omega)E
   -\int_0^L dx\,\omega\rho
   +[\omega E]_0^L .
\end{equation}
Its variation contains the boundary term
$
  [\omega\delta E]_0^L
$.
Hence one adds
\begin{equation}
  B[\omega] = -[\omega E]_0^L=-\omega(L)E(L)+\omega(0)E(0),
\end{equation}
and obtains the \textbf{improved generator}
\begin{equation}
  \boxed{
  \widetilde G[\omega]
  =\int_0^L dx\,\omega(\partial_xE-\rho)
   +\omega(0)E(0)-\omega(L)E(L)
  } .
  \label{eq:abelian-improved}
\end{equation}
This is the Abelian version of \eqref{eq:improved-gauss-interval},
whose main-text counterpart is Eq.~\eqref{improved generator}.
 On the constraint surface,
\begin{equation}
  \widetilde G[\omega]
  \simeq B[\omega]=\omega(0)E(0)-\omega(L)E(L).
\end{equation}
The status of \eqref{eq:abelian-improved} is exactly the same as that of \eqref{eq:improved-gauss-interval}.  Both are newly defined improved functionals obtained by imposing differentiability, not merely identities obtained by partial integration.  The difference is that in the Abelian theory the covariant derivative reduces to \(\partial_x\), and the boundary charges are simply the electric fluxes \(E(0)\) and \(-E(L)\).  In the gauge--Higgs model the bulk charge density \(\rho\) appears explicitly, and dynamical charged matter can screen electric flux.

\subsection*{What is a choice, and what is fixed}

In a gauge theory with boundary, the following three pieces of data must be specified together:
\begin{enumerate}
  \item the boundary conditions imposed on the allowed fields;
  \item the boundary behavior of the allowed gauge parameters;
  \item which boundary degrees of freedom are kept as physical or edge-mode data.
\end{enumerate}
Once this phase space has been chosen, the boundary term that makes the generator differentiable is fixed.  If one allows only \(\omega|_{\partial V}=0\), no improvement term is needed.  If one fixes the boundary electric flux and imposes \(\delta \mathscr{E}|_{\partial V}=0\), the boundary variation also disappears.  If one allows \(\omega\) to be nonzero $\omega|_{\partial V}\not=0$ at the boundary and treats boundary electric flux as a variable boundary datum, then the improvement term in \eqref{eq:improved-gauss-interval} is necessary.
Thus Eq.~\eqref{eq:improved-gauss-interval} is not an arbitrary operation.

\subsection*{Bibliographic positioning}

Historically, the role of surface terms in non-Abelian gauge theory was treated directly by Gervais, Sakita and Wadia \cite{GervaisSakitaWadia1976}, and by Wadia \cite{Wadia1977}. Regge and Teitelboim \cite{ReggeTeitelboim1974} gave the classical general formulation of the principle that surface terms are required for differentiable Hamiltonian generators. The standard constrained-Hamiltonian treatment is given in the textbook by Henneaux and Teitelboim \cite{HenneauxTeitelboim1992}. In the covariant phase-space language, Lee and Wald \cite{LeeWald1990} formulated the relation between local symmetries, constraints and Hamiltonian generators.

More recent discussions of gauge theory with boundaries, edge modes and electric flux include Donnelly and Freidel\cite{Donnelly:2016auv},  Troessaert\cite{Troessaert2013}, and Henneaux--Troessaert\cite{HenneauxTroessaert2018}, Riello\cite{Riello2021}.
Riello's treatment is especially close to the present formula.  In Yang--Mills theory on a bounded region, the relevant moment map is not merely the bulk Gauss constraint, but the integration-by-parts expression including the boundary electric flux.  Equation \eqref{eq:improved-gauss-interval} is precisely its specialization to a one-dimensional open interval \(V=I=(0,L)\), with the outward-normal convention stated above.

\section{Relational bulk--boundary charge identity and its open-interval specialization}
\label{app:boundary-charge}

\subsection*{General relational bulk--boundary identity}

Let $V$ be a spatial region with boundary $\partial V$. The Yang--Mills Gauss law is
\begin{equation}
 \mathscr D_j[\mathscr A]\mathscr E^j-\rho\approx0.
\end{equation}
Introduce a relational frame $h(x)$ and
\begin{align}
 \mathscr A_j^h
 &:=h^{-1}\mathscr A_jh
 +\frac{i}{g_{\rm YM}}h^{-1}\partial_jh,\\
 \mathscr E_h^j&:=h^{-1}\mathscr E^jh,
 \qquad
 \rho_h:=h^{-1}\rho h .
\end{align}
Then
\begin{equation}
 \mathscr D_j[\mathscr A^h]\mathscr E_h^j-\rho_h
 =h^{-1}\bigl(\mathscr D_j[\mathscr A]\mathscr E^j-\rho\bigr)h
 \approx0.
 \label{rel-Gauss}
\end{equation}
For an adjoint-valued test direction $\xi(x)$ in the $h$ frame, define
\begin{align}
 Q_V^{\rm matter}[h,\xi]
 &:=\int_V d^{D-1}x\,2\operatorname{tr}(\xi\rho_h),\\
 Q_V^{\rm gauge}[h,\xi]
 &:=\int_V d^{D-1}x\,2\operatorname{tr}
 \left\{\bigl(\mathscr D_j[\mathscr A^h]\xi\bigr)\mathscr E_h^j\right\}.
\end{align}
Here $Q_V^{\rm gauge}$ is best read as the gauge-field/transport contribution required to compare non-Abelian color directions over the region; it is not an additional independent matter charge.

Using invariance of the trace,
\begin{align}
 &2\operatorname{tr}\!\left\{\xi\,\mathscr D_j[\mathscr A^h]\mathscr E_h^j\right\}
 +2\operatorname{tr}\!\left\{(\mathscr D_j[\mathscr A^h]\xi)\mathscr E_h^j\right\}
 \nonumber\\
 &\hspace{25mm}=\partial_j\,2\operatorname{tr}(\xi\mathscr E_h^j).
\end{align}
Combining this identity with Eq.~\eqref{rel-Gauss} gives
\begin{equation}
 \boxed{
 Q_V^{\rm matter}[h,\xi]+Q_V^{\rm gauge}[h,\xi]
 \approx Q_{\partial V}^{\rm out}[h,\xi]}
 \label{eq:appendix-D-bulk-boundary}
\end{equation}
with the outward-flux convention
\begin{equation}
 Q_{\partial V}^{\rm out}[h,\xi]
 :=\int_{\partial V}dS_j\,2\operatorname{tr}(\xi\mathscr E_h^j).
\end{equation}
Equation~\eqref{eq:appendix-D-bulk-boundary} is precisely the general relation used in Sec.~II.C.

For comparison, in a fixed color basis $h=1$ and with a spatially constant generator $\xi=T^A$,
\begin{equation}
 Q_V^{{\rm gauge},A}
 =-g_{\rm YM}\int_Vd^{D-1}x\,
 f^{ABC}\mathscr A_j^B\mathscr E^{jC},
\end{equation}
which is the standard non-Abelian gauge-field contribution to the integrated Gauss law. The relational formulation simply rewrites the same structure without assigning invariant meaning to an unreferenced fixed color component.

For the open-interval specialization below we return to the
physical-connection convention of Sec.~II.  The Wilson-line frame therefore
contains $g_{\rm YM}$ explicitly.  Under the Sec.~IV rescaling
$\mathscr A_{\rm geom}=g_{\rm YM}\mathscr A_{\rm phys}$, the same frame is
written with no explicit $g_{\rm YM}$ in the exponent; no physical change is
involved.

\subsection*{Specialization to the open interval}

Now take $V=I=(0,L)$ in pure $(1+1)$-dimensional Yang--Mills theory. Since $\rho=0$,
\begin{equation}
 Q_I^{\rm matter}[h,\xi]=0.
\end{equation}
Choose the Wilson-line frame rooted at the left endpoint,
\begin{equation}
 h(x)=\mathcal P\exp\!\left(
 ig_{\rm YM}\int_0^xdy\,\mathscr A_1(y)
 \right),
 \qquad h(0)=1.
\end{equation}
It obeys $\mathscr A_1^h=0$. The Gauss law therefore reduces to
\begin{equation}
 \partial_x\mathscr E_h=0,
 \qquad
 \mathscr E_h(x)=q\in\mathfrak{su}(2).
\end{equation}
For an arbitrary test direction $\xi(x)$,
\begin{align}
 Q_I^{\rm gauge}[h,\xi]
 &=\int_0^Ldx\,2\operatorname{tr}\{(\partial_x\xi)q\}\nonumber\\
 &=2\operatorname{tr}\{[\xi(L)-\xi(0)]q\},
 \label{eq:appendix-D-interval-gauge}
\end{align}
and the outward boundary flux is
\begin{equation}
 Q_{\partial I}^{\rm out}[h,\xi]
 =2\operatorname{tr}\{\xi(L)q\}
 -2\operatorname{tr}\{\xi(0)q\}.
\end{equation}
Hence
\begin{equation}
 \boxed{Q_I^{\rm gauge}[h,\xi]
 =Q_{\partial I}^{\rm out}[h,\xi]},
\end{equation}
which is the interval specialization of Eq.~\eqref{eq:appendix-D-bulk-boundary}.

The important point is that the result depends on which boundary observable the smearing $\xi$ is designed to read.
If the body-frame test direction is constant,
\begin{equation}
 \xi(x)=\xi_0,
\end{equation}
then $Q_I^{\rm gauge}[h,\xi_0]=0$ and, because there is no matter,
\begin{equation}
 Q_I^{\rm total}[h,\xi_0]=0.
\end{equation}
This statement concerns the combined left-plus-right boundary charge expressed in one relational frame. It does \emph{not} imply that the individual endpoint charges vanish.

To isolate the left endpoint, choose a smooth $f_L(x)$ with
\begin{equation}
 f_L(0)=1,
 \qquad f_L(L)=0,
 \qquad
 \xi_L(x)=f_L(x)T^A.
\end{equation}
Then
\begin{equation}
 Q_I^{\rm gauge}[h,\xi_L]
 =-q^A.
\end{equation}
Similarly, $f_R(0)=0$, $f_R(L)=1$ and $\xi_R=f_RT^A$ give
\begin{equation}
 Q_I^{\rm gauge}[h,\xi_R]
 =+q^A.
\end{equation}
Thus endpoint-selective smearings recover nonzero boundary charges even though the total charge measured by a constant test direction vanishes.

\subsection*{Boundary-generator sign and the physical content}

The outward flux is
\begin{equation}
 Q_{\partial I}^{\rm out}[\omega]
 =\omega(L)\mathscr E(L)-\omega(0)\mathscr E(0).
\end{equation}
The differentiably improved canonical generator discussed in Appendix~\ref{app:differentiable_generator} satisfies, on the constraint surface,
\begin{equation}
 \widetilde G[\omega]
 \approx-Q_{\partial I}^{\rm out}[\omega].
\end{equation}
Accordingly, the endpoint generator charges may be written as
\begin{equation}
 Q_L=\mathscr E(0),
 \qquad
 Q_R=-\mathscr E(L),
\end{equation}
without introducing a second definition of outward flux. The relative sign is fixed entirely by the outward-normal convention.

\section{Operational QRF Criterion: algebraic supplement}
\label{app:operational-qrf-revised}
This appendix is the technical supplement to the confinement-specific Sec.~III.B; it is not part of the general QRF definitions collected in Sec.~II.

The operational definition itself is given in Sec.~\ref{subsec:operational-qrf}.
This appendix
retains only two technical points that are useful when applying that
definition: the distinction between the dressed field algebra and its
frame-independent fixed-point algebra, and the distinction between a
finite-volume charge subspace and an infinite-volume charged
representation.

\subsection*{Dressed algebra, fixed-point algebra, and twirling}

For a fixed boundary condition $B$, a fixed boundary/asymptotic sector $s$,
and an admissible frame $h$, let $\mathfrak F_{B,s}(h)$ denote the algebra of
gauge-invariant dressed relational fields.  If $\mathcal K^{\rm tw}_{B,s}$
is the compact subgroup of right-frame reorientations whose orientation is
operationally inaccessible in the specified protocol, define the
corresponding fixed-point algebra by
\begin{equation}
 \mathfrak A^{\rm fp}_{B,s}(h)
 :=
 \left(\mathfrak F_{B,s}(h)\right)^{\mathcal K^{\rm tw}_{B,s}} .
\end{equation}
For a continuous action of $\mathcal K^{\rm tw}_{B,s}$ on the dressed
$C^*$-algebra, the normalized Haar average
\begin{equation}
 P_KX=\int_K dk\,\alpha_k(X)
\end{equation}
is the conditional expectation onto the fixed-point algebra.  Only the
inaccessible part of the right-frame action enters this fixed-point
condition; accessible boundary symmetries and prepared matter frames are
not removed by it.  The common quasilocal observable algebra
$\mathfrak A_B$ used in the thermodynamic-limit criterion is part of the
operational protocol and should not, in general, be identified with a
direct sum over dynamical topological-defect sectors.  Dynamical
topological-defect sectors are treated separately through their physical
path-integral or ensemble weights.

\subsection*{Single-operator projection versus invariant correlation}

Let \(O_R^a\) transform in a nontrivial irreducible representation \(R\) of
\(K\):
\begin{equation}
 \alpha_k(O_R^a)=R(k)^a{}_{b}O_R^b.
\end{equation}
Then we obtain
\begin{equation}
 P_KO_R^a=\int_K dk\,\alpha_k(O_R^a)=\int_K dk\,R(k)^a{}_{b}O_R^b =0,
 \
 R\neq\mathbf{1}.
\end{equation}
This does not imply that every correlation built from \(O_R\) vanishes.  In
particular,
\begin{equation}
 C_R(x-y)
 =
 \frac{1}{d_R}
 \sum_{a=1}^{d_R}
 \omega\left(
 O_R^a(x)O_R^{a\dagger}(y)
 \right)
\end{equation}
is $K$-invariant (and is a singlet when $K$ is the full frame group) and may remain nonzero.
  Schur orthogonality gives
\begin{equation}
 \int_K dk\,
 R(k)^a{}_{b}R(k)^{*c}{}_{d}
 =
 \frac{1}{d_R}\delta^{ac}\delta_{bd}.
\end{equation}
Hence removing a single open right-frame index is not equivalent to removing
the physical information in an invariant pair correlator.  The latter must
still be interpreted dynamically: it may represent a neutral pair, a
line-attached excitation, a boundary-anchored object, or a genuine charged
sector.

\subsection*{Finite-volume sectors and the infinite-volume representation}

Let $\mathcal H^{B}_{\mathrm{phys},[R]}(L)$ denote the finite-volume
physical Hilbert space for the sector $[R]$, or a physical subspace of a
common Hilbert space when such a realization is available, and define
\begin{align}
 E^{B}_{[R]}(L)
 &=
 \inf_{\substack{
 |\Psi\rangle\in\mathcal H^{B}_{\mathrm{phys},[R]}(L)\\
 \langle\Psi|\Psi\rangle=1}}
 \langle\Psi|H_L|\Psi\rangle,
 \\
 \Delta E^{B}_{[R]}(L)
 &=
 E^{B}_{[R]}(L)-E^{B}_{[\mathbf{1}]}(L).
\end{align}
A bounded sequence of $\Delta E^{B}_{[R]}(L)$ is only a necessary
energetic condition.  One must in addition obtain a local weak-* limit
(or, when required, a convergent subnet) of the corresponding state
functionals on the common quasilocal algebra $\mathfrak A_B$,
\begin{equation}
 \omega_{[R]}(A)
 =
 \lim_{n\to\infty}
 \omega_{[R],L_n}(A),
\end{equation}
construct the associated GNS representation
\begin{equation}
 \left(
 \pi_{[R]},
 \mathcal H_{[R]},
 |\Omega_{[R]}\rangle
 \right),
\end{equation}
and verify that the nontrivial charge information survives in the
accessible boundary or asymptotic charge algebra rather than escaping to
infinity. Thus

\begin{equation}
\begin{gathered}
 \text{bounded finite-volume isolation energy}\\
 \not\Longrightarrow
 \text{nontrivial charged GNS representation}
\end{gathered}
\nonumber
\end{equation}
without the additional thermodynamic-limit assumptions stated in Sec.~\ref{subsec:operational-qrf}.

Finally, an isolated one-particle pole is only an auxiliary diagnostic when
the infrared theory is gapped and admits a stable particle interpretation.
Its absence is not a universal definition of confinement, since a charged
sector can be infraparticle-like or can fail to couple to a chosen
interpolating operator.  The primary criterion remains the existence or
absence of the finite-energy, physically accessible, asymptotically isolated
charged representation.

\section{Exact fixed-frame correlator in \((1+1)\)-dimensional Yang--Mills theory}
\label{app:example_QRF}
This appendix follows the geometric-connection convention of Sec.~IV; hence the interval holonomy contains no explicit $g_{\rm YM}$ in its exponent, while the spectrum retains the factor $g_{\rm YM}^2$.

Section IV already derives the interval Hilbert space, the electric-flux
spectrum, the linear sector energy, and the Wilson-loop area law.  The states considered here are \textbf{boundary electric-flux sectors for fixed endpoint-frame data}; they are not dynamical colored bulk-particle sectors of pure Yang--Mills theory.
The only additional calculation needed here is \textbf{the exact correlator of an open relational operator} that is invariant under the proper gauge group for fixed endpoint frames while remaining covariant under physical endpoint-frame transformations, before the endpoint-frame projection is performed.

Let
\begin{equation}
 U_{0\leftarrow L}(t)
 :=
 \mathcal{P}\exp\left(
 i\int_L^0 dx\,\mathscr{A}_1(t,x)
 \right)\in G
\end{equation}
be the interval holonomy oriented from the right endpoint $L$ to the left endpoint $0$.  This is the same orientation as $U_{0L}=h(L)^{-1}$
used in Sec.~\ref{subsec:ym-open-interval}; the arrow is displayed here to remove any ambiguity.  On \(L^2(G,dU)\), the matrix elements
\(D^R_{mn}(U)\) are eigenfunctions of the group Laplacian,
\begin{equation}
\begin{aligned}
 H D^R_{mn}(U)&=E_R(L)D^R_{mn}(U),\\
 E_R(L)&=\frac{g_{\rm YM}^2L}{2}C_2(R).
\end{aligned}
\end{equation}
Define the \textbf{fixed-endpoint-frame relational operator}
\begin{equation}
 O^R_{mn}(t):=D^R_{mn}(U_{0\leftarrow L}(t)).
\end{equation}
It is invariant under gauge transformations that are trivial at the endpoints, but under endpoint-frame transformations it carries one left and one right frame index:
\begin{equation}
 O^R_{mn}
 \longmapsto
 D^R(g_L)_{ma}\,
 O^R_{ab}\,
 D^R(g_R^{-1})_{bn}.
\end{equation}
With the vacuum represented by the constant wave function on \(G\), the
Euclidean correlator is rewritten into
\begin{align}
 G^R_{mn;pq}(\tau;L)
 &:={}\langle0|O^R_{mn}(\tau)
 O^{R\dagger}_{pq}(0)|0\rangle
 \nonumber\\
 &=\int_G dU\,D^R_{mn}(U)
 e^{-H|\tau|}D^R_{pq}(U)^* .
\end{align}

Using
\begin{equation}
 e^{-H|\tau|}D^R_{pq}(U)^*
 =
 e^{-E_R(L)|\tau|}D^R_{pq}(U)^*
\end{equation}
and the Schur orthogonality,
\begin{equation}
 \int_G dU\,
 D^R_{mn}(U)D^S_{pq}(U)^*
 =
 \frac{1}{d_R}
 \delta_{RS}\delta_{mp}\delta_{nq},
\end{equation}
one obtains
\begin{equation}
 G^R_{mn;pq}(\tau;L)
 =
 \frac{1}{d_R}
 \delta_{mp}\delta_{nq}
 \exp\left[-E_R(L)|\tau|\right].
\end{equation}
Thus a relational non-singlet correlator is exactly nonzero at finite separation.  Confinement is therefore not the statement that such a relational object cannot be constructed.

For every nontrivial \(R\),
\begin{equation}
 E_R(L)
 =
 \frac{g_{\rm YM}^2L}{2}C_2(R)
 \longrightarrow\infty
 \qquad
 (L\to\infty).
\end{equation}
Therefore, \textbf{the finite-\(L\) spectral line  represents an electric-flux string between endpoint frames rather than a finite-mass isolated colored particle.}

If the operational protocol makes the two endpoint-frame orientations
independently inaccessible, the appropriate frame-independent projection is
the Haar twirl over $G_L\times G_R$,
\begin{equation}
 \int_G dg_L\int_G dg_R\,
 D^R_{mn}(g_LUg_R^{-1})
 =
 0,
 \qquad
 R\neq\mathbf{1}.
\end{equation}
More generally, only the subgroup of endpoint-frame reorientations that is
operationally inaccessible is to be twirled.  Thus the vanishing above is
the result for the full independent endpoint twirl, not a universal
statement about every possible frame-independent projection.  This
projection does not generate the linear energy; the dynamical confinement
statement remains the independent divergence of $E_R(L)$.

The exactly solvable model therefore separates the three ingredients needed
in the QRF formulation:
\begin{enumerate}
\item for fixed endpoint frames, an open relational operator exists that is invariant under the proper gauge group but covariant under endpoint-frame transformations;
\item its fixed-frame correlator is nonzero, but its excitation energy grows
      linearly with the endpoint separation;
\item when endpoint frames are inaccessible, the remaining open frame
      indices are removed by the frame-independent projection.
\end{enumerate}
This is the nontrivial exact check of the operational criterion; the
derivation of the spectrum and of the area law itself is already contained
in Sec.~\ref{sec:ym2-qrf} and need not be repeated here.

\end{document}